\documentclass[%
 aip,jmp,
 amsmath,amssymb,
 reprint,%
]{revtex4-1}
\usepackage{stmaryrd}
\usepackage{dsfont}
\usepackage{graphicx}
\usepackage{tikz}
\usetikzlibrary{decorations,shapes,arrows,matrix,positioning}

\usepackage{dcolumn}
\usepackage{bbm}
\usepackage{amsthm}
\usepackage[utf8]{inputenc}
\usepackage[T1]{fontenc}
\usepackage{mathptmx}
\usepackage{tocloft}
\usepackage{bbold}
\usetikzlibrary{backgrounds}
\usetikzlibrary{patterns}
\usetikzlibrary{arrows.meta}
\usetikzlibrary{trees,snakes}
\tikzstyle{bag} = [align=center]
\usetikzlibrary{decorations.pathmorphing}
\usepackage{xcolor}
\usepackage{fancyhdr}
\usepackage{amssymb}
\usepackage{enumitem}
\usepackage{xcolor}
\usepackage{slashed}
\usepackage{amsmath}
\usepackage[export]{adjustbox}
\usepackage{hyperref}
\usepackage{subcaption}

\newtheorem{theorem}{Theorem}
\newtheorem{proposition}{Proposition}
\newtheorem{corollary}{ Corollary}
\newtheorem{lemma}{Lemma}
\newtheorem{remark}{Remark}

\theoremstyle{definition}
\newtheorem{definition}{Definition}

\begin{document}

\preprint{AIP/123-QED}

\title[Position-Space Renormalization and Half-Space Truncations in $\phi^4_4$]
{Position-Space Renormalization and Half-Space Truncations in $\phi^4_4$}

\author{Majdouline Borji}
\email{majdouline.borji@kfupm.edu.sa}
\affiliation{King Fahd Univeristy of Petroleum and Minerals, Department of Mathematics, Saudi Arabia}%

\date{\today}

\begin{abstract}
In this paper, we study half-space observables in the massive Euclidean
$\phi^4_4$ theory. We prove that the renormalized correlators can be multiplied by half-space truncations without requiring
any additional renormalization. More precisely, products with smooth
approximations of the half-space indicator converge to well-defined
distributional limits, uniformly in the ultraviolet cutoff.
The proof relies on a position-space renormalization framework for the Wilson--Polchinski flow equation based on a hierarchy of power-counting spaces adapted to the singularity structure of the correlators. This yields
uniform power-counting bounds and Besov--Hölder regularity estimates for
the renormalized correlators. As a consequence, the
correlators converge as the ultraviolet cutoff tends to infinity, and
the convergence takes place in explicitly identified Besov spaces. 
\end{abstract}
\maketitle
\section{Introduction}

Renormalised correlators in perturbative Euclidean quantum field theory are canonically constructed as distributions in position space with controlled singularities along partial diagonals. For the Euclidean $\phi^4_4$ model, this structure is studied within several rigorous frameworks, including Wilsonian renormalisation group methods and flow equations. These approaches provide a perturbative analysis of the theory together with ultraviolet bounds and scaling estimates.

However, the functional-analytic nature of renormalised correlators remains  poorly understood. Existing formulations based on flow equations control correlators only through families of scale-dependent estimates, typically obtained after smearing against heat kernels or other regularising test functions \cite{KopperMuller2007,BorjiKopper2}. While sufficient for establishing perturbative renormalisability, such estimates do not identify the correlators themselves as elements of explicit spaces of distributions, nor do they provide a natural topology in which the ultraviolet renormalisation limit is taken.

The first objective of the present work is to develop an intrinsic position-space formulation of perturbative renormalisation. More precisely, we construct a functional-analytic framework stable under the Wilson--Polchinski flow equation hierarchy and formulate renormalisation as a statement of stability and convergence in suitable spaces of distributions. Within this framework, the connected amputated correlators of the Euclidean $\phi^4_4$ theory are elements of Besov--Hölder spaces of negative regularity.

A principal consequence of this framework is the construction of renormalised correlators together with convergence in the ultraviolet limit. In contrast to existing position-space approaches, where correlators are controlled only after smearing against scale-dependent test functions, our framework identifies both the regularity class of the correlators and the topology in which the renormalisation limit is taken. Renormalised correlators therefore emerge as limits in spaces of distributions rather than merely as objects defined through families of scale-dependent estimates.

More precisely, we establish uniform Besov--Hölder bounds on the connected amputated correlators and prove convergence as the ultraviolet cutoff $\Lambda_0$ tends to infinity. To the best of our knowledge, this provides the first position-space formulation of perturbative renormalisation by flow equations in which both the regularity and the convergence of the correlators are established intrinsically at the level of distributions.

A second fundamental aspect of the paper concerns the behaviour of renormalised correlators under \emph{singular localisation in configuration space}, namely multiplication by discontinuous functions. We study products of the form
$$
\prod_{i\in\sigma}\mathds{1}^+(z_i)
\mathcal L_{l,n}^{\Lambda,\Lambda_0},
$$
where $\sigma\subseteq\{2,\dots,n\}$ is arbitrary and $\mathcal L_{l,n}^{\Lambda,\Lambda_0}$ denotes the connected amputated correlators of the Euclidean $\phi^4_4$ theory. Since the correlators are distributions of negative regularity, these products are a priori ill-defined and their existence cannot be deduced from classical multiplication theorems.

Our second principal result establishes that these products admit a canonical distributional interpretation without requiring any additional renormalisation. Furthermore, we show that singular localisation induces a sharp and universal loss of Besov regularity. More precisely, localisation by a half-space indicator lowers the regularity. This phenomenon does not correspond to the appearance of new ultraviolet divergences or additional counterterms.

The existence of these truncated correlators reveals a previously unexplored interaction between renormalisation, geometric localisation, and distributional regularity. It shows that singular localisation may be performed directly on renormalised correlators while preserving the renormalised structure of the underlying quantum field theory.

The present work is motivated by our study of the Euclidean $\phi^4_4$ theory on the half-space $\mathbb{R}_+\times\mathbb{R}^3$. A central difficulty in the analysis of such models is understanding how renormalised bulk correlators behave under localisation to a half-space. In \cite{BorjiKopper2}, we established perturbative renormalisability of the model on $\mathbb R_+\times\mathbb R^3$ in the sense that only finitely many counterterms are required to make the theory finite. Subsequently, in \cite{Borji3}, we initiated the study of multiplicative renormalisation for the half-space model. The latter requires proving that the effective action retains the same structural form as the tree-level interaction, implying in particular that the counterterms are position independent and correspond precisely to mass, coupling, wavefunction, and boundary renormalisation terms.

A central ingredient in this programme is understanding the origin of boundary singularities. The analysis of \cite{Borji3} identifies boundary counterterms with the reflective contribution of the half-space propagator, namely the part of the following form
$$
e^{-\sqrt{k^2+m^2}|z+z'|},
$$
where $k\in\mathbb R^3$ denotes the Fourier variable in the translationally invariant directions and $z,z'\in\mathbb R_+$ are the half-space coordinates.

By contrast, the contribution
$$
e^{-\sqrt{k^2+m^2}|z-z'|}
$$
is naturally associated with the bulk part of the theory. At first sight, one might therefore expect that it generates only the usual bulk counterterms. However, this intuition turns out to be incomplete. Indeed, the corresponding model describes the Euclidean $\phi^4_4$ theory with interaction supported on a half-space, and in separate work we show that boundary counterterms also arise in this setting despite the translation invariance of the propagator. These singularities originate from the discontinuous nature of the interaction at the boundary rather than from the propagator itself.

The truncated correlators studied in the present work provide the natural reference objects for analysing this phenomenon. They appear as comparison distributions in the study of interactions supported on a half-space and therefore constitute an essential ingredient for a rigorous treatment of boundary renormalisation.

It is important to emphasise that, despite the presence of a distinguished hyperplane, the objects constructed in this paper do not define a boundary or interface quantum field theory. The localisation is performed \emph{after} renormalisation and does not modify the action, the propagator, or the renormalisation conditions of the bulk theory. The resulting distributions therefore retain the renormalised structure of the underlying Euclidean $\phi^4_4$ model.

Beyond their role in boundary renormalisation, the truncated correlators admit a natural conceptual interpretation. Although the construction is carried out entirely in Euclidean space, localisation by half-space indicators introduces a geometric ordering structure absent from the full correlators.

In Lorentzian quantum field theory, ordering information is encoded through causality, time-ordering, and causal factorisation. Euclidean theories possess no intrinsic notion of time. Nevertheless, localisation to a half-space selects contributions supported on one side of a distinguished hyperplane while suppressing the complementary region. In this sense, the truncated correlators encode a form of geometric ordering analogous to the ordering structures that underlie causal perturbation theory.

The resulting construction is entirely intrinsic to Euclidean distribution theory and does not rely on analytic continuation or Wick rotation. Rather, it reveals a direct connection between localisation, regularity, and renormalisation within the Euclidean framework itself.

Our work lies at the intersection of several directions in mathematical quantum field theory while addressing, to the best of our knowledge, a problem that has not previously been studied in this form.

\paragraph{Renormalisation group and flow equations.}
Renormalisation group methods provide one of the principal approaches to perturbative quantum field theory. In particular, the Wilson--Polchinski flow equation \cite{Polchinski1984,KellerKopperSalmhofer1992,Kopper2007Flow} yields an inductive construction of perturbative renormalised Euclidean correlators together with precise quantitative control of their scaling behaviour. Mathematical implementations of this programme have led to rigorous renormalisability proofs and sharp momentum- and position-space bounds for Euclidean field theories \cite{KellerKopperSalmhofer1992,KopperMuller2007,KopperMeunier2012,KellerKopper1991QED,KellerKopper1992Composite,EfremovGuidaKopper2017YM,KopperMuller2000SBYM}. The present work builds upon this framework. Our objective is not to establish renormalisability itself, but rather to formulate the resulting correlators intrinsically as distributions of prescribed Besov regularity and to analyse their behaviour under singular localisation.

\paragraph{Causal perturbation theory.}
Ordering and localisation play a fundamental role in Epstein--Glaser renormalisation and its extensions \cite{EpsteinGlaser1973,BrunettiFredenhagen2000,HollandsWald2001}. There, causal factorisation provides the organising principle for the perturbative construction of Lorentzian quantum field theories. These approaches are, however, intrinsically Lorentzian and do not address localisation phenomena of Euclidean correlators, nor the regularity issues arising from multiplication by discontinuous functions.

\paragraph{Distribution theory, singular products, and regularity structures.}
The multiplication of distributions by non-smooth functions is a classical problem in analysis, closely related to the structure of singularities and wavefront interactions \cite{HormanderI}. In recent years, such questions have played a central role in the development of renormalisation methods for singular stochastic PDEs. Notable examples include Hairer's theory of regularity structures \cite{Hairer2014,HairerPhi43}, the paracontrolled calculus of Gubinelli, Imkeller, and Perkowski \cite{GubinelliImkellerPerkowski2015}, and the large-scale regularity and renormalisation methods developed by Otto and collaborators \cite{OttoWeber2019,FischerOtto2015}. These theories provide frameworks in which singular products can be interpreted after suitable renormalisation procedures.

Although the present work is deterministic and situated entirely within perturbative quantum field theory, it addresses a related question from a different perspective. Rather than renormalising singular products of distributions, we investigate the effect of singular localisation on already renormalised Euclidean correlators. Our results show that multiplication by half-space indicators can be performed directly on renormalised correlators without introducing additional counterterms. Furthermore, the resulting loss of regularity admits a precise quantitative description in Besov--Hölder spaces.

Taken together, these results establish a bridge between perturbative renormalisation, the regularity theory of distributions, and singular localisation in position space. They provide a functional-analytic framework in which renormalised correlators can be studied directly as distributions and in which localisation phenomena can be analysed intrinsically at the level of the correlators themselves.

The paper is organised as follows.

Section~\ref{sec1} recalls the flow equation framework, introduces the basic notation and definitions, and states the principal results of the paper. 
Section~\ref{SecTreeAmp} develops the combinatorial structures underlying our position-space analysis. We introduce a class of generalized tree amplitudes that extends the framework of \cite{KopperMuller2007} by allowing insertions of smooth localisation functions, including the approximations of the half-space indicator that play a central role in the present work. We further develop a calculus for these amplitudes and analyse their behaviour under the flow equation hierarchy.

Our construction  differs from that of \cite{KopperMuller2007}. Besides the localisation insertions, it also adresses a structural incompatibility between the tree representation of \cite{KopperMuller2007} and the natural inductive ordering generated by the flow equations. The resulting class of amplitudes is stable under the flow equation hierarchy and provides the combinatorial foundation for the analysis carried out in the remainder of the paper.

The results of this section supply the main combinatorial input for the subsequent position-space renormalisation procedure. While they are used throughout the paper, their proofs are largely self-contained and may be postponed on a first reading.

Section~\ref{SecProof} constitutes the analytic core of the paper. There, we introduce a family of \emph{power-counting spaces} together with their associated seminorms and develop a systematic calculus for distributions in these spaces. The purpose of this framework is to formulate renormalisability directly at the level of distributions and to obtain uniform control of the correlators under the flow equation. Applying this framework to the Euclidean $\phi^4_4$ theory, we establish Besov--Hölder regularity estimates for the correlators.
 We also construct the truncated correlators as limits of
$$
\prod_{i\in\sigma}
\chi_\varepsilon(z_i)
\mathcal L_{l,n}^{\Lambda,\Lambda_0}
\bigl(
(z_1,p_1),\dots,(z_n,p_n)
\bigr)
$$
 when $\varepsilon\to0$, where $\chi_\varepsilon$ denotes a smooth approximation of the half-space indicator function with nonnegative derivative and bounds uniform in $\varepsilon$. We then establish the existence, regularity, and convergence properties of the resulting truncated correlators and prove the associated regularity-loss phenomenon. These results provide the analytic foundation for the study of renormalisation in models with interactions supported on a half-space, which will be analysed elsewhere.

Finally, we establish in Section \ref{ConvSection} the ultraviolet convergence of both the full and truncated correlators as \(\Lambda_0\to\infty\) in Besov--Hölder spaces of suitable regularity. These results provide the mathematical foundation for the study of quantum field theories with interactions supported on a half-space, which will be investigated elsewhere.

\section{The $\phi^4_4$ model and the flow equation formalism}\label{sec1}

In this section we recall the Wilson-Polchinski flow equation formalism for the massive Euclidean $\phi^4_4$ theory and introduce the notation used throughout the paper. Since the localisation procedures studied later act only on one distinguished coordinate, it is convenient to work in a mixed position--momentum representation. More precisely, we keep one coordinate in position space and perform a Fourier transform in the remaining translationally invariant directions. This representation contains all the information of the full position-space theory while being particularly well adapted to the analysis of half-space localisation.

\subsection{Mixed position--momentum representation and regularised propagator}
We consider a massive scalar field on $\mathbb{R}^4$ with coordinates
$(z,x)\in\mathbb{R}\times\mathbb{R}^3$. We perform a Fourier transform with
respect to the transverse variables $x\in\mathbb{R}^3$ while keeping the
coordinate $z$ in position space. This leads to a mixed position--momentum
representation, which we denote by $(z,p)$.
The regularized flowing propagator is defined by
\begin{equation}\label{prop}
C^{\Lambda,\Lambda_0}(p;z,z')
=
\int_{\frac{1}{\Lambda_0^{2}}}^{\frac{1}{\Lambda^{2}}}
d\lambda\,
e^{-\lambda(p^2+m^2)}\,
K(\lambda;z,z') ,
\end{equation}
where $K$ denotes the one--dimensional Euclidean heat kernel
\begin{equation}
K(\lambda;z,z')
=
\frac{1}{\sqrt{2\pi\lambda}}
\exp\!\left(-\frac{(z-z')^2}{2\lambda}\right).
\end{equation}

The parameter $\Lambda_0$ plays the role of an ultraviolet cutoff, while
$\Lambda$ is the flow parameter. The full propagator is recovered in the limit
$\Lambda_0\to\infty$ and $\Lambda\to0$. Differentiation with respect to $\Lambda$ yields
\begin{equation}
\dot C^\Lambda(p;z,z')
=
\dot C_m^\Lambda(p)\,
K(\Lambda^{-2};z,z'),
\qquad
\dot C^\Lambda_m(p)
=
-\frac{2}{\Lambda^3}
\exp\!\left(-\frac{p^2+m^2}{\Lambda^2}\right).
\end{equation}
The factorisation of $\dot C^\Lambda$ into a momentum-dependent term and a heat kernel will play a central role throughout the paper. 

\subsection{Effective action and connected amputated correlators}

For finite $\Lambda_0$ the effective action is defined through the functional
integral
\begin{equation}\label{eg09}
\left\{
\begin{aligned}
\exp\!\left(
-\frac{1}{\hbar}\bigl(
L^{\Lambda,\Lambda_0}(\phi)
+
I^{\Lambda,\Lambda_0}
\bigr)
\right)
&:=
\int d\mu_{\Lambda,\Lambda_0}(\Phi)\,
\exp\!\left(
-\frac{1}{\hbar}
L^{\Lambda_0,\Lambda_0}(\Phi+\phi)
\right),\\
L^{\Lambda,\Lambda_0}(0)&=0 .
\end{aligned}
\right.
\end{equation}

Here $d\mu_{\Lambda,\Lambda_0}$ denotes the centered Gaussian measure with
covariance $\hbar C^{\Lambda,\Lambda_0}$.

The bare interaction is assumed to be local and of bounded canonical
dimension. For the $\phi^4_4$ model the most general local interaction
compatible with these assumptions takes the form
\begin{multline}\label{bareL}
L^{\Lambda_0,\Lambda_0}(\phi)
=
\frac{\lambda}{4!}
\int_{\mathbb{R}^4}
\phi^4(z,x)\,dz\,d^3x
+
\frac12
\int_{\mathbb{R}^4}
\Big(
a^{\Lambda_0}(z,x)\,\phi^2(z,x)
\\
-b_0^{\Lambda_0}(z,x)\,\phi(z,x)\,\partial_z^2\phi(z,x)
-\sum_{i=1}^3
b_i^{\Lambda_0}(z,x)\,\phi(z,x)\,\partial_{x_i}^2\phi(z,x)
\\
+s_0^{\Lambda_0}(z,x)\,\phi(z,x)\,\partial_z\phi(z,x)
+\sum_{i=1}^3
s_i^{\Lambda_0}(z,x)\,\phi(z,x)\,\partial_{x_i}\phi(z,x)
\\
+\frac{2}{4!}\,
c^{\Lambda_0}(z,x)\,\phi^4(z,x)
\Big)
dz\,d^3x .
\end{multline}

The coefficient functions $a^{\Lambda_0}$, $b_i^{\Lambda_0}$,
$s_i^{\Lambda_0}$ and $c^{\Lambda_0}$ represent the counterterms of the
theory. Their precise form depends on the symmetries of the model and on
the renormalization conditions imposed later.

In principle, the bare interaction is assumed to be local and compatible with power counting.
In particular it contains all local monomials in $\phi$ and its derivatives of
canonical dimension less than or equal to four compatible with the symmetry
$\phi\mapsto-\phi$. A priori, the precise form of the counterterms is not fixed
before renormalization. The only requirement is that
$L^{\Lambda_0,\Lambda_0}$ contains no operators of canonical dimension
greater than four.\\
\indent A convenient choice of renormalization conditions leads to a minimal
form of the effective action, corresponding to multiplicative
renormalization. In this case the counterterms take the same structure as
the tree-level interaction and the bare action reduces to
\begin{equation}
L^{\Lambda_0,\Lambda_0}(\phi)
=
\frac{\lambda}{4!}
\int_{\mathbb{R}^4}
\phi^4(z,x)\,dz\,d^3x
+
\frac12
\int_{\mathbb{R}^4}
\Big(
a^{\Lambda_0}\,\phi^2
-
b^{\Lambda_0}\,\phi\,\Delta\phi
+
\frac{2}{4!}\,c^{\Lambda_0}\,\phi^4
\Big)
dz\,d^3x ,
\end{equation}
where $a^{\Lambda_0}$, $b^{\Lambda_0}$ and $c^{\Lambda_0}$ are
position-independent counterterms. Expanding the effective action as a formal power series in $\hbar$,
\begin{equation}\label{loop}
L^{\Lambda,\Lambda_0}(\phi)
=
\sum_{l=0}^{\infty}
\hbar^l
L^{\Lambda,\Lambda_0}_l(\phi),
\end{equation}
we introduce the connected amputated Schwinger functions (CAS), also
referred to as the $n$-point correlators. They are defined by functional
differentiation of the effective action:
\begin{equation}\label{CASL}
\mathcal L^{\Lambda,\Lambda_0}_{l,n}
((z_1,x_1),\ldots,(z_n,x_n))
:=
\left.
\frac{\delta^n}{\delta\phi(z_1,x_1)\cdots\delta\phi(z_n,x_n)}
L^{\Lambda,\Lambda_0}_l(\phi)
\right|_{\phi=0}.
\end{equation}

Since translation invariance holds in the transverse directions, it is
convenient to work in the mixed position–momentum representation obtained
by Fourier transforming the transverse variables. In this representation
the CAS functions are defined by
\begin{multline}\label{eg015}
\delta^{(3)}(p_1+\cdots+p_n)\,
\mathcal L^{\Lambda,\Lambda_0}_{l,n}
((z_1,p_1),\ldots,(z_n,p_n))
\\
:=
(2\pi)^{3(n-1)}
\left.
\frac{\delta^n}{
\delta\phi(z_1,p_1)\cdots\delta\phi(z_n,p_n)}
L^{\Lambda,\Lambda_0}_l(\phi)
\right|_{\phi\equiv0}.
\end{multline}

The Dirac distribution expresses conservation of the transverse momentum.
\subsection{Flow equations, Polchinski operators and Boundary conditions}

Differentiating \eqref{eg09} with respect to the
flow parameter $\Lambda$ and using the standard identities satisfied by
Gaussian measures with covariance depending on a parameter
(see e.g.\ \cite{Glimm}) yields the flow equation
\begin{equation}\label{flowEq}
\partial_\Lambda
(L^{\Lambda,\Lambda_0}+I^{\Lambda,\Lambda_0})
=
\frac{\hbar}{2}
\Big\langle
\frac{\delta}{\delta\phi},
\dot C^\Lambda
\frac{\delta}{\delta\phi}
\Big\rangle
L^{\Lambda,\Lambda_0}
-
\frac12
\Big\langle
\frac{\delta L^{\Lambda,\Lambda_0}}{\delta\phi},
\dot C^\Lambda
\frac{\delta L^{\Lambda,\Lambda_0}}{\delta\phi}
\Big\rangle .
\end{equation}

Here $\dot C^\Lambda$ denotes the derivative of the covariance with
respect to $\Lambda$, and $\langle\cdot,\cdot\rangle$ stands for the
natural $L^2$ pairing with kernel $\dot C^\Lambda$. Taking functional derivatives of \eqref{flowEq} and expanding in the loop-order \eqref{loop} yields the hierarchy of flow equations satisfied by the CAS
functions:
\begin{multline}\label{FEL4}
\partial_\Lambda
\mathcal L_{l,n}^{\Lambda,\Lambda_0}
(\vec z_{1,n},\underline{p}_{n})
=
\frac12
\int_{\mathbb R}dz
\int_{\mathbb R}dz'
\int_k
\dot C^\Lambda(k;z,z')
\mathcal L_{l-1,n+2}^{\Lambda,\Lambda_0}
(\vec z_{1,n},\underline{p}_n,(z,k),(z',-k))
\\
-
\frac12
\sum_{l_1+l_2=l}
\sum_{\pi\in\mathcal P_n}
\int_{\mathbb R}dz
\int_{\mathbb R}dz'~
\dot C^\Lambda(p_{\pi};z,z')
\\
\times
\mathcal L_{l_1,n_1+1}^{\Lambda,\Lambda_0}
((z_i,p_i)_{i\in\pi_1},(z,p_{\pi}))
\,
\mathcal L_{l_2,n_2+1}^{\Lambda,\Lambda_0}
((z',-p_{\pi}),(z_i,p_i)_{i\in\pi_2}) .
\end{multline}
Here we use the notation
\[
p_\pi = \sum_{i\in\pi_1}p_i
      = -\sum_{i\in\pi_2}p_i,
\qquad
\int_k := \int_{\mathbb{R}^3}\frac{d^3k}{(2\pi)^3}
\]
and
$$
\underline p_n:=(p_1,\ldots,p_n)
$$
denotes an $n$-tuple of momenta satisfying the conservation law
$$
\sum_{i=1}^n p_i=0.
$$
Throughout the paper, the notation $\underline p_n$ will always be understood in this sense. Later, in the derivation of the renormalization bounds, we shall use the notation
$$
\|\underline{p}_n\|=\max_{1\leq i\leq n}\left|p_i\right|~.
$$
If $n=2$, we simply write $\underline{p}=(p,-p)$.
\\
\indent Furthermore, $\mathcal P_n$ denotes the set of partitions of
$\{1,\dots,n\}$ into two non-empty subsets.
For $\pi=\{\pi_1,\pi_2\}\in\mathcal P_n$ we set
\[
n_i := |\pi_i|.
\]
Equation \eqref{FEL4} expresses the hierarchical structure of the flow:
the derivative of an $n$-point function at loop order $l$ depends only on
correlators of strictly lower loop order or lower number of external
legs.

For later use it is convenient to rewrite the flow equation in a more
compact form by introducing the Polchinski operators. First we define the linear operator
\[
\mathbf L^\Lambda :
\mathcal S'(\mathbb R^{4n+8})
\to
\mathcal S'(\mathbb R^{4n})
\]
by
\begin{multline}\label{LinOp}
\mathbf L^\Lambda(T)\left(\vec z_{1,n};\vec{p}_{1,n}\right)
\\
=
\int_k\int_{\mathbb R}dz
\int_{\mathbb R}dz'~
T\left((z_1,p_1)\ldots,(z_n,p_n),(z,k),(z',-k)\right)~K\bigl(\Lambda^{-2};z,z'\bigr)\,\dot C^{\Lambda}_m(k)~.
\end{multline}
Similarly, for a partition $\pi=\{\pi_1,\pi_2\}$ of $\{1,\dots,n\}$ we
introduce the bilinear operator
\[
\mathbf B^\Lambda_\pi :
\mathcal S'(\mathbb R^{n_1+1})
\times
\mathcal S'(\mathbb R^{n_2+1})
\to
\mathcal S'(\mathbb R^n)
\]
defined by
\begin{multline}\label{QuOp}
\mathbf B_\pi^\Lambda(T_1,T_2)
(\vec{z}_{1,n};\vec{p}_{1,n})
:=
e^{-\frac{m^2}{2\Lambda^2}}
\int_{\mathbb R}dz
\int_{\mathbb R}dz'\,
T_1
\bigl(
\left(z_i,p_i\right)_{i\in\pi_1},\left(z,p_{\pi}\right)
\bigr)
\\
\times
T_2
\bigl(
\left(z_i,p_i\right)_{i\in\pi_2},\left(z',-p_{\pi}\right)
\bigr)
K\!\left(
\Lambda^{-2};z,z'
\right).
\end{multline}
In terms of these operators the flow equation
\eqref{FEL4} takes the following compact form
\begin{align}
\partial_\Lambda
\mathcal L_{l,n}^{\Lambda,\Lambda_0}
&=\frac12
\mathbf L^\Lambda
\big(
\mathcal L_{l-1,n+2}^{\Lambda,\Lambda_0}
\big)
-\frac12
\sum_{\pi\in\mathcal P_n}
\mathbf B^\Lambda_\pi
\big(
\mathcal L_{l_1,n_1+1}^{\Lambda,\Lambda_0},
\mathcal L_{l_2,n_2+1}^{\Lambda,\Lambda_0}
\big)
\,\dot C^\Lambda_{\frac{m}{\sqrt{2}}}(p_{\pi}).
\end{align}
This formulation will be convenient when analyzing the correlators later in the paper. It also makes explicit the two basic operations generated by the flow: the linear contraction of two additional variables and the quadratic fusion of two lower-order correlators along a heat kernel. \\
\indent For a multi-index
\[
w=(w_1,\cdots,w_n)\in\mathbb N_n^3,
\]
we set
\[
|w|:=\sum_{i=1}^nw_i,
\qquad
\partial^w
:=
\partial_{p_1}^{w_1}\cdots
\partial_{p_n}^{w_n}.
\]

To close the hierarchy \eqref{FEL4}, one must supplement the flow equations with suitable boundary conditions. Throughout this paper, we employ the standard mixed boundary conditions of the flow equation formalism. These consist of ultraviolet boundary conditions imposed at $\Lambda=\Lambda_0$ for irrelevant terms in the renormalization group sense and renormalization conditions imposed at $\Lambda=0$ for the relevant part of the theory.

More precisely, for all loop orders $l\ge1$ we impose the following conditions:

\begin{itemize}

\item[(i)] Ultraviolet scale: for all $n\ge2$, all multi-indices $w$ and $r\in\mathbb{N}_0$ such that $r,|w|\in\{0,\cdots,3\}$ and satisfying
$
n+|w|+r\ge5$
one has 
\begin{equation}\label{Bc1}
\left(z_1-z_i\right)^r\partial^w
\mathcal L_{l,n}^{\Lambda_0,\Lambda_0}
\bigl(
(z_1,p_1),\ldots,(z_n,p_n)
\bigr)
=
0,\qquad\forall i\in\left\{2,\cdots,n\right\}~.
\end{equation}

\item[(ii)] Renormalisation conditions: at $\Lambda=0$ we impose at zero momenta
\begin{equation}\label{renoc1}
\int_{\mathbb R}
(z_1-z_2)^r
\,
\partial^w
\mathcal L_{l,2}^{0,\Lambda_0}
\bigl(
(z_1,0),(z_2,0)
\bigr)
\,dz_2
=
0,
\qquad
r+|w|\le2,
\end{equation}
together with
\begin{equation}\label{renoc2}
\int_{\mathbb R^3}
\mathcal L_{l,4}^{0,\Lambda_0}
\bigl(
(z_1,0),
(z_2,0),
(z_3,0),
(z_4,0)
\bigr)
\,d\vec z_{2,4}
=
0.
\end{equation}

\end{itemize}
The role of these conditions is to uniquely determine the perturbative solution of the flow equation hierarchy. Their relation with power counting and renormalisability will become apparent in Section~\ref{SecProof}, where we introduce the functional-analytic framework used throughout the sequel.

\subsection{Main results}
\indent In momentum space, perturbative renormalisability is formulated through estimates of the form\cite{Polchinski1984,KellerKopperSalmhofer1992} 
\begin{equation}\label{EquMomentSp}
\left|\hat{\mathcal{L}}_{l,n}^{\Lambda,\Lambda_0}(\underline{ p}_{n})\right|
\le
\left(\Lambda+m\right)^{4-n}
\mathcal P\!\left(\log\frac{\Lambda+m}{m}\right)
\mathcal P\!\left(\frac{ \|\underline{p}_n\|}{\Lambda+m}\right)~.
\end{equation}
Such bounds simultaneously encode the power-counting behaviour of the theory, the logarithmic structure of the ultraviolet divergences, and the dependence on the external momenta. They provide a complete description of renormalisability within the momentum-space framework.

In position space, however, the correlators are distributions rather than functions, and estimates of the form \eqref{EquMomentSp} no longer admit a direct analogue. The main difficulty is therefore to identify a functional-analytic framework capable of encoding simultaneously the singular behaviour of the correlators, their scaling properties, and their stability under the flow equation hierarchy.

The principal objective of the present work is to construct half-space truncations of renormalised correlators and to show that this operation does not require additional renormalisation. Since the correlators themselves are distributions of negative regularity, this problem cannot be addressed within classical multiplication theory. The Besov regularity theory developed below provides the functional-analytic framework that makes this construction possible.

The first step is therefore to formulate perturbative renormalisation intrinsically in position space. To this end, we introduce a family of power-counting spaces that encode the scaling properties of the correlators directly at the level of distributions and are stable under the flow equation hierarchy. Within this framework, renormalisability can be formulated intrinsically in position space.

The Besov regularity estimates stated below arise as a consequence of the corresponding power-counting bounds. Define
\begin{equation}
K(\tau;x,y)
:=
\frac{1}{\sqrt{\tau}}
\exp\!\left(-\frac{(x-y)^2}{2\tau}\right).
\end{equation}
The regularity properties of the correlators will be measured in Besov--Hölder spaces\cite{BahouriCheminDanchin} through the norm
\begin{equation}\label{BesovCorr}
\|\mathcal{L}_{l,n}^{\Lambda,\Lambda_0}
(\underline p_n)\|_{\mathcal C^\alpha(\mathbb R^n)}
:=
\sup_{\tau\in(0,1]}
\tau^{-\frac{\alpha}{2}}
\left\|
K^{\otimes n}
\star
\mathcal L_{l,n}^{\Lambda,\Lambda_0}
(\underline p_n)
\right\|_\infty,
\end{equation}
where
\begin{equation}\label{BesovInf}
\left\|
K^{\otimes n}
\star
\mathcal L_{l,n}^{\Lambda,\Lambda_0}
(\underline p_n)
\right\|_\infty
:=
\sup_{\vec{y}_{1,n}\in\mathbb R^n}
\left|
\int_{\mathbb R^n}
\mathcal L_{l,n}^{\Lambda,\Lambda_0}
\bigl(
(z_1,p_1),\ldots,(z_n,p_n)
\bigr)
\prod_{i=1}^{n}
K(\tau;z_i,y_i)
\,d\vec z_{1,n}
\right|.
\end{equation}
The following theorem provides Besov regularity estimates that are uniform with respect to the ultraviolet cutoff $\Lambda_0$.

\begin{theorem}[Besov regularity of the correlators]\label{BesovNorm}
Let $0\leq \Lambda\leq \Lambda_0$ and set $\Lambda_m:=\Lambda+m$. Fix external momenta
$\underline p_n=(p_1,\ldots,p_n)$ satisfying momentum conservation and assume that the correlators satisfy the boundary conditions
\eqref{Bc1} and the renormalization conditions
\eqref{renoc1}--\eqref{renoc2}.

For $n\ge4$ and $l\ge1$ one has
\begin{align}
\left\|\mathcal{L}_{l,n}^{\Lambda,\Lambda_0}\left(\underline{p}_{n}\right)\right\|_{\mathcal{C}^{-2n+2}(\mathbb{R}^n)}
&\le
\Big(\sum_{i=0}^{n-1}\Lambda_m^{4-n-i}\Big)
\mathcal{P}_{l-1}\!\left(\log\frac{\Lambda_m}{m}\right)
\widetilde{\mathcal{P}}\!\left(\frac{\|\underline{p}_n\|}{\Lambda_m}\right).
\end{align}
For $n\ge4$ one further has
\begin{align}\label{BesovnTree1}
\|\mathcal{L}_{0,n}^{\Lambda,\Lambda_0}\left(\underline{p}_{n}\right)\|_{\mathcal{C}^{-n+1}(\mathbb{R}^n)}
\le
\sum_{i=0}^{n-1}\Lambda_m^{4-n-i}.
\end{align}
For the two–point functions and $l\ge2$,
\begin{align}
\|\mathcal{L}_{l,2}^{\Lambda,\Lambda_0}(\underline{p})\|_{\mathcal{C}^{-4}(\mathbb{R}^2)}
&\le
\Big(\sum_{i=0}^{3}\Lambda_m^{2-i}\Big)
\mathcal{P}_{l-1}\!\left(\log\frac{\Lambda_m}{m}\right)~\widetilde{\mathcal{P}}\!\left(\frac{|p|}{\Lambda_m}\right)~.
\end{align}
Finally,
\begin{equation}
\|\mathcal{L}_{1,2}^{\Lambda,\Lambda_0}\left(\underline{p}\right)\|_{\mathcal{C}^{-1}(\mathbb{R}^2)}
\le
\sum_{i=0}^{3}\Lambda_m^{2-i}~.
\label{Eqq74}
\end{equation}
The polynomials $\mathcal{P}_{l-1}$ and $\widetilde{\mathcal{P}}$ may vary from line to line. Their coefficients are non-negative and depend only on  $l$ and $n$, while the degree of
$\mathcal{P}_{l-1}$ is $l-1$.
\end{theorem}
\begin{corollary}[Existence of the renormalised correlators]
\label{CorUVLimit}
Fix external momenta
$\underline p_n=(p_1,\ldots,p_n)$ satisfying momentum conservation and let the $n$-point correlators $\mathcal{L}_{l,n}^{\Lambda,\Lambda_0}$ verify \eqref{Bc1} and \eqref{renoc1}-\eqref{renoc2}.  
Then for $n\ge4$, $l\ge1$, there exists a distribution $
\mathcal L_{l,n}(\underline p_n)$ in $
\mathcal C^{-2n+2}(\mathbb R^n)$
such that
\[
\mathcal L_{l,n}^{\Lambda,\Lambda_0}(\underline p_n)
\longrightarrow
\mathcal L_{l,n}(\underline p_n)
\]
as $\Lambda\to0$ and $\Lambda_0\to\infty$ in
$\mathcal C^{-2n+2}(\mathbb R^n)$.

For the tree-level sector, one has for every $n\ge4$
\[
\mathcal L_{0,n}^{\Lambda,\Lambda_0}(\underline p_n)
\longrightarrow
\mathcal L_{0,n}(\underline p_n)
\]
as $\Lambda\to0$ and $\Lambda_0\to\infty$ in
$\mathcal C^{-n+1}(\mathbb R^n)$.

For the two-point functions, there exist distributions $
\mathcal L_{l,2}(
\underline{p})$ in
$\mathcal C^{-4}(\mathbb R^2)$ with $l\ge2$ such that
\[
\mathcal L_{l,2}^{\Lambda,\Lambda_0}(\underline{p})
\longrightarrow
\mathcal L_{l,2}\left(\underline{p}\right)
\]
as $\Lambda\to0$ and $\Lambda_0\to\infty$ in
$\mathcal C^{-4}(\mathbb R^2)$.

Finally,
\[
\mathcal L_{1,2}^{\Lambda,\Lambda_0}\left(\underline{p}\right)
\longrightarrow
\mathcal L_{1,2}\left(\underline{p}\right)
\]
as $\Lambda\to0$ and $\Lambda_0\to\infty$ in
$\mathcal C^{-1}(\mathbb R^2)$.
\end{corollary}
\begin{remark}
    \noindent As a consistency check, observe that at $l=0$ and $n=4$ one has
\begin{equation}
\mathcal{L}_{0,4}^{\Lambda,\Lambda_0}\left(\vec{z}_{1,4};\underline{p}_4\right)
=
\lambda \prod_{i=2}^4 \delta(z_1-z_i),
\end{equation}
which belongs to $\mathcal{C}^{-3}(\mathbb{R}^4)$ and is therefore
consistent with \eqref{BesovnTree1}.
\end{remark}
\begin{remark}
The regularity exponents appearing in
\eqref{BesovnTree1} and in the first estimate of
Theorem~\ref{BesovNorm} differ.
At tree level, the correlators satisfy the stronger bound
\(
\mathcal L_{0,n}^{\Lambda,\Lambda_0}
\in \mathcal C^{-n+1},
\)
whereas loop corrections lead to the weaker regularity
\(
\mathcal L_{l,n}^{\Lambda,\Lambda_0}
\in \mathcal C^{-2n+2}
\)
for \(l\ge1\).
This loss of regularity reflects the additional short-distance
singularities generated by loop integrations.
\end{remark}
As follows from Theorem~\ref{BesovNorm}, the correlators are distributions whose regularity depends on the number of external arguments and on the loop order. Consequently, multiplication by the discontinuous functions \(\mathds 1^{+}\) and \(\mathds 1^{-}\) is, in general, not defined in distribution theory.

A natural question is therefore whether the renormalised correlators admit a meaningful restriction to a half-space and, more generally, whether products of the form
\[
\prod_{i\in\sigma}\mathds 1_{+}(z_i)\,
\mathcal L_{l,n}^{\Lambda,\Lambda_0}
\]
can be defined without introducing additional renormalisation.

To address this problem, we introduce a class
\(\mathcal K\) of smooth approximations of the half-space indicator.
The precise definition of \(\mathcal K\) will be given in
Definition~\ref{DefK}.

Given \(\chi_\varepsilon\in\mathcal K\) and
\(\sigma\subseteq\{2,\ldots,n\}\), we consider the regularised products
\[
\Bigl(
\prod_{i\in\sigma}
\chi_\varepsilon(z_i)
\Bigr)
\mathcal L_{l,n}^{\Lambda,\Lambda_0}
(\vec z_{1,n};\underline p_n).
\]

Our main result shows that these regularised correlators admit a limit as
\(\varepsilon\to0\), that this limit is independent of the choice of
\(\chi_\varepsilon\in\mathcal K\), and that no additional renormalisation is required. This allows us to define the truncated correlators by
\begin{equation}
\label{DefHalfSpaceCorr}
\bigl(
\mathds 1_{+}^{\sigma}
\mathcal L_{l,n}^{\Lambda,\Lambda_0}
\bigr)
(\vec z_{1,n};\underline p_n)
:=
\lim_{\varepsilon\to0}
\Bigl(
\prod_{i\in\sigma}
\chi_\varepsilon(z_i)
\Bigr)
\mathcal L_{l,n}^{\Lambda,\Lambda_0}
(\vec z_{1,n};\underline p_n).
\end{equation}

The limit is understood in a suitable Besov--Hölder space. The resulting truncated correlators satisfy the following regularity estimates.
\begin{theorem}[Half-space truncation]\label{BesovNorm2}
Under the same assumptions as in Theorem~\ref{BesovNorm}, and for
$n\ge4$ and $l\ge1$, one has 
\begin{align}
\left\|\mathds{1}^{\sigma}_{+}\mathcal{L}_{l,n}^{\Lambda,\Lambda_0}(\underline{p}_{n})\right\|_{\mathcal{C}^{-|\sigma|n+|\sigma|-n}_{\mathrm{loc}}(\mathbb{R}^n)}
&\le
\Big(\sum_{i=0}^{n-1}\Lambda_m^{4-n-i}\Big)
\mathcal{P}_{l-1}\!\left(\log\frac{\Lambda_m}{m}\right)
\mathcal{P}\!\left(\frac{\|\underline{p}_n\|}{\Lambda_m}\right).
\end{align}
For $n\ge4$ one further has
\begin{align}\label{BesovnTreeCutoff}
\left\|\mathds{1}^{\sigma}_+\mathcal{L}_{0,n}^{\Lambda,\Lambda_0}\left(\underline{p}_{n}\right)\right\|_{\mathcal{C}^{-n+1}_{\mathrm{loc}}\left(\mathbb{R}^n\right)}
\le
\sum_{i=0}^{n-1}\Lambda_m^{4-n-i}.
\end{align}
For the two–point functions and $l\ge2$,
\begin{align}
\left\|\mathds{1}^+\mathcal{L}_{l,2}^{\Lambda,\Lambda_0}(\underline{p})\right\|_{\mathcal{C}^{-4}(\mathbb{R}^2)}
&\le
\Big(\sum_{i=0}^{3}\Lambda_m^{2-i}\Big)
\mathcal{P}_{l-1}\!\left(\log\frac{\Lambda_m}{m}\right)
\mathcal{P}\!\left(\frac{|p|}{\Lambda_m}\right).
\end{align}
Finally,
\begin{align}
\left\|\mathds{1}^+\mathcal{L}_{1,2}^{\Lambda,\Lambda_0}(\underline{p})\right\|_{\mathcal{C}^{-1}(\mathbb{R}^2)}
&\le
\Big(\sum_{i=0}^{3}\Lambda_m^{2-i}\Big)
\mathcal{P}\!\left(\frac{|p|}{\Lambda_m}\right).
\end{align}
\end{theorem}
\noindent As a consequence, one obtains the existence of renormalised truncated correlators. 
\begin{corollary}[Existence of the truncated renormalised correlators]
\label{CorTruncatedUV}

Fix $n\ge4$, $l\ge1$, and external momenta
$\underline p_n=(p_1,\ldots,p_n)$ satisfying momentum conservation.
Then there exists a distribution
\[
\mathcal L^{+,\sigma}_{l,n}(\underline p_n)
\in
\mathcal C_{\mathrm{loc}}^{-|\sigma|n+|\sigma|-n}
(\mathbb R^n),
\]
such that
\[
 \mathds 1^{\sigma}_+
\mathcal L_{l,n}^{\Lambda,\Lambda_0}
\bigl(\underline{p}_n
\bigr)
\longrightarrow
\mathcal L^{+,\sigma}_{l,n}(\underline p_n)
\]
as $\Lambda\to0$ and $\Lambda_0\to\infty$, in
$\mathcal C_{\mathrm{loc}}^{-|\sigma|n+|\sigma|-n-\kappa}
(\mathbb R^n)$ for all $\kappa>0$. For the tree-level sector, there exists a distribution
\[
\mathcal L^{+,\sigma}_{0,n}(\underline p_n)
\in
\mathcal C_{\mathrm{loc}}^{-n+1}
(\mathbb R^n),
\]
such that
\[
\mathds 1^{\otimes\sigma}_+(z_i)\,
\mathcal L_{0,n}^{\Lambda,\Lambda_0}
\bigl(
(z_1,p_1),\ldots,(z_n,p_n)
\bigr)
\longrightarrow
\mathcal L^{+,\sigma}_{0,n}\left(\underline p_n\right)
\]
as $\Lambda\to0$ and $\Lambda_0\to\infty$, in
$\mathcal C_{\mathrm{loc}}^{-n+1-\kappa}
\left(\mathbb R^n\right)$ for all $\kappa>0$.\\
Fix $p\in\mathbb{R}^3$. For the two-point correlators and $l\ge2$, there exists a distribution
\[
\mathcal L^+_{l,2}(\underline{p})
\in
\mathcal C^{-4}(\mathbb R^2),
\]
such that
\[
\mathds 1_+(z_1)\,
\mathcal L_{l,2}^{\Lambda,\Lambda_0}
\bigl(\vec{z}_{1,2};\underline{p}\bigr)
\longrightarrow
\mathcal L^+_{l,2}(\underline{p})(\vec{z}_{1,2})
\]
as $\Lambda\to0$ and $\Lambda_0\to\infty$, in
$\mathcal C^{-4}(\mathbb R^2)$.

Finally, there exists a distribution
\[
\mathcal L^+_{1,2}(\underline{p})
\in
\mathcal C^{-1}(\mathbb R^2),
\]
such that
\[
\mathds 1_+(z_2)\,
\mathcal L_{1,2}^{\Lambda,\Lambda_0}
\bigl(\vec{z}_{1,2};\underline{p}\bigr)
\longrightarrow
\mathcal L^+_{1,2}(\underline{p})(\vec{z}_{1,2})
\]
as $\Lambda\to0$ and $\Lambda_0\to\infty$, in
$\mathcal C^{-1}(\mathbb R^2)$.
\end{corollary}
\begin{remark}
Comparing Theorems~\ref{BesovNorm} and \ref{BesovNorm2}, one observes that
multiplication by a hard half-space cutoff induces a loss of regularity whose magnitude depends explicitly on the cardinality of \(\sigma\). In particular, the truncated correlators belong to strictly less regular Besov--Hölder spaces than the corresponding renormalised correlators.

This loss of regularity is entirely due to the discontinuity of the cutoff and is independent of the ultraviolet behaviour of the theory. It does not signal the appearance of new ultraviolet divergences and does not require the introduction of additional counterterms. The standard renormalisation of the Euclidean \(\phi^4_4\) theory is therefore sufficient to construct all truncated correlators.

The loss should thus be viewed as a purely analytic consequence of multiplying a distribution of negative regularity by a discontinuous function. Moreover, the regularity exponents appearing in Theorem~\ref{BesovNorm2} are optimal within the Besov--Hölder scale considered in this work.

In this sense, Theorem~\ref{BesovNorm2} identifies a precise regularity threshold for the operation
\[
\mathcal L_{l,n}
\longmapsto
\prod_{i\in\sigma}
\mathds 1_{+}(z_i),
\mathcal L_{l,n},
\]
and shows that half-space truncation is compatible with renormalisation at the expense of a quantifiable and optimal loss of regularity.
\end{remark}

\section{Tree representation of the bounds}\label{SecTreeAmp}
The proof of the renormalization bounds relies on a family of tree
amplitudes that encode the iterative structure of the flow equation in
position space. The purpose of this section is to introduce these
combinatorial objects together with the associated amplitudes and to
establish the analytic estimates that will be used throughout the paper.

The appearance of trees is a direct consequence of the quadratic part of
the flow equation. Suppose that the correlators $
\mathcal L_{l_1,n_1+1}^{\Lambda,\Lambda_0}$ and
$\mathcal L_{l_2,n_2+1}^{\Lambda,\Lambda_0}$
are bounded by amplitudes represented by trees whose edges carry
Euclidean heat kernels. The quadratic term combines these amplitudes
through the kernel
\[
K\!\left(\frac{1}{\Lambda^2};z,z'\right),
\]
together with integrations over the variables \(z\) and \(z'\). At the
level of the tree representation, this operation joins the two trees by
creating a new internal edge connecting the vertices \(z\) and \(z'\).
Iterating the flow equation therefore generates a hierarchy of
increasingly complex trees.

Although the perturbative correlators contain contributions from Feynman
diagrams with arbitrarily many loops, the bounds established in this
paper are expressed entirely in terms of tree amplitudes. The
perturbative order is encoded in the combinatorial complexity of the
trees, and in particular in the growth of the number of internal
vertices.
\subsection{Combinatorics of trees}
The tree amplitudes introduced later are indexed by rooted trees.
We therefore begin by collecting the combinatorial operations that will
be used throughout the construction.
\begin{definition}
[Trees]
A tree is understood as a finite connected graph without cycles.
Fix $s\ge2$. We denote by $\mathcal{T}_1^s$ the set of trees with
$s-1$ external vertices and one distinguished vertex called the
\emph{root}. The root is neither internal nor external.

Let $\mathfrak{s}\subset\mathbb{N}$ with $|\mathfrak{s}|=s-1$.
The set of external vertices of $T\in\mathcal{T}_1^s$ is denoted by
$y_{\mathfrak{s}}\equiv\{y_i\}_{i\in\mathfrak{s}}$.

For a tree $T\in\mathcal{T}_1^s$ we denote by
$\mathcal{V}_i(T)$ and $\mathcal{V}_e(T)$
the sets of internal and external vertices respectively and we set
\[
v(T):=|\mathcal{V}_i(T)|
\]
for the number of internal vertices of $T$.
\end{definition}

\begin{remark}
Throughout the paper the vertices of a tree $T\in\mathcal{T}_1^s$
are identified with points of $\mathbb{R}$.
The root, internal vertices and external vertices are denoted
respectively by $z_1$, $\vec z_{2,r}$ and $y_{\mathfrak s}$.
Whenever it is convenient to make this dependence explicit we write
\[
T \equiv T(z_1;\vec z_{2,r};y_{\mathfrak{s}}).
\]
\end{remark}

\begin{definition}[Edges]
We denote by $\left( u,v\right)$ the edge connecting two vertices
$u$ and $v$ in a tree $T$. The sets of internal and external edges
are defined by
\[
\mathcal{I}(T)
:=
\{\langle i,j\rangle:\left( z_i,z_j\right)
\text{ is an internal edge of }T\},
\]
\[
\mathcal{E}(T)
:=
\{\langle i,j\rangle:\left( z_i,y_j\right)
\text{ is an external edge of }T\}.
\]
\end{definition}
\begin{definition}[Double/Triple rooted trees]\label{DiffTr}
Let $T\in\mathcal{T}_1^s$ and let $z',~z''\in\mathcal{V}_i(T)$  be two
internal vertices. The \emph{doubly rooted tree}
$T^{(2)}_{z'}$ is obtained from $T$ by declaring $z'$ to be an
additional root and the \emph{triple rooted tree}
$T^{(3)}_{z'z''}$ is obtained from $T$ by declaring $z'$ and $z''$ to be both
additional roots. In particular
\[
v(T^{(2)}_{z'}) = v(T)-1,~~~v(T^{(3)}_{z'z''}) = v(T)-2 .
\]
The sets of double/triple rooted trees are respectively denoted by $\mathcal{T}^s_2$ and $\mathcal{T}^s_3$.
\end{definition}
\begin{remark}
Multi-rooted trees arise naturally when one or several internal
integration variables are kept fixed. They will play a central role in
the definition of the modified amplitudes introduced below.
\end{remark}
\begin{definition}[Fusion of trees]\label{defFusi}
Let $s_1\ge1$ and $s_2\ge2$.
Let $T_1\in\mathcal{T}_1^{s_1+1}$ and
$T_2\in\mathcal{T}_1^{s_2+1}$ and assume that both trees possess
an external vertex labeled by $u$.

The \emph{fusion tree} $\mathcal F_u(T_1,T_2)$ is obtained by
identifying the two vertices $u$ and removing them, while connecting
the vertices that were adjacent to $u$ in $T_1$ and $T_2$ by a new
internal edge.
\end{definition}
\begin{remark}
The construction of the fusion tree implies that
\[
\mathcal F_u(T_1,T_2)\in\mathcal{T}_1^{s}.
\]
Moreover, if $v$ denotes the number of internal vertices of
$\mathcal F_u(T_1,T_2)$ then
\begin{equation}\label{eg30}
v = v(T_1) + v(T_2) + 1 .
\end{equation}
\end{remark}

\begin{definition}[Reduction of a tree]\label{RedT}
Let $s\ge2$ and $\mathfrak{s}\subset\mathbb N$ with $|\mathfrak{s}|=s-1$.
Let $T\in\mathcal{T}_1^s$ and let $y_j$ be an external vertex. The \emph{reduced tree} $\mathcal R_{y_j}(T)$ is obtained by the
following procedure:

\begin{itemize}
\item Remove the external edge $(z,y_j)$ attached to $y_j$.
\item If the adjacent vertex $z$ becomes of valency one, remove the
internal edge $(z',z)$.
\item Repeat the previous step until reaching a vertex of valency
at least two.
\end{itemize}
If $s=2$, the tree has a single external vertex and the reduction
procedure yields the empty tree. For $s\ge3$ we define the \emph{double reduction}
\[
\mathcal R_{y_j,y_k}(T)
:=
(\mathcal R_{y_j}\circ\mathcal R_{y_k})(T)
=
(\mathcal R_{y_k}\circ\mathcal R_{y_j})(T)
\]
for two distinct external vertices $y_j$ and $y_k$.
\end{definition}
\begin{remark}
The reduction operator $\mathcal R_{y_j}$ removes the external vertex
$y_j$ together with all edges and vertices that become redundant after
its removal. Equivalently, one follows the unique path connecting
$y_j$ to the root and removes successive vertices of valency one until
reaching the first vertex of valency at least two.
\end{remark}

\begin{remark}
The reduction operators decrease the number of external vertices.
More precisely, for every $s\ge2$ one has
\[
\mathcal R_{\cdot}\bigl(\mathcal{T}_1^{s}\bigr)
\subset
\mathcal{T}_1^{s-1},
\]
and for every $s\ge3$
\[
\mathcal R_{\cdot,\cdot}\bigl(\mathcal{T}_1^{s}\bigr)
\subset
\mathcal{T}_1^{s-2}.
\]
\end{remark}
\begin{lemma}\label{lem:red-root-commute}
Let $T$ be a tree and let $z$ be an internal vertex of $T$. Assume that
$z$ is not removed by the corresponding reduction procedure. Then one has
\begin{equation}\label{eq:red-root-commute}
\mathcal R_u\bigl(T^{(2)}_z\bigr)
=
\bigl(\mathcal R_u(T)\bigr)^{(2)}_z,
\qquad
\mathcal R_{u,u}\bigl(T^{(2)}_z\bigr)
=
\bigl(\mathcal R_{u,u}(T)\bigr)^{(2)}_z.
\end{equation}
\end{lemma}

\begin{proof}
The operation of double rooting changes only the designation of the
distinguished vertices and leaves the underlying graph unchanged. On the
other hand, the reduction maps act solely on the graph by removing the
prescribed external edge and iteratively removing vertices of valency
one. Since the vertex $z$ is assumed not to be removed during this
procedure, the two operations commute.
\end{proof}
\subsection{Tree amplitudes}

The trees introduced above serve as combinatorial objects encoding the
structure of the bounds appearing in the flow equations. Their edges
correspond to heat kernels while their internal vertices represent
integration variables. We now associate analytic amplitudes with these
trees.

The amplitudes introduced below will be used both for the analysis of
the full-space correlators and for correlators multiplied by smooth
approximations of half-space truncations. To treat these situations
within a common framework, we introduce a class of admissible cutoff
functions containing both the constant function and monotone smooth
approximations of the Heaviside function.

\begin{definition}\label{DefK}[Admissible cutoff functions]
Let \(\mathcal K\) denote the collection of functions
\[
\chi\in C^\infty(\mathbb R)\cap L^\infty(\mathbb R)
\]
such that
\[
\chi'(z)\ge0\quad\mathrm{and}\quad \chi(z)\ge 0,
\qquad \forall z\in\mathbb R.
\]
\end{definition}
\begin{definition}[Feynman rules]\label{FR}
Let $T\in\mathcal{T}_1^{s}$ and set
\[
\delta_{l,n}:=2-\frac{1}{2^{l}n}.
\]
The amplitudes are defined through the following Feynman rules.

Each internal edge $\left( z_i,z_j\right)$ carries the factor
\[
K\!\left(\frac{\delta_{l,n}}{\Lambda_{ij}^{2}};z_i,z_j\right).
\]
Let $y_j$ be an external vertex attached to $z_i$.
If $j\in\sigma$, the edge $\left( z_i,y_j\right)$ carries the factor
\[
\chi^{k_j}(z_i)\,
K\!\left(\delta_{l,n}\tau_j;z_i,y_j\right),
\]
where $k=(k_j)_{j\in\sigma}\in\mathbb N^{|\sigma|}$.
If $j\notin\sigma$, the corresponding factor is simply
\[
K\!\left(\delta_{l,n}\tau_j;z_i,y_j\right).
\]
The tree amplitude is obtained by integrating over all internal
vertices of $T$.
\end{definition}
\begin{remark}
    The parameter $
\delta_{l,n}$
is introduced in order to keep a small reserve of heat-kernel decay.
Throughout the proof, factors of the form $
|z-z'|^r$
arise from Taylor expansions and derivative estimates.
These factors are absorbed into the heat kernels at each order of perturbation theory at the expense of a
slight reduction of the diffusion parameter.
The choice \(\delta_{l,n}<2\) guarantees that
this procedure can be iterated finitely many times while preserving the
required heat-kernel bounds.
\end{remark}
\begin{definition}[Tree amplitudes]\label{TrA}
Let $T\in\mathcal{T}_1^{s}$ and $\sigma\subseteq\mathfrak{s}$.
For $\vec k=(k_j)_{j\in\sigma}\in\mathbb N^{|\sigma|}$ we define
\begin{multline}\label{AmT}
\mathcal{A}^{\Lambda,\Lambda_0}_{l,n}(T)
\bigl(z_1;(\tau_i,y_i)_{i\in\mathfrak{s}};
\chi^{\sigma}_{\vec k}\bigr)
\\
:=
\sup_{\Lambda\le\Lambda_{ij}\le\Lambda_0}
\int_{\mathbb R^{r-1}} d\vec z_{2,r}
\prod_{\langle i,j\rangle\in\mathcal I(T)}
K\!\left(\frac{\delta_{l,n}}{\Lambda_{ij}^{2}};z_i,z_j\right)
\\
\times
\prod_{\substack{\langle i,j\rangle\in\mathcal E(T)\\ j\in\sigma}}
\chi^{k_j}(z_i)\,
K(\delta_{l,n}\tau_j;z_i,y_j)
\prod_{\substack{\langle i,j\rangle\in\mathcal E(T)\\ j\notin\sigma}}
K(\delta_{l,n}\tau_j;z_i,y_j).
\end{multline}
To account for all possible distributions of the cutoff factors we set
\[
\Sigma(\sigma)
:=
\Bigl\{\vec k\in\mathbb N^{|\sigma|}:
\sum_{i\in\sigma} k_i = |\sigma|\Bigr\},
\]
and define
\begin{equation}\label{Sumkchi}
\mathcal{A}^{\Lambda,\Lambda_0}_{l,n}(T)
\bigl(z_1;(\tau_i,y_i)_{i\in\mathfrak{s}};
\chi^{\sigma}\bigr)
:=
\sum_{\vec k\in\Sigma(\sigma)}
\mathcal{A}^{\Lambda,\Lambda_0}_{l,n}(T)
\bigl(
z_1;(\tau_i,y_i)_{i\in\mathfrak{s}};
\chi^{\sigma}_{\vec k}
\bigr).
\end{equation}
If $\chi=\mathds{1}$ we simply write
\[
\mathcal{A}^{\Lambda,\Lambda_0}_{l,n}(T)
\bigl(z_1;(\tau_i,y_i)_{i\in\mathfrak{s}}\bigr).
\]
Note that \eqref{AmT} implies that 
\begin{equation}\label{eq:borneamplitudechi}
   \mathcal{A}^{\Lambda,\Lambda_0}_{l,n}(T)
\bigl(z_1;(\tau_i,y_i)_{i\in\mathfrak{s}};\chi^{\sigma}\bigr)\leq~\|\chi\|^{|\sigma|}_{\infty}~\mathcal{A}^{\Lambda,\Lambda_0}_{l,n}(T)
\bigl(z_1;(\tau_i,y_i)_{i\in\mathfrak{s}}\bigr)~. 
\end{equation}
If $T^{(2)}$ is a doubly rooted tree, the amplitude
$\mathcal{A}^{\Lambda,\Lambda_0}_{l,n}(T^{(2)})$
is defined in the same way, keeping both roots fixed.
\end{definition}
\begin{remark}
The amplitude
\(
\mathcal A_{l,n}^{\Lambda,\Lambda_0}(T)
\)
is obtained by assigning a heat kernel to each edge of \(T\) and then
integrating over all internal vertices except the root.
Multi-rooted amplitudes correspond to keeping one or several of these
integration variables fixed.
\end{remark}
The decomposition of the correlators into relevant and irrelevant
contributions is based on Taylor expansions around the root variable.
When the Taylor operator acts on the cutoff factors, the resulting
remainder terms are expressed through integral representations involving
derivatives of \(\chi\).
These remainders naturally lead to amplitudes in which one or two
internal vertices are distinguished. The corresponding contributions are
encoded by the modified amplitudes introduced below.
For $(z,z')\in\mathbb{R}^2$
we set
\[
\Delta^{(1)}_{z,z'}\chi
:=
\int_0^1
\chi'\!\bigl(tz+(1-t)z'\bigr)\,dt .
\]
\begin{definition}[$\wedge$--amplitudes]\label{def:hat-amplitude}
Let $T\in\mathcal{T}_1^{s}$ and let $\sigma\subset\mathfrak{s}$.
For every internal vertex $z\in\mathcal V_i(T)$, define
\[
\sigma_T(z)
:=
\{\, j\in\sigma : \left(z,y_j\right)~\mathrm{external~edge~in~}T\,\}.
\]
We also set
\[
V_{\mathrm{ext}}^{\sigma}(T)
:=
\{\, z\in\mathcal V_i(T): \sigma_T(z)\neq\varnothing\,\}.
\]
For $\vec k=(k_i)_{i\in\sigma}\in\mathbb N^{|\sigma|}$ we define
\begin{multline}\label{equation72}
\hat{\mathcal A}^{\Lambda,\Lambda_0}_{l,n}(T)
\bigl(
z_1;(\tau_i,y_i)_{i\in\mathfrak{s}};
\chi^\sigma_{\vec k}
\bigr)
:=\sum_{z\in V_{\mathrm{ext}}^\sigma(T)}
\int_{\mathbb R} dz~
\Delta^{(1)}_{z_1,z}\chi~
\mathcal A^{\Lambda,\Lambda_0}_{l,n}
\bigl(
T^{(2)}_{z}
\bigr)
\bigl(
z_1,z;(\tau_i,y_i)_{i\in\mathfrak{s}};
\chi^{\sigma}_{\vec k^{(z)}}
\bigr)\\
+\sum_{z,z'\in V_{\mathrm{ext}}^\sigma(T)}
\int_{\mathbb R} dz\int_{\mathbb{R}}dz'~
\Delta^{(1)}_{z,z'}\chi
\Bigl(
\mathcal A^{\Lambda,\Lambda_0}_{l,n}
\bigl(
T^{(3)}_{zz'}
\bigr)
\bigl(
z_1,z,z';(\tau_i,y_i)_{i\in\mathfrak{s}};
\chi^{\sigma}_{\vec k^{(z)}}
\bigr)
\\
+
\mathcal A^{\Lambda,\Lambda_0}_{l,n}
\bigl(
T^{(3)}_{zz'}
\bigr)
\bigl(
z_1,z,z';(\tau_i,y_i)_{i\in\mathfrak{s}};
\chi^{\sigma}_{\vec k^{(z')}}
\bigr)
\Bigr),
\end{multline}
where
\[
\vec k^{(z)}=(k_i^{(z)})_{i\in\sigma},
\qquad
k_i^{(z)} :=
\begin{cases}
\max(0,k_i-1), & i\in\sigma_T(z),\\
k_i, & i\notin\sigma_T(z).
\end{cases}
\]

The modified decoration \(k^{(z)}\) corresponds to removing one power
of \(\chi\) from every external edge attached to the distinguished
vertex \(z\).
The indices \(z\) and \(z'\) refer to internal vertices of \(T\)
belonging to \(V_{\mathrm{ext}}^\sigma(T)\).
By a slight abuse of notation the same symbols are used for the
corresponding integration variables.\\
The definition \eqref{equation72} relies on the assumption $|\sigma|\ge2$.
If $|\sigma|=1$, we instead define
\begin{equation}\label{equation72-one}
\hat{\mathcal A}^{\Lambda,\Lambda_0}_{l,n}(T)
\bigl(
z_1;(\tau_i,y_i)_{i\in\mathfrak{s}};
\chi
\bigr)
:=
\int_{\mathbb R} dz~
\Delta^{(1)}_{z_1,z}\chi
\,
\mathcal A^{\Lambda,\Lambda_0}_{l,n}
\bigl(
T^{(2)}_{z}
\bigr)
\bigl(
z_1,z;(\tau_i,y_i)_{i\in\mathfrak{s}}
\bigr).
\end{equation}
Summing over all decorations, we set
\begin{equation}
\hat{\mathcal A}^{\Lambda,\Lambda_0}_{l,n}(T)
\bigl(
z_1;(\tau_i,y_i)_{i\in\mathfrak{s}};
\chi^\sigma
\bigr)
:=
\sum_{\vec k\in\Sigma(\sigma)}
\hat{\mathcal A}^{\Lambda,\Lambda_0}_{l,n}(T)
\bigl(
z_1;(\tau_i,y_i)_{i\in\mathfrak{s}};
\chi^\sigma_{\vec k}
\bigr).
\end{equation}
\end{definition}
\begin{remark}
The monotonicity assumption \(\chi'\ge0\) enters at this stage.
Indeed, the localisation difference \(\Delta\chi\) is represented
through derivatives of the cutoff function. In order to estimate the
resulting terms by positive tree amplitudes, the derivative is
transferred onto the heat kernel through integration by parts in case the root is also integrated out. Since the argument is performed at the
level of absolute values, the identity
\(
|\chi'|=\chi'
\)
is crucial.
\end{remark}
\begin{lemma}[Root extraction]\label{lem:multi-root}
Let $T\in\mathcal{T}_1^{s}$ and let $z_k,z_{k'}\in\mathcal V_i(T)$ be
internal vertices of $T$. Then one has
\begin{align}
\mathcal A^{\Lambda,\Lambda_0}_{l,n}(T)
\bigl(
z_1;(\tau_i,y_i)_{i\in\mathfrak{s}};\chi^\sigma
\bigr)
&=
\int_{\mathbb R} dz_k\;
\mathcal A^{\Lambda,\Lambda_0}_{l,n}
\bigl(
T^{(2)}_{z_k}
\bigr)
\bigl(
z_1,z_k;(\tau_i,y_i)_{i\in\mathfrak{s}};\chi^\sigma
\bigr),
\label{ADRT}
\\
\mathcal A^{\Lambda,\Lambda_0}_{l,n}(T)
\bigl(
z_1;(\tau_i,y_i)_{i\in\mathfrak{s}};\chi^\sigma
\bigr)
&=
\int_{\mathbb R} dz_k\int_{\mathbb R} dz_{k'}\;
\mathcal A^{\Lambda,\Lambda_0}_{l,n}
\bigl(
T^{(3)}_{z_k z_{k'}}
\bigr)
\bigl(
z_1,z_k,z_{k'};(\tau_i,y_i)_{i\in\mathfrak{s}};\chi^\sigma
\bigr).
\label{ATRT}
\end{align}
\end{lemma}

\begin{proof}
Both identities follow directly from the definition of the tree
amplitudes \eqref{AmT}. Indeed, in the definition of
\(
\mathcal A_{l,n}^{\Lambda,\Lambda_0}(T)
\),
one may isolate the integration with respect to $z_k$, or with respect
to both $z_k$ and $z_{k'}$. The remaining integrand is then precisely
the amplitude of the corresponding doubly rooted tree
\(
T^{(2)}_{z_k}
\)
or triply rooted tree
\(
T^{(3)}_{z_k z_{k'}}
\),
respectively.
\end{proof}

We now introduce the global amplitudes obtained by summing the
contributions of all admissible trees.

\begin{definition}[Tree global amplitudes]
Fix $\mathfrak{s}\subseteq\{2,\dots,n\}$. 
The $(n,l)$ tree global amplitudes are defined by
\begin{align}\label{AsumT}
\mathcal A^{\Lambda,\Lambda_0}_{l,n}
\bigl(
z_1;(\tau_i,y_i)_{i\in\mathfrak s};\chi^\sigma
\bigr)
&:=
\sum_{\substack{T\in\mathcal T^{\mathfrak s}_1\\
0\le v(T)\le l+\frac n2-2}}
\mathcal A^{\Lambda,\Lambda_0}_{l,n}(T)
\bigl(
z_1;(\tau_i,y_i)_{i\in\mathfrak s};\chi^\sigma
\bigr),
\\
\hat{\mathcal A}^{\Lambda,\Lambda_0}_{l,n}
\bigl(
z_1;(\tau_i,y_i)_{i\in\mathfrak s};\chi^\sigma
\bigr)
&:=
\sum_{\substack{T\in\mathcal T^{\mathfrak s}_1\\
0\le v(T)\le l+\frac n2-2}}
\hat{\mathcal A}^{\Lambda,\Lambda_0}_{l,n}(T)
\bigl(
z_1;(\tau_i,y_i)_{i\in\mathfrak s};\chi^\sigma
\bigr).
\end{align}
\end{definition}
\subsection{Calculus of tree amplitudes}
We begin the calculus of tree amplitudes with a simple monotonicity
property with respect to the parameters $(l,n)$.
\subsubsection{Monotonicity of tree amplitudes}
\begin{lemma}\label{lem:monotone_l}
For every $T\in\mathcal T_1^s$, every $l'\le l$, and every $n'\le n$, one has
\begin{equation}\label{AlAl0_nohat}
\mathcal A^{\Lambda,\Lambda_0}_{l',n'}(T)
\;\lesssim\;
\mathcal A^{\Lambda,\Lambda_0}_{l,n}(T),
\qquad
\hat{\mathcal A}^{\Lambda,\Lambda_0}_{l',n'}(T)
\;\lesssim\;
\hat{\mathcal A}^{\Lambda,\Lambda_0}_{l,n}(T).
\end{equation}
\end{lemma}

\begin{proof}
Let $T\in\mathcal T_1^s$. Since $l'\le l$ and $n'\le n$, one has
\[
\delta_{l',n'}
=
2-\frac{1}{2^{l'}n'}
\le
2-\frac{1}{2^l n}
=
\delta_{l,n}.
\]
By definition, both $\mathcal A^{\Lambda,\Lambda_0}_{l,n}(T)$ and
$\hat{\mathcal A}^{\Lambda,\Lambda_0}_{l,n}(T)$ are obtained from
non-negative integrands by integration over internal vertices and by
taking the supremum over the internal scales
$\Lambda\le \Lambda_{ij}\le \Lambda_0$. More precisely, each internal edge $\langle i,j\rangle\in\mathcal I(T)$
contributes a factor
\[
K\!\left(\frac{\delta_{l,n}}{\Lambda_{ij}^2};z_i,z_j\right),
\]
while each external edge $\langle i,j\rangle\in\mathcal E(T)$ contributes
\[
K(\delta_{l,n}\tau_j;z_i,y_j),
\]
possibly multiplied by a non-negative cutoff factor. In the hatted case,
one also has the additional factor $\Delta^{(1)}_{z_1}\chi(z_j)$,
which is independent of $(l,n)$. Using the monotonicity of the heat kernel, namely
\begin{equation}\label{eq:Kmono_nohat}
K(\delta_{l',n'}t;z,z')
\;\lesssim\;
K(\delta_{l,n}t;z,z')
\qquad
\text{for all }t>0,\ (z,z')\in\mathbb R^2,
\end{equation}
we obtain a pointwise bound of the integrand defining
$\mathcal A^{\Lambda,\Lambda_0}_{l',n'}(T)$ by the one defining
$\mathcal A^{\Lambda,\Lambda_0}_{l,n}(T)$.
\end{proof}
\subsubsection{Stability under external insertions}
The next result shows that the class of tree amplitudes is stable under
the insertion of additional external lines attached to the root. This
property will be used repeatedly when Taylor expanding the four-point
correlator around its first argument, which generates additional
heat-kernel factors attached to the expansion point.

\begin{lemma}\label{MagnLemm}
Fix $l\in\mathbb{N}_0$ and $n\in\mathbb{N}$. Let
$\tilde{\mathfrak{s}}_1,\mathfrak{s}_2,\sigma_1\subseteq \mathcal{I}_n$
be pairwise disjoint, and let $\sigma_2\subseteq \mathfrak{s}_2$.
Set
\[
\sigma:=\sigma_1\cup\sigma_2,
\qquad
\mathfrak{s}:=\tilde{\mathfrak{s}}_1\cup\sigma_1\cup\mathfrak{s}_2.
\]
Then one has
\begin{equation}\label{Magnetic}
\prod_{i\in\tilde{\mathfrak{s}}_1}
K(\tau_i;z_1,y_i)
\prod_{i\in\sigma_1}
\chi(z_1)\,K(\tau_i;z_1,y_i)\,
\mathcal{A}_{l,n}^{\Lambda,\Lambda_0}
\bigl(
z_1;(\tau_i,y_i)_{i\in\mathfrak{s}_2};\chi^{\sigma_2}
\bigr)
\lesssim
\mathcal{A}_{l,n}^{\Lambda,\Lambda_0}
\bigl(
z_1;(\tau_i,y_i)_{i\in\mathfrak{s}};\chi^{\sigma}
\bigr).
\end{equation}
The same estimate holds with $\mathcal{A}_{l,n}^{\Lambda,\Lambda_0}$
replaced by $\hat{\mathcal{A}}_{l,n}^{\Lambda,\Lambda_0}$.
\end{lemma}

\begin{proof}
By definition of $\mathcal{A}_{l,n}^{\Lambda,\Lambda_0}$, we have
\begin{equation}
\mathcal{A}_{l,n}^{\Lambda,\Lambda_0}
\bigl(
z_1;(\tau_i,y_i)_{i\in\mathfrak{s}_2};\chi^{\sigma_2}
\bigr)
=
\sum_{\substack{T\in\mathcal{T}^{s_2}_1\\0\le v(T)\le l+\frac{n}{2}-2}}
\sum_{\vec{k}\in\Sigma(\sigma_2)}
\mathcal{A}_{l,n}^{\Lambda,\Lambda_0}(T)
\bigl(
z_1;(\tau_i,y_i)_{i\in\mathfrak{s}_2};\chi^{\sigma_2}_{\vec{k}}
\bigr),
\end{equation}
where we used the notation of Definition~\eqref{Sumkchi}. Multiplying by
\[
\prod_{i\in\tilde{\mathfrak{s}}_1}
K(\tau_i;z_1,y_i)
\prod_{i\in\sigma_1}
\chi(z_1)\,K(\tau_i;z_1,y_i)
\]
amounts to attaching additional external lines to $T$ at the root vertex $z_1$.
The lines indexed by $\sigma_1$ carry a factor $\chi(z_1)$, while those indexed
by $\tilde{\mathfrak{s}}_1$ are undecorated.

Denote by $\tilde{T}$ the resulting tree. Then $\tilde{T}\in\mathcal{T}^s_1$,
with $s-1=|\mathfrak{s}|$, and its internal structure coincides with that of $T$,
so that $v(\tilde{T})=v(T)$. Moreover, the decoration is given by
\[
k_i^{(\sigma)}=
\begin{cases}
k_i, & i\in\sigma_2,\\
1, & i\in\sigma_1.
\end{cases}
\]
In particular, one has $\vec{k}^{(\sigma)}\in\Sigma(\sigma)$. Therefore,
\begin{multline}
\prod_{i\in\tilde{\mathfrak{s}}_1}
K(\tau_i;z_1,y_i)
\prod_{i\in\sigma_1}
\chi(z_1)\,K(\tau_i;z_1,y_i)\,
\mathcal{A}_{l,n}^{\Lambda,\Lambda_0}(T)
\bigl(
z_1;(\tau_i,y_i)_{i\in\mathfrak{s}_2};\chi^{\sigma_2}_{\vec{k}}
\bigr)
\\
=
\mathcal{A}_{l,n}^{\Lambda,\Lambda_0}(\tilde{T})
\bigl(
z_1;(\tau_i,y_i)_{i\in\mathfrak{s}};\chi^{\sigma}_{\vec{k}^{(\sigma)}}
\bigr).
\end{multline}

Summing over all $T$ and $\vec{k}$ yields \eqref{Magnetic}.
The bound for $\hat{\mathcal{A}}$ is obtained in the same way,
starting from its definition~\eqref{equation72}.
\end{proof}
\subsubsection{Distance estimate}
We now turn to the first  geometric estimate in the tree
calculus. It shows that powers of distances between distinguished
vertices can be absorbed into the tree amplitude at the expense of
powers of the infrared scale $\Lambda^{-1}$. This mechanism is one of
the basic tools used to control Taylor remainders.

\begin{lemma}\label{LemDecPo}
Let $r\in\{1,2, 3\}$, and let $T\in\mathcal{T}^s_2$ be a double-rooted tree
with roots $z_1$ and $z_2$. Then the following bound holds for all $l'\ge l$, and $n'>n$
\begin{equation}\label{R83}
\left|z_1-z_2\right|^r\,
\mathcal{A}_{l,n}^{\Lambda,\Lambda_0}(T)
\bigl(
z_1,z_2;(\tau_i,y_i)_{i\in\mathfrak{s}};\chi^{\sigma}
\bigr)
\lesssim
\Lambda^{-r}\,
\mathcal{A}_{l',n'}^{\Lambda,\Lambda_0}(T)
\bigl(
z_1,z_2;(\tau_i,y_i)_{i\in\mathfrak{s}};\chi^{\sigma}
\bigr)~.
\end{equation}
In particular \eqref{R83} holds for $\chi=\mathds{1}$.
\end{lemma}

\begin{proof}
Since $T$ is a tree, there exists a unique path connecting the two roots
$z_1$ and $z_2$. Let $\langle z_1,z_2\rangle$ denote the set of internal vertices in $T$
forming the unique path from $z_1$ to $z_2$, including the endpoints.
We write
\begin{equation}
\langle z_1,z_2\rangle
=
\{v_a: 0\le a\le q\},
\qquad v_0=z_1,\quad v_q=z_2.
\end{equation}
We telescope $z_1-z_2$ along this path:
\begin{equation}
z_1-z_2=\sum_{a=0}^{q-1}(v_{a+1}-v_a).
\end{equation}

Recall that $\mathcal{A}_{l,n}^{\Lambda,\Lambda_0}(T)$ is given by a product
of heat kernels evaluated at adjacent vertices, weighted by parameters
$\Lambda\le \Lambda_a\le \Lambda_0$ if both vertices are internal, and by
$\tau_j$ with $j\in\mathfrak{s}$ if one of the vertices is external.
In addition, there is a product of cutoff factors $\chi$ evaluated at
internal vertices attached to external vertices with indices in $\sigma$.
These factors play no role in the present argument and can therefore be
ignored.
The contribution of the internal vertices along the path from $z_1$ to $z_2$
is given by
\begin{equation}
\prod_{a=0}^{q-1}
K\!\left(\frac{\delta_{l,n}}{\Lambda_a^2};v_a,v_{a+1}\right).
\end{equation}
We use the bound
\begin{equation}
|z-z'|^r\,K(\tau;z,z')
\lesssim
\tau^{-r/2}\,K(c\tau;z,z'),
\qquad \forall\, c>1~.
\end{equation}
By the triangle inequality along the path, and since \(r\le 3\), one has
\[
|z_1-z_2|^r
\lesssim
\sum_{\substack{\alpha_0,\ldots,\alpha_{q-1}\in\mathbb N_0\\
\alpha_0+\cdots+\alpha_{q-1}=r}}
\prod_{a=0}^{q-1}
|v_{a+1}-v_a|^{\alpha_a}.
\]
Since \(r\in\{1,2,3\}\), each term contains at most three nontrivial
factors, and each such factor can be absorbed into the corresponding
heat kernel on the path.\\
Combining these bounds and using that $\delta_{l,n}\le \delta_{l',n'}$ for all
$l'\ge l$ and $n'>n$, we obtain
\begin{equation}
|z_1-z_2|^r\,
\prod_{a=0}^{q-1}
K\!\left(\frac{\delta_{l,n}}{\Lambda_a^2};v_a,v_{a+1}\right)
\lesssim
\Lambda^{-r}\,
\prod_{a=0}^{q-1}
K\!\left(\frac{\delta_{l',n'}}{\Lambda_a^2};v_a,v_{a+1}\right).
\end{equation}
The remaining heat kernel factors in the amplitude are bounded using
\[
K(\delta_{l,n}\tau;z,z')
\lesssim
K(\delta_{l',n'}\tau;z,z'),
\qquad \forall\, l'\ge l,~\forall n'>n~.
\]
This yields \eqref{R83}.
\end{proof}
\subsubsection{Estimates for derivative insertions}
The modified amplitudes $\hat{\mathcal A}$ arise from the Taylor
remainder terms generated when the expansion acts on cutoff factors.
The purpose of this subsection is to show that these amplitudes can be
reduced to ordinary tree amplitudes. More precisely, we prove that each
derivative insertion either produces a Gaussian derivative factor
\(\tau_i^{-1/2}\) or can be absorbed into the cutoff decoration of the
tree. This ultimately yields a bound of \(\hat{\mathcal A}\) in terms
of \(\mathcal A\).\\
\indent The following notation will be used to keep track of the cutoff
decorations generated by derivative insertions. For $a\in\sigma$, we write
\[
\vec k+\mathbf 1_a\in\mathbb N_0^\sigma
\]
for the multi--index obtained from $\vec k$ by increasing the $a$--th
component by one, that is,
\[
(\vec k+\mathbf 1_a)_i=
\begin{cases}
k_a+1,& i=a,\\
k_i,& i\neq a.
\end{cases}
\]
Similarly, whenever $r\in\sigma$ with $k_r\ge1$, we denote by
\[
\vec k-\mathbf 1_r\in\mathbb N_0^\sigma
\]
the multi--index obtained from $\vec k$ by decreasing the $r$--th
component by one.

The next lemma describes how the localisation carried by the remainder
factor \(\Delta^{(1)}_{z,z'}\chi\) can be transferred to the cutoff
decorations attached to the tree.
\begin{lemma}\label{Lemmachi'}
Let $\mathfrak s\subset\mathbb N$ be finite with $|\mathfrak s|\ge2$,
and let $\sigma\subseteq\mathfrak s$. Let $T\in\mathcal T_1^s$ with
$v(T)\ge1$. Assume that $a,b\in\sigma$ are such that the internal
vertices $z$ and $z'$ are adjacent respectively to the external
vertices $y_a$ and $y_b$. Let $\vec k=(k_i)_{i\in\sigma}\in\mathbb N_0^\sigma$
satisfy
\[
\sum_{i\in\sigma}k_i=|\sigma|-1 
\]
and set 
\[
\sigma_{\vec k}:=\{r\in\sigma:\;k_r\ge1\}.
\]
Then, for every $l'\ge l$ and $n'>n$, one has
\begin{multline}\label{equation41-new}
\int_{\mathbb R}dz\int_{\mathbb R}dz'~
\Delta^{(1)}_{z,z'}\chi~
\mathcal A_{l,n}^{\Lambda,\Lambda_0}
\bigl(T^{(2)}_{z'}\bigr)
\bigl(
z,z';(\tau_i,y_i)_{i\in\mathfrak s};
\chi^\sigma_{\vec k}
\bigr)
\\
\lesssim
\sum_{i\in\mathfrak s}\tau_i^{-1/2}
\int_{\mathbb R}dz~
\Bigl(
\mathcal A_{l',n'}^{\Lambda,\Lambda_0}(T)
\bigl(
z;(\tau_i,y_i)_{i\in\mathfrak s};
\chi^\sigma_{\vec k+\mathbf 1_a}
\bigr)
+
\mathcal A_{l',n'}^{\Lambda,\Lambda_0}(T)
\bigl(
z;(\tau_i,y_i)_{i\in\mathfrak s};
\chi^\sigma_{\vec k+\mathbf 1_b}
\bigr)
\Bigr)
\\
+
\sum_{r\in\sigma_{\vec k}}
\int_{\mathbb R}dz~
\Bigl(
\hat{\mathcal A}_{l,n}^{\Lambda,\Lambda_0}(T)
\bigl(
z;(\tau_i,y_i)_{i\in\mathfrak s};
\chi'_r\,\chi^\sigma_{\vec k-\mathbf 1_r+\mathbf 1_a}
\bigr)
+
\hat{\mathcal A}_{l,n}^{\Lambda,\Lambda_0}(T)
\bigl(
z;(\tau_i,y_i)_{i\in\mathfrak s};
\chi'_r\,\chi^\sigma_{\vec k-\mathbf 1_r+\mathbf 1_b}
\bigr)
\Bigr),
\end{multline}
where $\chi'_r$ denotes the insertion of $\chi'$ at the internal vertex
attached to $y_r$.
\end{lemma}

\begin{proof}
Fix $T\in\mathcal T_1^s$ with $v(T)\ge1$. By definition of the double
rooted amplitude,
\begin{multline}\label{eq:double-rooted-proof-revised}
\mathcal A_{l,n}^{\Lambda,\Lambda_0}
\bigl(T^{(2)}_{z'}\bigr)
\bigl(
z,z';(\tau_i,y_i)_{i\in\mathfrak s};
\chi^\sigma_{\vec k}
\bigr)
\\
=
\sup_{\Lambda\le\Lambda_{uv}\le\Lambda_0}
\int d\vec z_{\mathrm{int}}
\prod_{\langle u,v\rangle\in\mathcal I(T)}
K\!\left(\frac{\delta_{l,n}}{\Lambda_{uv}^2};z_u,z_v\right)
\\
\times
\prod_{\substack{\langle u,v\rangle\in\mathcal E(T)\\ v\in\sigma}}
\chi^{k_v}(z_u)\,
K(\delta_{l,n}\tau_v;z_u,y_v)
\prod_{\substack{\langle u,v\rangle\in\mathcal E(T)\\ v\notin\sigma}}
K(\delta_{l,n}\tau_v;z_u,y_v),
\end{multline}
where the two roots $z$ and $z'$ are kept fixed. We translate every internal integration variable distinct from $z'$
according to
\[
z_u\mapsto z_u-z'.
\]
By translation invariance of the heat kernel, the internal kernels
become independent of $z'$. After the  change of variables, the remainder factor takes the form
\[
\Delta^{(1)}_{z,z'}\chi
=
\int_0^1 \chi'(tz+z')\,dt
=
\partial_{z'}
\left(
\int_0^1 \chi(tz+z')\,dt
\right).
\]
Hence the dependence on $z'$ appears only
through the cutoff factors and the external kernels, together with
\[
\Delta^{(1)}_{z,z'}\chi
=
\partial_{z'}
\left(
\int_0^1 \chi(tz+z')\,dt
\right).
\]
Substituting this representation into the left-hand side of
\eqref{equation41-new} and integrating by parts in $z'$ yields
\begin{equation}\label{eq:ibp-scheme}
\int_{\mathbb R}dz\int_{\mathbb R}dz'
\int_0^1\chi(tz+z')\,dt~
\partial_{z'}
\Bigl[
\mathcal A_{l,n}^{\Lambda,\Lambda_0}
\bigl(T^{(2)}_{z'}\bigr)
\bigl(
z,0;(\tau_i,y_i-z')_{i\in\mathfrak s};
\theta_{z'}\chi^\sigma_{\vec k}
\bigr)
\Bigr].
\end{equation}
No boundary terms appear since $\chi$ is bounded and the heat kernels
have sufficient decay. We now analyze the derivative $\partial_{z'}$:

\smallskip
\noindent
\textit{Derivative acting on cutoff factors.}
If $\partial_{z'}$ hits $\chi^{k_r}(z_u-z')$ with $r\in\sigma$ such that $k_r\geq 1$ we obtain
\[
|\partial_{z'}\chi^{k_r}(z_u-z')|
\lesssim
\chi'(z_u-z')\chi^{k_r-1}(z_u-z').
\]
Multiplying by $\chi(tz+(1-t)z')$ and using that $\chi$ is nonnegative and increasing, one has
\[
\chi(tz+(1-t)z')
\le
\chi(z)+\chi(z')
\]
and this produces the two hatted contributions appearing in the second line of
\eqref{equation41-new}.

\smallskip
\noindent
\textit{Derivative acting on external kernels.}
If $\partial_{z'}$ hits an external kernel attached to
$i\in\mathfrak s$, we use the Gaussian derivative bound
\[
|\partial_z K(\delta_{l,n}\tau_i;z,z')|
\lesssim
\tau_i^{-1/2}\,
K(\delta_{l',n'}\tau_i;z,z'),
\qquad l'\ge l,\;n'>n .
\]
Combining again with
\(
\chi(tz+(1-t)z')\le\chi(z)+\chi(z')
\)
yields the first line of the right-hand side of
\eqref{equation41-new}. The remaining kernels are controlled by the
monotonicity bound
\begin{equation}\label{HKMONO}
K(\delta_{l,n}\tau;z,z')\lesssim K(\delta_{l',n'}\tau;z,z').
\end{equation}
Summing over $i\in\mathfrak s$ concludes the proof.
\end{proof}
\begin{remark}
The assumption that the distinguished vertices
\(z,z'\) belong to \(V^{\sigma}_{\mathrm{ext}}(T)\) is used only to
preserve the cutoff bookkeeping. If neither \(z\) nor \(z'\) carries a
cutoff decoration, then the localisation factor
\(\Delta^{(1)}_{z,z'}\chi\) cannot be absorbed into the multi-index
\(\vec k\). In this case, the argument of Lemma~\ref{Lemmachi'} together with \eqref{eq:borneamplitudechi}
still yields the estimate
\begin{multline}\label{Transport3}
\int_{\mathbb R}dz\int_{\mathbb R}dz'~
\Delta^{(1)}_{z,z'}\chi~
\mathcal A_{l,n}^{\Lambda,\Lambda_0}
\bigl(T^{(2)}_{z'}\bigr)
\bigl(
z,z';(\tau_i,y_i)_{i\in\mathfrak s};
\chi^\sigma_{\vec k}
\bigr)
\\
\lesssim
\|\chi\|_{\infty}^{|\sigma|}
\sum_{i\in\mathfrak s}\tau_i^{-1/2}
\int_{\mathbb R}dz~
\mathcal A_{l',n'}^{\Lambda,\Lambda_0}(T)
\bigl(
z;(\tau_i,y_i)_{i\in\mathfrak s};
\chi^\sigma_{\vec k}
\bigr)
\\
+
\|\chi\|_{\infty}
\sum_{r\in\sigma_{\vec k}}
\int_{\mathbb{R}}dz~
\hat{\mathcal A}_{l,n}^{\Lambda,\Lambda_0}(T)
\bigl(
z;(\tau_i,y_i)_{i\in\mathfrak{s}};
\chi'_r\,\chi^{\sigma}_{\vec{k}-\mathbf{1}_r}
\bigr).
\end{multline}
Thus the same type of bound remains valid, but the localisation factor
is no longer converted into an additional cutoff power. Consequently,
the total cutoff weight is not preserved, which is why we restrict
attention to distinguished vertices in
\(V^{\sigma}_{\mathrm{ext}}(T)\) throughout the subsequent inductive
analysis.
\end{remark}
The previous lemma shows that the remainder factor
$\Delta^{(1)}_{z,z'}\chi$ carries one unit of localisation that can be
transported to either of the two external vertices adjacent to $z$ and
$z'$. The cost of this transport is of two possible types: either a
Gaussian derivative factor $\tau_i^{-1/2}$ when the derivative acts on an
external heat kernel, or a $\chi'$-insertion when it acts on a pre-existing
cutoff factor. In this way the remainder term can be rewritten in terms of
standard decorated amplitudes, which fit directly into the inductive
bookkeeping of cutoff powers and will be repeatedly used in the fusion
estimates below.

The previous lemma reduces the problem to amplitudes containing
\(\chi'\)-insertions. The next result shows that these derivative
decorations can in turn be absorbed into the original cutoff
decoration.
\begin{lemma}\label{LemmeDeriv}
Let $\mathfrak s$, $\sigma$, and $\vec k$ be as in Lemma~\ref{Lemmachi'}.
Let $T\in\mathcal T_1^{\mathfrak s}$ with $v(T)\ge1$. Then for every $l'\geq l$ and $n'>n$ one has
\begin{multline}\label{equation411}
\sum_{r\in\sigma_{\vec k}}
\int_{\mathbb R}dz~
\hat{\mathcal A}^{\Lambda,\Lambda_0}_{l,n}(T)
\bigl(
z;(\tau_i,y_i)_{i\in\mathfrak s};
\chi'_r\,\chi^\sigma_{\vec k-\mathbf 1_r}
\bigr)\\
\lesssim
\left(\sum_{i\in\mathfrak s}\tau_i^{-1/2}\right)
\int_{\mathbb R}dz~
\mathcal A^{\Lambda,\Lambda_0}_{l',n'}(T)
\bigl(
z;(\tau_i,y_i)_{i\in\mathfrak s};
\chi^\sigma_{\vec k}
\bigr).
\end{multline}
\end{lemma}
\begin{proof}
By the product rule, one has
\begin{equation}\label{eq:chi-derivative-identity}
(\chi^\sigma_{\vec k})'
=
\sum_{r\in\sigma_{\vec{k}}}
k_r\,\chi'_r\,\chi^\sigma_{\vec k-\mathbf 1_r}.
\end{equation}
Since $\chi'\ge0$, the right-hand side is a nonnegative linear combination of
the decorations $\chi'_r\,\chi^\sigma_{\vec k-\mathbf 1_r}$. By linearity of
the amplitude in the decoration and positivity of the integrand, it follows that
\begin{equation}\label{eq:chi-prime-bound}
\sum_{r\in\sigma_{\vec{k}}}
\hat{\mathcal A}_{l,n}^{\Lambda,\Lambda_0}
\bigl(
T
\bigr)
\bigl(
z;(\tau_i,y_i)_{i\in\mathfrak s};
\chi'_r\,\chi^\sigma_{\vec k-\mathbf 1_r}
\bigr)
\lesssim~
\hat{\mathcal A}_{l,n}^{\Lambda,\Lambda_0}
\bigl(
T
\bigr)
\bigl(
z;(\tau_i,y_i)_{i\in\mathfrak s};
(\chi^\sigma_{\vec k})'
\bigr).
\end{equation}
We now translate all internal integration variables according to
\[
z_u\mapsto z_u-z.
\]
By translation invariance of the heat kernel, all internal kernels remain
unchanged under this change of variables. After translation, each external
kernel is of the form
\[
K(\delta_{l,n}\tau_i;z_{a(i)},y_i-z),
\]
while each cutoff factor becomes a function of the form $\chi(z_u+z)$.
Therefore, all dependence on $z$ is carried either by the translated cutoff
factors or by the external kernels.\\
Set
\[
G_{l,n}(z)
:=
\mathcal A_{l,n}^{\Lambda,\Lambda_0}
\bigl(
T
\bigr)
\bigl(
0;(\tau_i,y_i-z)_{i\in\mathfrak s};
\Theta_z(\chi^\sigma_{\vec k})
\bigr),
\]
where $\Theta_z(\chi^\sigma_{\vec k})$ denotes the translated decoration.
Differentiating under the integral sign, one obtains
\begin{equation}\label{eq:derivative-splitting}
\partial_z G_{l,n}(z)
=
\hat{\mathcal A}_{l,n}^{\Lambda,\Lambda_0}
\bigl(
T
\bigr)
\bigl(
z;(\tau_i,y_i)_{i\in\mathfrak s};
(\chi^\sigma_{\vec k})'
\bigr)
+
\sum_{i\in\mathfrak s} B_{l,n;i}(z),
\end{equation}
where, for each $i\in\mathfrak s$, the term $B_{l,n;i}(z)$ is given by the same
translated amplitude as $G_{l,n}(z)$, except that the external heat kernel
corresponding to the label $i$ is replaced by its derivative with respect to
$z$. Since the external kernels provide Gaussian decay in $z$, one has
\[
G_{l,n}(z)\to0
\qquad\text{as }|z|\to\infty.
\]
Consequently,
\[
\int_{\mathbb R}\partial_z G_{l,n}(z)\,dz=0.
\]
Integrating \eqref{eq:derivative-splitting} over $z\in\mathbb R$ and using
\eqref{eq:chi-prime-bound}, we obtain
\[
\sum_{r\in\sigma_{\vec{k}}}
\int_{\mathbb R}dz~
\hat{\mathcal A}_{l,n}^{\Lambda,\Lambda_0}
\bigl(
T
\bigr)
\bigl(
z;(\tau_i,y_i)_{i\in\mathfrak s};
\chi'_r\,\chi^\sigma_{\vec k-\mathbf 1_r}
\bigr)
\lesssim
\sum_{i\in\mathfrak s}\int_{\mathbb R}dz\,|B_{l,n;i}(z)|.
\]
It therefore remains to bound the terms $B_{l,n;i}$. By the Gaussian derivative
estimate for the heat kernel, for every $l'\geq l$ and $n'>n$, one has $\delta_{l',n'}\geq \delta_{l,n}$ which gives 
\[
|\partial_z K(\delta_{l,n}\tau_i;z,z')|
\lesssim
\tau_i^{-1/2}\,
K(\delta_{l',n}\tau_i;z,z').
\]
Furthermore, by the monotonicity bound \eqref{HKMONO},
the remaining heat kernels are controlled at scale
\(\delta_{l',n'}\).
Combining these estimates with the positivity of the remaining factors
yields
\[
|B_{l,n;i}(z)|
\lesssim
\tau_i^{-1/2}\,
\mathcal A_{l',n'}^{\Lambda,\Lambda_0}
\bigl(
T
\bigr)
\bigl(
z;(\tau_j,y_j)_{j\in\mathfrak s};
\chi^\sigma_{\vec k}
\bigr).
\]
Integrating over $z$ and summing over $i\in\mathfrak s$, we obtain
\[
\sum_{r\in\sigma_{\vec{k}}}
\int_{\mathbb R}dz\,
\hat{\mathcal A}_{l,n}^{\Lambda,\Lambda_0}
\bigl(
T
\bigr)
\bigl(
z;(\tau_i,y_i)_{i\in\mathfrak s};
\chi'_r\,\chi^\sigma_{\vec k-\mathbf 1_r}
\bigr)
\lesssim
\left(\sum_{i\in\mathfrak s}\tau_i^{-1/2}\right)
\int_{\mathbb R}dz\,
\mathcal A_{l',n'}^{\Lambda,\Lambda_0}
\bigl(
T
\bigr)
\bigl(
z;(\tau_i,y_i)_{i\in\mathfrak s};
\chi^\sigma_{\vec k}
\bigr),
\]
which is exactly \eqref{equation411}.
\end{proof}
\begin{remark}
Lemma~\ref{LemmeDeriv} shows that a $\chi'$-insertion can be absorbed
into the cutoff decoration of the amplitude. More precisely, the factor
$\chi'_r\,\chi^{k_r-1}$ can be replaced by $\chi^{k_r}$ at the price of
a Gaussian derivative factor $\tau_i^{-1/2}$. Together with
Lemma~\ref{Lemmachi'}, this provides a systematic way of transporting
localisation factors into the multi-index $\vec k$, which is the
quantity controlled in the subsequent bounds.
\end{remark}
\begin{corollary}[Reduction of $\hat{\mathcal A}$ to $\mathcal A$]\label{CorTransport}
Let $\mathfrak s$, $\sigma$, and $T$ be as in Lemma~\ref{Lemmachi'}.
Then, for every $l'\ge l$ and $n'>n$, one has
\begin{multline}\label{equation41-new}
\int_{\mathbb R}dz\int_{\mathbb R}dz'~
\Delta^{(1)}_{z,z'}\chi~
\mathcal A_{l,n}^{\Lambda,\Lambda_0}
\bigl(
T^{(2)}_{z'}
\bigr)
\bigl(
z,z';(\tau_i,y_i)_{i\in\mathfrak s};
\chi^\sigma
\bigr)
\\
\lesssim
\left(\sum_{i\in\mathfrak s}\tau_i^{-1/2}\right)
\int_{\mathbb R}dz~
\mathcal A_{l',n'}^{\Lambda,\Lambda_0}
\bigl(
T
\bigr)
\bigl(
z;(\tau_i,y_i)_{i\in\mathfrak s};
\chi^\sigma
\bigr).
\end{multline}
As a consequence, using the definition of $\hat{\mathcal A}$ and summing
over all pairs of vertices in $V^\sigma_{\mathrm{ext}}(T)$, one obtains
\begin{equation}\label{equationNuevo}
\int_{\mathbb R}dz~
\hat{\mathcal A}_{l,n}^{\Lambda,\Lambda_0}
\bigl(
T
\bigr)
\bigl(
z;(\tau_i,y_i)_{i\in\mathfrak s};
\chi^\sigma
\bigr)
\lesssim
\left(\sum_{i\in\mathfrak s}\tau_i^{-1/2}\right)
\int_{\mathbb R}dz~
\mathcal A_{l',n'}^{\Lambda,\Lambda_0}
\bigl(
T
\bigr)
\bigl(
z;(\tau_i,y_i)_{i\in\mathfrak s};
\chi^\sigma
\bigr).
\end{equation}
Note that here $z\in V^{\sigma}_{\mathrm{ext}}(T)$.
\end{corollary}
\begin{proof}
By definition of the decorated amplitudes, one has
\[
\mathcal A_{l,n}^{\Lambda,\Lambda_0}(T^{(2)}_{z'})
\bigl(
z,z';(\tau_i,y_i)_{i\in\mathfrak s};\chi^\sigma
\bigr)
=
\sum_{\vec k\in\Sigma(\sigma)}
\mathcal A_{l,n}^{\Lambda,\Lambda_0}(T^{(2)}_{z'})
\bigl(
z,z';(\tau_i,y_i)_{i\in\mathfrak s};\chi^\sigma_{\vec k}
\bigr).
\]
Applying Lemma~\ref{Lemmachi'} to each summand yields terms involving either
decorations of the form $\chi^\sigma_{\vec k+\mathbf 1_a}$ and
$\chi^\sigma_{\vec k+\mathbf 1_b}$, or terms with derivative insertions,
namely
\[
\chi'_r\,\chi^\sigma_{\vec k-\mathbf 1_r+\mathbf 1_a},
\qquad
\chi'_r\,\chi^\sigma_{\vec k-\mathbf 1_r+\mathbf 1_b}.
\]
Summing over $\vec k\in\Sigma(\sigma)$ preserves this structure. The terms
containing derivatives of the cutoff are controlled by
Lemma~\ref{LemmeDeriv}, which yields a factor
\[
\sum_{i\in\mathfrak s}\tau_i^{-1/2}.
\]
Integrating in $z'$ and using \eqref{ADRT} from Lemma~\ref{lem:multi-root}
then yields \eqref{equation41-new}.\\
We now turn to the proof of \eqref{equationNuevo}. Recall the definition of
the $\wedge$--amplitude:
\begin{multline}\label{bobo}
\hat{\mathcal A}^{\Lambda,\Lambda_0}_{l,n}(T)
\bigl(
z;(\tau_i,y_i)_{i\in\mathfrak{s}};
\chi^\sigma_{\vec k}
\bigr)
:=\sum_{z'\in V_{\mathrm{ext}}^\sigma(T)}
\int_{\mathbb R} dz'~
\Delta^{(1)}_{z,z'}\chi~
\mathcal A^{\Lambda,\Lambda_0}_{l,n}
\bigl(
T^{(2)}_{z'}
\bigr)
\bigl(
z,z';(\tau_i,y_i)_{i\in\mathfrak{s}};
\chi^{\sigma}_{\vec k^{(z')}}
\bigr)\\
+\sum_{z',z''\in V_{\mathrm{ext}}^\sigma(T)}
\int_{\mathbb R} dz'\int_{\mathbb{R}}dz''~
\Delta^{(1)}_{z',z''}\chi
\Bigl(
\mathcal A^{\Lambda,\Lambda_0}_{l,n}
\bigl(
T^{(3)}_{z'z''}
\bigr)
\bigl(
z,z',z'';(\tau_i,y_i)_{i\in\mathfrak{s}};
\chi^{\sigma}_{\vec k^{(z')}}
\bigr)\\
+
\mathcal A^{\Lambda,\Lambda_0}_{l,n}
\bigl(
T^{(3)}_{z'z''}
\bigr)
\bigl(
z,z',z'';(\tau_i,y_i)_{i\in\mathfrak{s}};
\chi^{\sigma}_{\vec k^{(z'')}}
\bigr)
\Bigr).
\end{multline}
Integrating over $z$ in \eqref{bobo}, we can rewrite it as
\begin{multline}
\int_{\mathbb{R}}dz~
\hat{\mathcal A}^{\Lambda,\Lambda_0}_{l,n}(T)
\bigl(
z;(\tau_i,y_i)_{i\in\mathfrak{s}};
\chi^\sigma_{\vec k}
\bigr)
=\sum_{z'\in V_{\mathrm{ext}}^\sigma(T)}
\int_{\mathbb{R}}dz\int_{\mathbb R} dz'~
\Delta^{(1)}_{z,z'}\chi~
\mathcal A^{\Lambda,\Lambda_0}_{l,n}
\bigl(
T^{(2)}_{z'}
\bigr)
\bigl(
z,z';(\tau_i,y_i)_{i\in\mathfrak{s}};
\chi^{\sigma}_{\vec k^{(z')}}
\bigr)\\
+\sum_{z',z''\in V_{\mathrm{ext}}^\sigma(T)}
\int_{\mathbb R} dz''\int_{\mathbb{R}}dz'~
\Delta^{(1)}_{z',z''}\chi
\Bigl(
\mathcal A^{\Lambda,\Lambda_0}_{l,n}
\bigl(
\hat{T}^{(2)}_{z'z''}
\bigr)
\bigl(
z',z'';(\tau_i,y_i)_{i\in\mathfrak{s}};
\chi^{\sigma}_{\vec k^{(z')}}
\bigr)\\
+
\mathcal A^{\Lambda,\Lambda_0}_{l,n}
\bigl(
\hat{T}^{(2)}_{z'z''}
\bigr)
\bigl(
z',z'';(\tau_i,y_i)_{i\in\mathfrak{s}};
\chi^{\sigma}_{\vec k^{(z'')}}
\bigr)
\Bigr),
\end{multline}
where $\hat{T}$ denotes the tree obtained from $T$ by converting the root
vertex $z$ into an internal vertex.

Summing \eqref{equation41-new} over all pairs $z\in V_{\mathrm{ext}}^\sigma(T)$ and all pairs $(z',z'')$ in
$V_{\mathrm{ext}}^\sigma(T)\times V_{\mathrm{ext}}^\sigma(T)$, we obtain
\[
\int_{\mathbb R}dz~
\hat{\mathcal A}_{l,n}^{\Lambda,\Lambda_0}
\bigl(
T
\bigr)
\bigl(
z;(\tau_i,y_i)_{i\in\mathfrak s};\chi^\sigma
\bigr)
\lesssim
\left(\sum_{i\in\mathfrak s}\tau_i^{-1/2}\right)
\int_{\mathbb R}dz~
\mathcal A_{l',n'}^{\Lambda,\Lambda_0}
\bigl(
T
\bigr)
\bigl(
z;(\tau_i,y_i)_{i\in\mathfrak s};\chi^\sigma
\bigr),
\]
which is precisely \eqref{equationNuevo}.
\end{proof}
\begin{remark}
The condition
\(
z\in V^\sigma_{\mathrm{ext}}(T)
\)
is not needed in order to obtain bounds on the corresponding
\(\hat{\mathcal A}\)-amplitudes. Indeed, the argument of
Corollary~\ref{CorTransport} can be adapted to the case
\(
z\notin V^\sigma_{\mathrm{ext}}(T)
\),
yielding estimates of the same type.

The role of this assumption is instead to preserve the cutoff
bookkeeping encoded by the decoration. Recall that the localisation
factor
\(
\Delta^{(1)}\chi
\)
is transferred to one of the distinguished vertices and produces
decorations of the form
\(
\chi^\sigma_{\vec k+\mathbf 1_a}
\)
or
\(
\chi^\sigma_{\vec k+\mathbf 1_b}.
\)
Since the initial decoration satisfies
\[
\sum_{i\in\sigma}k_i=|\sigma|-1,
\]
the additional unit generated by the localisation restores the total
cutoff weight to
\[
\sum_{i\in\sigma}(k_i+\delta_{ia})
=
\sum_{i\in\sigma}(k_i+\delta_{ib})
=
|\sigma|.
\]

When one of the distinguished vertices belongs to
\(
V^\sigma_{\mathrm{ext}}(T),
\)
the localisation factor can therefore be absorbed into an existing
cutoff decoration while preserving the total number of cutoff powers.
As a consequence, after reduction of the \(\hat{\mathcal A}\)-amplitude,
one recovers an \(\mathcal A\)-amplitude carrying exactly the same
cutoff weight as the original decoration \(\chi^\sigma\). This
conservation property will be essential in the inductive estimates
below.

Combining \eqref{Transport3} with Lemma~\ref{LemmeDeriv} and \eqref{eq:borneamplitudechi}, one obtains
that if
\(
z\notin V^\sigma_{\mathrm{ext}}(T),
\)
then
\begin{equation}\label{Nuevo2}
\int_{\mathbb R}dz~
\hat{\mathcal A}_{l,n}^{\Lambda,\Lambda_0}
\bigl(
T
\bigr)
\bigl(
z;(\tau_i,y_i)_{i\in\mathfrak s};
\chi^\sigma
\bigr)
\lesssim
\|\chi\|_\infty^{|\sigma|}
\left(
\sum_{i\in\mathfrak s}\tau_i^{-1/2}
\right)
\int_{\mathbb R}dz~
\mathcal A_{l',n'}^{\Lambda,\Lambda_0}
\bigl(
T
\bigr)
\bigl(
z;(\tau_i,y_i)_{i\in\mathfrak s}
\bigr).
\end{equation}
In this case, the estimate remains valid, but the cutoff weight is no
longer preserved.
\end{remark}
\subsubsection{Reduction of amplitudes}
We now analyse the effect of the reduction map on tree amplitudes.
This operation corresponds to the linear part of the flow equation,
where two external lines are integrated at a common point. At the level
of tree amplitudes, the two associated heat kernels are merged by the
semigroup property. This produces a factor of order $\Lambda$ and
replaces the original tree by its reduced version.

\begin{lemma}\label{lemm2}
Let $s\ge2$ and let $\mathfrak{s}\subset\mathbb N$ be such that $|\mathfrak{s}|=s-1$.
Let $\sigma\subset\mathfrak{s}$ denote the subset of indices corresponding to
external edges carrying cutoff decorations. Fix
$\vec k=(k_i)_{i\in\sigma}\in\mathbb N^{|\sigma|}$.

Let $T\in\mathcal{T}^{s+2}_3$ be a tree rooted at the vertices
$z_1,z_2,z_3$, and assume that
\[
z_2,z_3\in V^{\sigma}_{\mathrm{ext}}(T),
\]
that is, the vertices $z_2,z_3$ are incident to external edges labelled by
indices in $\sigma$. Then one has
\begin{multline}\label{eq:lemm2-three-root}
\int_{\mathbb R}
\mathcal{A}^{\Lambda,\Lambda_0}_{l-1,n+2}(T)
\Bigl(
z_1,z_2,z_3;
(\tau_i,y_i)_{i\in\mathfrak s},
\bigl(\tfrac{1}{2\Lambda^2},u\bigr),
\bigl(\tfrac{1}{2\Lambda^2},u\bigr);
\chi^\sigma_{\vec k}
\Bigr)\,du
\\
\lesssim
\Lambda~
\mathcal{A}^{\Lambda,\Lambda_0}_{l,n}\bigl(\mathcal{R}_{u,u}(T)\bigr)
\Bigl(
z_1,z_2,z_3;
(\tau_i,y_i)_{i\in\mathfrak s};
\chi^\sigma_{\vec k}
\Bigr).
\end{multline}
The implicit constants are independent of $\Lambda$ and $\Lambda_0$.
\end{lemma}

\begin{proof}
The proof consists of three steps. We first merge the two external heat
kernels by the semigroup property. We then extract the resulting factor
of order $\Lambda$. Finally, we identify the remaining amplitude with
the amplitude of the reduced tree.

By Definition~\eqref{AmT}, the amplitude reads
\begin{multline}\label{eq:lemm2:def}
\mathcal{A}_{l-1,n+2}^{\Lambda,\Lambda_0}(T)
\Bigl(
z_1,z_2,z_3;
(\tau_i,y_i)_{i\in\mathfrak s},
\bigl(\tfrac{1}{2\Lambda^2},u\bigr),
\bigl(\tfrac{1}{2\Lambda^2},u\bigr);
\chi^\sigma_{\vec k}
\Bigr)
\\
=
\sup_{\Lambda\le \Lambda_{ij}\le\Lambda_0}
\int_{\mathbb R^{r-3}} d\vec z_{4,r}\;
\prod_{\langle i,j\rangle\in\mathcal I(T)}
K\!\left(\tfrac{\delta_{l-1,n+2}}{\Lambda_{ij}^2};z_i,z_j\right)
\prod_{\substack{\langle i,j\rangle\in\mathcal E(T)\\ j\in\sigma}}
\chi^{k_j}(z_i)\,K(\delta_{l-1,n+2}\tau_j;z_i,y_j)
\\
\times
\prod_{\substack{\langle i,j\rangle\in\mathcal E(T)\\ j\notin\sigma,\,j\in\mathfrak s}}
K(\delta_{l-1,n+2}\tau_j;z_i,y_j)
\;
K\!\left(\tfrac{\delta_{l-1,n+2}}{2\Lambda^2};z_k,u\right)
K\!\left(\tfrac{\delta_{l-1,n+2}}{2\Lambda^2};z_p,u\right).
\end{multline}
Since
\[
\delta_{l-1,n+2}
=
2-\frac{1}{2^{\,l-1}(n+2)}
\le
2-\frac{1}{2^l n}
=
\delta_{l,n},
\]
the monotonicity of the heat kernel yields
\begin{equation}\label{eq:lemm2:mon}
K(\delta_{l-1,n+2}\tau;z,z')
\lesssim
K(\delta_{l,n}\tau;z,z'),
\end{equation}
uniformly in $z,z'\in\mathbb R$. Hence every factor involving
$\delta_{l-1,n+2}$ can be bounded by the corresponding one with
$\delta_{l,n}$.
Integrating over $u$ and using the semigroup property of the heat kernel
gives
\begin{equation}\label{eq:lemm2:semigroup}
\int_{\mathbb R}
K\!\left(\tfrac{\delta_{l-1,n+2}}{2\Lambda^2};z_k,u\right)
K\!\left(\tfrac{\delta_{l-1,n+2}}{2\Lambda^2};z_p,u\right)\,du
=
K\!\left(\tfrac{\delta_{l-1,n+2}}{\Lambda^2};z_k,z_p\right).
\end{equation}
The resulting kernel has spatial scale $\Lambda^{-1}$ and therefore satisfies
\begin{equation}\label{eq:lemm2:Lambda}
\sup_{z_k,z_p\in\mathbb R}
K\!\left(\tfrac{\delta_{l-1,n+2}}{\Lambda^2};z_k,z_p\right)
\lesssim
\Lambda .
\end{equation}

Extracting this factor $\Lambda$ corresponds precisely to replacing the two
external edges attached to the integration variable $u$ by a single internal
edge connecting $z_k$ and $z_p$.\\
\noindent If this operation creates a vertex of valency one, say $z_j$ attached only
to a vertex $z_i$, then the corresponding integration yields
\begin{equation}\label{eq:lemm2:mass}
\int_{\mathbb R}
K\!\left(\tfrac{\delta_{l,n}}{\Lambda_{ij}^2};z_i,z_j\right)\,dz_j
=1,
\end{equation}
by conservation of mass of the heat kernel.

Observe that none of the vertices removed in this way can carry a cutoff
factor $\chi^{k'}(z_j)$. Indeed such a factor can only arise from an
external edge $\left( z_j,y_i\right)$ with $i\in\sigma$. In that case the
vertex $z_j$ remains incident to this external edge and therefore cannot
become of valency one during the reduction procedure $\mathcal R_{u,u}$.
Consequently
\[
V^{\sigma}_{\mathrm{ext}}\bigl(T\bigr)
=
V^{\sigma}_{\mathrm{ext}}\bigl(\mathcal{R}_{u,u}(T)\bigr).
\]

Iterating \eqref{eq:lemm2:mass} therefore removes only vertices attached
exclusively to (internal) heat-kernel factors and produces exactly the reduced tree
$\mathcal R_{u,u}(T)$.\\
Collecting the above bounds and using \eqref{eq:lemm2:mon} for the remaining
kernel factors yields \eqref{eq:lemm2-three-root}.
\end{proof}
\begin{remark}
The same argument applies to doubly rooted and one rooted trees.
In particular, for every $\tilde T\in\mathcal T_2^{s+2}$ and every
$T\in\mathcal T_1^{s+2}$ one obtains the corresponding estimates
\begin{multline}\label{eq:lemm2-two-root}
\int_{\mathbb R}
\mathcal{A}^{\Lambda,\Lambda_0}_{l-1,n+2}(\tilde T)
\Bigl(
z_1,z_2;
(\tau_i,y_i)_{i\in\mathfrak s},
\bigl(\tfrac{1}{2\Lambda^2},u\bigr),
\bigl(\tfrac{1}{2\Lambda^2},u\bigr);
\chi^\sigma_{\vec k}
\Bigr)\,du
\\
\lesssim
\Lambda~
\mathcal{A}^{\Lambda,\Lambda_0}_{l,n}\bigl(\mathcal{R}_{u,u}(\tilde T)\bigr)
\Bigl(
z_1,z_2;
(\tau_i,y_i)_{i\in\mathfrak s};
\chi^\sigma_{\vec k}
\Bigr).
\end{multline}
and
\begin{multline}\label{RedRooted}
\int_{\mathbb R}
\mathcal{A}^{\Lambda,\Lambda_0}_{l-1,n+2}(T)
\Bigl(
z_1;
(\tau_i,y_i)_{i\in\mathfrak s},
\bigl(\tfrac{1}{2\Lambda^2},u\bigr),
\bigl(\tfrac{1}{2\Lambda^2},u\bigr);
\chi^\sigma
\Bigr)\,du
\\
\lesssim
\Lambda\,
\mathcal{A}^{\Lambda,\Lambda_0}_{l,n}\bigl(\mathcal{R}_{u,u}(T)\bigr)
\Bigl(
z_1;
(\tau_i,y_i)_{i\in\mathfrak s};
\chi^\sigma
\Bigr).
\end{multline}.
\end{remark}
\begin{corollary}\label{Co1}
Let $s\ge2$ and let $\mathfrak{s}\subset\mathbb N$ with $|\mathfrak{s}|=s-1$.
For every $T\in\mathcal T_1^s$ one has
\begin{multline}\label{borneThat}
\int_{\mathbb R}
\hat{\mathcal A}^{\Lambda,\Lambda_0}_{l-1,n+2}(T)
\left(
z_1;(\tau_i,y_i)_{i\in\mathfrak s},
\bigl(\tfrac{1}{2\Lambda^2},u\bigr),
\bigl(\tfrac{1}{2\Lambda^2},u\bigr);\chi^\sigma
\right)du
\\
\lesssim
\Lambda~
\hat{\mathcal A}^{\Lambda,\Lambda_0}_{l,n}\!\left(\mathcal R_{u,u}(T)\right)
\left(
z_1;(\tau_i,y_i)_{i\in\mathfrak s};\chi^\sigma
\right).
\end{multline}
\end{corollary}

\begin{proof}
Recall from the definition of the $\hat{\mathcal A}$–amplitude that
\begin{multline}\label{Co1:def-wedge}
\hat{\mathcal A}^{\Lambda,\Lambda_0}_{l-1,n+2}(T)
\left(
z_1;(\tau_i,y_i)_{i\in\mathfrak s},
\bigl(\tfrac{1}{2\Lambda^2},u\bigr),
\bigl(\tfrac{1}{2\Lambda^2},u\bigr);\chi^\sigma
\right)
\\
=
\sum_{z\in V^\sigma_{\mathrm{ext}}(T)}
\sum_{\vec{k}\in\Sigma(\sigma)}
\int_{\mathbb R}dz\,
\Bigl(\int_0^1 \chi'\!\bigl(tz_1+(1-t)z\bigr)\,dt\Bigr)
\mathcal A^{\Lambda,\Lambda_0}_{l-1,n+2}
\!\left(T^{(2)}_z\right)
\left(
z_1,z;
(\tau_i,y_i)_{i\in\mathfrak s},
\bigl(\tfrac{1}{2\Lambda^2},u\bigr),
\bigl(\tfrac{1}{2\Lambda^2},u\bigr);
\chi^{\sigma}_{\vec k^{(z)}}
\right)
\\
+
\sum_{z,z'\in V^\sigma_{\mathrm{ext}}(T)}
\sum_{\vec{k}\in\Sigma(\sigma)}
\int_{\mathbb R}dz\int_{\mathbb R}dz'\,
\Bigl(\int_0^1 \chi'\!\bigl(tz+(1-t)z'\bigr)\,dt\Bigr)
\\
\times
\Bigl(
\mathcal A^{\Lambda,\Lambda_0}_{l-1,n+2}
\!\left(T^{(3)}_{z,z'}\right)
\left(
z_1,z,z';
(\tau_i,y_i)_{i\in\mathfrak s},
\bigl(\tfrac{1}{2\Lambda^2},u\bigr),
\bigl(\tfrac{1}{2\Lambda^2},u\bigr);
\chi^{\sigma}_{\vec k^{(z)}}
\right)
\\
+
\mathcal A^{\Lambda,\Lambda_0}_{l-1,n+2}
\!\left(T^{(3)}_{z,z'}\right)
\left(
z_1,z,z';
(\tau_i,y_i)_{i\in\mathfrak s},
\bigl(\tfrac{1}{2\Lambda^2},u\bigr),
\bigl(\tfrac{1}{2\Lambda^2},u\bigr);
\chi^{\sigma}_{\vec k^{(z')}}
\right)
\Bigr).
\end{multline}
We integrate \eqref{Co1:def-wedge} with respect to $u$.
Applying Lemma~\ref{lemm2} to the twice-rooted trees
$T^{(2)}_z\in\mathcal T^{s+2}_2$ yields
\begin{multline}\label{Co1:apply-lemm2}
\int_{\mathbb R}
\mathcal A^{\Lambda,\Lambda_0}_{l-1,n+2}
\!\left(T^{(2)}_z\right)
\left(
z_1,z;
(\tau_i,y_i)_{i\in\mathfrak s},
\bigl(\tfrac{1}{2\Lambda^2},u\bigr),
\bigl(\tfrac{1}{2\Lambda^2},u\bigr);
\chi^{\sigma}_{\vec k^{(z)}}
\right)du
\\
\lesssim
\Lambda~
\mathcal A^{\Lambda,\Lambda_0}_{l,n}
\!\left(\mathcal R_{u,u}(T^{(2)}_z)\right)
\left(
z_1,z;
(\tau_i,y_i)_{i\in\mathfrak s};
\chi^{\sigma}_{\vec k^{(z)}}
\right).
\end{multline}
Similarly, applying Lemma~\ref{lemm2} to the three-rooted trees
$T^{(3)}_{z,z'}\in\mathcal T^{s+2}_3$ gives
\begin{multline}
\int_{\mathbb R}
\mathcal A^{\Lambda,\Lambda_0}_{l-1,n+2}
\!\left(T^{(3)}_{z,z'}\right)
\left(
z_1,z,z';
(\tau_i,y_i)_{i\in\mathfrak s},
\bigl(\tfrac{1}{2\Lambda^2},u\bigr),
\bigl(\tfrac{1}{2\Lambda^2},u\bigr);
\chi^{\sigma}_{\vec k^{(\cdot)}}
\right)du
\\
\lesssim
\Lambda~
\mathcal A^{\Lambda,\Lambda_0}_{l,n}
\!\left(\mathcal R_{u,u}(T^{(3)}_{z,z'})\right)
\left(
z_1,z,z';
(\tau_i,y_i)_{i\in\mathfrak s};
\chi^{\sigma}_{\vec k^{(\cdot)}}
\right).
\end{multline}
Using the commutation between reduction and rooting,
\begin{equation}\label{Co1:root-red}
\mathcal R_{u,u}(T^{(2)}_z)
=
(\mathcal R_{u,u}(T))^{(2)}_z,
\qquad
\mathcal R_{u,u}(T^{(3)}_{z,z'})
=
(\mathcal R_{u,u}(T))^{(3)}_{z,z'},
\end{equation}
together with the invariance of the set of external vertices
\begin{equation}\label{Co1:Vext}
V^\sigma_{\mathrm{ext}}(T)
=
V^\sigma_{\mathrm{ext}}(\mathcal R_{u,u}(T)),
\end{equation}
we obtain
\[
\int_{\mathbb R}
\hat{\mathcal A}^{\Lambda,\Lambda_0}_{l-1,n+2}(T)\,du
\lesssim
\Lambda~
\hat{\mathcal A}^{\Lambda,\Lambda_0}_{l,n}
\!\left(\mathcal R_{u,u}(T)\right),
\]
which proves \eqref{borneThat}.
\end{proof}

\begin{corollary}[Reduction of global amplitudes]\label{CoRe}
Fix $n\ge2$ and $l\ge1$. Let $\mathfrak{s}\subset\mathbb N$ with
$s:=|\mathfrak{s}|+1$, and let $\sigma\subseteq\mathfrak{s}$.
Then one has
\begin{align}
\int_{\mathbb R}
\mathcal A_{l-1,n+2}^{\Lambda,\Lambda_0}\!\left(
z_1;(\tau_i,y_i)_{i\in\mathfrak s},
\left(\tfrac{1}{2\Lambda^2},u\right),
\left(\tfrac{1}{2\Lambda^2},u\right);
\chi^\sigma
\right)\,du
&\lesssim
\Lambda~
\mathcal A_{l,n}^{\Lambda,\Lambda_0}\!\left(
z_1;(\tau_i,y_i)_{i\in\mathfrak s};
\chi^\sigma
\right),
\label{Redu1}
\end{align}
and 
\begin{multline}
    \int_{\mathbb R}
\hat{\mathcal A}_{l-1,n+2}^{\Lambda,\Lambda_0}\!\left(
z_1;(\tau_i,y_i)_{i\in\mathfrak s},
\left(\tfrac{1}{2\Lambda^2},u\right),
\left(\tfrac{1}{2\Lambda^2},u\right);\chi^\sigma
\right)\,du\lesssim
\Lambda~
\hat{\mathcal A}_{l,n}^{\Lambda,\Lambda_0}\!\left(
z_1;(\tau_i,y_i)_{i\in\mathfrak s};\chi^\sigma
\right).
\label{Redu2}
\end{multline}
The implicit constants are independent of $\Lambda$ and $\Lambda_0$.
\end{corollary}
\begin{proof}
We prove \eqref{Redu1}. The proof of \eqref{Redu2} is identical. Recall that
$|\mathfrak{s}|=s-1$ and that, by definition of the global amplitude at
order $l-1$,
\begin{multline}\label{eq:CoRe:treeexp}
\mathcal A_{l-1,n+2}^{\Lambda,\Lambda_0}\!\left(
z_1;(\tau_i,y_i)_{i\in\mathfrak{s}},
\left(\tfrac{1}{2\Lambda^2},u\right),
\left(\tfrac{1}{2\Lambda^2},u\right);
\chi^\sigma
\right)
\\
=
\sum_{\substack{
T\in\mathcal T^{s+2}_1\\
0\le v(T)\le l+\frac{n}{2}-2}}
\mathcal A_{l-1,n+2}^{\Lambda,\Lambda_0}(T)\!\left(
z_1;(\tau_i,y_i)_{i\in\mathfrak{s}},
\left(\tfrac{1}{2\Lambda^2},u\right),
\left(\tfrac{1}{2\Lambda^2},u\right);
\chi^\sigma
\right).
\end{multline}
Integrating \eqref{eq:CoRe:treeexp} with respect to $u$ and applying
\eqref{RedRooted} to each summand yields
\begin{multline}\label{eq:CoRe:afterCo1}
\int_{\mathbb R}
\mathcal A_{l-1,n+2}^{\Lambda,\Lambda_0}\!\left(
z_1;(\tau_i,y_i)_{i\in\mathfrak{s}},
\left(\tfrac{1}{2\Lambda^2},u\right),
\left(\tfrac{1}{2\Lambda^2},u\right);
\chi^\sigma
\right)du
\\
\lesssim
\Lambda
\sum_{\substack{
T\in\mathcal T^{s+2}_1\\
0\le v(T)\le l+\frac{n}{2}-2}}
\mathcal A_{l,n}^{\Lambda,\Lambda_0}\!\left(\mathcal R_{u,u}(T)\right)\!\left(
z_1;(\tau_i,y_i)_{i\in\mathfrak{s}};
\chi^\sigma
\right).
\end{multline}

By construction, the reduction $\mathcal R_{u,u}$ removes the two
external lines attached at $\{u,u\}$ and then eliminates all induced
valency-one vertices. Consequently,
$\mathcal R_{u,u}(T)\in\mathcal T^{s}_1$ whenever
$T\in\mathcal T^{s+2}_1$.

For fixed \(l,n\), \(s\), and \(\tilde T\in\mathcal T_1^s\), the number
of trees \(T\in\mathcal T_1^{s+2}\) satisfying
\[
v(T)\le l+\frac n2-2,
\qquad
\mathcal R_{u,u}(T)=\tilde T,
\]
is finite. Since the number of internal vertices is bounded by
\[
v(T)\le l+\frac n2-2,
\]
the length of every chain of valency-one vertices inserted before
reduction is uniformly bounded. Hence there
exists a constant \(C_{l,n,s}\) such that
\begin{equation}\label{eq:CoRe:multiplicity}
\#\Bigl\{
T\in\mathcal T^{s+2}_1:
v(T)\le l+\frac n2-2,\ 
\mathcal R_{u,u}(T)=\tilde T
\Bigr\}
\le C_{l,n,s}.
\end{equation}
Rewriting the sum according to the value of the reduced tree yields
\[
\sum_T
\mathcal A_{l,n}\!\left(\mathcal R_{u,u}(T)\right)
=
\sum_{\tilde T}
\#\Bigl\{
T:\mathcal R_{u,u}(T)=\tilde T
\Bigr\}
\mathcal A_{l,n}(\tilde T).
\]
Applying \eqref{eq:CoRe:multiplicity} gives 
\begin{equation}\label{eq:CoRe:reindex}
\sum_{\substack{
T\in\mathcal T^{s+2}_1\\
0\le v(T)\le l+\frac{n}{2}-2}}
\mathcal A_{l,n}^{\Lambda,\Lambda_0}\!\left(\mathcal R_{u,u}(T)\right)\!\left(
z_1;(\tau_i,y_i)_{i\in\mathfrak{s}};
\chi^\sigma
\right)
\lesssim
\sum_{\substack{
\tilde T\in\mathcal T^{s}_1\\
0\le v(\tilde T)\le l+\frac{n}{2}-2}}
\mathcal A_{l,n}^{\Lambda,\Lambda_0}(\tilde T)\!\left(
z_1;(\tau_i,y_i)_{i\in\mathfrak{s}};
\chi^\sigma
\right).
\end{equation}
The right-hand side is precisely the definition of
\[
\mathcal A_{l,n}^{\Lambda,\Lambda_0}\!\left(
z_1;(\tau_i,y_i)_{i\in\mathfrak{s}};
\chi^\sigma
\right),
\]
and inserting \eqref{eq:CoRe:reindex} into
\eqref{eq:CoRe:afterCo1} yields \eqref{Redu1}.
\end{proof}
\begin{remark}
    The same result holds for the twice-rooted global amplitudes. Namely we have 
\begin{align}
\int_{\mathbb R}
\mathcal A_{l-1,n+2}^{\Lambda,\Lambda_0}\!\left(
z_1,z_2;(\tau_i,y_i)_{i\in\mathfrak s},
\left(\tfrac{1}{2\Lambda^2},u\right),
\left(\tfrac{1}{2\Lambda^2},u\right)
\right)\,du
&\lesssim
\Lambda~
\mathcal A_{l,n}^{\Lambda,\Lambda_0}\!\left(
z_1,z_2;(\tau_i,y_i)_{i\in\mathfrak s}
\right).
\label{DouRed}
\end{align}
\end{remark}
\subsubsection{Fusion of tree amplitudes}
The quadratic part of the flow equation involves products of two
correlators connected by a covariance kernel. At the level of tree
amplitudes, this operation joins two rooted trees through an additional
internal edge and produces a new tree amplitude. The purpose of this
subsection is to formalize this construction and to establish the
corresponding bounds.

To describe the distribution of the external variables among the two
factors, we introduce the following notation: let $n\ge2$ and set
\[
\mathfrak I_n:=\{2,\dots,n\}.
\]
Let $\pi=\{\pi_1,\pi_2\}$ be a partition of $\mathfrak I_n$.
Given $\mathfrak s\subseteq\mathfrak I_n$ and $\sigma\subseteq\mathfrak s$
we define
\[
\mathfrak s_i(\pi):=\mathfrak s\cap\pi_i,
\qquad
\sigma_i(\pi):=\sigma\cap\pi_i,
\qquad
s_i:=|\mathfrak s_i(\pi)|.
\]
In the sequel we slightly abuse notation and write
\[
\mathfrak s_i:=\mathfrak s_i(\pi),
\qquad
\sigma_i:=\sigma_i(\pi).
\]
We also introduce the difference operators
\[
\Delta_{z_1}\chi(z):=\chi(z)-\chi(z_1),
\qquad
\Delta^{(1)}_{z,z'}\chi
:=
\int_0^1\chi'\left(tz+(1-t)z'\right)\,dt,
\]
as well as the difference of heat kernels
\begin{equation}
\Delta K_{\tau,y}(z_1,z_2)
:=
K(\tau;z_1,y)-K(\tau;z_2,y).
\end{equation}
Let $T_1\in\mathcal T^{s_1+1}_1$ and $T_2\in\mathcal T^{s_2+1}_1$. 
For amplitudes $\mathcal B_1,\mathcal B_2$ of the form
$\mathcal A$ or $\hat{\mathcal A}$ we define the fusion product
\begin{multline}\label{fusion-general}
\mathcal B_1(T_1)
\otimes_\pi
\mathcal B_2(T_2)
\Bigl(
z_1;
(\tau_i,y_i)_{i\in\mathfrak s};
\chi^\sigma
\Bigr)
\\
:=
\int_{\mathbb R}du\;
\mathcal B_{1}^{\Lambda,\Lambda_0}(T_1)
\Bigl(
z_1;
(\tau_i,y_i)_{i\in\mathfrak s_1},
(\tfrac{1}{2\Lambda^2},u);
\chi^{\sigma_1}
\Bigr)
\\
\times
\int_{\mathbb R}dz\;
\mathcal B_{2}^{\Lambda,\Lambda_0}(T_2)
\Bigl(
z;
(\tau_i,y_i)_{i\in\mathfrak s_2},
(\tfrac{1}{2\Lambda^2},u);
\chi^{\sigma_2}
\Bigr).
\end{multline}
The covariance connecting the two factors is represented by the common
integration variable \(u\), while the second root variable \(z\) is
integrated against the second amplitude.
If the label $\sigma$ is omitted on the left-hand side,
this means that $\chi\equiv\mathds 1$ is taken in the
right-hand side of \eqref{fusion-general}. If a weight \(W(z_1,z)\) is present (for example
\(\Delta_{z_1}\chi(z)\) or \(\Delta K_{\tau,y}(z_1,z)\)),
the fusion product is defined by inserting this factor into the second
integral in \eqref{fusion-general}. This convention allows us to treat
the various remainder terms generated by Taylor expansions within the
same formalism.\\
\paragraph{{Fusion of Tree Amplitudes without insertions}}
The purpose of this part is to control the fusion product introduced
above. We first prove that ordinary tree amplitudes are stable under
fusion and that the fusion operation is represented combinatorially by
the fusion tree introduced in Definition~\ref{defFusi}. We then extend
this result to amplitudes containing Taylor remainder insertions,
encoded by the hatted amplitudes. The final estimate shows that the
fusion of two hatted amplitudes can again be controlled by a single
hatted amplitude at the expense of the factor
$$
\sum_i \tau_i^{-1/2}+\Lambda
$$
which comes from the reduction of derivative insertions.

We begin with the fusion of ordinary amplitudes. The following result
shows that the fusion product is represented at the level of trees by
the combinatorial fusion operation introduced in
Definition~\ref{defFusi}.

\begin{lemma}[Stability of amplitudes under fusion]\label{lem:FusionBound_tight}
Let $l_1,l_2\in\mathbb N_0$ and let
$T_1\in\mathcal T^{s_1+1}_1$, $T_2\in\mathcal T^{s_2+1}_1$.
Denote by $\mathcal F_u(T_1,T_2)$ the fusion tree introduced in
Definition~\ref{defFusi}, and set
\[
\hat n=\max\left(n_1+1,n_2+1\right),\qquad
l_1\vee l_2:=\max(l_1,l_2).
\]
Then one has
\begin{align}
\mathcal A^{\Lambda,\Lambda_0}_{l_1,n_1+1}(T_1)\otimes_\pi
\mathcal A^{\Lambda,\Lambda_0}_{l_2,n_2+1}(T_2)\left(z_1,z_2;\cdot\right)
&\lesssim
\mathcal A^{\Lambda,\Lambda_0}_{l_1\vee l_2,\hat n}
\bigl(\mathcal F_u(T_1,T_2)^{(2)}_{z_2}\bigr)\left(z_1,z_2;\cdot\right),
\label{eq:FusionTight0}
\\
\mathcal A^{\Lambda,\Lambda_0}_{l_1,n_1+1}(T_1)\otimes_\pi
\mathcal A^{\Lambda,\Lambda_0}_{l_2,n_2+1}(T_2)
&\lesssim
\mathcal A^{\Lambda,\Lambda_0}_{l_1\vee l_2,\hat n}
\bigl(\mathcal F_u(T_1,T_2)\bigr),
\label{eq:FusionTight1}
\\[0.3em]
\hat{\mathcal A}^{\Lambda,\Lambda_0}_{l_1,n_1+1}(T_1)\otimes_\pi
\mathcal A^{\Lambda,\Lambda_0}_{l_2,n_2+1}(T_2)
&\lesssim
\hat{\mathcal A}^{\Lambda,\Lambda_0}_{l_1\vee l_2,\hat n}
\bigl(\mathcal F_u(T_1,T_2)\bigr).
\label{eq:FusionTight5}
\end{align}
\end{lemma}

\begin{proof}
The proof consists of three steps.
First, we enlarge all diffusion parameters to the common scale
\(\delta_{l_1\vee l_2,\hat n}\).
Second, we merge the two fusion kernels using the semigroup property.
Third, we identify the resulting expression with the amplitude
associated with the fused tree. First we prove \eqref{eq:FusionTight0}. By definition of the
$\otimes_\pi$--product,
\begin{multline}\label{eq:FusionRep_tight}
\mathcal A_{l_1,n_1+1}^{\Lambda,\Lambda_0}(T_1)\otimes_\pi
\mathcal A_{l_2,n_2+1}^{\Lambda,\Lambda_0}(T_2)
\left(z_1,z_2;(\tau_i,y_i)_{i\in\mathfrak{s}};\chi^{\sigma}\right)
\\
=
\int_{\mathbb R}du\;
\mathcal A_{l_1,n_1+1}^{\Lambda,\Lambda_0}(T_1)
\Bigl(
z_1;
(\tau_i,y_i)_{i\in\mathfrak s_1},
(\tfrac{1}{2\Lambda^2},u);
\chi^{\sigma_1}
\Bigr)
\\
\times
\mathcal A_{l_2,n_2+1}^{\Lambda,\Lambda_0}(T_2)
\Bigl(
z_2;
(\tau_i,y_i)_{i\in\mathfrak s_2},
(\tfrac{1}{2\Lambda^2},u);
\chi^{\sigma_2}
\Bigr).
\end{multline}
Using the representation of the amplitudes given by the Feynman rules in \eqref{AmT}
and the definition of the decorations in \eqref{Sumkchi}, the integrand in
\eqref{eq:FusionRep_tight} can be written as a sum over pairs of admissible
multi-indices
\[
\vec k^{(1)}\in\Sigma(\sigma_1),
\qquad
\vec k^{(2)}\in\Sigma(\sigma_2).
\]
Fix such a pair. The corresponding term contains all internal and
external kernels of $T_1$ at scale $\delta_{l_1,n_1+1}$, all kernels of
$T_2$ at scale $\delta_{l_2,n_2+1}$, together with the two fusion kernels
\[
K\!\left(\frac{\delta_{l_1,n_1+1}}{2\Lambda^2};z,u\right),
\qquad
K\!\left(\frac{\delta_{l_2,n_2+1}}{2\Lambda^2};z',u\right).
\]
Since $l_i\le l_1\vee l_2$ and $n_i+1\le \hat n$, we obtain
\[
\delta_{l_i,n_i+1}\le \delta_{l_1\vee l_2,\hat n},
\qquad i\in\{1,2\}.
\]
By monotonicity of the heat kernel, every kernel in the integrand can
therefore be bounded by the corresponding one at scale
$\delta_{l_1\vee l_2,\hat n}$. In particular,
\[
K\!\left(\frac{\delta_{l_1,n_1+1}}{2\Lambda^2};z,u\right)
\lesssim
K\!\left(\frac{\delta_{l_1\vee l_2,\hat n}}{2\Lambda^2};z,u\right),
\qquad
K\!\left(\frac{\delta_{l_2,n_2+1}}{2\Lambda^2};z',u\right)
\lesssim
K\!\left(\frac{\delta_{l_1\vee l_2,\hat n}}{2\Lambda^2};z',u\right).
\]
Applying the semigroup identity then yields
\[
\int_{\mathbb R}
K\!\left(\frac{\delta_{l_1\vee l_2,\hat n}}{2\Lambda^2};z,u\right)
K\!\left(\frac{\delta_{l_1\vee l_2,\hat n}}{2\Lambda^2};z',u\right)\,du
=
K\!\left(\frac{\delta_{l_1\vee l_2,\hat n}}{\Lambda^2};z,z'\right),
\]
which is precisely the kernel associated with the new internal edge in
the fused tree $\mathcal F_u(T_1,T_2)$.

The remaining kernels and cutoff factors coincide with those of the
fused tree, with decoration $\chi^\sigma_{\vec k}$, where
$\vec k\in\Sigma(\sigma)$ is obtained by combining
$\vec k^{(1)}$ and $\vec k^{(2)}$ on the disjoint sets
$\sigma_1$ and $\sigma_2$. Hence each summand is bounded by
\[
\mathcal A^{\Lambda,\Lambda_0}_{l_1\vee l_2,\hat n}
\bigl(
\mathcal F_u(T_1,T_2)^{(2)}_{z_2};\chi^\sigma_{\vec k}
\bigr).
\]
Summing over all admissible pairs
$(\vec k^{(1)},\vec k^{(2)})$ yields \eqref{eq:FusionTight0}.\\
The estimate \eqref{eq:FusionTight1} is obtained by integrating
\eqref{eq:FusionTight0} over $z_2$ and using \eqref{ADRT}.\\
The proof of \eqref{eq:FusionTight5} is identical. Expanding the
decorations using \eqref{equation72} and repeating the above argument
shows that the hatted factor $\Delta_{\cdot,\cdot}^{(1)}\chi$ is transported
unchanged through the fusion, which gives
\[
\hat{\mathcal A}^{\Lambda,\Lambda_0}_{l_1,n_1+1}(T_1)\otimes_\pi
\mathcal A^{\Lambda,\Lambda_0}_{l_2,n_2+1}(T_2)
\lesssim
\hat{\mathcal A}^{\Lambda,\Lambda_0}_{l_1\vee l_2,\hat n}
\bigl(\mathcal F_u(T_1,T_2)\bigr),
\]
as claimed.
\end{proof}

\begin{lemma}\label{lem:FusionDeltaChi}
Let $(l_1,l_2)\in\mathbb N_0\times\mathbb N$, $(n_1,n_2)\in\mathbb N^2$
be such that $n_1\ge1$ and $n_2\ge3$, and set $n:=n_1+n_2$.
Let $T_1\in\mathcal T^{s_1+1}_1$ and $T_2\in\mathcal T^{s_2+1}_1$. Then one has
\begin{equation}\label{eq:FusionDeltaChi}
\mathcal A_{l_1,n_1+1}^{\Lambda,\Lambda_0}(T_1)
\otimes_\pi
\hat{\mathcal A}_{l_2,n_2+1}^{\Lambda,\Lambda_0}(T_2)
\lesssim
\hat{\mathcal A}_{l_1\vee l_2,n}^{\Lambda,\Lambda_0}
\bigl(\mathcal F_u(T_1,T_2)\bigr).
\end{equation}
\end{lemma}

\begin{proof}
Recall that
\begin{multline}\label{eq:tensor-hat-compact-proof}
\Bigl(
\mathcal A_{l_1,n_1+1}^{\Lambda,\Lambda_0}(T_1)\otimes_\pi
\hat{\mathcal A}_{l_2,n_2+1}^{\Lambda,\Lambda_0}(T_2)
\Bigr)
\bigl(
z_1;(\tau_i,y_i)_{i\in\mathfrak s};\chi^\sigma
\bigr)
\\
=
\int_{\mathbb R}du\;
\mathcal A_{l_1,n_1+1}^{\Lambda,\Lambda_0}(T_1)
\Bigl(
z_1;
(\tau_i,y_i)_{i\in\mathfrak s_1},
\bigl(\tfrac{1}{2\Lambda^2},u\bigr);
\chi^{\sigma_1}
\Bigr)
\\
\times
\int_{\mathbb R}dz_2\;
\hat{\mathcal A}_{l_2,n_2+1}^{\Lambda,\Lambda_0}(T_2)
\Bigl(
z_2;
(\tau_i,y_i)_{i\in\mathfrak s_2},
\bigl(\tfrac{1}{2\Lambda^2},u\bigr);
\chi^{\sigma_2}
\Bigr).
\end{multline}
We expand both factors according to the definitions of
$\mathcal A$ and $\hat{\mathcal A}$. For the first factor, we write
\begin{multline}\label{eq:proof-fusion-hat:first-factor}
\mathcal A_{l_1,n_1+1}^{\Lambda,\Lambda_0}(T_1)
\bigl(
z_1;
(\tau_i,y_i)_{i\in\mathfrak s_1},
\bigl(\tfrac{1}{2\Lambda^2},u\bigr);
\chi^{\sigma_1}
\bigr)
\\
=
\sum_{\vec k^{(1)}\in\Sigma(\sigma_1)}
\mathcal A_{l_1,n_1+1}^{\Lambda,\Lambda_0}(T_1)
\bigl(
z_1;
(\tau_i,y_i)_{i\in\mathfrak s_1},
\bigl(\tfrac{1}{2\Lambda^2},u\bigr);
\chi^{\sigma_1}_{\vec k^{(1)}}
\bigr).
\end{multline}
For the second factor, using the definition of $\hat{\mathcal A}$ and
separating the contributions in which one of the two distinguished vertices is
the root $z_2$ from those in which both distinguished vertices belong to
$V^{\sigma_2}_{\mathrm{ext}}(T_2)$, we obtain
\begin{multline}\label{eq:proof-fusion-hat:second-factor}
\hat{\mathcal A}_{l_2,n_2+1}^{\Lambda,\Lambda_0}(T_2)
\bigl(
z_2;
(\tau_i,y_i)_{i\in\mathfrak s_2},
\bigl(\tfrac{1}{2\Lambda^2},u\bigr);
\chi^{\sigma_2}
\bigr)
\\
=
\sum_{\vec k^{(2)}\in\Sigma(\sigma_2)}
\sum_{z\in V^{\sigma_2}_{\mathrm{ext}}(T_2)}
\int_{\mathbb R}dz\;
\Delta^{(1)}_{z_2,z}\chi\,
\mathcal A_{l_2,n_2+1}^{\Lambda,\Lambda_0}
\bigl(
T^{(2)}_{2,z}
\bigr)
\bigl(
z_2,z;
(\tau_i,y_i)_{i\in\mathfrak s_2},
\bigl(\tfrac{1}{2\Lambda^2},u\bigr);
\chi^{\sigma_2}_{(\vec k^{(2)})^{(z)}}
\bigr)
\\
\quad+
\sum_{\vec k^{(2)}\in\Sigma(\sigma_2)}
\sum_{z,z'\in V^{\sigma_2}_{\mathrm{ext}}(T_2)}
\int_{\mathbb R}dz\int_{\mathbb R}dz'\;
\Delta^{(1)}_{z,z'}\chi\,
\mathcal A_{l_2,n_2+1}^{\Lambda,\Lambda_0}
\bigl(
T^{(3)}_{2,zz'}
\bigr)
\bigl(
z_2,z,z';
(\tau_i,y_i)_{i\in\mathfrak s_2},
\bigl(\tfrac{1}{2\Lambda^2},u\bigr);
\chi^{\sigma_2}_{(\vec k^{(2)})^{(z,z')}}
\bigr).
\end{multline}
Here $(\vec k^{(2)})^{(z)}$ denotes the decoration obtained from
$\vec k^{(2)}$ by decreasing by one the exponent attached to the vertex $z$,
while $(\vec k^{(2)})^{(z,z')}$ denotes the symmetrised decoration appearing
in the definition of $\hat{\mathcal A}$ when the two distinguished vertices
are $z$ and $z'$. In particular, the vertices $z$ and $z'$ play the same role.

Substituting \eqref{eq:proof-fusion-hat:first-factor} and
\eqref{eq:proof-fusion-hat:second-factor} into
\eqref{eq:tensor-hat-compact-proof}, we obtain a sum of nonnegative terms.
Fix one such term. It is indexed by
$\vec k^{(1)}\in\Sigma(\sigma_1)$, $\vec k^{(2)}\in\Sigma(\sigma_2)$, and
either by a vertex $z\in V^{\sigma_2}_{\mathrm{ext}}(T_2)$ in the first sum of
\eqref{eq:proof-fusion-hat:second-factor}, or by a pair
$z,z'\in V^{\sigma_2}_{\mathrm{ext}}(T_2)$ in the second one.

We first consider a term coming from the first sum in
\eqref{eq:proof-fusion-hat:second-factor}. It has the form
\begin{multline}\label{eq:proof-fusion-hat:first-kind}
\int_{\mathbb R}du\;
\mathcal A_{l_1,n_1+1}^{\Lambda,\Lambda_0}(T_1)
\bigl(
z_1;
(\tau_i,y_i)_{i\in\mathfrak s_1},
\bigl(\tfrac{1}{2\Lambda^2},u\bigr);
\chi^{\sigma_1}_{\vec k^{(1)}}
\bigr)
\\
\times
\int_{\mathbb R}dz_2\int_{\mathbb R}dz\;
\Delta^{(1)}_{z_2,z}\chi\,
\mathcal A_{l_2,n_2+1}^{\Lambda,\Lambda_0}
\bigl(
T^{(2)}_{2,z}
\bigr)
\bigl(
z_2,z;
(\tau_i,y_i)_{i\in\mathfrak s_2},
\bigl(\tfrac{1}{2\Lambda^2},u\bigr);
\chi^{\sigma_2}_{(\vec k^{(2)})^{(z)}}
\bigr).
\end{multline}
We fuse the two amplitudes along the common external variable $u$.
By the fusion estimate for amplitudes, this yields an amplitude on
the fused tree $\mathcal F_u(T_1,T_2)$ with scale index $l_1\vee l_2$.
The appearance of $l_1\vee l_2$ is precisely the content of the
fusion bound: when two amplitudes of scales $l_1$ and $l_2$ are fused, the
resulting internal kernels are controlled at the larger of the two scales.
Hence the fused contribution is one of the terms appearing in the definition of
\[
\hat{\mathcal A}_{l_1\vee l_2,n_1+n_2}^{\Lambda,\Lambda_0}
\bigl(
\mathcal F_u(T_1,T_2)
\bigr).
\]
We next consider a term coming from the second sum in
\eqref{eq:proof-fusion-hat:second-factor}. It has the form
\begin{multline}\label{eq:proof-fusion-hat:second-kind}
\int_{\mathbb R}du\;
\mathcal A_{l_1,n_1+1}^{\Lambda,\Lambda_0}(T_1)
\bigl(
z_1;
(\tau_i,y_i)_{i\in\mathfrak s_1},
\bigl(\tfrac{1}{2\Lambda^2},u\bigr);
\chi^{\sigma_1}_{\vec k^{(1)}}
\bigr)
\\
\times
\int_{\mathbb R}dz_2\int_{\mathbb R}dz\int_{\mathbb R}dz'\;
\Delta^{(1)}_{z,z'}\chi\,
\mathcal A_{l_2,n_2+1}^{\Lambda,\Lambda_0}
\bigl(
T^{(3)}_{2,zz'}
\bigr)
\bigl(
z_2,z,z';
(\tau_i,y_i)_{i\in\mathfrak s_2},
\bigl(\tfrac{1}{2\Lambda^2},u\bigr);
\chi^{\sigma_2}_{(\vec k^{(2)})^{(z,z')}}
\bigr).
\end{multline}
Again, fusing the two ordinary amplitudes along the variable $u$ and applying
the ordinary fusion bound produces an amplitude on
$\mathcal F_u(T_1,T_2)$ with scale index $l_1\vee l_2$. The vertices $z$ and
$z'$ are unaffected by the fusion and remain vertices of
$V^\sigma_{\mathrm{ext}}(\mathcal F_u(T_1,T_2))$. Therefore, the fused term is
again one of the contributions entering the definition of
\[
\hat{\mathcal A}_{l_1\vee l_2,n_1+n_2}^{\Lambda,\Lambda_0}
\bigl(
\mathcal F_u(T_1,T_2)
\bigr).
\]

Thus every term arising in the expansion of
\eqref{eq:tensor-hat-compact-proof} is bounded by a corresponding term in the
definition of
\[
\hat{\mathcal A}_{l_1\vee l_2,n_1+n_2}^{\Lambda,\Lambda_0}
\bigl(\mathcal F_u(T_1,T_2)\bigr).
\]
Since all terms are nonnegative, summing over all
$\vec k^{(1)}\in\Sigma(\sigma_1)$, $\vec k^{(2)}\in\Sigma(\sigma_2)$, over all
$z\in V^{\sigma_2}_{\mathrm{ext}}(T_2)$, and over all
$z,z'\in V^{\sigma_2}_{\mathrm{ext}}(T_2)$ yields
\[
\mathcal A_{l_1,n_1+1}^{\Lambda,\Lambda_0}(T_1)
\otimes_\pi
\hat{\mathcal A}_{l_2,n_2+1}^{\Lambda,\Lambda_0}(T_2)
\lesssim
\hat{\mathcal A}_{l_1\vee l_2,n_1+n_2}^{\Lambda,\Lambda_0}
\bigl(
\mathcal F_u(T_1,T_2)
\bigr),
\]
which is exactly \eqref{eq:FusionDeltaChi}.
\end{proof}
The next estimate treats the mixed fusion
\(\hat{\mathcal A}\otimes_\pi \hat{\mathcal A}\).  Its proof uses the
derivative-insertion estimate of Corollary~\ref{CorTransport} on the
second factor. For this reason we assume that the root of \(T_2\) is
attached to an external vertex \(y_j\) with \(j\in\sigma_2\); equivalently,
the root carries one of the localisation factors. Under this assumption,
the hatted amplitude in the second factor can first be reduced to an
ordinary decorated amplitude, and the mixed fusion bound can then be
applied.
\begin{lemma}\label{LemIm}
Let $(l_1,l_2)\in\mathbb{N}_0^2$ and let $(n_1,n_2)\in\mathbb{N}^2$ such that $n_1,n_2\ge3$.
Let $T_1\in\mathcal T^{s_1+1}_1$ and $T_2\in\mathcal T^{s_2+1}_1$, and assume that the root vertex $z_2$ of $T_2$ is attached to an external vertex
$y_j$ with $j\in\sigma_2$. 
Then one has
\begin{equation}\label{Aa11}
\hat{\mathcal A}_{l_1,n_1+1}^{\Lambda,\Lambda_0}(T_1)
\otimes_\pi
\hat{\mathcal A}_{l_2,n_2+1}^{\Lambda,\Lambda_0}(T_2)
\;\lesssim\;
\left(\sum_{i\in\mathfrak s_2}\tau_i^{-1/2}+\Lambda\right)
\hat{\mathcal A}_{l_1\vee l_2,n_1+n_2}^{\Lambda,\Lambda_0}
\bigl(\mathcal F_u(T_1,T_2)\bigr).
\end{equation}
\end{lemma}
\begin{proof}
By definition of the fusion product with two hatted amplitudes, one has
\begin{multline}\label{eq:tensor-hat-hat}
\hat{\mathcal A}_{l_1,n_1+1}^{\Lambda,\Lambda_0}(T_1)\otimes_\pi
\hat{\mathcal A}_{l_2,n_2+1}^{\Lambda,\Lambda_0}(T_2)
\\
=
\int_{\mathbb R}du\;
\hat{\mathcal A}_{l_1,n_1+1}^{\Lambda,\Lambda_0}(T_1)
\Bigl(
z_1;
(\tau_i,y_i)_{i\in\mathfrak s_1},
(\tfrac{1}{2\Lambda^2},u);\chi^{\sigma_1}
\Bigr)
\\
\times
\int_{\mathbb R}dz_2\;
\hat{\mathcal A}_{l_2,n_2+1}^{\Lambda,\Lambda_0}(T_2)
\Bigl(
z_2;
(\tau_i,y_i)_{i\in\mathfrak s_2},
(\tfrac{1}{2\Lambda^2},u);\chi^{\sigma_2}
\Bigr)~,
\end{multline}
where $z_2\in V^{\sigma_2}_{\mathrm{ext}}(T_2)$. We now apply the bound \eqref{equationNuevo} from Corollary~\ref{CorTransport}
to the second factor. Since the external variable corresponding to the fusion
line carries $
\tau=(2\Lambda^2)^{-1}$,
the contribution $\tau^{-1/2}$ produced by the corollary is bounded by a
constant multiple of $\Lambda$. Hence
\begin{multline}
\int_{\mathbb R}dz_2\;
\hat{\mathcal A}_{l_2,n_2+1}^{\Lambda,\Lambda_0}(T_2)
\Bigl(
z_2;
(\tau_i,y_i)_{i\in\mathfrak s_2},
(\tfrac{1}{2\Lambda^2},u);
\chi^{\sigma_2}
\Bigr)
\\
\lesssim
\left(\sum_{i\in\mathfrak s_2}\tau_i^{-1/2}+\Lambda\right)
\int_{\mathbb R}dz_2\;
\mathcal A_{l_2,n_1+n_2}^{\Lambda,\Lambda_0}(T_2)
\Bigl(
z_2;
(\tau_i,y_i)_{i\in\mathfrak s_2},
(\tfrac{1}{2\Lambda^2},u);
\chi^{\sigma_2}
\Bigr).
\end{multline}
Substituting this bound into \eqref{eq:tensor-hat-hat} yields
\begin{multline}
\hat{\mathcal A}_{l_1,n_1+1}^{\Lambda,\Lambda_0}(T_1)\otimes_\pi
\hat{\mathcal A}_{l_2,n_2+1}^{\Lambda,\Lambda_0}(T_2)
\\
\lesssim
\left(\sum_{i\in\mathfrak s_2}\tau_i^{-1/2}+\Lambda\right)
\int_{\mathbb R}du\;
\hat{\mathcal A}_{l_1,n_1+1}^{\Lambda,\Lambda_0}(T_1)
\Bigl(
z_1;
(\tau_i,y_i)_{i\in\mathfrak s_1},
(\tfrac{1}{2\Lambda^2},u);
\chi^{\sigma_1}
\Bigr)
\\
\times
\int_{\mathbb R}dz_2\;
\mathcal A_{l_2,n_1+n_2}^{\Lambda,\Lambda_0}(T_2)
\Bigl(
z_2;
(\tau_i,y_i)_{i\in\mathfrak s_2},
(\tfrac{1}{2\Lambda^2},u);
\chi^{\sigma_2}
\Bigr).
\end{multline}
The right-hand side is precisely
\[
\left(\sum_{i\in\mathfrak s_2}\tau_i^{-1/2}+\Lambda\right)
\hat{\mathcal A}_{l_1,n_1+1}^{\Lambda,\Lambda_0}(T_1)
\otimes_\pi
\mathcal A_{l_2,n_1+n_2}^{\Lambda,\Lambda_0}(T_2).
\]
Finally, applying the mixed fusion bound \eqref{eq:FusionTight5} from
Lemma~\ref{lem:FusionBound_tight} gives
\[
\hat{\mathcal A}_{l_1,n_1+1}^{\Lambda,\Lambda_0}(T_1)\otimes_\pi
\mathcal A_{l_2,n_1+n_2}^{\Lambda,\Lambda_0}(T_2)
\lesssim
\hat{\mathcal A}_{l_1\vee l_2,n_1+n_2}^{\Lambda,\Lambda_0}
\bigl(\mathcal F_u(T_1,T_2)\bigr),
\]
which proves \eqref{Aa11}.
\end{proof}

\paragraph{Fusion of amplitudes with localisation insertions.}
We now state the fusion estimates needed when a localisation difference
is produced by the Taylor expansion at the root.  The typical term contains
a factor
\[
\Delta_{z_1}\chi(z_2):=\chi(z_1)-\chi(z_2),
\]
where \(z_1\) and \(z_2\) are the two root variables before fusion.  The
point is that whenever one of the localization factors is attached to the
root of the second tree, this difference can be written as
\[
\chi(z_1)-\chi(z_2)
=
\left(z_1-z_2\right)~\Delta^{(1)}_{z_1,z_2}\chi .
\]
The distance factor is then absorbed by the geometric estimate
Lemma~\ref{LemDecPo}, producing a factor \(\Lambda^{-1}\), while the
remaining derivative insertion is exactly of the type encoded by the
hatted amplitudes introduced in Definition~\ref{def:hat-amplitude}.

The following lemma gives the precise form of this mechanism.  We write
\(\sigma_i\) for the subset of localisation-decorated external vertices
belonging to \(T_i\), and after fusion we set
\[
\sigma:=\sigma_1\cup\sigma_2 .
\]
\begin{lemma}\label{lem:Delta-chi-fusion}
Fix $(l_1,l_2)\in\mathbb{N}_0^2$ and
\((n_1,n_2)\in\mathbb N^2\) with \(n_1,n_2\ge2\).
Let $T_1\in\mathcal{T}^{s_1+1}_1$ and $T_2\in\mathcal{T}^{s_2+1}_1$,
and denote by $z_1$ and $z_2$ their respective root variables.
Assume that the root $z_i$ of $T_i$ is adjacent to an external
vertex $y_{a_i}$ with $a_i\in\sigma_i$, $i\in\{1,2\}$.
Set
\[
l:=l_1+l_2,
\qquad
n:=n_1+n_2.
\]
Then the following bounds hold:
\begin{align}
\mathcal{A}_{l_1,n_1+1}^{\Lambda,\Lambda_0}(T_1)\otimes_{\pi}
\Delta_{z_1}\chi\,
\mathcal{A}_{l_2,n_2+1}^{\Lambda,\Lambda_0}(T_2)
&\lesssim
\Lambda^{-1}\,
\hat{\mathcal{A}}_{l,n}^{\Lambda,\Lambda_0}
\bigl(\mathcal{F}_{u}(T_1,T_2)\bigr),
\label{a-108}
\\
\mathcal{A}_{l_1,n_1+1}^{\Lambda,\Lambda_0}(T_1)\otimes_{\pi}
\Delta_{z_1}\chi\,
\hat{\mathcal{A}}_{l_2,n_2+1}^{\Lambda,\Lambda_0}(T_2)
&\lesssim
\hat{\mathcal{A}}_{l,n}^{\Lambda,\Lambda_0}
\bigl(\mathcal{F}_{u}(T_1,T_2)\bigr),
\label{a-109}
\\
\hat{\mathcal{A}}_{l_1,n_1+1}^{\Lambda,\Lambda_0}(T_1)\otimes_{\pi}
\Delta_{z_1}\chi\,
\mathcal{A}_{l_2,n_2+1}^{\Lambda,\Lambda_0}(T_2)
&\lesssim
\hat{\mathcal{A}}_{l,n}^{\Lambda,\Lambda_0}
\bigl(\mathcal{F}_{u}(T_1,T_2)\bigr),
\label{a-110}
\\
\hat{\mathcal{A}}_{l_1,n_1+1}^{\Lambda,\Lambda_0}(T_1)\otimes_{\pi}
\Delta_{z_1}\chi\,
\hat{\mathcal{A}}_{l_2,n_2+1}^{\Lambda,\Lambda_0}(T_2)
&\lesssim
\Bigl(\sum_{i\in\mathfrak{s}_2}\tau_i^{-1/2}+\Lambda\Bigr)
\hat{\mathcal{A}}_{l,n}^{\Lambda,\Lambda_0}
\bigl(\mathcal{F}_{u}(T_1,T_2)\bigr).
\label{a-111}
\end{align}
\end{lemma}
\begin{proof}
We first prove \eqref{a-108}. Since the root $z_2$ of $T_2$ is adjacent
to an external vertex $y_{a_2}$ with $a_2\in\sigma_2$, one localisation
factor can be extracted from the amplitude of $T_2$ and the definition of the
$\otimes_\pi$-product gives
\begin{multline}\label{eq:DeltaChiFusion1}
\mathcal{A}^{\Lambda,\Lambda_0}_{l_1,n_1+1}(T_1)
\otimes_{\pi}
\Delta_{z_1}\chi\,
\mathcal{A}^{\Lambda,\Lambda_0}_{l_2,n_2+1}(T_2)
\\
=
\int_{\mathbb R}dz_2\,
\bigl(\chi(z_1)-\chi(z_2)\bigr)
\Bigl(
\mathcal{A}^{\Lambda,\Lambda_0}_{l_1,n_1+1}(T_1)
\otimes_{\pi}
\mathcal{A}^{\Lambda,\Lambda_0}_{l_2,n_2+1}(T_2)
\Bigr)
\bigl(
z_1,z_2;
(\tau_i,y_i)_{i\in\mathfrak s};
\chi^{\sigma\setminus\{a_2\}}
\bigr).
\end{multline}
Applying the fusion estimate \eqref{eq:FusionTight0} yields
\begin{equation}\label{eq:DeltaChiFusion2}
\eqref{eq:DeltaChiFusion1}
\lesssim
\int_{\mathbb R}dz_2\,
\bigl|\chi(z_1)-\chi(z_2)\bigr|
\,
\mathcal A^{\Lambda,\Lambda_0}_{l_1\vee l_2,\hat{n}}
\bigl(
\mathcal F_u(T_1,T_2)^{(2)}_{z_2}
\bigr)
\bigl(
z_1,z_2;
(\tau_i,y_i)_{i\in\mathfrak s};
\chi^{\sigma\setminus\{a_2\}}
\bigr).
\end{equation}
Using
\[
\chi(z_1)-\chi(z_2)
=
\left(z_1-z_2\right)\Delta^{(1)}_{z_1,z_2}\chi,
\]
together with Lemma~\ref{LemDecPo} and the monotonicity property of the amplitude \eqref{AlAl0_nohat} , we obtain
\[
|z_1-z_2|
\,
\mathcal A^{\Lambda,\Lambda_0}_{l_1\vee l_2,\hat{n}}
\bigl(
\mathcal F_u(T_1,T_2)^{(2)}_{z_2}
\bigr)
\lesssim
\Lambda^{-1}
\mathcal A^{\Lambda,\Lambda_0}_{l_1\vee l_2,\hat{n}}
\bigl(
\mathcal F_u(T_1,T_2)^{(2)}_{z_2}
\bigr)\lesssim
\Lambda^{-1}
\mathcal A^{\Lambda,\Lambda_0}_{l,n}
\bigl(
\mathcal F_u(T_1,T_2)^{(2)}_{z_2}
\bigr)~.
\]
Substituting this estimate into \eqref{eq:DeltaChiFusion2}, the resulting
integrand coincides with the first contribution appearing in the
definition of
\(
\hat{\mathcal A}^{\Lambda,\Lambda_0}_{l,n}
(\mathcal F_u(T_1,T_2))
\).
Using the root integration identity \eqref{ADRT}, we therefore obtain
\[
\mathcal A^{\Lambda,\Lambda_0}_{l_1,n_1+1}(T_1)
\otimes_\pi
\Delta_{z_1}\chi\,
\mathcal A^{\Lambda,\Lambda_0}_{l_2,n_2+1}(T_2)
\lesssim
\Lambda^{-1}
\hat{\mathcal A}^{\Lambda,\Lambda_0}_{l,n}
\bigl(
\mathcal F_u(T_1,T_2)
\bigr),
\]
which proves \eqref{a-108}. For \eqref{a-109}, we use
\[
|\chi(z_1)-\chi(z_2)|
\le
\chi(z_1)+\chi(z_2).
\]
Since the root $z_1$ is adjacent to the $\sigma_1$-decorated external
vertex $y_{a_1}$, the factor $\chi(z_1)$ can be absorbed into the
localisation factors already present in
\(
\mathcal A^{\Lambda,\Lambda_0}_{l_1,n_1+1}(T_1)
\).
Similarly, the factor $\chi(z_2)$ is absorbed into the localisation
structure of
\(
\hat{\mathcal A}^{\Lambda,\Lambda_0}_{l_2,n_2+1}(T_2)
\).
Consequently,
\[
\mathcal A^{\Lambda,\Lambda_0}_{l_1,n_1+1}(T_1)
\otimes_\pi
\Delta_{z_1}\chi\,
\hat{\mathcal A}^{\Lambda,\Lambda_0}_{l_2,n_2+1}(T_2)
\lesssim
\mathcal A^{\Lambda,\Lambda_0}_{l_1,n_1+1}(T_1)
\otimes_\pi
\hat{\mathcal A}^{\Lambda,\Lambda_0}_{l_2,n_2+1}(T_2).
\]
The bound \eqref{a-109} now follows from \eqref{eq:FusionTight5}.\\
The same argument yields
\[
\hat{\mathcal A}^{\Lambda,\Lambda_0}_{l_1,n_1+1}(T_1)
\otimes_\pi
\Delta_{z_1}\chi\,
\mathcal A^{\Lambda,\Lambda_0}_{l_2,n_2+1}(T_2)
\lesssim
\hat{\mathcal A}^{\Lambda,\Lambda_0}_{l_1,n_1+1}(T_1)
\otimes_\pi
\mathcal A^{\Lambda,\Lambda_0}_{l_2,n_2+1}(T_2),
\]
and
\[
\hat{\mathcal A}^{\Lambda,\Lambda_0}_{l_1,n_1+1}(T_1)
\otimes_\pi
\Delta_{z_1}\chi\,
\hat{\mathcal A}^{\Lambda,\Lambda_0}_{l_2,n_2+1}(T_2)
\lesssim
\hat{\mathcal A}^{\Lambda,\Lambda_0}_{l_1,n_1+1}(T_1)
\otimes_\pi
\hat{\mathcal A}^{\Lambda,\Lambda_0}_{l_2,n_2+1}(T_2).
\]
The bounds \eqref{a-110} and \eqref{a-111} then follow respectively from
\eqref{eq:FusionDeltaChi} and \eqref{Aa11}.
\end{proof}
\begin{corollary}\label{cor:Delta-chi-fusion-combined}
Under the assumptions and notation of Lemma~\ref{lem:Delta-chi-fusion}, one has for $\left(l_1,l_2\right)\in\mathbb{N}^2$ and $n_1,n_2\ge 3$ that
\begin{multline}
\Bigl(
\mathcal{A}_{l_1,n_1+1}^{\Lambda,\Lambda_0}(T_1)
+\Lambda_m^{-1}\hat{\mathcal{A}}_{l_1,n_1+1}^{\Lambda,\Lambda_0}(T_1)
\Bigr)
\otimes_{\pi}
\Delta_{z_1}\chi\,
\Bigl(
\mathcal{A}_{l_2,n_2+1}^{\Lambda,\Lambda_0}(T_2)
+\Lambda_m^{-1}\hat{\mathcal{A}}_{l_2,n_2+1}^{\Lambda,\Lambda_0}(T_2)
\Bigr)
\\
\lesssim
\Lambda^{-1}
\left(
1+\sum_{i\in\mathfrak{s}_2}\frac{\tau_i^{-1/2}}{\Lambda_m}
\right)
\hat{\mathcal{A}}_{l,n}^{\Lambda,\Lambda_0}\bigl(\mathcal{F}_u(T_1,T_2)\bigr).
\end{multline}
\end{corollary}

\begin{proof}
By bilinearity of the $\otimes_\pi$–product, the left-hand side expands into
\begin{align*}
&\mathcal{A}_{l_1,n_1+1}^{\Lambda,\Lambda_0}(T_1)\otimes_{\pi}
\Delta_{z_1}\chi\,\mathcal{A}_{l_2,n_2+1}^{\Lambda,\Lambda_0}(T_2),
\\
&\Lambda_m^{-1}\,
\mathcal{A}_{l_1,n_1+1}^{\Lambda,\Lambda_0}(T_1)\otimes_{\pi}
\Delta_{z_1}\chi\,\hat{\mathcal{A}}_{l_2,n_2+1}^{\Lambda,\Lambda_0}(T_2),
\\
&\Lambda_m^{-1}\,
\hat{\mathcal{A}}_{l_1,n_1+1}^{\Lambda,\Lambda_0}(T_1)\otimes_{\pi}
\Delta_{z_1}\chi\,\mathcal{A}_{l_2,n_2+1}^{\Lambda,\Lambda_0}(T_2),
\\
&\Lambda_m^{-2}\,
\hat{\mathcal{A}}_{l_1,n_1+1}^{\Lambda,\Lambda_0}(T_1)\otimes_{\pi}
\Delta_{z_1}\chi\,\hat{\mathcal{A}}_{l_2,n_2+1}^{\Lambda,\Lambda_0}(T_2).
\end{align*}
Applying the bounds \eqref{a-108}--\eqref{a-111} to these four terms yields
\[
\Lambda^{-1}\hat{\mathcal{A}}_{l,n}^{\Lambda,\Lambda_0}\bigl(\mathcal{F}_u(T_1,T_2)\bigr),
\qquad
\Lambda_m^{-1}\hat{\mathcal{A}}_{l,n}^{\Lambda,\Lambda_0}\bigl(\mathcal{F}_u(T_1,T_2)\bigr),
\]
and
\[
\Lambda_m^{-2}
\Bigl(\sum_{i\in\mathfrak{s}_2}\tau_i^{-1/2}+\Lambda\Bigr)
\hat{\mathcal{A}}_{l,n}^{\Lambda,\Lambda_0}\bigl(\mathcal{F}_u(T_1,T_2)\bigr).
\]
Since $\Lambda_m=\Lambda+m\ge \Lambda$, one has $\Lambda_m^{-1}\le\Lambda^{-1}$.
Therefore each of the above contributions is bounded by
\[
\Lambda^{-1}
\left(
1+\sum_{i\in\mathfrak{s}_2}\frac{\tau_i^{-1/2}}{\Lambda_m}
\right)
\hat{\mathcal{A}}_{l,n}^{\Lambda,\Lambda_0}\bigl(\mathcal{F}_u(T_1,T_2)\bigr),
\]
which proves the claim.
\end{proof}
The next lemma shows that polynomial factors in the distance between the
two root variables can be absorbed into the fused tree amplitude at the
expense of powers of $\Lambda^{-1}$. Such factors arise repeatedly from
Taylor expansions and localisation remainders.
\begin{lemma}\label{Rlemma}
Fix $(l_1,l_2)\in\mathbb N_0^2$ and $n_1,~n_2\ge 2$. Define $$
l=l_1+l_2,
\qquad
n=n_1+n_2.
$$
Let $
T_1\in\mathcal T_1^{s_1+1}$ and $
T_2\in\mathcal T_1^{s_2+1}$
with root variables $z_1$ and $z_2$, respectively.
For $r\in\{1,2,3\}$ one has
\begin{multline}\label{Rrelation}
\int_{\mathbb R}dz_2\,
|z_1-z_2|^r
\Bigl(
\mathcal A^{\Lambda,\Lambda_0}_{l_1,n_1+1}(T_1)
\otimes_\pi
\mathcal A^{\Lambda,\Lambda_0}_{l_2,n_2+1}(T_2)
\Bigr)
\bigl(
z_1,z_2;(\tau_i,y_i)_{i\in\mathfrak s}
\bigr)\\
\lesssim
\Lambda^{-r}
\mathcal A^{\Lambda,\Lambda_0}_{l,n}
\bigl(
\mathcal F_u(T_1,T_2)
\bigr).
\end{multline}
The amplitudes are evaluated with $\chi\equiv1$.
\end{lemma}

\begin{proof}
Applying the fusion estimate \eqref{eq:FusionTight0} together with \eqref{AlAl0_nohat}, we obtain
\begin{multline}\label{eq:Rlemma-step1}
\int_{\mathbb R}dz_2\,
|z_1-z_2|^r
\Bigl(
\mathcal A^{\Lambda,\Lambda_0}_{l_1,n_1+1}(T_1)
\otimes_\pi
\mathcal A^{\Lambda,\Lambda_0}_{l_2,n_2+1}(T_2)
\Bigr)
\bigl(
z_1,z_2;(\tau_i,y_i)_{i\in\mathfrak s}
\bigr)
\\
\lesssim
\int_{\mathbb R}dz_2\,
|z_1-z_2|^r
\mathcal A^{\Lambda,\Lambda_0}_{l_1\vee l_2,\hat{n}}
\bigl(
\mathcal F_u(T_1,T_2)^{(2)}_{z_2}
\bigr)
\bigl(
z_1,z_2;(\tau_i,y_i)_{i\in\mathfrak s}
\bigr),
\end{multline}
where $\hat{n}=\max\left(n_1+1,n_2+1\right)$. Using Lemma~\ref{LemDecPo}, we obtain
\[
|z_1-z_2|^r
\mathcal A^{\Lambda,\Lambda_0}_{l_1\vee l_2,\hat{n}}
\bigl(
\mathcal F_u(T_1,T_2)^{(2)}_{z_2}
\bigr)
\lesssim
\Lambda^{-r}
\mathcal A^{\Lambda,\Lambda_0}_{l,n}
\bigl(
\mathcal F_u(T_1,T_2)^{(2)}_{z_2}
\bigr).
\]
Substituting this estimate into \eqref{eq:Rlemma-step1} and using the
root extraction identity \eqref{ADRT}, we obtain
\[
\int_{\mathbb R}dz_2\,
|z_1-z_2|^r
\Bigl(
\mathcal A^{\Lambda,\Lambda_0}_{l_1,n_1+1}(T_1)
\otimes_\pi
\mathcal A^{\Lambda,\Lambda_0}_{l_2,n_2+1}(T_2)
\Bigr)
\bigl(
z_1,z_2;(\tau_i,y_i)_{i\in\mathfrak s}
\bigr)
\lesssim
\Lambda^{-r}
\mathcal A^{\Lambda,\Lambda_0}_{l,n}
\bigl(
\mathcal F_u(T_1,T_2)
\bigr),
\]
which proves the claim.
\end{proof}

We now consider fusion estimates involving differences of heat kernels.
The basic mechanism is that a kernel difference between two vertices can
be expanded along the unique path connecting them in the tree. Each
increment is then absorbed into the heat kernel associated with the
corresponding edge, producing an additional external leg and a factor
$\Lambda^{-1}$. We begin by introducing some notation that will be used
throughout this section. Let $T$ be a tree and let $u,v\in V(T)$. We denote by
$\langle u,v\rangle$ the unique simple path in $T$ connecting $u$ and $v$,
viewed as the ordered set of vertices along this path and including both
endpoints. We also use the interval-type notations
\[
\langle u,v)
:= \langle u,v\rangle\setminus\{v\}, 
\qquad
(u,v\rangle
:= \langle u,v\rangle\setminus\{u\}~.
\]
All these sets are understood as subsets of $V(T)$.
\begin{lemma}\label{Lemme16}
Fix $l\ge2$, $n\ge2$, $1\le n_1\le n-1$ and $1\le l_1<l$. Let
$T_1\in\mathcal{T}^{s}_1$ and $\tau'>0$. Assume that
\begin{equation}\label{eq:Lemme16_ass}
\hat{\delta}_{l,n}\Lambda \ge b_{l_1,l}\,(\tau')^{-1/2},
\qquad
b_{l_1,l}:=\frac{(2-2^{-l-1})(2-2^{-l_1-1})}{2^{-l_1-1}-2^{-l-1}},
\qquad
\hat{\delta}_{l,n}:=\delta_{l,n}^{1/3}-1.
\end{equation}
Then one has
\begin{multline}\label{eq:Lemme16}
\bigl|K(\tau';z_1,y)-K(\tau';u,y)\bigr|\,
\mathcal{A}^{\Lambda,\Lambda_0}_{l_1,n_1+1}(T_1)
\bigl(
z_1;(\tau_i,y_i)_{i\in\mathfrak{s}},(\tfrac{1}{2\Lambda^{2}},u)
\bigr)
\\
\lesssim
\tau'^{-1/2}\Lambda^{-1}
\sum_{v\in\langle z_1,u)}
\mathcal{A}^{\Lambda,\Lambda_0}_{l,n}\!\left(T_{1,v}\right)
\bigl(
z_1;(\tau_i,y_i)_{i\in\mathfrak{s}},(\tau',y),(\tfrac{1}{2\Lambda^{2}},u)
\bigr),
\end{multline}
where $T_{1,v}\in\mathcal{T}^{s+1}_1$ is obtained from $T_1$ by attaching
an additional external leg $(\tau',y)$ to the vertex $v$.
\end{lemma}

\begin{proof}
Let $\langle z_1,u\rangle=\{v_0,\dots,v_q\}$ with $v_0=z_1$ and $v_q=u$.
Using the telescopic decomposition along the unique path connecting
$z_1$ and $u$, we write
\[
K(\tau';z_1,y)-K(\tau';u,y)
=
\sum_{i=0}^{q-1}
\bigl(
K(\tau';v_i,y)-K(\tau';v_{i+1},y)
\bigr).
\]

\medskip
\noindent
For every $i\in\{0,\dots,q-1\}$, a first-order Taylor expansion together
with the standard gradient bound
\[
|\nabla_x K(\tau;x,y)|
\lesssim
\tau^{-1/2}K(c\tau;x,y),
\qquad c>1,
\]
yields
\begin{equation}\label{eq:Lemme16_Taylor}
\bigl|K(\tau';v_i,y)-K(\tau';v_{i+1},y)\bigr|
\lesssim
\tau'^{-1/2}|v_i-v_{i+1}|
\int_0^1
K\bigl((1+\hat{\delta}_{l,n})\tau';tv_i+(1-t)v_{i+1},y\bigr)\,dt .
\end{equation}

Multiplying by the heat-kernel factor
$K(\delta_{l_1,n_1+1}/\Lambda_i^2;v_i,v_{i+1})$ attached to the edge
$(v_i,v_{i+1})$ in the representation of
$\mathcal{A}^{\Lambda,\Lambda_0}_{l_1,n_1+1}(T_1)$ and using, for
$1\le l_1\le l-1$,
\begin{equation}\label{eq:Lemme16_edge}
|v_i-v_{i+1}|
K\!\left(\frac{\delta_{l_1,n_1+1}}{\Lambda_i^2};v_i,v_{i+1}\right)
\lesssim
\Lambda^{-1}
K\!\left(\frac{\delta_{l,n}}{\Lambda_i^2};v_i,v_{i+1}\right)
\exp\!\Bigl(-\frac{|v_i-v_{i+1}|^2\Lambda_i^2}{2b_{l_1,l}}\Bigr),
\end{equation}
we invoke the assumption \eqref{eq:Lemme16_ass}. The assumption \eqref{eq:Lemme16_ass} guarantees that the Gaussian factor
in \eqref{eq:Lemme16_edge} can be absorbed into the enlarged heat kernel
appearing in \eqref{eq:Lemme16_Taylor}. A direct completion-of-squares
argument yields
\begin{equation}\label{equ76}
\frac{|v_i-v_{i+1}|^2\Lambda_i^2}{b_{l_1,l}}
+
\frac{|tv_i+(1-t)v_{i+1}-y|^2}
{\tau'(1+\hat\delta_{l,n})}
\gtrsim
\frac{|v_i-y|^2}
{\tau'(1+\hat\delta_{l,n})^3}.
\end{equation}
Using \eqref{equ76}, we obtain the following estimate
\begin{multline*}\label{eq:Lemme16_conv}
K\!\left(\frac{\delta_{l_1,n_1+1}}{\Lambda_i^2};v_i,v_{i+1}\right)
\int_0^1
K\bigl((1+\hat{\delta}_{l,n})\tau';tv_i+(1-t)v_{i+1},y\bigr)\,dt
\\
\lesssim
K\!\left(\frac{\delta_{l,n}}{\Lambda_i^2};v_i,v_{i+1}\right)
K(\delta_{l,n}\tau';v_i,y).
\end{multline*}

\medskip
\noindent
Combining
\eqref{eq:Lemme16_Taylor},
\eqref{eq:Lemme16_edge}
and the preceding convolution estimate,
the contribution associated with the edge
$(v_i,v_{i+1})$
is bounded by
\[
\tau'^{-1/2}\Lambda^{-1}
\mathcal A^{\Lambda,\Lambda_0}_{l,n}(T_{1,v_i})
\Bigl(
z_1;
(\tau_j,y_j)_{j\in\mathfrak s},
(\tau',y),
(\tfrac1{2\Lambda^2},u)
\Bigr),
\]
where $T_{1,v_i}$ is obtained by attaching the additional external leg
$(\tau',y)$ to the vertex $v_i$.
Summing over all vertices
$v_i\in\langle z_1,u)$
yields \eqref{eq:Lemme16}.
\end{proof}
\begin{lemma}\label{LemmaDeltaFusion}
Let $l\ge2$, $n\ge2$, $1\le n_1\le n-1$, and $1\le l_1<l$, and set
\[
l_2:=l-l_1,
\qquad
n_2:=n-n_1.
\]
Let $
T_1\in\mathcal{T}^{s_1+1}_1$ and $
T_2\in\mathcal{T}^{s_2}_1$
rooted respectively at $z_1$ and $z_2$. For $\tau_2>0$ and $y_2\in\mathbb R$,
define
\[
\Delta_{z_1}K_{\tau_2,y_2}(z_2)
:=
K(\tau_2;z_1,y_2)-K(\tau_2;z_2,y_2).
\]
Then one has
\begin{multline}\label{DeltaFusionBound}
\int_{\mathbb R}dz_2\,
\bigl|\Delta_{z_1}K_{\tau_2,y_2}(z_2)\bigr|\,
\mathcal{A}_{l_1,n_1+1}^{\Lambda,\Lambda_0}(T_1)
\otimes_{\pi}
\mathcal{A}_{l_2,n_2+1}^{\Lambda,\Lambda_0}(T_2)
\Bigl(
z_1,z_2;(\tau_i,y_i)_{i\in\mathfrak{s}}
\Bigr)
\\
\lesssim
\tau_2^{-1/2}\Lambda^{-1}
\sum_{v\in\langle z_1,z_2\rangle}
\mathcal{A}_{l,n}^{\Lambda,\Lambda_0}
\bigl(
\mathcal{F}_{u}(T_1,T_2)_{v}
\bigr)
\Bigl(
z_1;
(\tau_i,y_i)_{i\in\mathfrak s},(\tau_2,y_2)
\Bigr),
\end{multline}
where $\mathcal{F}_{u}(T_1,T_2)_{v}$ denotes the tree obtained from
$\mathcal{F}_{u}(T_1,T_2)$ by attaching an additional external leg
$(v,y_2)$ to the vertex $v$ and where the path
\(
\langle z_1,z_2\rangle
\)
is understood in the fused tree
\(
\mathcal F_u(T_1,T_2)
\).
\end{lemma}

\begin{proof}
The proof combines Lemma~\ref{Lemme16} with the fusion estimate
\eqref{eq:FusionTight1}. The argument depends on whether the scale
$\tau_2$ is sufficiently large compared with $\Lambda^{-2}$:

\medskip
\noindent
\emph{Case 1:}
\[
\hat{\delta}_{l,n}\Lambda \ge b_{l_1,l}\,\tau_2^{-1/2},
\qquad
b_{l_1,l}:=\frac{(2-2^{-l-1})(2-2^{-l_1-1})}{2^{-l_1-1}-2^{-l-1}},
\qquad
\hat{\delta}_{l,n}:=\delta_{l,n}^{1/3}-1.
\]
We insert and subtract the intermediate kernel $K(\tau_2;u,y_2)$ in
$\Delta_{\tau_2,y_2}K$ and decompose the left-hand side of
\eqref{DeltaFusionBound} into the sum of the two contributions
\begin{multline}\label{eq:case1_a}
\int_{\mathbb R}
\bigl(K(\tau_2;z_1,y_2)-K(\tau_2;u,y_2)\bigr)
\mathcal A_{l_1,n_1+1}^{\Lambda,\Lambda_0}\left(T_1\right)
\left(
z_1;(\tau_i,y_i)_{i\in\mathfrak s_1},\left(\tfrac{1}{2\Lambda^2},u\right)
\right)
\\
\times
\mathcal A_{l_2,n_2+1}^{\Lambda,\Lambda_0}(T_2)
\left(
z_2;(\tau_i,y_i)_{i\in\mathfrak s_2},\left(\tfrac{1}{2\Lambda^2},u\right)
\right)\,du
\end{multline}
and
\begin{multline}\label{eq:case1_b}
\int_{\mathbb R}
\mathcal A_{l_1,n_1+1}^{\Lambda,\Lambda_0}(T_1)
\left(
z_1;(\tau_i,y_i)_{i\in\mathfrak s_1},\left(\tfrac{1}{2\Lambda^2},u\right)
\right)
\\
\times
\bigl(K(\tau_2;u,y_2)-K(\tau_2;z_2,y_2)\bigr)
\mathcal A_{l_2,n_2+1}^{\Lambda,\Lambda_0}(T_2)
\left(
z_2;(\tau_i,y_i)_{i\in\mathfrak s_2},\left(\tfrac{1}{2\Lambda^2},u\right)
\right)\,du.
\end{multline}
Applying Lemma~\ref{Lemme16} and then the fusion estimate
\eqref{eq:FusionTight1}, we obtain
\begin{equation}\label{eq:case1_est1}
\eqref{eq:case1_a}
\lesssim
\tau_2^{-1/2}\Lambda^{-1}
\sum_{v\in\langle z_1,u)}
\mathcal A^{\Lambda,\Lambda_0}_{l,n}
\bigl(
\mathcal F_u(T_{1,v},T_2)
\bigr).
\end{equation}
An entirely analogous argument applied to \eqref{eq:case1_b}, now under the
assumption
\[
\hat{\delta}_{l,n}\Lambda \ge b_{l_2,l}\,\tau_2^{-1/2},
\qquad
b_{l_2,l}:=\frac{(2-2^{-l-1})(2-2^{-l_2-1})}{2^{-l_2-1}-2^{-l-1}},
\]
yields
\begin{equation}\label{eq:case1_est2}
\eqref{eq:case1_b}
\lesssim
\tau_2^{-1/2}\Lambda^{-1}
\sum_{v\in\langle z_2,u)}
\mathcal A^{\Lambda,\Lambda_0}_{l,n}
\bigl(
\mathcal F_u(T_{1},T_{2,v})
\bigr)~.
\end{equation}
The vertices produced by the first contribution lie on the
segment joining $z_1$ to the fusion vertex $u$, whereas those
produced by the second contribution lie on the segment joining
$u$ to $z_2$. Together they form the full path
$
\langle z_1,z_2\rangle
$
in the fused tree
$
\mathcal F_u(T_1,T_2)$.
Hence the contribution of Case~1 is bounded by the right-hand
side of \eqref{DeltaFusionBound}.\\
\medskip
\noindent
\emph{Case 2:}
Since
$$
\hat\delta_{l,n}\Lambda
\le
\max(b_{l_1,l},b_{l_2,l})\tau_2^{-1/2},
$$
we have
$
1\lesssim \tau_2^{-1/2}\Lambda^{-1}$.
Hence
\[
|\Delta_{z_1}K_{\tau_2,y_2}(z_2)|
\lesssim
\tau_2^{-1/2}\Lambda^{-1}
\Bigl(
K(\tau_2;z_1,y_2)
+
K(\tau_2;z_2,y_2)
\Bigr).
\]
The two heat kernels may therefore be viewed as additional external
legs attached at $z_1$ and $z_2$, respectively. Applying the fusion
estimate \eqref{eq:FusionTight1} yields a contribution bounded by the
right-hand side of \eqref{DeltaFusionBound}.
Combining the estimates obtained in Cases~1 and~2 yields
\eqref{DeltaFusionBound}.
\end{proof}

\paragraph{Calculus of global $\otimes_{\pi}$ fusion products }
We now move from treewise fusion estimates to estimates for the global
amplitudes. Since the latter are obtained by summing over all admissible
trees, it is necessary to control the combinatorial multiplicity of the
fusion map. The crucial observation is that any fused tree admits only a
uniformly bounded number of preimages under the fusion operation.
Consequently, the estimates established in the previous subsections
remain stable after summation over trees. The next lemma makes this
principle precise.
\begin{lemma}[Fusion pushforward bound]\label{lem:fusion_pushforward_general2}
Fix $s_1,s_2\in\mathbb N$, and set $s:=s_1+s_2$.
Recall the fusion map
\[
\mathcal F_u:\mathcal T^{s_1+1}_1\times \mathcal T^{s_2+1}_1
\longrightarrow
\mathcal T^s_1
\]
 obtained by fusing the distinguished external legs of
both trees at the vertex $u$: the external vertex $u$ is contracted, the two
edges $(z,u)$ and $(z',u)$ are replaced by a single edge $(z,z')$, the root of
the first factor is kept as the root of the fused tree, and the root of the
second factor becomes an internal vertex.

Then the following hold:
\begin{enumerate}
\item For all $T_1\in\mathcal T^{s_1+1}_1$ and $T_2\in\mathcal T^{s_2+1}_1$,
\begin{equation}\label{eq:v_fusion_identity_i}
v\bigl(\mathcal F_u(T_1,T_2)\bigr)=v(T_1)+v(T_2)+1.
\end{equation}

\item For every nonnegative functional $\mathsf F:\mathcal T^s_1\to[0,\infty)$
and all integers $M_1,M_2\ge1$, there exists a constant
$C=C(M_1,M_2,s_1,s_2)$ such that
\begin{equation}\label{eq:fusion_pushforward_general_i}
\sum_{\substack{T_1\in\mathcal T^{s_1+1}_1\\ v(T_1)\le M_1}}
\ \sum_{\substack{T_2\in\mathcal T^{s_2+1}_1\\ v(T_2)\le M_2}}
\mathsf F\bigl(\mathcal F_u(T_1,T_2)\bigr)
\;\le\;
C
\sum_{\substack{T\in\mathcal T^s_1\\ v(T)\le M_1+M_2+1}}
\mathsf F(T).
\end{equation}
\end{enumerate}
\end{lemma}

\begin{proof}
The identity \eqref{eq:v_fusion_identity_i} follows directly from the
definition of the fusion map. The root of the first tree remains the root
of the fused tree, whereas the root of the second tree becomes an internal
vertex. The fusion also creates one additional internal vertex along the
new internal edge. Hence
\[
v\bigl(\mathcal F_u(T_1,T_2)\bigr)
=
v(T_1)+v(T_2)+1.
\]
We now prove \eqref{eq:fusion_pushforward_general_i}. Since $\mathsf F$ is
nonnegative, it is enough to bound, for fixed
$T\in\mathcal T^s_1$ with $v(T)\le M_1+M_2+1$, the cardinality of
\[
\bigl\{(T_1,T_2)\in\mathcal T^{s_1+1}_1\times\mathcal T^{s_2+1}_1:
\mathcal F_u(T_1,T_2)=T\bigr\}.
\]
Fix \(T\in\mathcal T^s_1\) with
\(v(T)\le M_1+M_2+1\). Any pair
\((T_1,T_2)\) satisfying
\[
\mathcal F_u(T_1,T_2)=T
\]
is obtained by cutting one internal edge of \(T\), declaring the component
containing the root of \(T\) to be the first factor, and inserting the
distinguished external vertex \(u\) on both sides of the cut. One must also
choose which \(s_1\) of the \(s\) external labels belong to the first
factor; the remaining labels then belong to the second factor. Since the
number of internal edges of \(T\) is bounded in terms of
\(M_1,M_2,s_1,s_2\), the number of such decompositions is bounded by a
constant
\[
C=C(M_1,M_2,s_1,s_2).
\]
Thus
\[
\#\bigl\{
(T_1,T_2):
\mathcal F_u(T_1,T_2)=T,\,
v(T_i)\le M_i
\bigr\}
\le
C(M_1,M_2,s_1,s_2).
\]
Using also \eqref{eq:v_fusion_identity_i}, we see that
\[
v(T_1)\le M_1,\qquad v(T_2)\le M_2
\quad\Longrightarrow\quad
v\bigl(\mathcal F_u(T_1,T_2)\bigr)\le M_1+M_2+1.
\]
Therefore,
\begin{align*}
\sum_{\substack{T_1\in\mathcal T^{s_1+1}_1\\ v(T_1)\le M_1}}
\ \sum_{\substack{T_2\in\mathcal T^{s_2+1}_1\\ v(T_2)\le M_2}}
\mathsf F\bigl(\mathcal F_u(T_1,T_2)\bigr)
&=
\sum_{\substack{T\in\mathcal T^s_1\\ v(T)\le M_1+M_2+1}}
\sum_{\substack{(T_1,T_2)\in\mathcal T^{s_1+1}_1\times\mathcal T^{s_2+1}_1\\
\mathcal F_u(T_1,T_2)=T}}
\mathsf F(T)
\\
&\le
C(M_1,M_2,s_1,s_2)
\sum_{\substack{T\in\mathcal T^s_1\\ v(T)\le M_1+M_2+1}}
\mathsf F(T),
\end{align*}
which is precisely \eqref{eq:fusion_pushforward_general_i}.
\end{proof}
\noindent For convenience, we introduce the shorthand
\begin{equation}\label{eq:MixedFusionNotation}
\langle \mathcal A,\hat{\mathcal A}\rangle_{l,n}
:=
\mathcal A_{l,n}^{\Lambda,\Lambda_0}
+
(1-\delta_{l,0})\mathbf{1}_{\{\sigma\neq \varnothing\}}
\Lambda_m^{-1}
\hat{\mathcal A}_{l,n}^{\Lambda,\Lambda_0}.
\end{equation}
Furthermore, if $j\in\sigma_2$, we write
\begin{equation}\label{eq:MixedFusionDeltaNotation}
\mathcal X_1
\otimes_\pi^{(j)}
\Delta_{z_1}\chi\,
\mathcal X_2
\end{equation}
for the fusion product in which the localisation factor indexed by
$j$ is extracted from the second amplitude and replaced by the insertion
$\Delta_{z_1}\chi$ before fusion.

The preceding estimates were established at the level of individual
rooted trees. Combining them with the pushforward estimate of
Lemma~\ref{lem:fusion_pushforward_general2} allows one to pass to the
corresponding bounds for the global amplitudes. The resulting estimates
are collected in the following corollary, which will be used repeatedly
in the analysis of the quadratic part of the flow equation.
\begin{corollary}[Global fusion bounds]\label{CorFus}
Let $(l_1,l_2)\in\mathbb N_0^2$ and $n_1,~n_2\ge2$, and set
\[
l:=l_1+l_2,
\qquad
n:=n_1+n_2.
\]
\emph{Fusion without insertions.}
One has
\begin{align}
\mathcal A_{l_1,2}^{\Lambda,\Lambda_0}
\otimes_\pi
\mathcal A_{l_2,2}^{\Lambda,\Lambda_0}
&\lesssim
\mathcal A_{l,2}^{\Lambda,\Lambda_0},
\label{157CorFus}
\\
\Bigl(
\mathcal A_{l_1,n}^{\Lambda,\Lambda_0}
+\Lambda_m^{-1}\hat{\mathcal A}_{l_1,n}^{\Lambda,\Lambda_0}
\Bigr)
\otimes_\pi
\mathcal A_{l_2,2}^{\Lambda,\Lambda_0}
&\lesssim
\Bigl(
\mathcal A_{l,n}^{\Lambda,\Lambda_0}
+\Lambda_m^{-1}\hat{\mathcal A}_{l,n}^{\Lambda,\Lambda_0}
\Bigr),
\qquad l_1\ge1,\label{CorFuss}\\
\mathcal A_{l_1,n_1+1}^{\Lambda,\Lambda_0}
\otimes_\pi
\mathcal A_{l_2,n_2+1}^{\Lambda,\Lambda_0}
&\lesssim
\mathcal A_{l,n}^{\Lambda,\Lambda_0}~,
\label{CorFuss1}\\
\mathcal A_{l_1,n_1+1}^{\Lambda,\Lambda_0}
\otimes_\pi
\hat{\mathcal A}_{l_2,n_2+1}^{\Lambda,\Lambda_0}
&\lesssim
\hat{\mathcal A}_{l,n}^{\Lambda,\Lambda_0}~,\qquad l_2\ge 1,
\label{CorFuss2}\\
\hat{\mathcal A}_{l_1,n_1+1}^{\Lambda,\Lambda_0}
\otimes_\pi
\hat{\mathcal A}_{l_2,n_2+1}^{\Lambda,\Lambda_0}
&\lesssim
~\left(\sum_{i\in\mathfrak{s}_2}\tau_i^{-\frac{1}{2}}+\Lambda\right)~\hat{\mathcal A}_{l,n}^{\Lambda,\Lambda_0}~,\qquad l_1\ge 1,~~l_2\ge 1~.
\label{CorFuss3}
\end{align}
In particular for $n_1,n_2\ge 3$ we have 
\begin{equation}\label{eq:CorFus_2}
    \langle\mathcal A,\hat{\mathcal A}\rangle_{l_1,n_1+1}
\otimes_\pi
\langle\mathcal A,\hat{\mathcal A}\rangle_{l_2,n_2+1}
\lesssim
\left(
1+
(1-\delta_{l_2,0})
\sum_{i\in\mathfrak s_2}
\frac{\tau_i^{-1/2}}{\Lambda_m}
\right)
\langle
\mathcal A,\hat{\mathcal A}\rangle_{l,n}~.
\end{equation}
\smallskip
\noindent
\emph{Fusion with one insertion of $\Delta_{z_1}K_{\tau_2,y_2}$.}
For $n_1,n_2\ge1$,
\begin{equation}\label{eq:CorFus_3}
\mathcal A_{l_1,n_1+1}^{\Lambda,\Lambda_0}
\otimes_\pi
\Delta_{z_1}K_{\tau_2,y_2}\,
\mathcal A_{l_2,n_2+1}^{\Lambda,\Lambda_0}
\bigl(
z_1;(\tau_i,y_i)_{i\in\mathfrak s}
\bigr)
\lesssim
\tau_2^{-1/2}~\Lambda^{-1}~
\mathcal A_{l,n}^{\Lambda,\Lambda_0}
\bigl(
z_1;
(\tau_i,y_i)_{i\in\mathfrak s},
(\tau_2,y_2)
\bigr).
\end{equation}

\smallskip
\noindent
\emph{Fusion with one insertion of $\Delta_{z_1}\chi$.}
For $n_1,n_2\ge2$ and $j\in\sigma_2$,
\begin{equation}\label{eq:CorFus_4}
\langle\mathcal A,\hat{\mathcal A}\rangle_{l_1,n_1+1}
\otimes_\pi^{(j)}
\Delta_{z_1}\chi\,
\langle\mathcal A,\hat{\mathcal A}\rangle_{l_2,n_2+1}
\lesssim
\Lambda^{-1}
\left(
1+
(1-\delta_{l_2,0})\mathbf{1}_{\{\sigma_2^{(j)}\neq\varnothing\}}
\sum_{i\in\mathfrak s_2}
\frac{\tau_i^{-1/2}}{\Lambda_m}
\right)
\hat{\mathcal A}_{l,n}^{\Lambda,\Lambda_0}.
\end{equation}
All amplitudes are evaluated at
$
z_1$, $(\tau_i,y_i)_{i\in\mathfrak s}$ and $\chi^\sigma$.
The implicit constants depend only on $l$, $n$, and $\mathfrak s$.
\end{corollary}

\begin{proof}
Expanding the global amplitudes into sums over rooted trees, each term
appearing on the left-hand side becomes a double sum indexed by pairs
\(
(T_1,T_2)
\).
Applying the corresponding treewise fusion estimate yields a bound by a
nonnegative functional of the fused tree
\[
\mathcal F_u(T_1,T_2).
\]
Since all amplitudes are nonnegative,
Lemma~\ref{lem:fusion_pushforward_general2}
converts the resulting double sum into a single sum over fused trees,
up to a combinatorial constant depending only on the number of external
legs.

The estimates
\eqref{157CorFus}--\eqref{CorFuss}
follow directly from the ordinary treewise fusion bounds established in
Lemma~\ref{lem:FusionBound_tight}. The corresponding mixed estimates
involving hatted amplitudes are obtained in exactly the same way.

To prove \eqref{eq:CorFus_2}, we expand the two brackets and apply the
appropriate treewise estimate to each of the four resulting terms
\[
\mathcal A\otimes_\pi\mathcal A,
\qquad
\mathcal A\otimes_\pi\hat{\mathcal A},
\qquad
\hat{\mathcal A}\otimes_\pi\mathcal A,
\qquad
\hat{\mathcal A}\otimes_\pi\hat{\mathcal A}.
\]
The first three terms are controlled by
Lemma~\ref{lem:FusionBound_tight}, while the last term is controlled by
Lemma~\ref{LemIm}. Summation over trees then yields
\eqref{eq:CorFus_2}.

Estimate \eqref{eq:CorFus_3} follows from
Lemma~\ref{LemmaDeltaFusion}. After applying the corresponding treewise
bound, Lemma~\ref{lem:fusion_pushforward_general2} yields the global
estimate with the factor $\tau_2^{-1/2}$.

Finally, \eqref{eq:CorFus_4} follows from
Corollary~\ref{cor:Delta-chi-fusion-combined}. Expanding the brackets,
applying the corresponding treewise estimates, and summing over trees
using Lemma~\ref{lem:fusion_pushforward_general2} yields the factor
$\Lambda^{-1}$ together with the hatted amplitude on the right-hand
side.
\end{proof}
\subsection{Norms of tree amplitudes}
The following lemma provides the basic estimate that will be used repeatedly
in the sequel. It shows that amplitudes containing cutoff decorations can
be controlled by the corresponding undecorated amplitudes at one higher
loop order, up to a factor involving the external scales $\tau_i$.
\begin{lemma}\label{lem:Linfty_forest}
Let $\mathfrak{s}\subseteq \left\{1,\cdots,n\right\}$ and $\sigma\subsetneq \mathfrak{s}$ with
$\mathfrak{s}\neq\varnothing$. Then, uniformly for $\Lambda\le \Lambda_0$ and
for $(\tau_i,y_i)_{i\in\mathfrak{s}}$, one has
\begin{equation}\label{eq:BornDomi_AHP}
\left(
\mathcal{A}_{l,n}^{\Lambda,\Lambda_0}
+\Lambda_m^{-1}\hat{\mathcal{A}}_{l,n}^{\Lambda,\Lambda_0}
\right)
\bigl((\tau_i,y_i)_{i\in\mathfrak{s}};\chi^{\sigma}\bigr)
\;\lesssim~\|\chi\|_{\infty}^{|\sigma|}
\left(1+
\sum_{i\in\mathfrak{s}}
\frac{\tau_i^{-1/2}}{\Lambda_m}
\right)
\mathcal{A}_{l+1,n}^{\Lambda,\Lambda_0}
\bigl((\tau_i,y_i)_{i\in\mathfrak{s}}\bigr).
\end{equation}
The implicit constant depends only on $l$, $n$, and $\mathfrak{s}$.
\end{lemma}

\begin{proof}
Recall the decomposition
\begin{equation}\label{eq:Ahat_decomp_clean}
\hat{\mathcal{A}}^{\Lambda,\Lambda_0}_{l,n}
\Bigl(
z_1;(\tau_i,y_i)_{i\in\mathfrak{s}};\chi^{\sigma}
\Bigr)
=
\sum_{\substack{
T\in\mathcal{T}^{\mathfrak{s}}_1\\
0\le v(T)\le l+\frac{n}{2}-2}}
\hat{\mathcal{A}}^{\Lambda,\Lambda_0}_{l,n}(T)
\Bigl(
z_1;(\tau_i,y_i)_{i\in\mathfrak{s}};\chi^{\sigma}
\Bigr).
\end{equation}
Integrating with respect to the distinguished variable $z_1$ yields
\begin{equation}\label{eq:Ahat_integrated_clean}
\hat{\mathcal{A}}^{\Lambda,\Lambda_0}_{l,n}\left(T\right)
\Bigl(
(\tau_i,y_i)_{i\in\mathfrak{s}};\chi^{\sigma}
\Bigr)
=
\int_{\mathbb{R}}
\hat{\mathcal{A}}^{\Lambda,\Lambda_0}_{l,n}\left(T\right)
\Bigl(
z_1;(\tau_i,y_i)_{i\in\mathfrak{s}};\chi^{\sigma}
\Bigr)\,dz_1.
\end{equation}
Using \eqref{Nuevo2} and \eqref{equationNuevo} from Corollary~\ref{CorTransport}, one obtains for every $l'>l$
\begin{equation}\label{bBo}
\int_{\mathbb{R}}
\hat{\mathcal{A}}^{\Lambda,\Lambda_0}_{l,n}\left(T\right)
\Bigl(
z_1;(\tau_i,y_i)_{i\in\mathfrak{s}};\chi^{\sigma}
\Bigr)\,dz_1
\;\lesssim\;
\|\chi\|_{\infty}^{|\sigma|}\left(
\sum_{i\in\mathfrak{s}}\tau_i^{-1/2}
\right)
\int_{\mathbb{R}}
\mathcal{A}^{\Lambda,\Lambda_0}_{l',n}
\Bigl(
z_1;(\tau_i,y_i)_{i\in\mathfrak{s}}
\Bigr)\,dz_1 .
\end{equation}
Combining \eqref{bBo} with  Lemma \ref{lem:monotone_l},
and choosing $l'=l+1$ yields \eqref{eq:BornDomi_AHP}.
\end{proof}
\begin{lemma}[Uniform $\Lambda_0$ bound via external $L^\infty$ estimate and semigroup contraction]
\label{lem:uniform_L0_Hairer_style}
Let $K$ denote the one--dimensional Euclidean heat kernel
\[
K(t;z,\bar z):=(4\pi t)^{-1/2}\exp\!\Bigl(-\frac{(z-\bar z)^2}{4t}\Bigr),
\qquad t>0.
\]
Let $T$ be a finite rooted tree with internal vertices $\{1,\dots,r\}$ (root $1$), internal edges
$\mathcal I(T)$, and external vertices indexed by a finite set $\mathfrak{s}$ with $|\mathfrak{s}|\ge 1$,
together with incidence relation $\mathcal E(T)\subset \{1,\dots,r\}\times\mathfrak{s}$.

Then there exists a constant $C_T<\infty$, depending only on the combinatorics of $T$, such that,
uniformly in $\Lambda_0$,
\begin{equation}\label{eq:uniform_L0_Hairer_min}
\sup_{(y_j)_{j\in\mathfrak{s}}}
\int_{\mathbb R}
\mathcal A^{\Lambda,\Lambda_0}_{l,n}(T)
\Bigl(z_1;(\tau_j,y_j)_{j\in\mathfrak{s}}\Bigr)\,dz_1
\;\le\;
C_T\,\delta_l^{-\frac{|\mathfrak{s}|-1}{2}}\,
\tau_*^{-\frac{|\mathfrak{s}|-1}{2}},
\end{equation}
where $\tau_*:=\min_{j\in\mathfrak{s}}\tau_j$.
\end{lemma}

\begin{proof}
Fix $j_0\in\mathfrak{s}$ and let $i_0$ be the unique internal vertex such that
$\langle i_0,j_0\rangle\in\mathcal E(T)$.
We first bound all external kernels except the one indexed by $j_0$ in $L^\infty$.
Since
\begin{equation}\label{eq:K_Linfty}
K(\delta_l\tau_j;z,y)\le (4\pi\delta_l\tau_j)^{-1/2},
\qquad j\in\mathfrak{s},
\end{equation}
inserting \eqref{eq:K_Linfty} into \eqref{AmT} for every $j\neq j_0$ yields
\begin{multline}\label{eq:reduce_to_one_external}
\mathcal A^{\Lambda,\Lambda_0}_l(T)
\Bigl(z_1;(\tau_j,y_j)_{j\in\mathfrak{s}}\Bigr)
\le
\Biggl(\prod_{\substack{j\in\mathfrak{s}\\ j\neq j_0}}
(4\pi\delta_l\tau_j)^{-1/2}\Biggr)\\\times
\sup_{\Lambda\le \Lambda_{ab}\le \Lambda_0}
\int_{\mathbb R^{r-1}} d\vec z_{2,r}\;
\Biggl(
\prod_{\langle a,b\rangle\in\mathcal I(T)}
K\!\left(\frac{\delta_l}{\Lambda_{ab}^2};z_a,z_b\right)
\Biggr)
K(\delta_l\tau_{j_0};z_{i_0},y_{j_0}).
\end{multline}
We claim that for every choice of nonnegative time parameters $(t_e)_{e\in\mathcal I(T)}$
and every $t>0$,
\begin{equation}\label{eq:tree_markov_identity}
\int_{\mathbb R^{r}} dz_1\cdots dz_r\;
\Biggl(
\prod_{\langle a,b\rangle\in\mathcal I(T)}
K(t_{\langle a,b\rangle};z_a,z_b)
\Biggr)
K(t;z_{i_0},y)
=1 ,
\end{equation}
for every $y\in\mathbb R$. \\
\indent Assuming \eqref{eq:tree_markov_identity} for the moment, we apply it to the integrand in
\eqref{eq:reduce_to_one_external} with $t=\delta_l\tau_{j_0}$ and
$t_{\langle a,b\rangle}=\delta_l/\Lambda_{ab}^2$. It follows that the integral over
$z_1$ and $\vec z_{2,r}$ is bounded by $1$, uniformly in the internal scales and hence
uniformly in $\Lambda_0$. Taking the supremum over $(y_j)$ does not change the bound,
since the right-hand side is independent of $y$. Absorbing $(4\pi)^{-(|\mathfrak{s}|-1)/2}$
into the constant yields \eqref{eq:uniform_L0_Hairer_min} after choosing $j_0$ such that
$\tau_{j_0}=\tau_*$.\\
\indent It therefore remains to prove \eqref{eq:tree_markov_identity}.
We argue by induction on the number $r$ of internal vertices.\\
For $r=1$, the left-hand side equals
\[
\int_{\mathbb R}K(t;z_1,y)\,dz_1=1.
\]

Assume now that \eqref{eq:tree_markov_identity} holds for all rooted trees with
$r-1$ internal vertices and consider a tree with $r$ internal vertices.
Since $T$ is a tree, there exists an index $u\in\{1,\dots,r\}$ with $u\neq i_0$
such that the corresponding internal vertex is connected to the rest of the tree
by a single internal edge (whenever $r\ge2$).
Let $v$ denote the unique index such that $\langle u,v\rangle\in\mathcal I(T)$.\\
In the integrand of \eqref{eq:tree_markov_identity}, the variable $z_u$ appears only
in the factor
\[
K(t_{\langle u,v\rangle};z_u,z_v).
\]
Integrating out $z_u$ and using the identity
\[
\int_{\mathbb R}K(s;z_u,z_v)\,dz_u=1
\]
removes this factor without altering the remaining integrand. The resulting
expression is precisely of the form \eqref{eq:tree_markov_identity} for the
reduced tree obtained by removing $u$ and the edge $\langle u,v\rangle$.
The induction hypothesis then yields $1$, completing the proof of
\eqref{eq:tree_markov_identity}.
\end{proof}
\begin{corollary}[Uniform short-time $L^\infty$ bounds]\label{UniLinfty}
Let $n\ge1$ and $0<\tau_i<1$ for $1\le i\le n$. Set
\[
\tau_*:=\min_{1\le i\le n}\tau_i .
\]
Then, uniformly for $\Lambda\le\Lambda_0$, one has
\begin{equation}\label{eq:LinftyAmp}
\sup_{(y_1,\dots,y_n)\in\mathbb{R}^n}
\mathcal{A}_{l,n}^{\Lambda,\Lambda_0}
\bigl((\tau_i,y_i);\chi^{\sigma}\bigr)
\;\lesssim\;\|\chi\|_{\infty}^{|\sigma|}~
\tau_*^{-\frac{n-1}{2}} .
\end{equation}
In particular,
\begin{equation}\label{eq:A-bound-cor}
\sup_{(y_1,\dots,y_n)\in\mathbb{R}^n}
\left(
\mathcal{A}_{l,n}^{\Lambda,\Lambda_0}
+\Lambda_m^{-1}\hat{\mathcal{A}}_{l,n}^{\Lambda,\Lambda_0}
\right)
\bigl((\tau_i,y_i);\chi^{\sigma}\bigr)
\;\lesssim\;
\\\|\chi\|_{\infty}^{|\sigma|}~\left(1+\frac{\tau_*^{-1/2}}{m}\right)\,
\tau_*^{-\frac{n-1}{2}} .
\end{equation}
The implicit constants depend only on $l$, $n$ and $|\mathfrak{s}|$.
\end{corollary}
\begin{proof}
Set
\[
\tau_*:=\min_{i\in\{1,\dots,n\}}\tau_i\in(0,1).
\]
All implicit constants below are uniform for $\Lambda\le \Lambda_0$ and in
$(\tau_i,y_i)$, and may depend on $l$ and $n$.
Using \eqref{eq:borneamplitudechi} and \eqref{eq:uniform_L0_Hairer_min}, we obtain
\begin{align}\label{A25}
\sup_{(y_1,\dots,y_n)\in\mathbb{R}^n}
\mathcal{A}_{l,n}^{\Lambda,\Lambda_0}
\bigl((\tau_i,y_i);\chi^{\sigma}\bigr)
&\leq \|\chi\|_{\infty}^{|\sigma|}~
\sup_{(y_1,\dots,y_n)\in\mathbb{R}^n}
\mathcal{A}_{l,n}^{\Lambda,\Lambda_0}
\bigl((\tau_i,y_i);\mathds{1}^{\sigma}\bigr)
\nonumber\\
&\leq \|\chi\|_{\infty}^{|\sigma|}~
\sum_{\substack{T\in\mathcal{T}^s_1\\0\le v(T)\le l+\frac{n}{2}-2}}
\sup_{(y_1,\dots,y_n)\in\mathbb{R}^n}
\mathcal{A}_{l,n}^{\Lambda,\Lambda_0}(T)
\bigl((\tau_i,y_i)\bigr)
\nonumber\\
&\lesssim \|\chi\|_{\infty}^{|\sigma|}~
\tau_*^{-\frac{n-1}{2}} .
\end{align}
Applying Lemma~\ref{lem:Linfty_forest} together with \eqref{A25} and using $\Lambda_m\ge m$, we obtain
\begin{align*}
\sup_{(y_1,\dots,y_n)\in\mathbb{R}^n}
\left(
\mathcal{A}_{l,n}^{\Lambda,\Lambda_0}
+\Lambda_m^{-1}\hat{\mathcal{A}}_{l,n}^{\Lambda,\Lambda_0}
\right)
\bigl((\tau_i,y_i);\chi^{\sigma}\bigr)
&\lesssim
\left(
1+\frac{\tau_*^{-1/2}}{m}
\right)~\sup_{(y_1,\dots,y_n)\in\mathbb{R}^n}
\mathcal{A}_{l+1,n}^{\Lambda,\Lambda_0}
\bigl((\tau_i,y_i);\chi^{\sigma}\bigr)
\nonumber\\
&\lesssim
\tau_*^{-\frac{n-1}{2}}
\left(
1+\frac{\tau_*^{-1/2}}{m}
\right).
\end{align*}
\end{proof}
\section{Perturbative renormalization in position space}\label{SecProof}

The proofs of Theorems~\ref{BesovNorm} and \ref{BesovNorm2} are based on an inductive argument
within the Polchinski renormalization framework. The main point of the
present section is to reformulate the required uniform bounds that establishes both Theorems in terms
of a scale of distribution spaces adapted to power counting. More precisely, to each scaling degree $\rho$ and loop order $l$ we
associate a space
\[
\mathcal D^{\Lambda_m}_{\rho,l}(\mathbb M^n),
\qquad
\Lambda_m=\Lambda+m,
\]
whose seminorms encode the expected behaviour of a distribution at the
scale $\Lambda$. In this formulation, the estimates stated in
Theorem~\ref{BesovNorm} are equivalent to proving that the correlators $
\mathcal L_{l,n}^{\Lambda,\Lambda_0}
$ satisfy the bounds associated with the corresponding power-counting
space $\mathcal D^{\Lambda_m}_{\rho,l}(\mathbb M^n)$, uniformly in the
ultraviolet cutoff $\Lambda_0$.

The Polchinski flow is then interpreted as a stability property of this
scale of spaces. Indeed, the linear and bilinear contributions appearing
in the flow equation are encoded by operators $
\mathbf L^\Lambda$ and $\mathbf B^\Lambda_\pi$,
and the key analytic strategy consists in proving that these operators map
the spaces $\mathcal D^{\Lambda_m}_{\rho,l}$ into spaces with the
appropriate shifted scaling degree.

The proof is organised as an induction on the total degree $n+2l$, and for fixed $n+2l$ by increasing loop order $l$. The tree-level bounds are
verified directly. In the inductive step, one differentiates the
correlators with respect to the flow parameter $\Lambda$ and estimates
each contribution generated by the flow equation using the stability
properties of the Polchinski operators. The resulting differential
bounds are then integrated in $\Lambda$.

The integration procedure is determined by power counting. Terms with
strictly negative scaling degree are irrelevant and are integrated from
$\Lambda$ to the ultraviolet scale $\Lambda_0$. When the scaling degree
is nonnegative, one first isolates the relevant part by a Taylor
expansion, either in the external momenta or in the first spatial
variable. The remainder then acquires improved scaling behaviour and can
be treated by the irrelevant integration argument. This mechanism forms
the core of the renormalization procedure developed below.

The purpose of this section is therefore to establish the abstract
analytic framework underlying the proof of
Theorems~\ref{BesovNorm} and \ref{BesovNorm2}: the definition of the spaces
$\mathcal D^{\Lambda_m}_{\rho,l}$, their stability under the
Polchinski operators, the Taylor remainder estimates for relevant
distributions, and the integration results in the scale parameter
$\Lambda$.
\subsection{Scale-dependent power-counting spaces}
In this subsection, we introduce a family of seminorms adapted to the
renormalization-group analysis in position space. These seminorms measure
distributions after convolution with heat kernels and compare them to
the canonical tree amplitudes governing the perturbative expansion.\\
The polynomial weights introduced below arise naturally from the
renormalization procedure in position space. More precisely, Taylor
expansions around the first spatial argument generate polynomial factors
in the variables $\tau_i$, while derivatives acting on the heat kernels
produce additional powers of these variables. Controlling the order of
the resulting polynomials is therefore essential for the regularity
analysis of the correlators. The superscript $\sigma$ keeps track of the
insertion of the cutoff factors
\[
\prod_{i\in\sigma}\chi(z_i)
\]
into the corresponding $n$-point distributions.
\begin{definition}\label{DefPol}
Let $\mathfrak s$ be a finite index set and let $q\in\mathbb N$. We define
\begin{equation}\label{DefQ}
Q_q\bigl((x_i)_{i\in\mathfrak s}\bigr)
:=
\prod_{i\in\mathfrak s}
\left(
1+x_i+x_i^2+\cdots+x_i^q
\right).
\end{equation}
For $\sigma\subseteq\mathfrak s$, we furthermore define
\begin{equation}\label{DefQsigma}
Q_q^\sigma
\bigl(
(x_i)_{i\in\mathfrak s}
\bigr)
:=
\mathbf 1_{\{\sigma\neq\varnothing\}}
Q_q
\bigl(
(x_i)_{i\in\mathfrak s}
\bigr)
+
\mathbf 1_{\{\sigma=\varnothing\}}
Q_1
\bigl(
(x_i)_{i\in\mathfrak s}
\bigr),
\end{equation}
where $\mathbf 1_A$ denotes the indicator function of the set $A$.
\end{definition} 
    The distinction between
\(
\sigma=\varnothing
\)
and
\(
\sigma\neq\varnothing
\)
reflects the different power-counting behaviour of correlators with and
without cutoff insertions. The modified polynomial
\(
Q_q^\sigma
\)
is defined such that it encodes this distinction directly at the level of the
seminorms.

The following lemma collects the basic stability properties of the
polynomial $Q_q$. These weights are stable under the
operations naturally arising in the flow equations, namely restriction
to subsets of variables, multiplication, and the insertion of additional
factors of low degree. These estimates will be used repeatedly in the
sequel.
\begin{lemma}[Stability of polynomial weights]\label{QLemma}
Let $q\in\mathbb N$, and let $Q_q$ be the polynomial defined in
Definition~\ref{DefPol}. Assume throughout that all variables are
non-negative.

\medskip

\noindent
\textit{(i)}
Let $\{\mathfrak s_1,\mathfrak s_2\}$ be a partition of
$\mathfrak s$, and let $q_1,q_2\in\mathbb N$ satisfy
$q=q_1+q_2$. Then, 
\begin{align}
Q_{q_1}
\bigl(
(x_i)_{i\in\mathfrak s_1}
\bigr)
Q_{q_2}
\bigl(
(x_i)_{i\in\mathfrak s_2}
\bigr)
&\lesssim
Q_q
\bigl(
(x_i)_{i\in\mathfrak s}
\bigr),
\label{Prop1Q}
\\
Q_q
\bigl(
(x_i)_{i\in\mathfrak s_1}
\bigr)
Q_q
\bigl(
(x_i)_{i\in\mathfrak s_2}
\bigr)
&\lesssim
Q_q
\bigl(
(x_i)_{i\in\mathfrak s}
\bigr),
\label{Prop2Q}
\\
Q_q
\bigl(
(x_i)_{i\in\mathfrak s},y
\bigr)
&\lesssim
Q_q
\bigl(
(x_i)_{i\in\mathfrak s}
\bigr),\quad \forall~0\le y\le1
\label{Prop2QBis}
\\
Q_q
\bigl(
(x_i)_{i\in\mathfrak s}
\bigr)
&\le
Q_{q'}
\bigl(
(x_i)_{i\in\mathfrak s}
\bigr),
\qquad
q'\ge q,
\label{Prop3Q}
\\
Q_q
\bigl(
(x_i)_{i\in\mathfrak s\setminus\{j\}}
\bigr)
&\le
Q_q
\bigl(
(x_i)_{i\in\mathfrak s}
\bigr).
\label{Prop6Q}
\end{align}

\medskip

\noindent
\textit{(ii)}
For every $j\in\mathfrak s$,
\begin{equation}\label{Prop4Q}
x_j\,
Q_q
\bigl(
(x_i)_{i\in\mathfrak s\setminus\{j\}}
\bigr)
\le
Q_q
\bigl(
(x_i)_{i\in\mathfrak s}
\bigr).
\end{equation}

\medskip

\noindent
\textit{(iii)}
For every $k\in\mathfrak s$,
\begin{equation}\label{Prop5Q}
x_k\,
Q_q
\bigl(
(x_i)_{i\in\mathfrak s}
\bigr)
\lesssim
Q_{q+1}
\bigl(
(x_i)_{i\in\mathfrak s}
\bigr).
\end{equation}
\end{lemma}
\begin{proof}
All variables being non-negative, the factors
\(1+x_i+\cdots+x_i^q\)
are monotone in \(q\). Properties
\eqref{Prop1Q},
\eqref{Prop2Q},
\eqref{Prop3Q}
and
\eqref{Prop6Q}
follow immediately from this monotonicity and from the fact that
\(\{\mathfrak s_1,\mathfrak s_2\}\)
forms a partition of
\(\mathfrak s\). Since \(0\le y\le1\),
\[
1+y+\cdots+y^q\le q+1,
\]
which yields \eqref{Prop2QBis}. Property \eqref{Prop4Q} follows from
\[
x_j\le1+x_j+\cdots+x_j^q,
\]
while \eqref{Prop5Q} follows from
\[
x_k(1+x_k+\cdots+x_k^q)
\le
1+x_k+\cdots+x_k^{q+1}.
\]
The claim follows.
\end{proof}
\noindent In the sequel, we use the following notations
\[
\mathcal I_n:=\{2,\dots,n\},
\qquad
\sigma\subseteq\mathfrak s\subseteq\mathcal I_n,
\]
and, for
$(\tau_i,y_i)\in\mathbb R_+\times\mathbb R$ and $z_i\in\mathbb{R}$
\[
\underline y_{\mathfrak s}
:=
\bigl(
\tau_i,y_i
\bigr)_{i\in\mathfrak s},\qquad
\tau_{\mathfrak s}^{\Lambda_m}
:=
\left(
\frac{\tau_i^{-1/2}}{\Lambda_m}
\right)_{i\in\mathfrak s},
\qquad \vec{z}_{1,n}:=\left(z_1,\cdots,z_n\right).
\]
We also set
\[
\mathbb M:=\mathbb R\times\mathbb R^3.
\]
\begin{definition}[Power-counting seminorms]\label{DefPowerCountingSeminorms}
Fix $\rho\in\mathbb Z$, $l\in\mathbb N_0$, $\Lambda\ge0$, and $m>0$,
and set $
\Lambda_m:=\Lambda+m$.
We consider distributions $
T(\vec{z}_{1,n};\underline p_n)$
on $\mathbb M^n$, partially Fourier transformed in the spatial
variables. We denote by $\overline{\mathcal K}$ the set of non-constant
elements of $\mathcal K$.

\medskip

\noindent
Let first $n\ge4$ and $l\ge1$. For $\chi\in\mathcal K$, we define
\begin{multline}\label{e203}
\mathcal N^{\Lambda_m}_{\rho,l}
\bigl(
\chi T(\underline p_n)
\bigr)
:=
\Lambda_m^{-\rho}
\sup_{z_1\in\mathbb R}
\sup_{\substack{
\sigma\subseteq\mathfrak s\subseteq\mathcal I_n\\
(\tau_i,y_i)\in\mathbb R_+\times\mathbb R
}}
\left(
Q_{|\sigma|}^{\sigma}
\bigl(
\tau_{\mathfrak s}^{\Lambda_m}
\bigr)
\right)^{-1}
\frac{
\left|
\left(
K^{\otimes\mathfrak s}
\star
\chi^{\otimes\sigma}T
\right)
\bigl(
z_1;\underline y_{\mathfrak s};\underline p_n
\bigr)
\right|
}{
\langle
\mathcal{A},\hat{\mathcal{A}}\rangle_{l,n}
\bigl(
z_1;\underline y_{\mathfrak s};\chi^\sigma
\bigr)
}.
\end{multline}
For $l=0$ and $n\ge4$, we set
\begin{equation}\label{e203Bis}
\mathcal N^{\Lambda_m}_{\rho,0}
\bigl(
\chi T(\underline p_n)
\bigr)
:=
\Lambda_m^{-\rho}
\sup_{z_1\in\mathbb R}
\sup_{\substack{
\sigma\subseteq\mathfrak s\subseteq\mathcal I_n\\
(\tau_i,y_i)\in\mathbb R_+\times\mathbb R
}}
\frac{
\left|
\left(
K^{\otimes\mathfrak s}
\star
\chi^{\otimes\sigma}T
\right)
\bigl(
z_1;\underline y_{\mathfrak s};\underline p_n
\bigr)
\right|
}{
\mathcal A_{0,n}^{\Lambda,\Lambda_0}
\bigl(
z_1;\underline y_{\mathfrak s};\chi^\sigma
\bigr)
}.
\end{equation}
The case $\sigma=\varnothing$ corresponds to the absence of cutoff
insertions. In all supremums above, the choices
$\sigma=\varnothing$ and $\mathfrak s=\varnothing$ are allowed. In this
case, no heat kernel is inserted and the numerator reduces to
\[
t(z_1;\underline p_n)
:=
\int_{\mathbb R^{n-1}}
T(\vec{z}_{1,n};\underline p_n)\,
d\vec z_{2,n}.
\]
Moreover, since all heat kernels are integrated out, the denominator
\[
\mathcal A_{l,n}^{\Lambda,\Lambda_0}\left(z_1;\varnothing;1\right)
\]
reduces to a positive constant depending only on $l$ and $n$. Thus this
contribution to the seminorm controls
\[
\Lambda_m^{-\rho}
|t(z_1;\underline p_n)|.
\]
For $n=2$ and $l\ge1$, we define for $\rho\in\mathbb{N}_0$ 
\begin{equation}\label{e203-n2}
\mathcal N^{\Lambda_m}_{\rho,l}
\bigl(
T(\underline{p})
\bigr)
:=
\Lambda_m^{-\rho}
\sup_{z_1\in\mathbb R}
\sup_{(\tau,y)\in\mathbb R_+\times\mathbb R}
\left(
\sum_{i=0}^{\rho+1}
\left(
\frac{\tau^{-1/2}}{\Lambda_m}
\right)^i
\right)^{-1}
\frac{
\left|
(K\star T)
(z_1;\tau,y;\underline{p})
\right|
}{
\mathcal A_{l,2}^{\Lambda,\Lambda_0}
(z_1;\tau,y)
}.
\end{equation}
If $\rho<0$ the polynomial reduces to a constant and we define 
\begin{equation}\label{e203-n2Bis}
\mathcal N^{\Lambda_m}_{\rho,l}
\bigl(
T(\underline{p})
\bigr)
:=
\Lambda_m^{-\rho}
\sup_{z_1\in\mathbb R}
\sup_{(\tau,y)\in\mathbb R_+\times\mathbb R}
\frac{
\left|
(K\star T)
(z_1;\tau,y;\underline{p})
\right|
}{
\mathcal A_{l,2}^{\Lambda,\Lambda_0}
(z_1;\tau,y)
}.
\end{equation}
In the supremum above, the empty configuration is also allowed. In this
case, no heat kernel is inserted and the numerator reduces to
\[
t(z_1;\underline{p})
:=
\int_{\mathbb R}
T(z_1,z_2;\underline{p})\,dz_2,
\]
while the denominator $
\mathcal A_{l,2}^{\Lambda,\Lambda_0}(z_1;\varnothing)
$
reduces to a positive constant depending only on $l$.
\\
\indent We also introduce seminorms involving discrete derivatives. These are
defined for $n\ge4$ and $l\ge1$. We define
the discrete derivative of $T$ in the $i$-th variable, relative to the
base point $z_1$, by
\[
\bigl\langle
\Delta_{z_1}^{(i)}T,
\phi_2\otimes\cdots\otimes\phi_n
\bigr\rangle
=
\bigl\langle
T,\Delta_{z_1}^{(i)}\left(\phi_2\otimes\cdots\otimes\phi_n\right)
\bigr\rangle
\]
where 
$$\Delta_{z_1}^{(i)}\left(\phi_2\otimes\cdots\otimes\phi_n\right):=\bigl(\phi_i(z_i)-\phi_i(z_1)\bigr)
\prod_{\substack{k=2\\k\neq i}}^n
\phi_k(z_k)~.$$
\smallskip

\noindent
\textit{(i) Derivative on the kernel:} we define for $j\in\mathcal{I}_n$
\begin{equation}\label{e203-DeltaK}
\mathcal N^{\Lambda_m}_{\rho,l}
\bigl(
\Delta^{(j)}T(\underline p_n)
\bigr)
:=
\Lambda_m^{-\rho+1}
\sup_{z_1\in\mathbb R}
\sup_{\substack{
\mathfrak s\subseteq\mathcal I_n\\
(\tau_i,y_i)_{i\in\mathfrak s}\in(\mathbb R_+\times\mathbb R)^{\mathfrak s}}}~\mathbf 1_{\{j\in\mathfrak s\}}~
\sqrt{\tau_j}~
Q_1^{-1}
\bigl(
\tau_{\mathfrak s\setminus\{j\}}^{\Lambda_m}
\bigr)
\frac{
\left|
\left(
K^{\otimes\mathfrak s}
\star
\Delta_{z_1}^{(j)}T
\right)
\bigl(
z_1;\underline y_{\mathfrak s};\underline p_n
\bigr)
\right|
}{
{\mathcal A}_{l,n}^{\Lambda,\Lambda_0}
\bigl(
z_1;\underline y_{\mathfrak s}
\bigr)
}.
\end{equation}

\smallskip

\noindent
\textit{(ii) Derivative on the cutoff.}
For $\chi\in\overline{\mathcal K}$, we define for $j\in\mathcal{I}_n$
\begin{multline}\label{e203-Deltachi}
\mathcal N^{\Lambda_m}_{\rho,l}
\bigl(
\Delta^{(j)}\chi\,T(\underline p_n)
\bigr)\\
:=
\Lambda_m^{-\rho+1}
\sup_{z_1\in\mathbb R}
\sup_{\substack{
\sigma\subseteq\mathfrak s\subseteq\mathcal I_n\\
(\tau_i,y_i)\in\mathbb R_+\times\mathbb R\\
q:=|\sigma|
}}~\mathbf{1}_{\{j\in\sigma\}}
\left(
Q_{q-1}^{\sigma\setminus\{j\}}
\bigl(
\tau_{\mathfrak s}^{\Lambda_m}
\bigr)
\right)^{-1}
\frac{
\left|
\left(
K^{\otimes\mathfrak s}
\star
\left(\Delta_{z_1}^{(j)}\chi^{\otimes\sigma}\right)T
\right)\bigl(
z_1;\underline y_{\mathfrak s};\underline p_n
\bigr)
\right|
}{
\hat{\mathcal A}_{l,n}^{\Lambda,\Lambda_0}\left(z_1;\underline{y}_{\mathfrak{s}};\chi^{\sigma}\right)
}.
\end{multline}
By convention,
\[
\Delta_{z_1}^{(j)}\chi^{\otimes\sigma}=0
\qquad\text{whenever }j\notin\sigma.
\]
Note that the denominators are strictly
positive, since they are given by sums of tree amplitudes, hence by sums
of products of heat kernels.\\
Finally, we set
\begin{align}
\mathcal N^{\Lambda_m}_{\rho,l;\infty}
\bigl(
T(\underline p_n)
\bigr)
&:=
\sup_{\chi\in\mathcal K}
\mathcal N^{\Lambda_m}_{\rho,l}
\bigl(
\chi T(\underline p_n)
\bigr),
\label{sup2}
\\
\mathcal N^{\Lambda_m}_{\rho,l;\infty}
\bigl(
\Delta T(\underline p_n)
\bigr)
&:=
\sup_{j\in\mathcal{I}_n}\left[\mathcal N^{\Lambda_m}_{\rho,l}
\bigl(
\Delta^{(j)} T(\underline p_n)
\bigr)
+
\sup_{\substack{\chi\in\overline{\mathcal K}}}
\mathcal N^{\Lambda_m}_{\rho,l}
\bigl(
\Delta^{(j)}\chi\,T(\underline p_n)
\bigr)\right].
\label{Ninf2}
\end{align}
\end{definition}
\begin{definition}[Power-counting space]
A distribution $T$ on $\mathbb M^n$ is said to belong to $
\mathcal D^{\Lambda_m}_{\rho,l}(\mathbb M^n)$ if the seminorm
\[
\mathcal N^{\Lambda_m}_{\rho,l;\infty}
\bigl(
T(\underline p_n)
\bigr)
\]
is finite, and if there exist polynomials
$\mathcal P_{l-1}$ and $\widetilde{\mathcal P}$, both with positive
coefficients, such that
\begin{equation}\label{DefDcal}
\mathcal N^{\Lambda_m}_{\rho,l;\infty}
\bigl(
T(\underline p_n)
\bigr)
\le
\mathcal P_{l-1}
\left(
\log\frac{\Lambda_m}{m}
\right)
\widetilde{\mathcal P}
\left(
\frac{\|\underline p_n\|}{\Lambda_m}
\right),
\qquad
\forall\,\underline p_n\in\mathbb R^{3n}.
\end{equation}
Here, the degree of $\mathcal P_{l-1}$ is at most $l-1$, while the
degree of $\widetilde{\mathcal P}$ depends only on $l$ and $n$.

Similarly, we write
\[
\Delta T
\in
\mathcal D^{\Lambda_m}_{\rho,l}(\mathbb M^n)
\]
whenever
\[
\mathcal N^{\Lambda_m}_{\rho,l;\infty}
\bigl(
\Delta T(\underline p_n)
\bigr)
\]
satisfies the bound \eqref{DefDcal}. The coefficients of the polynomials
$\mathcal P_{l-1}$ and $\widetilde{\mathcal P}$ may depend on
$\rho$, $l$, and $n$, but are independent of
$\Lambda$, $\Lambda_0$, $m$ and $\underline p_n$.
\end{definition}
The spaces
\(
\mathcal D^{\Lambda_m;(2)}_{\rho,l}
\)
are designed for the analysis of irrelevant two-point contributions.
Indeed, after multiplication by sufficiently high powers of the relative
coordinate
\(
(z_1-z_2),
\)
the resulting distributions acquire improved regularity in both
variables. Since this gain plays a crucial role in the inductive
estimates, we introduce seminorms in which the variables \(z_1\) and
\(z_2\) are kept explicit, while the remaining variables are treated as
in the ordinary power-counting spaces.
\begin{definition}[Second-order power-counting spaces]
\label{DefSecondOrderSpaces}
We define the space $
\mathcal D^{\Lambda_m;(2)}_{\rho,l}(\mathbb M^n)
$
exactly as in Definition~\ref{DefPowerCountingSeminorms}, except that
the seminorms are replaced by the following seminorms in which the first two variables are fixed.

\medskip

\noindent
For $n\ge4$ and $l\ge0$, we set
\begin{equation}\label{e203-second}
\mathcal N^{\Lambda_m;(2)}_{\rho,l}
\bigl(
T(\underline p_n)
\bigr)
:=
\Lambda_m^{-\rho}
\sup_{(z_1,z_2)\in\mathbb R^2}
\sup_{\substack{
\mathfrak s\subseteq\{3,\dots,n\}\\
(\tau_i,y_i)\in\mathbb R_+\times\mathbb R
}}
Q_1^{-1}
\bigl(
\tau_{\mathfrak s}^{\Lambda_m}
\bigr)
\frac{
\left|
\left(
K^{\otimes\mathfrak s}
\star T
\right)
\bigl(
z_1,z_2;\underline y_{\mathfrak s};\underline p_n
\bigr)
\right|
}{
\mathcal A_{l,n}^{\Lambda,\Lambda_0}
\bigl(
z_1,z_2;\underline y_{\mathfrak s}
\bigr)
}.
\end{equation}
For $n=2$ and $l\ge1$, we define
\begin{equation}\label{e203-second-n2}
\mathcal N^{\Lambda_m;(2)}_{\rho,l}
\bigl(
T(\underline{p})
\bigr)
:=
\Lambda_m^{-\rho}
\sup_{(z_1,z_2)\in\mathbb R^2}
\frac{
\left|
T(z_1,z_2;\underline{p})
\right|
}{
\mathcal A_{l,2}^{\Lambda,\Lambda_0}(z_1,z_2)
}.
\end{equation}
A distribution
$
T\in\mathcal S'(\mathbb M^n)
$ belongs to $
\mathcal D^{\Lambda_m;(2)}_{\rho,l}(\mathbb M^n)
$
if the seminorm $\mathcal N^{\Lambda_m;(2)}_{\rho,l}
\bigl(
T(\underline p_n)
\bigr)$ satisfies the bound \eqref{DefDcal}.
\end{definition}
For every $\Lambda\in[0,\Lambda_0]$, the seminorms
$\mathcal N^{\Lambda_m}_{\rho,l}$ define scale-dependent spaces
$\mathcal D^{\Lambda_m;(i)}_{\rho,l}$.
All statements below are understood pointwise in $\Lambda$.
In particular, whenever operators such as
$\mathbf L^\Lambda$ or $\mathbf B^\Lambda_{\pi}$
act on distributions
$T\in\mathcal D^{\Lambda_m}_{\rho,l}$,
the parameter $\Lambda$ entering the operator and the seminorms is the same.
\begin{definition}\label{FamilyT}[Uniform power-counting bounds]
Let $\Lambda_0\ge 0$ and
\[
\Lambda\in[0,\Lambda_0]\longmapsto T^\Lambda
\]
be a family of distributions on $\mathbb M^n$. We say that this family
satisfies uniform power-counting bounds of degree $\rho=\rho(T)$ and
loop order $l$ if, for every $\Lambda\in[0,\Lambda_0]$, the distribution
$T^\Lambda$ is measured in the space
\[
\mathcal D^{\Lambda_m}_{\rho,l}(\mathbb M^n),
\qquad
\Lambda_m=\Lambda+m,
\]
with the same value of $\Lambda$ as in $T^\Lambda$, and if there exist
polynomials $\mathcal P_{l-1}$ and $\widetilde{\mathcal P}$, whose
degrees and coefficients depend only on $l$ and $n$, such that
\begin{equation}
\mathcal N^{\Lambda_m}_{\rho,l;\infty}
\bigl(
T^\Lambda(\underline p_n)
\bigr)
\lesssim
\mathcal P_{l-1}
\left(
\log\frac{\Lambda_m}{m}
\right)
\widetilde{\mathcal P}
\left(
\frac{\|\underline p_n\|}{\Lambda_m}
\right),
\end{equation}
with an implicit constant independent of $\Lambda$ and $\Lambda_0$. The same definition applies to the spaces
\[
\mathcal D^{\Lambda_m;(2)}_{\rho,l}(\mathbb M^n).
\]
\end{definition}
The purpose of the preceding definitions is that
the operations appearing in the Polchinski flow
equation act continuously on these spaces.
The remainder of this section establishes
precisely these stability properties:
multiplication by momentum factors,
the linear operator,
the quadratic operator,
Taylor remainders,
and integration in the flow parameter.
\begin{proposition}\label{PropPoSp}
Let $\beta\in\mathbb N_0$ and let
$T\in\mathcal D^{\Lambda_m}_{\rho,l}(\mathbb M^n)$. Let
$g^{\Lambda_m}\in \mathcal{C}^\infty(\mathbb R^{3n})$ and assume that there exists
a polynomial $\widetilde P$ with positive coefficients, independent of
$m$ and $\Lambda$, such that
\begin{equation}\label{eq:gBound}
|g^{\Lambda_m}(\underline p_n)|
\le
\Lambda_m^\beta~
\widetilde P
\left(
\frac{\|\underline p_n\|}{\Lambda_m}
\right).
\end{equation}
Then
\begin{equation}\label{StabPo}
g^{\Lambda_m}(\underline p_n)
T(\underline z_n;\underline p_n)
\in
\mathcal D^{\Lambda_m}_{\rho+\beta,l}(\mathbb M^n).
\end{equation}
In particular, for every $i\in\{1,\dots,n\}$ and every momentum component
$\mu$,
\[
|p_{i,\mu}|^\beta T
\in
\mathcal D^{\Lambda_m}_{\rho+\beta,l}(\mathbb M^n).
\]
\end{proposition}

\begin{proof}
Since $g^{\Lambda_m}$ depends only on the momentum variables, it factors
out of all convolutions in the spatial variables. Hence, by the
definition of the seminorms,
\[
\mathcal N^{\Lambda_m}_{\rho+\beta,l;\infty}
\bigl(
g^{\Lambda_m}T(\underline p_n)
\bigr)
\le
\Lambda_m^{-(\rho+\beta)}
|g^{\Lambda_m}(\underline p_n)|
\,
\Lambda_m^\rho
\mathcal N^{\Lambda_m}_{\rho,l;\infty}
\bigl(
T(\underline p_n)
\bigr).
\]
Using \eqref{eq:gBound}, this gives
\[
\mathcal N^{\Lambda_m}_{\rho+\beta,l;\infty}
\bigl(
g^{\Lambda_m}T(\underline p_n)
\bigr)
\lesssim
\widetilde P
\left(
\frac{\|\underline p_n\|}{\Lambda_m}
\right)
\mathcal N^{\Lambda_m}_{\rho,l;\infty}
\bigl(
T(\underline p_n)
\bigr).
\]
Since $T\in\mathcal D^{\Lambda_m}_{\rho,l}(\mathbb M^n)$, the last
factor is bounded by
\[
\mathcal P_{l-1}
\left(
\log\frac{\Lambda_m}{m}
\right)
\widehat{\mathcal P}
\left(
\frac{\|\underline p_n\|}{\Lambda_m}
\right).
\]
Absorbing the product of the two momentum polynomials into a new
polynomial gives
\[
\mathcal N^{\Lambda_m}_{\rho+\beta,l;\infty}
\bigl(
g^{\Lambda_m}T(\underline p_n)
\bigr)
\lesssim
\mathcal P_{l-1}
\left(
\log\frac{\Lambda_m}{m}
\right)
\widehat{\mathcal P}
\left(
\frac{\|\underline p_n\|}{\Lambda_m}
\right).
\]
This proves \eqref{StabPo}.\\
The particular case follows by taking
\[
g^{\Lambda_m}(\underline p_n)=p_{i,\mu}^\beta,
\]
since
\[
|p_{i,\mu}|^\beta
\le
\Lambda_m^\beta
\left(
\frac{\|\underline p_n\|}{\Lambda_m}
\right)^\beta.
\]
\end{proof}
\subsection{Uniform power-counting bounds for the correlators}

The main objective of this section is to establish uniform
power-counting bounds for the families of correlators
\[
\Lambda\longmapsto
\mathcal L_{l,n}^{\Lambda,\Lambda_0}.
\]
As explained above, the power-counting spaces
\(
\mathcal D^{\Lambda_m}_{\rho,l}
\)
and
\(
\mathcal D^{\Lambda_m;(2)}_{\rho,l}
\)
are designed so that the renormalization problem can be formulated
entirely in terms of stability properties of these spaces under the
Polchinski flow.

The following theorem constitutes the main inductive statement of this
section. It asserts that the correlators satisfy the scaling behaviour
predicted by power counting, uniformly in the ultraviolet cutoff
\(\Lambda_0\).

\begin{theorem}[Uniform power-counting bounds]
\label{ThmPrinc}
Fix \(\Lambda_0\ge0\) and let \(\Lambda\in[0,\Lambda_0]\).
Let \(w\) be a multi-index and let \(r\in\mathbb N_0\).

Then, for every \(r,|w|\in\{0,\dots,3\}\), one has
\begin{align}
(z_1-z_i)^r
\partial_{p_i}^{\,w}
\mathcal L_{l,n}^{\Lambda,\Lambda_0}
\bigl(
\vec{z}_{1,n};
\underline{p}_{n}
\bigr)
&\in
\mathcal D^{\Lambda_m}_{4-n-|w|-r,\,l}
(\mathbb M^n),
\qquad
n\ge2,
\label{StaAssum1}
\\
(z_1-z_2)^r
\partial^{w}
\mathcal L_{l,n}^{\Lambda,\Lambda_0}
&\in
\mathcal D^{\Lambda_m;(2)}_{1-n,\,l}
(\mathbb M^n),
\qquad
r+|w|=3,
\quad
n\ge2.
\label{StaAssum2}\\
\Delta
\mathcal L_{l,n}^{\Lambda,\Lambda_0}
&\in
\mathcal D^{\Lambda_m}_{4-n,\,l}
(\mathbb M^n),
\qquad
n\ge4,
\label{StatAssum3}
\end{align}
\end{theorem}

The remainder of this section is devoted to the proof of
Theorem~\ref{ThmPrinc}. We first establish the stability properties of
the power-counting spaces under the operations appearing in the
Polchinski flow equation. These results are then combined with the
renormalization conditions and the flow equation to perform the
induction on \(n+2l\).

The bounds of Theorems~\ref{BesovNorm} and \ref{BesovNorm2} are obtained
as consequences of Theorem~\ref{ThmPrinc}. More precisely, the
power-counting estimates imply uniform control of the corresponding
Besov seminorms. The bounds for the truncated correlators are obtained
by applying the same estimates to correlators multiplied by smooth
approximations of the half-space indicator and subsequently passing to
the limit in the approximation parameter. This construction is carried
out in subsection~\ref{BesovTrSec}.
\subsection{Stability of the power-counting spaces under the Polchinski operators}\label{Continuitt}
In this subsection, we establish the mapping properties of the
Polchinski operators on the power-counting spaces introduced above.
These estimates constitute the main analytic part in the inductive
bounds  of Theorem~\ref{ThmPrinc}.

The linear contribution to the Polchinski flow improves the
power-counting degree by one. This reflects the gain of one scale factor
coming from the reduction of the amplitudes. The following proposition
makes this statement precise.
\begin{proposition}\label{PropL}
Let $n\ge2$ and $l\ge1$. Then the linear Polchinski operator
$\mathbf L^\Lambda$ satisfies
\begin{equation}\label{L1}
\mathbf L^\Lambda :
\mathcal D^{\Lambda_m}_{\rho,l-1}(\mathbb M^{n+2})
\to
\mathcal D^{\Lambda_m}_{\rho+1,l}(\mathbb M^n),
\end{equation}
as well as
\begin{equation}\label{L2}
\mathbf L^\Lambda :
\mathcal D^{\Lambda_m;(2)}_{\rho,l-1}(\mathbb M^{n+2})
\to
\mathcal D^{\Lambda_m;(2)}_{\rho+1,l}(\mathbb M^n).
\end{equation}
Moreover, the following estimate holds
\begin{equation}\label{Lquant}
\mathcal N^{\Lambda_m}_{\rho+1,l;\infty}
\bigl(
\mathbf L^\Lambda(T)(\underline p_n)
\bigr)
\le
\int_k
\mathcal N^{\Lambda_m}_{\rho,l-1;\infty}
\bigl(
T(\underline p_n,k,-k)
\bigr)
\dot C_m^\Lambda(k)
\end{equation}
and
\begin{equation}\label{DeltaL}
\Delta T
\in
\mathcal D^{\Lambda_m}_{\rho,l-1}(\mathbb M^{n+2})\implies
\Delta \mathbf L^\Lambda(T)
\in
\mathcal D^{\Lambda_m}_{\rho+1,l}(\mathbb M^n).
\end{equation}
\end{proposition}
\begin{proof}
The argument consists of three steps. First, we rewrite the linear term
using the semigroup property of the heat kernel, thereby expressing
\(\mathbf L^\Lambda(T)\) in terms of the seminorms of \(T\) with two
additional heat-kernel insertions at scale \(\Lambda^{-1}\). Second, we
combine the defining bounds of
\(\mathcal D^{\Lambda_m}_{\rho,l-1}\)
with the reduction estimate for the tree amplitudes, which yields the
gain of one power of \(\Lambda_m\) predicted by power counting. Finally, we repeat the same argument for the seminorms involving
discrete derivatives, thereby establishing the stability of the
$\Delta$-bounds.\\
\indent Throughout the proof, the symbols
\(
\mathcal P_{l-2}
\)
and
\(
\widetilde{\mathcal P}
\)
denote generic polynomials of the indicated degrees, with positive
coefficients independent of
\(\Lambda\), \(\Lambda_0\), \(m\), and the external momenta.
They may change from line to line.\\
\indent Let $T \in \mathcal{D}^{\Lambda_m}_{\rho,l-1}\bigl(\mathbb{M}^{n+2}\bigr)$. Recall that
\begin{equation}\label{E209}
\mathbf{L}^{\Lambda}(T)\bigl(\vec{z}_{1,n};\underline{p}_n\bigr) 
= \int_k \int_{\mathbb{R}} dz \int_{\mathbb{R}} dz'\,
T\bigl(\vec{z}_{1,n},z,z';\underline{p}_n,k,-k\bigr)\,
K\bigl(\Lambda^{-2};z,z'\bigr)\,\dot C^{\Lambda}_m(k).
\end{equation}
Fix $\mathfrak{s}\subseteq \mathcal{I}_n$, $\sigma\subseteq \mathfrak{s}$ and $\chi\in\mathcal{K}$. We use the notation
\[
\bar u:=\bigl((2\Lambda^2)^{-1},u\bigr)~.
\]
Using the semigroup property
\begin{equation}\label{SemiGroup}
K\bigl(\Lambda^{-2};z,z'\bigr)
=
\int_{\mathbb{R}} du~
K\bigl((2\Lambda^2)^{-1};z,u\bigr)\,
K\bigl((2\Lambda^2)^{-1};z',u\bigr),
\end{equation}
we rewrite \eqref{E209} as
\begin{equation*}
\bigl(K^{\otimes \mathfrak{s}} \star \chi^{\otimes \sigma}\mathbf{L}^{\Lambda}(T)\bigr)
\bigl(z_1;\underline{p}_n;\underline{y}_{\mathfrak{s}}\bigr)
=
\int_k \int_{\mathbb{R}} du\,
\bigl(K^{\otimes (\mathfrak{s}+2)} \star \chi^{\otimes \sigma} T\bigr)
\bigl(z_1;\underline{p}_n,k,-k;\underline{y}_{\mathfrak{s}},\bar u,\bar u\bigr)\,
\dot C^{\Lambda}_m(k).
\end{equation*}
Since \(T \in \mathcal{D}^{\Lambda_m}_{\rho,l-1}(\mathbb{M}^{n+2})\), by~\eqref{e203} we obtain for $l\ge 2$
\begin{multline*}
\Bigl|
\bigl(K^{\otimes (\mathfrak{s}+2)} \star \chi^{\otimes \sigma} T\bigr)
\bigl(z_1;\underline{p}_n,k,-k;\underline{y}_{\mathfrak{s}},\bar u,\bar u\bigr)
\Bigr|
\\
\le
\Lambda_m^\rho\,
\mathcal{N}^{\Lambda_m}_{\rho,l-1;\infty}\bigl(T(\underline{p}_n,k,-k)\bigr)\,
\langle
\mathcal{A},\hat{\mathcal{A}}\rangle_{l-1,n+2}
\bigl(z_1;\underline{y}_{\mathfrak{s}},\bar u,\bar u\bigr)~
{Q}^{\sigma}_{q}\!\left(
\tau_{\mathfrak{s}}^{\Lambda_m},
y^{\Lambda},y^{\Lambda}
\right),
\end{multline*}
where $q:=|\sigma|$. For $l=1$, we use instead \eqref{e203Bis}. Using Definition \eqref{DefQ} we have 
\[
{Q}^{\sigma}_{q}\!\left(
\tau_{\mathfrak{s}}^{\Lambda_m},
y^{\Lambda},y^{\Lambda}
\right)
\lesssim
{Q}^{\sigma}_{q}\!\left(\tau_{\mathfrak{s}}^{\Lambda_m}\right).
\]
We now integrate in $u$ and apply Corollary~\ref{CoRe} which yields
\begin{multline}\label{ici}
\bigl|
\bigl(K^{\otimes \mathfrak{s}} \star \chi^{\otimes \sigma}\mathbf{L}^{\Lambda}(T)\bigr)
\left(z_1;\underline{p}_n;\underline{y}_{\mathfrak{s}}\right)
\bigr|
\le
\Lambda_m^{\rho+1}\,
\int_k \mathcal{N}^{\Lambda_m}_{\rho,l-1;\infty}\bigl(T(\underline{p}_n,k,-k)\bigr)~\dot{C}^{\Lambda}_m(k)~
{Q}^{\sigma}_{q}\!\left(
\tau_{\mathfrak{s}}^{\Lambda}
\right)
\\
\times
\langle
\mathcal{A},\hat{\mathcal{A}}\rangle_{l,n}^{\Lambda,\Lambda_0}
\bigl(z_1;\underline{y}_{\mathfrak{s}};\chi^\sigma\bigr).
\end{multline}
Taking the supremum in~\eqref{ici} using the definition \eqref{e203}, we conclude that
\begin{equation}\label{b--1}
\mathcal{N}^{\Lambda_m}_{\rho+1,l;\infty}\bigl(\mathbf{L}^{\Lambda}(T)(\underline{p}_n)\bigr)
\le
\int_k \mathcal{N}^{\Lambda_m}_{\rho,l-1;\infty}\bigl(T(\underline{p}_n,k,-k)\bigr)~\dot{C}^{\Lambda}_m(k)~.
\end{equation}
Furthermore, since
\(
T\in \mathcal{D}^{\Lambda_m}_{\rho,l-1}(\mathbb{M}^{n+2}),
\)
there exist a polynomial $\mathcal{P}_{l-2}$ of degree at most $l-2$ and a polynomial $\widetilde{\mathcal{P}}$ such that
\[
\mathcal{N}^{\Lambda_m}_{\rho,l-1;\infty}\bigl(T(\underline{p}_n,k,-k)\bigr)
\le
\mathcal{P}_{l-2}\!\left(\log \frac{\Lambda_m}{m}\right)
\widetilde{\mathcal{P}}\!\left(\frac{\|\underline{p}_n\|+|k|}{\Lambda_m}\right).
\]
Combining this with the covariance bound
\begin{equation}\label{PolCovB}
\int_{\mathbb{R}^3}
\left(\frac{|k|}{\Lambda_m}\right)^\alpha
\dot C_m^\Lambda(k)\,d^3k
\le C_\alpha,
\qquad \alpha\in\mathbb{N},
\end{equation}
together with \eqref{b--1} we obtain
\[
\mathcal{N}^{\Lambda_m}_{\rho+1,l;\infty}\bigl(\mathbf{L}^{\Lambda}(T)(\underline{p}_n)\bigr)
\le
\mathcal{P}_{l-2}\!\left(\log \frac{\Lambda_m}{m}\right)
\widetilde{\mathcal{P}}\!\left(\frac{\|\underline{p}_n\|}{\Lambda_m}\right).
\]
The estimate 
\[
\mathcal{N}^{\Lambda_m;(2)}_{\rho+1,l;\infty}\bigl(\mathbf{L}^{\Lambda}(T)(\underline{p}_n)\bigr)
\le
\mathcal{P}_{l-2}\!\left(\log \frac{\Lambda_m}{m}\right)
\widetilde{\mathcal{P}}\!\left(\frac{\|\underline{p}_n\|}{\Lambda_m}\right).
\]
is obtained following similar steps. We instead use the reduction of the double rooted amplitude estimate \eqref{DouRed} from Corollary \ref{CorFus}.
We now treat the case with $\Delta K$ and $\Delta \chi$ insertions. Assume that $\Delta T\in\mathcal{D}^{\Lambda_m}_{\rho,l-1}\left(\mathbb{M}^{n+2}\right)$. This means that we have for $j\in\mathcal{I}_n$
\begin{multline}\label{Borne*}
\bigl(K^{\otimes \mathfrak{s}}\star \Delta_{z_1}^{(j)}\mathbf{L}^{\Lambda}(T)\bigr)
\bigl(z_1;\underline{p}_n;\underline{y}_{\mathfrak{s}}\bigr)\\
=
\int_k \int_{\mathbb{R}} du\,
\bigl(K^{\otimes (\mathfrak{s}\setminus\{j\}+2)}\otimes\Delta_{z_1}K_{\tau_j,y_j}  \star T\bigr)
\bigl(z_1;\underline{p}_n,k,-k;\underline{y}_{\mathfrak{s}},\bar u,\bar u\bigr)\,
\dot C^{\Lambda}(k).
\end{multline}
Since $\Delta T\in\mathcal{D}^{\Lambda_m}_{\rho,l-1}\left(\mathbb{M}^{n+2}\right)$, we deduce that \eqref{Borne*} is bounded for $l\ge 2$ by 
\begin{equation}
\tau_j^{-\frac{1}{2}}\Lambda_m^{\rho-1}\,
\mathcal{N}^{\Lambda_m}_{\rho,l-1}\bigl(\Delta^{(j)} T(\underline{p}_n,k,-k)\bigr)\,
\mathcal{A}_{l-1,n+2}^{\Lambda,\Lambda_0}
\bigl(z_1;\underline{y}_{\mathfrak{s}},\bar u,\bar u\bigr)
{Q}_{1}\!\left(
\tau^{\Lambda_m}_{\mathfrak{s}},
y^{\Lambda},
y^{\Lambda}
\right).
\end{equation}
For $l=1$, the polynomial $Q_1$ is reduced to a constant. Following similar steps as before, we deduce that there exist polynomials $\mathcal{P}_{l-2}$ of degree at most $l-2$ and $\widetilde{\mathcal{P}}$ of degree depending only on $n$ and $l$ such that 
\begin{equation}\label{b1love}
\mathcal{N}^{\Lambda_m}_{\rho+1,l}\bigl(\Delta^{(j)} \mathbf{L}^{\Lambda}(T)(\underline{p}_n)\bigr)
\le
\mathcal{P}_{l-2}\!\left(\log\frac{\Lambda_m}{m}\right)
\widetilde{\mathcal{P}}\!\left(\frac{\|\underline{p}_n\|}{\Lambda_m}\right).
\end{equation}
Next, bounding the following for a given $\chi\in\overline{\mathcal{K}}$,
\begin{multline}
\bigl(K^{\otimes \mathfrak{s}}\star \chi^{\otimes \sigma\setminus\{j\}}\otimes\Delta_{z_1}\chi\,\mathbf{L}^{\Lambda}(T)\bigr)
\bigl(z_1;\underline{p}_n;\underline{y}_{\mathfrak{s}}\bigr)\\
=
\int_k \int_{\mathbb{R}} du\,
\bigl(K^{\otimes \mathfrak{s}+2}\star \left(\Delta^{(j)}_{z_1}\chi^{\otimes\sigma}\right)T\bigr)
\bigl(z_1;\underline{p}_n,k,-k;\underline{y}_{\mathfrak{s}},\bar u,\bar u\bigr)\,
\dot C^{\Lambda}(k),
\end{multline}
we deduce that 
\begin{equation}\label{b2love}
\mathcal{N}^{\Lambda_m}_{\rho+1,l}\bigl(\Delta^{(j)}\chi\,\mathbf{L}^{\Lambda}(T)(\underline{p}_n)\bigr)
\le
\mathcal{P}_{l-2}\!\left(\log\frac{\Lambda_m}{m}\right)
\widetilde{\mathcal{P}}\!\left(\frac{\|\underline{p}_n\|}{\Lambda_m}\right).
\end{equation}
Combining \eqref{b1love} and \eqref{b2love} with \eqref{Ninf2} and taking the sup over $j\in\mathcal{I}_n$, we deduce that 
\begin{equation}
\mathcal{N}^{\Lambda_m}_{\rho+1,l;\infty}\bigl(\Delta \mathbf{L}^{\Lambda}(T)(\underline{p}_n)\bigr)
\le
\mathcal{P}_{l-2}\!\left(\log\frac{\Lambda_m}{m}\right)
\widetilde{\mathcal{P}}\!\left(\frac{\|\underline{p}_n\|}{\Lambda_m}\right),
\end{equation}
and this establishes that $\Delta \mathbf{L}^{\Lambda}(T)\in\mathcal{D}^{\Lambda_m}_{\rho+1,l}\left(\mathbb{M}^n\right)$.
\end{proof}
The bilinear contribution to the Polchinski flow preserves the
power-counting degree. This reflects the fact that the contraction
kernel is not bounded away and therefore produces no
additional gain of scale. The corresponding estimate is a consequence of
the fusion bounds established in Section~\ref{SecTreeAmp}.

We begin by introducing some notation that will be used throughout the
analysis of the bilinear operator: fix a partition $\pi=\{\pi_1,\pi_2\}$ of $\left\{1,\cdots,n\right\}$, and let
\[
\sigma\subseteq\mathfrak s\subseteq\mathcal I_n.
\]
Throughout this subsection, we write, for $r\in\{1,2\}$,
\[
\mathfrak s_r:=\mathfrak s\cap\pi_r,
\qquad
\sigma_r:=\sigma\cap\pi_r,
\qquad
q_r:=|\sigma_r|.
\]
If $\mathfrak a\subseteq\mathcal I_n$ and $j\in\mathfrak a$, we further
set
\[
\mathfrak a^{(j)}
:=
\mathfrak a\setminus\{j\}.
\]
We also define for $i\in\mathcal{I}_n$
\begin{equation}\label{Seti}
\mathfrak{a}(i):=
\begin{cases}
\mathfrak{a}^{(i)}, & \text{if } i\in\mathfrak{a},\\[0.3em]
\mathfrak{a}, & \text{otherwise}.
\end{cases}
\end{equation}
We also use the shorthand notation
\[
\underline y_{\mathfrak s_i}
:=
\bigl(
(\tau_j,y_j)
\bigr)_{j\in\mathfrak s_i},
\qquad
\underline p_{\pi_i}
:=
(p_j)_{j\in\pi_i},
\qquad
p_\pi
:=
\sum_{j\in\pi_1}p_j,
\]
as well as
\[
\bar u
:=
\bigl(
(2\Lambda^2)^{-1},u
\bigr),
\qquad
y^\Lambda
:=
\sqrt2\,\frac{\Lambda}{\Lambda_m}.
\]
In the next propositions, we analyze the bilinear Polchinski operator
$\mathbf B_\pi^\Lambda$, defined by
\begin{multline}\label{Bpi}
\mathbf B_\pi^\Lambda(T_1,T_2)
(\vec{z}_{1,n};\underline p_n)
:=
e^{-\frac{m^2}{2\Lambda^2}}
\int_{\mathbb R}dz
\int_{\mathbb R}dz'\,
T_1
\bigl(
\underline z_{\pi_1},z;
\underline p_{\pi_1},p_\pi
\bigr)
\\
\times
T_2
\bigl(
\underline z_{\pi_2},z';
\underline p_{\pi_2},-p_\pi
\bigr)
K\!\left(
\frac1{\Lambda^2};z,z'
\right).
\end{multline}
We repeatedly use the semigroup property of the heat kernel,
\begin{equation}\label{BSG}
K\!\left(
\frac1{\Lambda^2};z,z'
\right)
=
\int_{\mathbb R}du~
K\!\left(
\frac1{2\Lambda^2};z,u
\right)
K\!\left(
\frac1{2\Lambda^2};z',u
\right).
\end{equation}
\begin{proposition}[Stability of the bilinear Polchinski operator]\label{propQu}
Let $n_1,n_2\ge1$, $l_1,l_2\ge0$, and
$(\rho_1,\rho_2)\in\mathbb Z^2$. Set
\[
\rho:=\rho_1+\rho_2,
\qquad
l:=l_1+l_2,\qquad n:=n_1+n_2
\]
Then, for every partition $
\pi=\{\pi_1,\pi_2\}$
of $\left\{1,\cdots,n\right\}$
the bilinear operator $
\mathbf B_\pi^\Lambda
$
extends continuously to
\begin{equation}\label{Qu3}
\mathbf B_\pi^\Lambda :
\mathcal D^{\Lambda_m}_{\rho_1,l_1}
(\mathbb M^{n_1+1})
\times
\mathcal D^{\Lambda_m}_{\rho_2,l_2}
(\mathbb M^{n_2+1})
\longrightarrow
\mathcal D^{\Lambda_m}_{\rho,l}
(\mathbb M^{n}).
\end{equation}
The same statement remains valid if one of the spaces
$\mathcal D^{\Lambda_m}_{\rho_i,l_i}$ on the left-hand side and the
corresponding target space on the right-hand side are replaced by their
second-order counterparts
$\mathcal D^{\Lambda_m;(2)}_{\rho_i,l_i}$ and
$\mathcal D^{\Lambda_m;(2)}_{\rho,l}$. More precisely, the following bounds hold.

\medskip

\noindent
\textit{(i) Two-point case.}
For $
(T_1,T_2)
\in
\mathcal D^{\Lambda_m}_{\rho_1,l_1}(\mathbb M^2)
\times
\mathcal D^{\Lambda_m}_{\rho_2,l_2}(\mathbb M^2)$
one has
\begin{equation}
\mathcal N^{\Lambda_m}_{\rho,l}
\Bigl(
\mathbf B^\Lambda(T_1,T_2)(\underline{p})
\Bigr)
\le
\mathcal N^{\Lambda_m}_{\rho_1,l_1}
\bigl(
T_1(p)
\bigr)
\mathcal N^{\Lambda_m}_{\rho_2,l_2}
\bigl(
T_2(-p)
\bigr).
\end{equation}

\medskip

\noindent
\textit{(ii)}
For
$
(T_1,T_2)
\in
\mathcal D^{\Lambda_m}_{\rho_1,l_1}(\mathbb M^2)
\times
\mathcal D^{\Lambda_m}_{\rho_2,l_2}(\mathbb M^n)$,
one has
\begin{equation}
\mathcal N^{\Lambda_m}_{\rho,l}
\Bigl(
\mathbf B^\Lambda_{\pi}(T_1,T_2)(\underline p_n)
\Bigr)
\le
\mathcal 
 N^{\Lambda_m}_{\rho_1,l_1}
\bigl(
T_1(p_{\pi})
\bigr)\mathcal N^{\Lambda_m}_{\rho_2,l_2}
\left(
T_2\Bigl(\sum_{i\in\pi_2}^n p_i\Bigr)
\right).
\end{equation}
\medskip

\noindent
\textit{(iii) General case.}
If $n_1,n_2>1$, then
\[
\mathcal N^{\Lambda_m}_{\rho,l;\infty}
\Bigl(
\mathbf B_\pi^\Lambda(T_1,T_2)(\underline p_n)
\Bigr)
\le
\mathcal N^{\Lambda_m}_{\rho_1,l_1;\infty}
\bigl(
T_1(\underline p_{\pi_1})
\bigr)
\mathcal N^{\Lambda_m}_{\rho_2,l_2;\infty}
\bigl(
T_2(\underline p_{\pi_2})
\bigr).
\]
\end{proposition}
\begin{proof}
The proof is divided into three cases. The first two correspond to
situations where one or both factors belong to the two-point sector and
can therefore be estimated directly using the special seminorms
introduced in \eqref{e203-n2}. All cases require the fusion
estimates of Section~\ref{SecTreeAmp}, which allow us to combine the
tree amplitudes associated with the two factors into a single amplitude
of total degree $
(\rho_1+\rho_2,l_1+l_2)$.\\
\emph{(i) Two-point Case:} using the definition \eqref{Bpi} together with the semigroup property
\eqref{BSG}, we write, for
$
T_1\in
\mathcal D^{\Lambda_m}_{\rho_1,l_1}(\mathbb M^2)$ and 
$T_2\in
\mathcal D^{\Lambda_m}_{\rho_2,l_2}(\mathbb M^2),$
that
\begin{multline*}
\bigl(
K\star\mathbf B^\Lambda(T_1,T_2)
\bigr)
\bigl(
z_1;\tau_2,y_2;\underline{p}
\bigr)
\\
=
\int_{\mathbb R}du\,
\bigl(
K\star T_1
\bigr)
\bigl(
z_1;\bar u;p
\bigr)
\int_{\mathbb R}dz_2\,
K(\tau_2;z_2,y_2)
\bigl(
K\star T_2
\bigr)
\bigl(
z_2;\bar u;-p
\bigr).
\end{multline*}
\eqref{e203-n2} implies that
\begin{multline}\label{11nn}
\left|
\bigl(
K\star\mathbf B^\Lambda(T_1,T_2)
\bigr)
\bigl(
z_1;\tau_2,y_2;\underline p
\bigr)
\right|
\le
\Lambda_m^{\rho_1+\rho_2}
\mathcal N^{\Lambda_m}_{\rho_1,l_1}
\bigl(
T_1(p)
\bigr)
\mathcal N^{\Lambda_m}_{\rho_2,l_2}
\bigl(
T_2(-p)
\bigr)
\left(
\sum_{i=0}^3
(y^\Lambda)^i
\right)
\\
\times
\mathcal A_{l_1,2}^{\Lambda,\Lambda_0}
\otimes
\mathcal A_{l_2,2}^{\Lambda,\Lambda_0}
\bigl(
z_1;\tau_2,y_2
\bigr).
\end{multline}
Using \eqref{157CorFus} from Corollary~\ref{CorFus}, together with the
bound $
y^\Lambda\le\sqrt2,
$
and taking the supremum in \eqref{11nn}, we obtain
\[
\mathcal N^{\Lambda_m}_{\rho_1+\rho_2,l_1+l_2}
\bigl(
\mathbf B^\Lambda(T_1,T_2)(\underline p)
\bigr)
\le
\mathcal N^{\Lambda_m}_{\rho_1,l_1}
\bigl(
T_1(p)
\bigr)
\mathcal N^{\Lambda_m}_{\rho_2,l_2}
\bigl(
T_2(-p)
\bigr).
\]

\medskip

\noindent
\emph{(ii)}
Proceeding as above, we write
\begin{multline*}
\bigl(
K^{\otimes\mathfrak s}
\star
\chi^{\otimes\sigma}
\mathbf B^\Lambda(T_1,T_2)
\bigr)
\bigl(
z_1;\underline y_{\mathfrak s};\underline p
\bigr)
=
\int_{\mathbb R}du\,
\bigl(
K\star T_1
\bigr)
\bigl(
z_1;\bar u;p_{\pi}
\bigr)
\\
\times
\int_{\mathbb R}dz'\,
K\!\left(
\frac1{2\Lambda^2};z',u
\right)
\bigl(
K^{\otimes\mathfrak s}
\star
\chi^{\otimes\sigma}T_2
\bigr)
\left(
z';
\underline y_{\mathfrak s};
\sum_{i\in\pi_2} p_i
\right).
\end{multline*}
Since $
T_i
\in
\mathcal D^{\Lambda_m}_{\rho_i,l_i}
(\mathbb M^{n_i+1})$, we obtain
\begin{multline}\label{12nn}
\left|
\bigl(
K^{\otimes\mathfrak s}
\star
\chi^{\otimes\sigma}
\mathbf B^\Lambda(T_1,T_2)
\bigr)
\bigl(
z_1;\underline y_{\mathfrak s};\underline p
\bigr)
\right|
\le
\Lambda_m^{\rho_1+\rho_2}
\mathcal N^{\Lambda_m}_{\rho_1,l_1}
\bigl(
T_1(p_{\pi})
\bigr)
\mathcal N^{\Lambda_m}_{\rho_2,l_2}
\left(
T_2\!\left(
\sum_{i\in\pi_2} p_i
\right)
\right)
\left(
\sum_{i=0}^3
(y^\Lambda)^i
\right)
\\
\times
Q_q^\sigma
\bigl(
\tau_{\mathfrak s}^{\Lambda_m}
\bigr)
\mathcal A_{l_1,2}^{\Lambda,\Lambda_0}
\otimes
\langle
\mathcal A,
\hat{\mathcal A}\rangle_{l_2,n}
\bigl(
z_1;\underline y_{\mathfrak s};\chi^\sigma
\bigr).
\end{multline}
Combining \eqref{CorFuss1} and \eqref{CorFuss2} from Corollary~\ref{CorFus}, and taking the
supremum as in \eqref{e203} and \eqref{sup2}, we deduce
\[
\mathcal N^{\Lambda_m}_{\rho_1+\rho_2,l_1+l_2}
\bigl(
\mathbf B^\Lambda(T_1,T_2)(\underline p_n)
\bigr)
\le
\mathcal N^{\Lambda_m}_{\rho_1,l_1}
\bigl(
T_1(p_{\pi})
\bigr)
\mathcal N^{\Lambda_m}_{\rho_2,l_2}
\left(
T_2\!\left(
\sum_{i\in\pi_2} p_i
\right)
\right).
\]

\medskip

\noindent
\emph{(iii) General case:} assume first that $\sigma_2\neq\varnothing$. Fix
$j\in\sigma_2$. Using \eqref{Bpi} together with \eqref{BSG}, we write
\begin{multline}\label{qtte}
\bigl(
K^{\otimes\mathfrak s}
\star
\chi^{\otimes\sigma}
\mathbf B_\pi^\Lambda(T_1,T_2)
\bigr)
\bigl(
z_1;\underline y_{\mathfrak s},\underline p_n
\bigr)
=
\int_{\mathbb R}du\,
\bigl(
K^{\otimes(\mathfrak s_1+1)}
\star
\chi^{\otimes\sigma_1}T_1
\bigr)
\bigl(
z_1;
\underline y_{\mathfrak s_1},
\bar u;
\underline p_{\pi_1},
p_\pi
\bigr)
\\
\times
\int_{\mathbb R}dz_j\,
\chi(z_j)\,
K(\tau_j;z_j,y_j)
\bigl(
K^{\otimes(\mathfrak s_2^{(j)}+1)}
\star
\chi^{\otimes\sigma_2^{(j)}}T_2
\bigr)
\bigl(
z_j;
\underline y_{\mathfrak s_2^{(j)}},
\bar u;
\underline p_{\pi_2},
-p_\pi
\bigr).
\end{multline}
Using the seminorm bounds for $T_1$ and $T_2$, we have that
\begin{multline}\label{IMPORTANT}
\Bigl|
\bigl(
K^{\otimes\mathfrak s}
\star
\chi^{\otimes\sigma}
\mathbf B_\pi^\Lambda(T_1,T_2)
\bigr)
\bigl(
z_1;\underline y_{\mathfrak s},\underline p_n
\bigr)
\Bigr|
\le
\Lambda_m^{\rho_1+\rho_2}
\mathcal N^{\Lambda_m}_{\rho_1,l_1;\infty}
\bigl(
T_1(\underline p_{\pi_1})
\bigr)
\mathcal N^{\Lambda_m}_{\rho_2,l_2;\infty}
\bigl(
T_2(\underline p_{\pi_2})
\bigr)\\\times
\left(
Q^\sigma
\otimes_\pi^{(j)}
Q^\sigma
\right)_{l_1,l_2}
\bigl(
\tau_{\mathfrak s^{(j)}}^{\Lambda_m},
y^\Lambda
\bigr)
\langle
\mathcal A,\hat{\mathcal A}
\rangle_{l_1,n_1+1}
\otimes_\pi
\langle
\mathcal A,\hat{\mathcal A}
\rangle_{l_2,n_2+1}
\bigl(
z_1;\underline y_{\mathfrak s};\chi^\sigma
\bigr),
\end{multline}
where
\begin{multline}\label{QqQ}
\left(
Q^\sigma
\otimes_\pi^{(j)}
Q^\sigma
\right)_{l_1,l_2}
\bigl(
\tau_{\mathfrak s^{(j)}}^\Lambda,
y^\Lambda
\bigr)
:=
\delta_{l_1,0}\delta_{l_2,0}
+
(1-\delta_{l_1,0})
Q_{q_1}^{\sigma_1}
\bigl(
\tau_{\mathfrak s_1}^\Lambda,
y^\Lambda
\bigr)
+
(1-\delta_{l_2,0})
Q_{q_2-1}^{\sigma_2^{(j)}}
\bigl(
\tau_{\mathfrak s_2^{(j)}}^\Lambda,
y^\Lambda
\bigr)
\\
+
(1-\delta_{l_1,0})
(1-\delta_{l_2,0})
\bigl(
Q^\sigma
\otimes_\pi
Q^\sigma
\bigr)
\bigl(
\tau_{\mathfrak s^{(j)}}^\Lambda,
y^\Lambda
\bigr).
\end{multline}
Observe that \eqref{QqQ} reduces to $1$ when $l_1=l_2=0$. 
The case \(\sigma_2\neq\varnothing\) allows us to exploit directly the
distinguished cutoff insertion carried by the second factor. The
remaining case \(\sigma_2=\varnothing\) is treated separately by rooting
the second factor at the integration variable \(z'\).
Assume now that $\sigma_2=\varnothing$. In this case, we rewrite
\eqref{qtte} so that the second factor is rooted at $z'$:
\begin{multline}\label{siVar}
\bigl(
K^{\otimes\mathfrak s}
\star
\chi^{\otimes\sigma}
\mathbf B_\pi^\Lambda(T_1,T_2)
\bigr)
\bigl(
z_1;\underline y_{\mathfrak s},\underline p_n
\bigr)
=
\int_{\mathbb R}du\,
\bigl(
K^{\otimes(\mathfrak s_1+1)}
\star
\chi^{\otimes\sigma}T_1
\bigr)
\bigl(
z_1;
\underline y_{\mathfrak s_1},
\bar u;
\underline p_{\pi_1},
p_\pi
\bigr)
\\
\times
\int_{\mathbb R}dz'\,
K\!\left(
\frac1{2\Lambda^2};z',u
\right)
\bigl(
K^{\otimes\mathfrak s_2}
\star
T_2
\bigr)
\bigl(
z';
\underline y_{\mathfrak s_2};
\underline p_{\pi_2},
-p_\pi
\bigr).
\end{multline}
Using again the seminorm bounds for $T_1$ and $T_2$, we obtain
\begin{multline}\label{IMPORTANT2}
\Bigl|
\bigl(
K^{\otimes\mathfrak s}
\star
\mathbf B_\pi^\Lambda(T_1,T_2)
\bigr)
\bigl(
z_1;\underline y_{\mathfrak s},\underline p_n
\bigr)
\Bigr|
\le
\Lambda_m^{\rho_1+\rho_2}
\mathcal N^{\Lambda_m}_{\rho_1,l_1;\infty}
\bigl(
T_1(\underline p_{\pi_1})
\bigr)
\mathcal N^{\Lambda_m}_{\rho_2,l_2;\infty}
\bigl(
T_2(\underline p_{\pi_2})
\bigr)
\\
\times
\Bigl(
\delta_{l_1,0}
+
(1-\delta_{l_1,0})
Q_{|\sigma|}^\sigma
(
\tau_{\mathfrak s_1}^{\Lambda_m},
y^\Lambda
)
\Bigr)
Q_1
\bigl(
\tau_{\mathfrak s_2}^{\Lambda_m}
\bigr)
\\
\times
\langle
\mathcal A,\hat{\mathcal A}
\rangle_{l_1,n_1+1}
\otimes_\pi
\mathcal A_{l_2,n_2+1}^{\Lambda,\Lambda_0}
\bigl(
z_1;\underline y_{\mathfrak s};\chi^\sigma
\bigr).
\end{multline}
For every $\sigma'\subseteq\sigma$, one has
\[
\mathbf 1_{\{\sigma'\neq\varnothing\}}
\le
\mathbf 1_{\{\sigma\neq\varnothing\}}.
\]
Combining this with
\eqref{Prop1Q}--\eqref{Prop6Q} and the bound
$
y^\Lambda\le\sqrt2$
we obtain, for every $j\in\sigma_2$,
\begin{equation}\label{277-}
\left(
Q^\sigma
\otimes_\pi^{(j)}
Q^\sigma
\right)_{l_1,l_2}
\bigl(
\tau_{\mathfrak s^{(j)}}^\Lambda,
y^\Lambda
\bigr)
\lesssim
Q_{|\sigma|-1}^{\sigma^{(j)}}
\bigl(
\tau_{\mathfrak s^{(j)}}^\Lambda
\bigr).
\end{equation}
Furthermore, since $
\mathfrak s_1\subseteq\mathfrak s$,
we also have
\begin{equation}\label{278-}
\delta_{l_1,0}
+
(1-\delta_{l_1,0})
Q_{|\sigma|}^\sigma
\bigl(
\tau_{\mathfrak s_1}^{\Lambda_m},
y^\Lambda
\bigr)
\lesssim
Q_{|\sigma|}^\sigma
\bigl(
\tau_{\mathfrak s}^{\Lambda_m}
\bigr).
\end{equation}
In the case $\sigma_2\neq\varnothing$, we use
\eqref{eq:CorFus_2} from Corollary~\ref{CorFus}. For
$n_1,n_2\ge3$, this yields
\begin{multline}\label{279-}
\langle
\mathcal A,\hat{\mathcal A}
\rangle_{l_1,n_1+1}
\otimes_\pi^{(j)}
\langle
\mathcal A,\hat{\mathcal A}
\rangle_{l_2,n_2+1}
\bigl(
z_1;\underline y_{\mathfrak s};\chi^\sigma
\bigr)
\lesssim
\left(
1
+
(1-\delta_{l_2,0})
\sum_{i\in\mathfrak s_2}
\tau_i^{\Lambda_m}
\right)
\\
\times
\left(
\mathcal A_{l,n}^{\Lambda,\Lambda_0}
+
(1-\delta_{l_1,0})
(1-\delta_{l_2,0})
\Lambda_m^{-1}
\hat{\mathcal A}_{l,n}^{\Lambda,\Lambda_0}
\right)
\bigl(
z_1;\underline y_{\mathfrak s};\chi^\sigma
\bigr).
\end{multline}
Using \eqref{Prop5Q} together with \eqref{277-}, we obtain
\begin{equation}\label{280-}
\left(
1+
\sum_{i\in\mathfrak s_2}
\tau_i^{\Lambda_m}
\right)
\left(
Q^\sigma
\otimes_\pi^{(j)}
Q^\sigma
\right)_{l_1,l_2}
\bigl(
\tau_{\mathfrak s^{(j)}}^\Lambda,
y^\Lambda
\bigr)
\lesssim
Q_{|\sigma|}^\sigma
\bigl(
\tau_{\mathfrak s}^\Lambda
\bigr).
\end{equation}
Combining \eqref{IMPORTANT}, \eqref{279-}, and \eqref{280-}, we deduce
that, for $n_1,n_2\ge3$ and $l_1,l_2\ge0$,
\begin{multline}\label{IMPORTANT0}
\Bigl|
\bigl(
K^{\otimes\mathfrak s}
\star
\chi^{\otimes\sigma}
\mathbf B_\pi^\Lambda(T_1,T_2)
\bigr)
\bigl(
z_1;\underline y_{\mathfrak s},\underline p_n
\bigr)
\Bigr|
\le
\Lambda_m^{\rho_1+\rho_2}
\mathcal N^{\Lambda_m}_{\rho_1,l_1;\infty}
\bigl(
T_1(\underline p_{\pi_1})
\bigr)
\mathcal N^{\Lambda_m}_{\rho_2,l_2;\infty}
\bigl(
T_2(\underline p_{\pi_2})
\bigr)
\\
\times
Q_{|\sigma|}^\sigma
\bigl(
\tau_{\mathfrak s}^{\Lambda_m}
\bigr)
\langle
\mathcal A,
\hat{\mathcal A}\rangle_{l,n}
\bigl(
z_1;\underline y_{\mathfrak s};\chi^\sigma
\bigr).
\end{multline}
In the case $\sigma_2=\varnothing$, the same estimate follows from
\eqref{IMPORTANT2}, \eqref{278-}, and \eqref{CorFuss1}-\eqref{CorFuss2} from
Corollary~\ref{CorFus}.\\
Taking the supremum finally yields
\[
\mathcal N^{\Lambda_m}_{\rho,l;\infty}
\bigl(
\mathbf B_\pi^\Lambda(T_1,T_2)(\underline p_n)
\bigr)
\le
\mathcal N^{\Lambda_m}_{\rho_1,l_1}
\bigl(
T_1(\underline p_{\pi_1})
\bigr)
\mathcal N^{\Lambda_m}_{\rho_2,l_2}
\bigl(
T_2(\underline p_{\pi_2})
\bigr).
\]
This concludes the proof.
\end{proof}
The next estimate quantifies the gain obtained from powers of the
distance between vertices belonging to different components of the
partition. The corresponding factors are controlled through the
geometric decay encoded in Lemma~\ref{Rlemma}.
\begin{proposition}\label{Bilinearr}
Let $n_1,n_2\ge1$ and $l_1,l_2\ge0$ be such that
$
l:=l_1+l_2\ge1,
$
and let $(\rho_1,\rho_2)\in\mathbb Z^2$. Set $
\rho:=\rho_1+\rho_2$.
Assume that
$
T_i
\in
\mathcal D^{\Lambda_m}_{\rho_i,l_i}
(\mathbb M^{n_i+1})$ with $i\in\{1,2\}$. Then, for every partition
$\pi=\left\{\pi_1\cup\pi_2\right\}$ of $\left\{1,\cdots,n\right\}$ and every index $j$ such that
$1$ and $j$ belong to different components of the partition $\pi$, one
has
\begin{equation}\label{popert}
\mathcal N^{\Lambda_m}_{\rho-r,l}
\Bigl(
\mathbb Z_{1,j}^{(r)}
\mathbf B_\pi^\Lambda(T_1,T_2)
(\underline p_n)
\Bigr)
\lesssim
\mathcal N^{\Lambda_m}_{\rho_1,l_1}
\bigl(
T_1(\underline p_{\pi_1})
\bigr)
\mathcal N^{\Lambda_m}_{\rho_2,l_2}
\bigl(
T_2(\underline p_{\pi_2})
\bigr),
\end{equation}
for every $r\in\{0,1,2,3\}$.
\end{proposition}
\begin{proof}
Fix
$
T_1\in
\mathcal D^{\Lambda_m}_{\rho_1,l_1}
(\mathbb M^{n_1+1})$ and $T_2\in
\mathcal D^{\Lambda_m}_{\rho_2,l_2}
(\mathbb M^{n_2+1})$,
and let $
\pi=\pi_1\cup\pi_2$
be a partition of $\left\{1,\cdots,n\right\}$. The proof is a direct consequence of the seminorm bounds for
\(T_1\) and \(T_2\) together with the geometric estimate of
Lemma~\ref{Rlemma}. The latter converts the factor
\((z_1-z_j)^r\) into a gain of order \(\Lambda^{-r}\), which is then
absorbed into the power-counting degree. Without loss of generality, we
assume that $
1\in\pi_1$ and $
j\in\pi_2$. Recall the definition \eqref{Bpi} of the bilinear Polchinski operator.
Applying \(\mathbb Z_{1,j}^{(r)}\) to \eqref{Bpi} and convolving with the
heat kernels defining the seminorms yields
\begin{multline}
\bigl(
K^{\otimes\mathfrak s}
\star
\mathbb Z_{1,j}^{(r)}
\mathbf B_\pi^\Lambda(T_1,T_2)
\bigr)
\bigl(
z_1;\underline y_{\mathfrak s};\underline p_n
\bigr)
=
e^{-\frac{m^2}{2\Lambda^2}}
\int_{\mathbb R}du
\int_{\mathbb R}dz_j\,
\bigl(
K^{\otimes(\mathfrak s_1+1)}
\star
T_1
\bigr)
\bigl(
z_1;
\underline y_{\mathfrak s_1},
\bar u;
\underline p_{\pi_1},
p_\pi
\bigr)
\\
\times
(z_1-z_j)^r
\phi_j(z_j)
\bigl(
K^{\otimes(\mathfrak s_2(j)+1)}
\star
T_2
\bigr)
\bigl(
z_j;
\underline y_{\mathfrak s_2(j)},
\bar u;
\underline p_{\pi_2},
-p_\pi
\bigr),
\end{multline}
where we use the notation \eqref{Seti} and 
\[
\phi_j:=
\begin{cases}
K_{\tau_j,y_j}, & \text{if } j\in\mathfrak{s},\\[0.3em]
1, & \text{otherwise}.
\end{cases}
\]
Since no cutoff decorations are involved in the
definition of \(\mathbb Z_{1,j}^{(r)}\),
it is sufficient to consider the seminorm with
\(\chi=\mathbf 1\). We obtain
\begin{multline}
\Bigl|
\bigl(
K^{\otimes\mathfrak s}
\star
\mathbb Z_{1,j}^{(r)}
\mathbf B_\pi^\Lambda(T_1,T_2)
\bigr)
\bigl(
z_1;\underline y_{\mathfrak s};\underline p_n
\bigr)
\Bigr|
\le
\Lambda_m^{\rho_1+\rho_2}
e^{-\frac{m^2}{2\Lambda^2}}
\mathcal N^{\Lambda_m}_{\rho_1,l_1}
\bigl(
T_1(\underline p_{\pi_1})
\bigr)
\mathcal N^{\Lambda_m}_{\rho_2,l_2}
\bigl(
T_2(\underline p_{\pi_2})
\bigr)
\\
\times
Q_1
\bigl(
\tau_{\mathfrak s_1}^{\Lambda_m},
y^\Lambda
\bigr)
Q_1
\bigl(
\tau_{\mathfrak s_2^{(j)}}^{\Lambda_m},
y^\Lambda
\bigr)
\int_{\mathbb R}dz_j\,
|z_1-z_j|^r
\left(
\mathcal A_{l_1,n_1+1}^{\Lambda,\Lambda_0}
\otimes_\pi
\mathcal A_{l_2,n_2+1}^{\Lambda,\Lambda_0}
\right)
\bigl(
z_1,z_j;\underline y_{\mathfrak s}
\bigr).
\end{multline}
By Lemma~\ref{Rlemma},
\[
\int_{\mathbb R} dz_j\,
|z_1-z_j|^r
\bigl(
\mathcal A^{\Lambda,\Lambda_0}_{l_1,n_1+1}
\otimes_\pi
\mathcal A^{\Lambda,\Lambda_0}_{l_2,n_2+1}
\bigr)\left(z_1,z_j;\underline{y}_{\mathfrak{s}}\right)
\lesssim
\Lambda^{-r}
\mathcal A^{\Lambda,\Lambda_0}_{l,n}\left(z_1;\underline{y}_{\mathfrak{s}}\right).
\]
This together with
\eqref{Prop2Q}--\eqref{Prop2QBis}, yields
\begin{multline*}
\Bigl|
\bigl(
K^{\otimes\mathfrak s}
\star
\mathbb Z_{1,j}^{(r)}
\mathbf B_\pi^\Lambda(T_1,T_2)
\bigr)
\bigl(
z_1;\underline y_{\mathfrak s};\underline p_n
\bigr)
\Bigr|
\lesssim
\Lambda_m^\rho~
\Lambda^{-r}~
e^{-\frac{m^2}{2\Lambda^2}}
\mathcal N^{\Lambda_m}_{\rho_1,l_1}
\bigl(
T_1(\underline p_{\pi_1})
\bigr)
\mathcal N^{\Lambda_m}_{\rho_2,l_2}
\bigl(
T_2(\underline p_{\pi_2})
\bigr)
\\
\times
Q_1
\bigl(
\tau_{\mathfrak s}^{\Lambda_m}
\bigr)
\mathcal A_{l,n}^{\Lambda,\Lambda_0}
\bigl(
z_1;\underline y_{\mathfrak s}
\bigr).
\end{multline*}
Using the bound
\[
\Lambda^{-r}
e^{-\frac{m^2}{2\Lambda^2}}
\lesssim
\Lambda_m^{-r},
\]
and taking the supremum over $
z_1$, $
\underline y_{\mathfrak s}$, and $
\mathfrak s\subseteq\mathcal I_n$
we conclude that
$$
\mathcal N^{\Lambda_m}_{\rho-r,l}
\bigl(
\mathbb Z_{1,j}^{(r)}
\mathbf B_\pi^\Lambda(T_1,T_2)
(\underline p_n)
\bigr)
\lesssim
\mathcal N^{\Lambda_m}_{\rho_1,l_1}
\bigl(
T_1(\underline p_{\pi_1})
\bigr)
\mathcal N^{\Lambda_m}_{\rho_2,l_2}
\bigl(
T_2(\underline p_{\pi_2})
\bigr).
$$
This concludes the proof.
\end{proof}
The next lemma provides the fusion estimates required for the analysis
of the discrete derivative operator
\(
\Delta \mathbf B_\pi^\Lambda.
\)
The main point is that the polynomial weights introduced in
Definition~\ref{DefPol} remain stable under the fusion procedures
associated with \(\Delta\chi\)- and \(\Delta K\)-insertions.\\
We begin by introducing some additional notation. For
$j\in\sigma_2$, we define the
$\otimes_\pi^{(j)}$-product of the polynomial weights by
\begin{equation}\label{OtimesQ}
\left(
Q^\sigma
\otimes_\pi^{(j)}
Q^\sigma
\right)
\bigl(
\tau_{\mathfrak s}^{\Lambda_m},
y^\Lambda
\bigr)
:=
Q_{q_1}^{\sigma_1}
\bigl(
\tau_{\mathfrak s_1}^{\Lambda_m},
y^\Lambda
\bigr)
\,
Q_{q_2-1}^{\sigma_2^{(j)}}
\bigl(
\tau_{\mathfrak s_2^{(j)}}^{\Lambda_m},
y^\Lambda
\bigr)
\end{equation}
where $q_i:=|\sigma_i|$. We also use the shorthand notation
\begin{equation}
(A\otimes B)
(\tau_{\mathfrak s}^{\Lambda_m},y^\Lambda)
:=
A(\tau_{\mathfrak s_1}^{\Lambda_m},y^\Lambda)
B(\tau_{\mathfrak s_2}^{\Lambda_m},y^\Lambda).
\end{equation}
Moreover, assume that $\sigma_2\neq\varnothing$, and fix
$j\in\sigma_2$. We also recall the shorthand notation introduced in
Section~\ref{SecTreeAmp}:
\begin{multline}
\langle\mathcal A,\hat{\mathcal A}\rangle_{l_1,n_1+1}
\otimes_\pi \Delta^{(j)}_{z_1}\chi~
\langle\mathcal A,\hat{\mathcal A}\rangle_{l_2,n_2+1}
\\
:=
\int_{\mathbb{R}}\int_{\mathbb{R}}dz_j~\left(
\mathcal A_{l_1,n_1+1}^{\Lambda,\Lambda_0}
+
\mathbf1_{\{\sigma_1\neq\varnothing\}}\left(1-\delta_{l_1,0}\right)
\Lambda_m^{-1}
\hat{\mathcal A}_{l_1,n_1+1}^{\Lambda,\Lambda_0}
\right)
\otimes_\pi
\Delta^{(j)}_{z_1}\chi
\\
\times
\left(
\mathcal A_{l_2,n_2+1}^{\Lambda,\Lambda_0}
+
\mathbf1_{\{\sigma_2^{(j)}\neq\varnothing\}}\left(1-\delta_{l_2,0}\right)
\Lambda_m^{-1}
\hat{\mathcal A}_{l_2,n_2+1}^{\Lambda,\Lambda_0}
\right).
\end{multline}
The next lemma establishes the stability of these polynomial weights
under the fusion estimates associated with the bilinear operator
$\mathbf B_\pi^\Lambda$.
\begin{lemma}[Stability under fusion]\label{LemmaFusionPoly}
Let \(l_1,l_2\ge0\), \(n_1,n_2\ge1\), and set
\[
l:=l_1+l_2\ge1,
\qquad
n:=n_1+n_2.
\]
Fix a partition $\{\pi_1,\pi_2\}$ of $\left\{1,\cdots,n\right\}$ and let
\(\sigma\subseteq\mathcal{I}_n\) be such that
\(\sigma_2\neq\varnothing\).
Then the following bounds hold.

\medskip

\noindent
\emph{(i) One \(\chi\)-insertion:}
\begin{multline}\label{AAchi}
\left(
\langle\mathcal{A},\hat{\mathcal{A}}\rangle_{l_1,n_1+1}
\otimes_{\pi}
\Delta_{z_1}\chi\,
\langle\mathcal{A},\hat{\mathcal{A}}\rangle_{l_2,n_2+1}
\right)
\left(
z_1;
\underline{y}_{\mathfrak{s}};
\chi^\sigma
\right)
\left(
Q^\sigma\otimes_{\pi}^{(j)}Q^\sigma
\right)
\left(
\tau_{\mathfrak{s}\setminus\{j\}}^{\Lambda_m},
y^\Lambda
\right)
\\
\lesssim
\Lambda^{-1}
\hat{\mathcal{A}}_{l,n}^{\Lambda,\Lambda_0}\left(z_1;\underline{y}_{\mathfrak{s}};\chi^{\sigma}\right)
\,
Q_{|\sigma|}^{\sigma^{(j)}}
\left(
\tau_{\mathfrak{s}}^{\Lambda_m}
\right).
\end{multline}

\medskip

\noindent
\emph{(ii) One \(\Delta_{z_1} K\)-insertion:}
for $j\in\mathfrak{s}$ the following bound holds
\begin{multline}\label{eq283}
\left(
\mathcal{A}^{\Lambda,\Lambda_0}_{l_1,n_1+1}
\otimes_{\pi}
\Delta_{z_1}K_{\tau_j,y_j}\,
\mathcal{A}^{\Lambda,\Lambda_0}_{l_2,n_2+1}
\right)
\left(
z_1;
(\tau_i,y_i)_{i\in\mathfrak{s}}
\right)~
\left(
Q_1\otimes_{\pi}Q_1
\right)
\left(
\tau_{\mathfrak{s}\setminus\{j\}}^{\Lambda_m},
y^\Lambda
\right)
\\
\lesssim
\tau_j^{-\frac{1}{2}}~\Lambda^{-1}~
\mathcal{A}_{l,n}^{\Lambda,\Lambda_0}
\left(
z_1;
\underline{y}_{\mathfrak{s}}
\right)
Q_1\left(
\tau_{\mathfrak{s}\setminus\{j\}}^{\Lambda_m}
\right).
\end{multline}

\end{lemma}
\begin{proof}
The proof uses Corollary~\ref{CorFus} together with the polynomial stability
estimates of Lemma~\ref{QLemma}. All implicit constants depend only on the fixed parameters of the
lemma. We also use that $0\le y^\Lambda \le \sqrt{2}$
so that the polynomial estimates of Lemma~\ref{QLemma} may be applied to
$y^\Lambda$, up to numerical constants.\\
Fix $j\in\sigma_2$ and set
\[
q:=|\sigma|,\qquad q_1:=|\sigma_1|,\qquad q_2:=|\sigma_2|,
\qquad |\sigma^{(j)}|=q-1.
\]
By definition,
\begin{equation}\label{eq:TensorQDefDetailed}
\left(Q^\sigma\otimes_\pi^{(j)}Q^\sigma\right)
\left(\tau_{\mathfrak{s}^{(j)}}^{\Lambda_m},y^\Lambda\right)
=
Q_{q_1}^{\sigma_1}
\left(\tau_{\mathfrak{s}_1}^{\Lambda_m},y^\Lambda\right)
Q_{q_2-1}^{\sigma_2^{(j)}}
\left(\tau_{\mathfrak{s}_2^{(j)}}^{\Lambda_m},y^\Lambda\right).
\end{equation}
Using Definition~\ref{DefPol}, we write
\begin{multline}\label{eq:TensorQDecompositionDetailed}
\left(Q^\sigma\otimes_\pi^{(j)}Q^\sigma\right)
\left(\tau_{\mathfrak{s}^{(j)}}^{\Lambda_m},y^\Lambda\right)
\le
\mathbf{1}_{\{\sigma^{(j)}=\varnothing\}}
\left(Q_1\otimes_{\pi}Q_1\right)
\left(\tau_{\mathfrak{s}^{(j)}}^{\Lambda_m},y^\Lambda\right)\\
+
\Big[
\mathbf{1}_{\{\sigma_1\neq\varnothing,\,
\sigma_2^{(j)}\neq\varnothing\}}
\left(Q_{q_1}\otimes_{\pi}Q_{q_2-1}\right)
\\
\hspace{2.5cm}
+
\mathbf{1}_{\{\sigma_1=\varnothing,\,
\sigma_2^{(j)}\neq\varnothing\}}
\left(Q_1\otimes_{\pi}Q_{q_2-1}\right)
+
\mathbf{1}_{\{\sigma_1\neq\varnothing,\,
\sigma_2^{(j)}=\varnothing\}}
\left(Q_{q_1}\otimes_{\pi}Q_1\right)
\Big]
\left(\tau_{\mathfrak{s}^{(j)}}^{\Lambda_m},y^\Lambda\right).
\end{multline}
Hence, by 
Lemma~\ref{QLemma} we obtain
\begin{equation}\label{eq:PolyFusionBound}
\bigl(
Q^\sigma\otimes_\pi^{(j)}Q^\sigma
\bigr)
\bigl(
\tau_{\mathfrak s^{(j)}}^{\Lambda_m},
y^\Lambda
\bigr)
\lesssim
Q_{q-1}^{\sigma^{(j)}}
\bigl(
\tau_{\mathfrak s^{(j)}}^{\Lambda_m}
\bigr).
\end{equation}
Furthermore, another application of Lemma~\ref{QLemma} gives
\begin{equation}\label{eq:PolyFusionBound2}
\left(
1+
\sum_{i\in\mathfrak s_2}
\tau_i^{\Lambda_m}
\right)
\bigl(
Q^\sigma\otimes_\pi^{(j)}Q^\sigma
\bigr)
\bigl(
\tau_{\mathfrak s^{(j)}}^{\Lambda_m},
y^\Lambda
\bigr)
\lesssim
Q_q^{\sigma^{(j)}}
\bigl(
\tau_{\mathfrak s}^{\Lambda_m}
\bigr).
\end{equation}
Combining \eqref{eq:PolyFusionBound} and \eqref{eq:PolyFusionBound2}, we obtain the bound
\begin{equation}\label{eq:CommonPolynomialBoundDetailed}
\left(
1+
(1-\delta_{l_2,0})
\mathbf{1}_{\{\sigma_2^{(j)}\neq\varnothing\}}
\sum_{i\in\mathfrak{s}_2}\tau_i^{\Lambda_m}
\right)
\left(Q^\sigma\otimes_\pi^{(j)}Q^\sigma\right)
\left(\tau_{\mathfrak{s}^{(j)}}^{\Lambda_m},y^\Lambda\right)
\lesssim
Q_{q}^{\sigma^{(j)}}
\left(\tau_{\mathfrak{s}}^{\Lambda_m}\right).
\end{equation}
We now prove the two estimates. The first estimate corresponds to a distinguished
\(\Delta\chi\)-insertion, while the second one treats a
\(\Delta K\)-insertion. In both cases, the proof consists of combining
the corresponding fusion estimate with the polynomial stability bound
\eqref{eq:CommonPolynomialBoundDetailed}. For (i), Corollary~\ref{CorFus} gives
\begin{multline}\label{eq:FusionChiDetailed}
\langle\mathcal{A},\hat{\mathcal{A}}\rangle_{l_1,n_1+1}
\otimes_{\pi}
\Delta^{(j)}_{z_1}\chi\,
\langle\mathcal{A},\hat{\mathcal{A}}\rangle_{l_2,n_2+1}
\lesssim
\Lambda^{-1}
\left(
1+
(1-\delta_{l_2,0})
\mathbf{1}_{\{\sigma_2^{(j)}\neq\varnothing\}}
\sum_{i\in\mathfrak{s}_2}\tau_i^{\Lambda_m}
\right)
\hat{\mathcal{A}}_{l,n}^{\Lambda,\Lambda_0}.
\end{multline}
Combining \eqref{eq:FusionChiDetailed} and \eqref{eq:CommonPolynomialBoundDetailed}, we obtain
\[
\langle\mathcal{A},\hat{\mathcal{A}}\rangle_{l_1,n_1+1}
\otimes_{\pi}
\Delta^{(j)}_{z_1}\chi\,
\langle\mathcal{A},\hat{\mathcal{A}}\rangle_{l_2,n_2+1}
\left(
Q^\sigma\otimes_\pi^{(j)}Q^\sigma
\right)
\lesssim
\Lambda^{-1}
\hat{\mathcal{A}}_{l,n}^{\Lambda,\Lambda_0}
Q_q^{\sigma^{(j)}}.
\]
This is the desired estimate in the case of one $\chi$-insertion.\\
Next we prove (ii). By Corollary~\ref{CorFus},
\begin{equation}\label{eq:FusionKDetailed}
\mathcal{A}^{\Lambda,\Lambda_0}_{l_1,n_1+1}
\otimes_{\pi}
\Delta_{z_1}K_{\tau_j,y_j}\,
\mathcal{A}^{\Lambda,\Lambda_0}_{l_2,n_2+1}
\lesssim
\tau_j^{-1/2}~\Lambda^{-1}~
\mathcal{A}_{l,n}^{\Lambda,\Lambda_0}~.
\end{equation}
Using \eqref{Prop2Q} from Lemma~\ref{QLemma}, we obtain
\begin{equation}\label{eq:TaujAbsorptionDetailed}
\left(Q_1\otimes_\pi Q_1\right)
\left(\tau_{\mathfrak{s}^{(j)}}^{\Lambda_m},y^\Lambda\right)
\lesssim
Q_{1}
\left(\tau_{\mathfrak{s}^{(j)}}^{\Lambda_m}\right).
\end{equation}
Multiplying \eqref{eq:FusionKDetailed} by the polynomial factor and
using \eqref{eq:TaujAbsorptionDetailed}, we obtain
\begin{multline}\label{eq:KAfterTaujDetailed}
\left(\mathcal{A}^{\Lambda,\Lambda_0}_{l_1,n_1+1}
\otimes_{\pi}
\Delta_{z_1}K_{\tau_j,y_j}\,
\mathcal{A}^{\Lambda,\Lambda_0}_{l_2,n_2+1}\right)\left(z_1;(\tau_i,y_i)_{i\in\mathfrak{s}}\right)~
\left(Q_1\otimes_\pi Q_1
\right)\left(\tau_{\mathfrak{s}^{(j)}}^{\Lambda_m},y^{\Lambda}\right)\\
\lesssim
\tau_j^{-\frac{1}{2}}~\Lambda^{-1}~
\mathcal{A}_{l,n}^{\Lambda,\Lambda_0}\left(z_1;(\tau_i,y_i)_{i\in\mathfrak{s}}\right)~Q_1\left(\tau_{\mathfrak s^{(j)}}^{\Lambda_m}\right)~.
\end{multline}
This proves \eqref{eq283}.
\end{proof}
\begin{proposition}\label{prop:B-estimate}
Let $n\ge4$ and write $
n=n_1+n_2$ with $n_1,n_2\ge1$. Let $l_1,l_2\ge0$ and $(\rho_1,\rho_2)\in\mathbb Z^2$, and set
$$
l:=l_1+l_2\ge1,
\qquad
\rho:=\rho_1+\rho_2.
$$
We use the notation of Proposition~\ref{propQu}.

\medskip

\noindent
\textit{(i)}
Let
$
(T_1,T_2)
\in
\mathcal D^{\Lambda_m}_{\rho_1,l_1}(\mathbb M^2)
\times
\mathcal D^{\Lambda_m}_{\rho_2,l_2}(\mathbb M^n)
$. Then
\begin{equation}
\mathcal N^{\Lambda_m}_{\rho,l}
\Bigl(
\Delta
\mathbf B^\Lambda(T_1,T_2)
(\underline p_n)
\Bigr)
\lesssim
\mathcal N^{\Lambda_m}_{\rho_1,l_1}
\bigl(
T_1(p_1)
\bigr)
\mathcal N^{\Lambda_m}_{\rho_2,l_2}
\left(
T_2\Bigl(
\sum_{i=2}^n p_i
\Bigr)
\right).
\end{equation}

\medskip

\noindent
\textit{(ii)}
Assume that $n_1,n_2>1$, and let $
\pi=\pi_1\cup\pi_2$
be a partition of $\left\{1,\cdots,n\right\}$. Let $
(T_1,T_2)
\in
\mathcal D^{\Lambda_m}_{\rho_1,l_1}
(\mathbb M^{n_1+1})
\times
\mathcal D^{\Lambda_m}_{\rho_2,l_2}
(\mathbb M^{n_2+1})$,
and fix $j\in\mathcal I_n$.

\begin{itemize}

\item[(a)]
Assume that $
\Delta^{(j)}T_1
\in
\mathcal D^{\Lambda_m}_{\rho_1,l_1}
(\mathbb M^{n_1+1})$. If $
\{1,j\}\subseteq\pi_1$,
then
\begin{equation}
\mathcal N^{\Lambda_m}_{\rho,l}
\Bigl(
\Delta^{(j)}
\mathbf B_\pi^\Lambda(T_1,T_2)
(\underline p_n)
\Bigr)
\lesssim
\mathcal N^{\Lambda_m}_{\rho_1,l_1}
\bigl(
\Delta^{(j)}T_1(\underline p_{\pi_1})
\bigr)
\mathcal N^{\Lambda_m}_{\rho_2,l_2}
\bigl(
T_2(\underline p_{\pi_2})
\bigr).
\end{equation}

\item[(b)]
Assume that $
\Delta^{(j)}T_2
\in
\mathcal D^{\Lambda_m}_{\rho_2,l_2}
(\mathbb M^{n_2+1})$.
If $
\{1,j\}\subseteq\pi_2$,
then
\begin{equation}
\mathcal N^{\Lambda_m}_{\rho,l}
\Bigl(
\Delta^{(j)}
\mathbf B_\pi^\Lambda(T_1,T_2)
(\underline p_n)
\Bigr)
\lesssim
\mathcal N^{\Lambda_m}_{\rho_1,l_1}
\bigl(
T_1(\underline p_{\pi_1})
\bigr)
\mathcal N^{\Lambda_m}_{\rho_2,l_2}
\bigl(
\Delta^{(j)}T_2(\underline p_{\pi_2})
\bigr).
\end{equation}

\item[(c)]
If $1$ and $j$ belong to different components of the partition $\pi$,
then
\begin{equation}
\mathcal N^{\Lambda_m}_{\rho,l}
\Bigl(
\Delta^{(j)}
\mathbf B_\pi^\Lambda(T_1,T_2)
(\underline p_n)
\Bigr)
\lesssim
\mathcal N^{\Lambda_m}_{\rho_1,l_1}
\bigl(
T_1(\underline p_{\pi_1})
\bigr)
\mathcal N^{\Lambda_m}_{\rho_2,l_2}
\bigl(
T_2(\underline p_{\pi_2})
\bigr).
\end{equation}

\end{itemize}
\end{proposition}
\begin{proof}
The proof follows the location of the discrete derivative insertion.
When the insertion belongs to one of the factors, it is absorbed into
the corresponding seminorm and the estimate follows directly from the
fusion bounds. When the distinguished variable (j) belongs to a
different component of the partition than the root variable (1), the
discrete derivative acts across the fusion edge and the estimate
requires the $\Delta K$-fusion bound established in
Lemma~\ref{LemmaFusionPoly}. We therefore treat these situations
separately.

We first establish \emph{(i)}. Let
$
(T_1,T_2)
\in
\mathcal D^{\Lambda_m}_{\rho_1,l_1}(\mathbb M^2)
\times
\mathcal D^{\Lambda_m}_{\rho_2,l_2}(\mathbb M^n)$. Fix
$
\sigma\subseteq\mathfrak s\subseteq\mathcal I_n$
and $j\in\mathfrak s$. Using the definition of $\mathbf B^\Lambda$ together with the semigroup
property of the heat kernel, we write
\begin{multline}
\left(
K^{\otimes\mathfrak s^{(j)}}
\otimes
\Delta_{z_1}K_{\tau_j,y_j}
\right)
\star
\mathbf B^\Lambda(T_1,T_2)
\bigl(
z_1;\underline y_{\mathfrak s};\underline p_n
\bigr)
=
e^{-\frac{m^2}{2\Lambda^2}}
\int_{\mathbb R}du\,
\bigl(
K\star T_1
\bigr)
\bigl(
z_1;\bar u;p_1
\bigr)
\,
\Delta_{z_1}K(\tau_j;z_j,y_j)
\\
\times
\left(
K^{\otimes\mathfrak s^{(j)}}
\otimes
K
\star
T_2
\right)
\left(
z_j;
\underline y_{\mathfrak s^{(j)}},
\bar u;
\sum_{i=2}^n p_i
\right).
\end{multline}
Using the seminorm bounds for $T_1$ and $T_2$, we obtain
\begin{multline}\label{est:i:1-final-corr}
\Bigl|
\left(
K^{\otimes\mathfrak s^{(j)}}
\otimes
\Delta_{z_1}K_{\tau_j,y_j}
\right)
\star
\mathbf B^\Lambda(T_1,T_2)
\bigl(
z_1;\underline y_{\mathfrak s};\underline p_n
\bigr)
\Bigr|
\le
e^{-\frac{m^2}{2\Lambda^2}}~
\Lambda_m^\rho~
\mathcal N^{\Lambda_m}_{\rho_1,l_1}
\bigl(
T_1(p_1)
\bigr)
\mathcal N^{\Lambda_m}_{\rho_2,l_2}
\left(
T_2\Bigl(
\sum_{i=2}^n p_i
\Bigr)
\right)
\\
\times
\left(
\sum_{r=0}^3
(y^\Lambda)^r
\right)
Q_1
\bigl(
\tau_{\mathfrak s^{(j)}}^{\Lambda_m},
y^\Lambda
\bigr)
\mathcal A_{l_1,2}^{\Lambda,\Lambda_0}
\otimes_\pi
\Delta_{z_1}K_{\tau_j,y_j}\,
\mathcal A_{l_2,n}^{\Lambda,\Lambda_0}
\bigl(
z_1;\underline y_{\mathfrak s}
\bigr).
\end{multline}
Applying the fusion estimate \eqref{eq283} together with
\eqref{Prop2QBis} and the elementary bound
\begin{equation}\label{CovEB}
\Lambda^{-r}e^{-m^2/(2\Lambda^2)}
\lesssim
\Lambda_m^{-r},~~~\forall r\in\mathbb{N}
\end{equation}
we deduce that
\begin{multline}\label{est:i:2-final-corr}
\Bigl|
\left(
K^{\otimes\mathfrak s^{(j)}}
\otimes
\Delta_{z_1}K_{\tau_j,y_j}
\right)
\star
\mathbf B^\Lambda(T_1,T_2)
\bigl(
z_1;\underline y_{\mathfrak s};\underline p_n
\bigr)
\Bigr|
\lesssim
\Lambda_m^{\rho-1}
\tau_j^{-1/2}
\mathcal N^{\Lambda_m}_{\rho_1,l_1}
\bigl(
T_1(p_1)
\bigr)
\mathcal N^{\Lambda_m}_{\rho_2,l_2}
\left(
T_2\Bigl(
\sum_{i=2}^n p_i
\Bigr)
\right)
\\
\times
Q_1
\bigl(
\tau_{\mathfrak s^{(j)}}^{\Lambda_m}
\bigr)
~\mathcal A_{l,n}^{\Lambda,\Lambda_0}
\bigl(
z_1;\underline y_{\mathfrak s}
\bigr).
\end{multline}
Taking the supremum over all variables entering the seminorm yields
\[
\mathcal N^{\Lambda_m}_{\rho,l}
\Bigl(
\Delta^{(j)}
\mathbf B^\Lambda(T_1,T_2)
(\underline p_n)
\Bigr)
\lesssim
\mathcal N^{\Lambda_m}_{\rho_1,l_1}
\bigl(
T_1(p_1)
\bigr)
\mathcal N^{\Lambda_m}_{\rho_2,l_2}
\left(
T_2\Bigl(
\sum_{i=2}^n p_i
\Bigr)
\right).
\]

\medskip

\noindent
We now establish \emph{(ii)}:

\medskip

\noindent
\emph{Case a): $\{1,j\}\in\pi_1$:}
proceeding as above, and using that $
\Delta^{(j)}T_1
\in
\mathcal D^{\Lambda_m}_{\rho_1,l_1}$
we obtain
\begin{multline}
\Bigl|
\left(
K^{\otimes\mathfrak s^{(j)}}
\otimes
\Delta_{z_1}K_{\tau_j,y_j}
\right)
\star
\mathbf B_\pi^\Lambda(T_1,T_2)
\bigl(
z_1;\underline y_{\mathfrak s};\underline p_n
\bigr)
\Bigr|
\le
\Lambda_m^{\rho-1}
\mathcal N^{\Lambda_m}_{\rho_1,l_1}
\bigl(
\Delta^{(j)}T_1(\underline p_{\pi_1})
\bigr)
\mathcal N^{\Lambda_m}_{\rho_2,l_2}
\bigl(
T_2(\underline p_{\pi_2})
\bigr)
\\
\times
\tau_j^{-1/2}
Q_1
\bigl(
\tau_{\mathfrak s_1^{(j)}}^{\Lambda_m},
y^\Lambda
\bigr)
Q_1
\bigl(
\tau_{\mathfrak s_2}^{\Lambda_m}
\bigr)
\left(
\mathcal A_{l_1,n_1+1}^{\Lambda,\Lambda_0}
\otimes_\pi
\mathcal A_{l_2,n_2+1}^{\Lambda,\Lambda_0}
\right)
\bigl(
z_1;\underline y_{\mathfrak s}
\bigr).
\end{multline}
The crucial observation is that the kernel
\(
\Delta_{z_1}K_{\tau_j,y_j}
\)
is already accounted for in the seminorm defining
\(
\Delta^{(j)}T_1
\).
Consequently, the fusion estimate does not generate any additional
derivative insertion. Consequently, no additional
$\tau_j^{-1/2}$ factor arises from the fusion estimate itself. Applying
\eqref{Prop1Q}, we obtain
\begin{multline}\label{case1pi1-corr}
\Bigl|
\left(
K^{\otimes\mathfrak s^{(j)}}
\otimes
\Delta_{z_1}K_{\tau_j,y_j}
\right)
\star
\mathbf B_\pi^\Lambda(T_1,T_2)
\bigl(
z_1;\underline y_{\mathfrak s};\underline p_n
\bigr)
\Bigr|
\lesssim
\Lambda_m^{\rho-1}
\mathcal N^{\Lambda_m}_{\rho_1,l_1}
\bigl(
\Delta^{(j)}T_1(\underline p_{\pi_1})
\bigr)
\mathcal N^{\Lambda_m}_{\rho_2,l_2}
\bigl(
T_2(\underline p_{\pi_2})
\bigr)
\\
\times
\tau_j^{-1/2}
Q_1
\bigl(
\tau_{\mathfrak s^{(j)}}^{\Lambda_m}
\bigr)
\left(
\mathcal A_{l_1,n_1+1}^{\Lambda,\Lambda_0}
\otimes_\pi
\mathcal A_{l_2,n_2+1}^{\Lambda,\Lambda_0}
\right)
\bigl(
z_1;\underline y_{\mathfrak s}
\bigr).
\end{multline}
The proof of case~(b) is identical after interchanging the roles of
\(T_1\) and \(T_2\). We therefore omit the details.
\medskip

\noindent
\emph{Case c): $1$ and $j$ belong to different components of $\pi$.} Without loss of generality we take $1\in\pi_1$ and $j\in\pi_2$.
In contrast to Cases~(a) and~(b), the discrete derivative now acts
across the fusion edge connecting the two components of the partition.
The estimate therefore relies on the $\Delta K$-fusion bound of
Lemma~\ref{LemmaFusionPoly}. We have
\begin{multline}
\Bigl|
\left(
K^{\otimes\mathfrak s^{(j)}}
\otimes
\Delta_{z_1}K_{\tau_j,y_j}
\right)
\star
\mathbf B_\pi^\Lambda(T_1,T_2)
\bigl(
z_1;\underline y_{\mathfrak s};\underline p_n
\bigr)
\Bigr|
\le
e^{-\frac{m^2}{2\Lambda^2}}
\Lambda_m^\rho~
\mathcal N^{\Lambda_m}_{\rho_1,l_1}
\bigl(
T_1(\underline p_{\pi_1})
\bigr)
\mathcal N^{\Lambda_m}_{\rho_2,l_2}
\bigl(
T_2(\underline p_{\pi_2})
\bigr)
\\
\times
Q_1
\bigl(
\tau_{\mathfrak s_1}^{\Lambda_m},
y^\Lambda
\bigr)
Q_1
\bigl(
\tau_{\mathfrak s_2^{(j)}}^{\Lambda_m}
\bigr)
\langle
\mathcal A,\hat{\mathcal A}
\rangle_{l_1,n_1+1}
\otimes_\pi
\Delta_{z_1}K_{\tau_j,y_j}\,
\langle
\mathcal A,\hat{\mathcal A}
\rangle_{l_2,n_2+1}
\bigl(
z_1;\underline y_{\mathfrak s}
\bigr).
\end{multline}
Using \eqref{eq283} from Lemma~\ref{LemmaFusionPoly}, together with
\eqref{Prop1Q} and \eqref{CovEB} we deduce that
\begin{multline}\label{case2pi2-corr}
\Bigl|
\left(
K^{\otimes\mathfrak s^{(j)}}
\otimes
\Delta_{z_1}K_{\tau_j,y_j}
\right)
\star
\mathbf B_\pi^\Lambda(T_1,T_2)
\bigl(
z_1;\underline y_{\mathfrak s};\underline p_n
\bigr)
\Bigr|
\lesssim
\Lambda_m^{\rho-1}
\mathcal N^{\Lambda_m}_{\rho_1,l_1}
\bigl(
T_1(\underline p_{\pi_1})
\bigr)
\mathcal N^{\Lambda_m}_{\rho_2,l_2}
\bigl(
T_2(\underline p_{\pi_2})
\bigr)
\\
\times
\tau_j^{-1/2}
Q_1
\bigl(
\tau_{\mathfrak s^{(j)}}^{\Lambda_m}
\bigr)
\mathcal A_{l,n}^{\Lambda,\Lambda_0}
\bigl(
z_1;\underline y_{\mathfrak s}
\bigr).
\end{multline}
Taking the supremum over the variables entering the seminorms in
\eqref{case1pi1-corr} and \eqref{case2pi2-corr} yields the claimed
bounds.
\end{proof}
\begin{proposition}\label{prop:B-estimateInsertion}
Fix $n\ge4$ and write
$n=n_1+n_2$ and $n_1,n_2\ge1$.
Let $l_1,l_2\ge0$ and $\rho_1,\rho_2\in\mathbb{Z}$, and define
\[
l:=l_1+l_2\ge 1,
\qquad
\rho:=\rho_1+\rho_2.
\]
We use the notation of Proposition~\ref{propQu} for momentum variables and partitions.

\medskip

\noindent
\textit{(i)}
Assume that $n_1,n_2>1$ and let $\pi=\pi_1\cup\pi_2$ be a partition.
Let $
(T_1,T_2)\in
\mathcal{D}^{\Lambda_m}_{\rho_1,l_1}\bigl(\mathbb{M}^{n_1+1}\bigr)
\times
\mathcal{D}^{\Lambda_m}_{\rho_2,l_2}\bigl(\mathbb{M}^{n_2+1}\bigr)$
and assume furthermore that for every $\chi\in\overline{\mathcal K}$,
\[
\mathcal N^{\Lambda_m}_{\rho_1,l_1;\infty}
\bigl(
\Delta\chi T_1(\underline p_{\pi_1})
\bigr)
<\infty.
\]
Then, for every $\chi\in\overline{\mathcal K}$,
\begin{align}
&
\mathcal N^{\Lambda_m}_{\rho,l}
\Bigl(
\Delta\chi
\mathbf B^\Lambda_\pi(T_1,T_2)(\underline p_n)
\Bigr)
\nonumber\lesssim
\mathcal N^{\Lambda_m}_{\rho_1,l_1;\infty}
\bigl(
\Delta\chi T_1(\underline p_{\pi_1})
\bigr)\,
\mathcal N^{\Lambda_m}_{\rho_2,l_2;\infty}
\bigl(
T_2(\underline p_{\pi_2})
\bigr)
\nonumber\\
&\hspace{1.7cm}
+
\mathcal N^{\Lambda_m}_{\rho_1,l_1;\infty}
\bigl(
T_1(\underline p_{\pi_1})
\bigr)\,
\mathcal N^{\Lambda_m}_{\rho_2,l_2;\infty}
\bigl(
T_2(\underline p_{\pi_2})
\bigr).
\label{eq:BInsertionEstimate}
\end{align}

\medskip

\noindent
\textit{(ii)}
Let $
(T_1,T_2)\in
\mathcal{D}^{\Lambda_m}_{\rho_1,l_1}(\mathbb{M}^{2})
\times
\mathcal{D}^{\Lambda_m}_{\rho_2,l_2}(\mathbb{M}^{n})$.
Then, for every $\chi\in\overline{\mathcal K}$ and $j\in\mathcal{I}_{2,n}$
\begin{equation}
\mathcal N^{\Lambda_m}_{\rho,l}
\Bigl(
\Delta^{(j)}\chi
\mathbf B^\Lambda(T_1,T_2)(\underline p_n)
\Bigr)
\lesssim
\mathcal N^{\Lambda_m}_{\rho_1,l_1}
\bigl(
T_1(p_1)
\bigr)\,
\mathcal N^{\Lambda_m}_{\rho_2,l_2}
\Bigl(
T_2\Bigl(\sum_{i=2}^n p_i\Bigr)
\Bigr).
\end{equation}

\end{proposition}

\begin{proof}
The proof follows the location of the insertion. If the distinguished
vertex (j) belongs to the first component of the partition, the
factor $\Delta^{(j)}_{z_1}\chi$ can be absorbed directly into the seminorm
of $T_1$, and the estimate reduces to the standard bilinear bound.
When $j$ belongs to the second component, the insertion acts across
the fusion edge and the estimate follows from the $\Delta\chi$-fusion
bound established in Lemma~\ref{LemmaFusionPoly}. We treat these two
situations separately.

\medskip

\noindent
\textit{Case 1: $j\in\pi_1$.}
\medskip

Since \(j\in\pi_1\), the insertion
\(\Delta^{(j)}_{z_1}\chi\)
acts entirely on the first factor and can therefore be absorbed into
the seminorm of \(T_1\).
Assume first that $\sigma_2\neq\varnothing$ and pick
$i\in\sigma_2$. By definition of the seminorms,
\begin{multline}\label{eq:ProofBInsertionA}
\bigl|
K^{\otimes\mathfrak s}
\star
\chi^{\otimes\sigma\setminus\{j\}}
\mathbf B^\Lambda_\pi
\bigl(
\Delta^{(j)}_{z_1}\chi\,T_1,
T_2
\bigr)
(z_1;\underline y_{\mathfrak s};\underline p_n)
\bigr|
\\
\lesssim
\Lambda_m^{\rho+1}
\,
\mathcal N^{\Lambda_m}_{\rho_1,l_1}
\bigl(
\Delta\chi\,T_1(\underline p_{\pi_1})
\bigr)
\,
\mathcal N^{\Lambda_m}_{\rho_2,l_2;\infty}
\bigl(
T_2(\underline p_{\pi_2})
\bigr)
\\
\times
\mathcal Q^{\sigma_1^{(j)}}_{q_1}
\bigl(
\tau^{\Lambda_m}_{\mathfrak s_1},
y^\Lambda
\bigr)
\,
\mathcal Q^{\sigma_2^{(i)}}_{q_2-1}
\bigl(
\tau^{\Lambda_m}_{\mathfrak s_2^{(i)}},
y^\Lambda
\bigr)
\hat{\mathcal A}^{\Lambda,\Lambda_0}_{l_1,n_1+1}
\otimes_\pi
\langle
\mathcal A,
\hat{\mathcal A}
\rangle_{l_2,n_2+1}
\bigl(
z_1;
\underline y_{\mathfrak s};
\chi^\sigma
\bigr).
\end{multline}
Using Corollary~\eqref{CorFus}, we obtain
\begin{equation}\label{eq:ProofBInsertionB}
\hat{\mathcal A}^{\Lambda,\Lambda_0}_{l_1,n_1+1}
\otimes_\pi
\langle
\mathcal A,
\hat{\mathcal A}
\rangle_{l_2,n_2+1}
\lesssim
\Bigl(
1+\sum_{i\in\mathfrak s_2}
\tau_i^{\Lambda_m}
\Bigr)
\hat{\mathcal A}^{\Lambda,\Lambda_0}_{l,n}.
\end{equation}
Furthermore, the polynomial estimates
\eqref{Prop2Q}--\eqref{Prop2QBis} together with
\eqref{Prop5Q} from Lemma~\ref{QLemma} imply
\begin{equation}\label{eq:ProofBInsertionC}
\Bigl(
1+\sum_{i\in\mathfrak s_2}
\tau_i^{\Lambda_m}
\Bigr)
\,
\mathcal Q^{\sigma_1^{(j)}}_{q_1-1}
\bigl(
\tau^{\Lambda_m}_{\mathfrak s_1},
y^\Lambda
\bigr)
\,
\mathcal Q^{\sigma_2^{(i)}}_{q_2-1}
\bigl(
\tau^{\Lambda_m}_{\mathfrak s_2^{(i)}},
y^\Lambda
\bigr)
\lesssim
\mathcal Q^{\sigma^{(j)}}_{q-1}
\bigl(
\tau^{\Lambda_m}_{\mathfrak s}
\bigr).
\end{equation}
Combining
\eqref{eq:ProofBInsertionA}--\eqref{eq:ProofBInsertionC}
yields
\begin{multline}\label{eq:ProofBInsertionD}
\bigl|
K^{\otimes\mathfrak s}
\star
\chi^{\otimes\sigma\setminus\{j\}}
\mathbf B^\Lambda_\pi
\bigl(
\Delta^{(j)}_{z_1}\chi\,T_1,
T_2
\bigr)
(z_1;\underline y_{\mathfrak s};\underline p_n)
\bigr|
\\
\lesssim
\Lambda_m^{\rho-1}
\,
\mathcal N^{\Lambda_m}_{\rho_1,l_1}
\bigl(
\Delta\chi\,T_1(\underline p_{\pi_1})
\bigr)
\,
\mathcal N^{\Lambda_m}_{\rho_2,l_2;\infty}
\bigl(
T_2(\underline p_{\pi_2})
\bigr)
\\
\times
\mathcal Q^{\sigma^{(j)}}_{q-1}
\bigl(
\tau^{\Lambda_m}_{\mathfrak s}
\bigr)
\,
\hat{\mathcal A}^{\Lambda,\Lambda_0}_{l,n}
\bigl(
z_1;
\underline y_{\mathfrak s};
\chi^\sigma
\bigr).
\end{multline}
If instead $\sigma_2=\varnothing$, we use the simpler fusion bound
\[
\hat{\mathcal A}^{\Lambda,\Lambda_0}_{l_1,n_1+1}
\otimes_\pi
\mathcal A^{\Lambda,\Lambda_0}_{l_2,n_2+1}
\lesssim
\hat{\mathcal A}^{\Lambda,\Lambda_0}_{l,n},
\]
together with
\[
\mathcal Q^{\sigma_1^{(j)}}_{q_1-1}
\bigl(
\tau^{\Lambda_m}_{\mathfrak s_1},
y^\Lambda
\bigr)
\,
\mathcal Q_{1}
\bigl(
\tau^{\Lambda_m}_{\mathfrak s_2},
y^\Lambda
\bigr)
\lesssim
\mathcal Q^{\sigma^{(j)}}_{q-1}
\bigl(
\tau^{\Lambda_m}_{\mathfrak s}
\bigr).
\]
This again yields the bound
\eqref{eq:ProofBInsertionD}.

\medskip

\noindent
\textit{Case 2: \(j\in\pi_2\).}
In contrast to Case~1, the insertion now acts across the fusion edge.
The estimate therefore relies on the \(\Delta\chi\)-fusion bound
\eqref{AAchi}. We write 
\begin{multline}
\bigl|
K^{\otimes\mathfrak s}
\star
\chi^{\otimes\sigma\setminus\{j\}}
\Delta^{(j)}_{z_1}\chi\,
\mathbf B^\Lambda_\pi(T_1,T_2)
(z_1;\underline y_{\mathfrak s};\underline p_n)
\bigr|
\lesssim
e^{-m^2/\Lambda^2}
\Lambda_m^\rho
\,
\mathcal N^{\Lambda_m}_{\rho_1,l_1;\infty}
\bigl(
T_1(\underline p_{\pi_1})
\bigr)
\,
\mathcal N^{\Lambda_m}_{\rho_2,l_2;\infty}
\bigl(
T_2(\underline p_{\pi_2})
\bigr)
\\\times
\bigl(
Q^\sigma\otimes_\pi^{(j)}Q^\sigma
\bigr)
\bigl(
\tau^{\Lambda_m}_{\mathfrak s},
y^\Lambda
\bigr)~
\langle
\mathcal A,
\hat{\mathcal A}
\rangle_{l_1,n_1+1}
\otimes_\pi^{(j)}
\Delta_{z_1}\chi\,
\langle
\mathcal A,
\hat{\mathcal A}
\rangle_{l_2,n_2+1}
\bigl(
z_1;
\underline y_{\mathfrak s};
\chi^\sigma
\bigr).
\end{multline}
Using \eqref{AAchi}, we infer that
\[
\langle
\mathcal A,
\hat{\mathcal A}
\rangle_{l_1,n_1+1}
\otimes_\pi^{(j)}
\Delta_{z_1}\chi\,
\langle
\mathcal A,
\hat{\mathcal A}
\rangle_{l_2,n_2+1}
\,
\bigl(
Q^\sigma\otimes_\pi^{(j)}Q^\sigma
\bigr)
\lesssim
\Lambda^{-1}
\hat{\mathcal A}^{\Lambda,\Lambda_0}_{l,n}
Q^{\sigma^{(j)}}_q.
\]
This together with the covariance estimate \eqref{CovEB} yields
\begin{multline}
\bigl|
K^{\otimes\mathfrak s}
\star
\chi^{\otimes\sigma\setminus\{j\}}
\Delta^{(j)}_{z_1}\chi\,
\mathbf B^\Lambda_\pi(T_1,T_2)
(z_1;\underline y_{\mathfrak s};\underline p_n)
\bigr|
\\
\lesssim
\Lambda_m^{\rho-1}
\,
\mathcal N^{\Lambda_m}_{\rho_1,l_1;\infty}
\bigl(
T_1(\underline p_{\pi_1})
\bigr)
\,
\mathcal N^{\Lambda_m}_{\rho_2,l_2;\infty}
\bigl(
T_2(\underline p_{\pi_2})
\bigr)
\\
\times
Q_q^{\sigma^{(j)}}
\bigl(
\tau^{\Lambda_m}_{\mathfrak s}
\bigr)
\,
\hat{\mathcal A}^{\Lambda,\Lambda_0}_{l,n}
\bigl(
z_1;
\underline y_{\mathfrak s};
\chi^\sigma
\bigr).
\end{multline}
Taking the supremum over all admissible parameters in the definition of
$\mathcal N^{\Lambda_m}_{\rho,l}$ concludes the proof of~\textit{(i)}.
The proof of part~(ii) is identical. The only difference is that the
first factor belongs to the two-point sector, so that one uses the
special form of the seminorms introduced in
\eqref{e203-n2}. We omit the details.
\end{proof}
\subsection{Integration and Taylor expansion of distributions in $\mathcal{D}^{\Lambda_m}_{\rho,l}$}
The irrelevant sector is controlled by integrating the flow equation
from the ultraviolet scale \(\Lambda_0\) down to \(\Lambda\). The
stability estimates established above imply bounds on the scale
derivative \(\partial_\Lambda T^\Lambda\). The following proposition
shows that integration in the scale variable improves the
power-counting degree by one and therefore provides the basic induction
step for irrelevant distributions.
\begin{proposition}[Integration in the scale for irrelevant distributions]
\label{integration}
Let $\rho<0$, $l\in\mathbb N_0$, and let
\[
\Lambda\in[0,\Lambda_0]
\longmapsto
T^\Lambda
\]
be a family of distributions on $\mathbb M^n$. Assume that there exists
$c_l^{\Lambda_0}\in\mathbb R$ such that
\begin{equation}\label{eq:UVboundaryIrrelevant}
T^{\Lambda_0}
\left(\underline{z}_n;\underline{p}_{n}\right)
=
c_l^{\Lambda_0}(z_1)
\prod_{i=2}^n
\delta(z_1-z_i).
\end{equation}
Assume furthermore that
$
\Delta\left(\partial_\Lambda T^\Lambda\right)\in
\mathcal D^{\Lambda_m}_{\rho,l}
(\mathbb M^n)
$ in the sense of Definition~\ref{FamilyT}. Then
\begin{equation}\label{eq:irrelevantIntegratedBound}
\Delta T^\Lambda
\in
\mathcal D^{\Lambda_m}_{\rho+1,l}
(\mathbb M^n).
\end{equation}
Moreover, if $\rho<-1$, $
T^{\Lambda_0}=0$,
and the distribution $
\partial_\Lambda T^\Lambda
\in
\mathcal D^{\Lambda_m}_{\rho,l}
(\mathbb M^n)$ then
\begin{equation}\label{eq:irrelevantIntegratedBound2}
T^\Lambda
\in
\mathcal D^{\Lambda_m}_{\rho+1,l}
(\mathbb M^n).
\end{equation}
\end{proposition}
\begin{proof}
   The proposition expresses the fact that integration in the scale
variable improves the power-counting degree by one. The argument is a
direct consequence of the defining seminorm bounds together with the
elementary estimate \eqref{irrePol}.

We first prove \eqref{eq:irrelevantIntegratedBound}. Let
$
\Lambda\longmapsto T^\Lambda
$ satisfy the assumptions of the proposition. By
Definition~\ref{FamilyT}, for every $j\in\mathcal I_n$ and every
$\chi\in\overline{\mathcal K}$, one has
    \begin{equation}\label{321-}
        \mathcal N^{\Lambda_m}_{\rho,l}
\Bigl(
\Delta^{(j)}
\partial_\Lambda
T^\Lambda
(\underline p_n)
\Bigr)+\mathcal N^{\Lambda_m}_{\rho,l}
\Bigl(
\Delta^{(j)}\chi
\partial_\Lambda
T^\Lambda
(\underline p_n)
\Bigr)
\lesssim
\mathcal P_{l-1}
\left(
\log\frac{\Lambda_m}{m}
\right)
\widetilde{\mathcal P}
\left(
\frac{\|\underline p_n\|}{\Lambda_m}
\right)~.
    \end{equation}
    In particular, for all $\sigma\subseteq\mathfrak{s}\subseteq\mathcal{I}_{n}$ we have that 
    \begin{multline}\label{322-}
\bigl|
K^{\otimes\mathfrak s^{(j)}}\otimes\Delta^{(j)}_{z_1}K
\star
\partial_{\Lambda}T^{\Lambda}
(z_1;\underline y_{\mathfrak s};\underline p_n)
\bigr|\\
\lesssim
\Lambda_m^{\rho-1}
\,
\mathcal N^{\Lambda_m}_{\rho,l}
\bigl(
\Delta^{(j)}
\partial_\Lambda
T^\Lambda\left(\underline{p}_n\right)
\bigr)~
\tau_{j}^{-\frac{1}{2}}Q_1
\bigl(
\tau^{\Lambda_m}_{\mathfrak s}
\bigr)
\,
{\mathcal A}^{\Lambda,\Lambda_0}_{l,n}
\bigl(
z_1;
\underline y_{\mathfrak s}
\bigr).
\end{multline}
Similarly,
\begin{multline}\label{323-}
\bigl|
K^{\otimes\mathfrak s}
\star
\chi^{\otimes\sigma\setminus\{j\}}
\Delta^{(j)}_{z_1}\chi\,
\mathbf \partial_{\Lambda}T^{\Lambda}
(z_1;\underline y_{\mathfrak s};\underline p_n)
\bigr|
\lesssim
\Lambda_m^{\rho-1}
\,
\mathcal N^{\Lambda_m}_{\rho,l}
\bigl(
T(\underline p_n)
\bigr)\\\times
Q_{|\sigma|}^{\sigma^{(j)}}
\bigl(
\tau^{\Lambda_m}_{\mathfrak s}
\bigr)
\,
\left(\mathcal{A}_{l,n}^{\Lambda,\Lambda_0}+\Lambda_m^{-1}\hat{\mathcal A}^{\Lambda,\Lambda_0}_{l,n}\right)
\bigl(
z_1;
\underline y_{\mathfrak s};
\chi^\sigma
\bigr).
\end{multline}
Combining \eqref{321-} with \eqref{322-} and \eqref{323-}, we obtain
\begin{multline}\label{322-}
\bigl|
K^{\otimes\mathfrak s^{(j)}}\otimes\Delta^{(j)}_{z_1}K
\star
\partial_{\Lambda}T^{\Lambda}
(z_1;\underline y_{\mathfrak s};\underline p_n)
\bigr|
\lesssim
\Lambda_m^{\rho-1}~
\tau_{j}^{-\frac{1}{2}}\\\times \mathcal P_{l-1}
\left(
\log\frac{\Lambda_m}{m}
\right)
\widetilde{\mathcal P}
\left(
\frac{\|\underline p_n\|}{\Lambda_m}
\right) Q_1
\bigl(
\tau^{\Lambda_m}_{\mathfrak s}
\bigr)
\,
{\mathcal A}^{\Lambda,\Lambda_0}_{l,n}
\bigl(
z_1;
\underline y_{\mathfrak s}
\bigr)
\end{multline}
and 
\begin{multline}\label{323-}
\bigl|
K^{\otimes\mathfrak s}
\star
\chi^{\otimes\sigma\setminus\{j\}}
\Delta^{(j)}_{z_1}\chi\,
\mathbf \partial_{\Lambda}T^{\Lambda}
(z_1;\underline y_{\mathfrak s};\underline p_n)
\bigr|
\lesssim
\Lambda_m^{\rho-1}\mathcal P_{l-1}
\left(
\log\frac{\Lambda_m}{m}
\right)
\\\times
Q_{|\sigma|}^{\sigma^{(j)}}
\bigl(
\tau^{\Lambda_m}_{\mathfrak s}
\bigr)\widetilde{\mathcal P}
\left(
\frac{\|\underline p_n\|}{\Lambda_m}
\right)
\,
\left(\mathcal{A}_{l,n}^{\Lambda,\Lambda_0}+\Lambda_m^{-1}\hat{\mathcal A}^{\Lambda,\Lambda_0}_{l,n}\right)
\bigl(
z_1;
\underline y_{\mathfrak s};
\chi^\sigma
\bigr).
\end{multline}
Observe now that, for every $
\lambda\in[\Lambda,\Lambda_0]$ one has
\begin{equation}
    \mathcal{A}_{l,n}^{\lambda,\Lambda_0}\le \mathcal{A}_{l,n}^{\lambda,\Lambda_0},~~~~\hat{\mathcal{A}}_{l,n}^{\lambda,\Lambda_0}\le \hat{\mathcal{A}}_{l,n}^{\lambda,\Lambda_0},~~~~\tau_{\mathfrak{s}}^{\lambda_m}\le \tau_{\mathfrak{s}}^{\Lambda_m}~,~~~\frac{\|\underline{p}_n\|}{\lambda_m}\le \frac{\|\underline{p}_n\|}{\Lambda_m}~.
\end{equation}
The amplitudes estimates follow directly from the definition of the amplitudes.
Indeed, in $\mathcal{A}_{l,n}^{\lambda,\Lambda_0}$ the dependence on the scale $\lambda$ enters only through the supremum over the parameters $\Lambda_{ij}$, which represent the weights of the kernels associated with the internal lines of the trees or forests.
Since these parameters are constrained to lie in the interval $[\lambda,\Lambda_0]$, decreasing $\lambda$ enlarges the domain of the supremum and hence can only increase the value of the amplitude. Next, the boundary condition \eqref{eq:UVboundaryIrrelevant} implies that 
\begin{equation}
    \Delta^{(j)}_{z_1}T^{\Lambda_0}= \Delta^{(j)}_{z_1}\chi T^{\Lambda_0}
=0~.
\end{equation}
We may therefore integrate \eqref{322-} and \eqref{323-} from
$\Lambda$ to $\Lambda_0$. Using the polynomial estimate
\eqref{irrePol}, together with the assumption $
\rho-1\le-2$ we obtain
\begin{multline}\label{322bis}
\bigl|
K^{\otimes\mathfrak s^{(j)}}\otimes\Delta^{(j)}_{z_1}K
\star
T^{\Lambda}
(z_1;\underline y_{\mathfrak s};\underline p_n)
\bigr|
\lesssim
\Lambda_m^{\rho}~
\tau_{j}^{-\frac{1}{2}}\\\times \mathcal P_{l-1}
\left(
\log\frac{\Lambda_m}{m}
\right)
\widetilde{\mathcal P}
\left(
\frac{\|\underline p_n\|}{\Lambda_m}
\right) Q_1
\bigl(
\tau^{\Lambda_m}_{\mathfrak s}
\bigr)
\,
{\mathcal A}^{\Lambda,\Lambda_0}_{l,n}
\bigl(
z_1;
\underline y_{\mathfrak s}
\bigr)
\end{multline}
and 
\begin{multline}\label{323bis}
\bigl|
K^{\otimes\mathfrak s}
\star
\chi^{\otimes\sigma\setminus\{j\}}
\Delta^{(j)}_{z_1}\chi\,
T^{\Lambda}
(z_1;\underline y_{\mathfrak s};\underline p_n)
\bigr|
\lesssim
\Lambda_m^{\rho}\mathcal P_{l-1}
\left(
\log\frac{\Lambda_m}{m}
\right)
\\\times
Q_q^{\sigma^{(j)}}
\bigl(
\tau^{\Lambda_m}_{\mathfrak s}
\bigr)\widetilde{\mathcal P}
\left(
\frac{\|\underline p_n\|}{\Lambda_m}
\right)
\,
\left(\mathcal{A}_{l,n}^{\Lambda,\Lambda_0}+\Lambda_m^{-1}\hat{\mathcal A}^{\Lambda,\Lambda_0}_{l,n}\right)
\bigl(
z_1;
\underline y_{\mathfrak s};
\chi^\sigma
\bigr).
\end{multline}
The bounds \eqref{322bis} and \eqref{323bis} are precisely the defining
estimates for
\[
\Delta T^\Lambda
\in
\mathcal D^{\Lambda_m}_{\rho+1,l}
(\mathbb M^n).
\]
The proof of \eqref{eq:irrelevantIntegratedBound2} is identical. One
uses the seminorm bounds for
\(
\partial_\Lambda T^\Lambda
\)
instead of those for
\(
\Delta(\partial_\Lambda T^\Lambda)
\),
integrates from \(\Lambda\) to \(\Lambda_0\), and invokes
\(T^{\Lambda_0}=0\) together with \eqref{irrePol}. This yields the
defining seminorm bounds for
\[
T^\Lambda
\in
\mathcal D^{\Lambda_m}_{\rho+1,l}
(\mathbb M^n),
\]
and completes the proof.
\end{proof}
The two-point sector contains both relevant and irrelevant
contributions. Unlike the irrelevant sector, the flow equation alone is
not sufficient to recover the required power-counting bounds. The
relevant information is encoded in finitely many Taylor coefficients,
while the remaining part is controlled by irrelevant power counting.

From the perspective of the induction, the irrelevant contributions are
treated first by integrating the flow equation and applying the
stability estimates established above. This yields control of the
Taylor remainders. The relevant Taylor coefficients are then estimated
separately through the renormalisation conditions. The following
proposition shows that these two ingredients are sufficient to
reconstruct the full two-point distribution. More precisely, it proves
that uniform bounds on the Taylor coefficients together with suitable
control of the Taylor remainders imply the expected power-counting bound
for the complete distribution.

\begin{proposition}[$2$-point relevant distributions]\label{prop:relevant}
Fix $\rho\in\mathbb{N}$ and let $T\in\mathcal{S}'(\mathbb{M}^2)$. Assume that for every multi-index $w$ such that $|w|=\rho+1$ we have 
\begin{equation}\label{Asum1}
\partial_{\underline p}^{\,w}T
\in
\mathcal{D}^{\Lambda_m}_{-1,l}(\mathbb{M}^2)~
\end{equation}
and for $|w|\le \rho$
\begin{equation}\label{Est1*}
(z_1-z_2)^{\rho-|w|+1}
\partial_{\underline p}^{\,w}T
\in
\mathcal{D}^{\Lambda_m;(2)}_{-1,l}\!\left(\mathbb{M}^2\right)~.
\end{equation}
Furthermore, if there exists a polynomial $\mathcal{P}_{l-1}$ of order $l-1$ and positive coefficients such that, for all $r+|w|\le \rho$ we have 
\begin{equation}\label{Asum3}
\left\|\partial^{w}_{p} t_{r}\left(\cdot;0\right)\right\|_{L^{\infty}(\mathbb{R})}
\le
\Lambda_m^{\rho-r-|w|}
\mathcal{P}_{l-1}\!\left(\log\frac{\Lambda_m}{m}\right),
\end{equation}
where
\[
t_{r}\left(z_1;p,-p\right)
:=
\int_{\mathbb{R}}
(z_1-z_2)^r
T\bigl((z_1,p),(z_2,-p)\bigr)\,d{z}_{2}.
\]
Then
\begin{equation}\label{2190*}
   T\in\mathcal{D}^{\Lambda_m}_{\rho,l}(\mathbb{M}^2).
\end{equation}
In particular, we have 
\begin{multline}\label{274~-}
\mathcal{N}^{\Lambda_m}_{\rho,l;\infty}\left(T\left(p\right)\right)\lesssim~\sum_{|w|=0}^{\rho}|p|^{|w|}~\left(\sum_{r=0}^{\rho-|w|}\left\|\partial_p^{w}t_{r}\left(\cdot;0\right)\right\|_{L^{\infty}\left(\mathbb{R}\right)}+\mathcal{N}^{\Lambda_m;(2)}_{-1,l;\infty}\left((z_1-z_2)^{\rho-|w|+1}\partial_p^{w} T(0)\right)\right)\\+\|p\|^{\rho+1}\sum_{|w|=\rho+1}\int_0^1dt~\mathcal{N}^{\Lambda_m}_{-1,l;\infty}\left(\partial_p^{w} T\left(tp\right)\right)~.
\end{multline} 
\end{proposition}

\begin{proof}
The argument proceeds in three steps. We first perform a Taylor
expansion in the momentum variable around \(p=0\) and show that the
corresponding momentum remainder is controlled by the assumptions on
the highest-order momentum derivatives. We then analyse the Taylor coefficients at vanishing momentum. For each
coefficient, a second Taylor expansion around the diagonal
\(z_2=z_1\) separates the relevant contributions, encoded in the moments
\(t_r\), from a remainder controlled by the second-order estimate
\eqref{Est1*}. Finally, combining the bounds for the relevant Taylor
coefficients, the spatial remainder, and the momentum remainder, we
reconstruct the full distribution and obtain the claimed
power-counting estimate.

We first perform a Taylor expansion in the momentum variable around
\(p=0\) to order \(\rho\):
\begin{equation}\label{TaylorPolynomial}
T\left((z_1,p),(z_2,-p)\right)
=
\sum_{|w|\le \rho}\frac{p^w}{w!}\,
\partial_{\underline p}^{\,w}T\left((z_1,0),(z_2,0)\right)
+
\mathcal{R}^{(\rho+1)}_p\bigl(T(p)\bigr)\left(z_1,z_2\right),
\end{equation}
where
\[
\partial_{\underline p}^{\,w}T\left((z_1,0),(z_2,0)\right)
:=
\partial_{p}^{\,w}
T\bigl((z_1,p),(z_2,-p)\bigr)\left.\right|_{p\equiv 0},
\]
and where the remainder is given by
\begin{equation}\label{RpDef}
\mathcal{R}^{(\rho+1)}_p\bigl(T(p)\bigr)\left(z_1,z_2\right)
:=
\sum_{|w|=\rho+1}\frac{p^w}{w!}
\int_0^1 (1-t)^{\rho}
\partial^{\,w}
T\left((z_1,tp),(z_2,-tp)\right)~dt.
\end{equation}
\eqref{Asum1} gives for $|w|=\rho+1$ that
\[
\partial_{\underline p}^{\,w}T
\in
\mathcal{D}^{\Lambda_m}_{-1,l}(\mathbb{M}^2)~.
\]
Equivalently, there exist polynomials $\mathcal{P}_{l-1}$ and $\widetilde{\mathcal P}$ such that
\[
\mathcal{N}^{\Lambda_m}_{-1,l}
\Bigl(
\partial_{\underline p}^{\,w}T(tp)
\Bigr)
\le
\mathcal{P}_{l-1}\!\left(\log\frac{\Lambda_m}{m}\right)
\widetilde{\mathcal P}\!\left(\frac{t|p|}{\Lambda_m}\right).
\]
Using \eqref{RpDef}, we deduce that
\begin{align}
\mathcal{N}^{\Lambda_m}_{-1,l}
\Bigl(
\mathcal{R}^{(\rho+1)}_p(T(p))
\Bigr)
&\lesssim~~
|p|^{\rho+1}~\sum_{|w|=\rho+1}\int_0^1dt~\mathcal{N}^{\Lambda_m}_{-1,l;\infty}\left(\partial^w_{\underline{p}}T\left(tp\right)\right)~,\label{BorSem}\\
&\lesssim \Lambda_m^{\rho+1}~
\left(\frac{|p|}{\Lambda_m}\right)^{\rho+1}~\mathcal{P}_{l-1}\!\left(\log\frac{\Lambda_m}{m}\right)
\widetilde{\mathcal P}\!\left(\frac{|p|}{\Lambda_m}\right).\nonumber
\end{align}
Since
\[
\Lambda_m^{\rho+1}~
\mathcal{N}^{\Lambda_m}_{-1,l}
\Bigl(
\mathcal{R}^{(\rho+1)}_p(T(p))
\Bigr)
=
\mathcal{N}^{\Lambda_m}_{\rho,l}
\Bigl(
\mathcal{R}^{(\rho+1)}_p(T(p))
\Bigr),
\]
it follows that
\[
\mathcal{R}^{(\rho+1)}_p(T(p))
\in
\mathcal{D}^{\Lambda_m}_{\rho,l}(\mathbb{M}^2).
\]
We are thus left with the Taylor coefficients $
\partial_{\underline p}^{\,w}T\left((z_1,0),(z_2,0)\right)$ with $|w|\le \rho$. For such a multi-index $w$, we now perform a Taylor expansion in the second position variable around the diagonal $z_2=z_1$, to order $\rho-|w|$. This gives
\begin{multline}\label{e222}
\partial_{\underline p}^{\,w}T(0)
=
\sum_{r=0}^{\rho-|w|}
\frac{(-1)^r}{r!}\,
\partial_{\underline p}^{\,w}t_r(z_1;0)\,\delta^{(r)}(z_1-z_2)
\\
+(-1)^{\rho-|w|+1}
\int_0^1 \frac{(1-t)^{\rho-|w|}}{(\rho-|w|)!}\,
(z_1-z_2)^{\rho-|w|+1}\,
\partial_{z_2}^{\,\rho-|w|+1}
\partial_{\underline p}^{\,w}
T\bigl((z_1,0),(z_1+t(z_2-z_1),0)\bigr)\,dt,
\end{multline}
where
\begin{equation}\label{deftr2}
t_r(z_1;0)
:=
\int_{\mathbb{R}}
(z_1-z_2)^r\,
T\bigl((z_1,0),(z_2,0)\bigr)\,dz_2,
\qquad r=0,1,2.
\end{equation}
We now convolve \eqref{e222} with the heat kernel in the second variable. This yields
\begin{multline}\label{ConvTaylor}
\bigl(K\star \partial_{\underline p}^{\,w}T\bigr)\bigl(z_1;(\tau,y);0\bigr)
=
\sum_{r=0}^{\rho-|w|}
\frac{(-1)^r}{r!}\,
\partial_{\underline p}^{\,w}t_r(z_1;0)\,
K^{(r)}(\tau;z_1,y)
\\
+
\int_{\mathbb{R}}\int_0^1
\frac{(1-t)^{\rho-|w|}}{(\rho-|w|)!}\,
(z_1-z_2)^{\rho-|w|+1}\,
\partial_{\underline p}^{\,w}T\left((z_1,0),(z_2,0)\right)\,
\partial_Z^{\,\rho-|w|+1}K(\tau;Z,y)\big|_{Z=t z_1+(1-t)z_2}\,dt\,dz_2.
\end{multline}
We first estimate the relevant terms. By \eqref{Asum3}, we have for all $0\le r+|w|\le \rho$
\[
\bigl|
\partial_{\underline p}^{\,w}t_r(z_1;0)
\bigr|
\le
\Lambda_m^{\rho-r-|w|}
\mathcal{P}_{l-1}\!\left(\log\frac{\Lambda_m}{m}\right)~.
\]
Furthermore, the heat-kernel derivative bound implies that, for every $0<\delta<1$,
\begin{equation}\label{Kr}
|K^{(r)}(\tau;z,y)|
\lesssim
\tau^{-r/2}K((1+\delta)\tau;z,y),
\end{equation}
with an implicit constant depending only on $r$ and $\delta$. Choosing $\delta=\delta_l$, we obtain
\[
|K^{(r)}(\tau;z,y)|
\lesssim
\tau^{-r/2}\mathcal{A}_{l,2}^{\Lambda,\Lambda_0}\bigl(z;(\tau,y)\bigr).
\]
Therefore
\[
\bigl|
\partial_{\underline p}^{\,w}t_r(z_1;0)\,
K^{(r)}(\tau;z_1,y)
\bigr|
\lesssim
\Lambda_m^{\rho-|w|}
\left(\frac{\tau^{-1/2}}{\Lambda_m}\right)^r
\mathcal{P}_{l-1}\!\left(\log\frac{\Lambda_m}{m}\right)
\mathcal{A}_{l,2}^{\Lambda,\Lambda_0}\bigl(z_1;(\tau,y)\bigr).
\]
The preceding bound is precisely the defining estimate of
\(
\mathcal D_{\rho-|w|,l}^{\Lambda_m}
\)
for the distribution
$$
\partial_{\underline p}^{\,w}t_r(z_1;0)\delta^{(r)}(z_1-z_2)
$$ which shows that 
\[
\partial_{\underline p}^{\,w}t_r(z_1;0)\,\delta^{(r)}(z_1-z_2)\in
\mathcal{D}^{\Lambda_m}_{\rho-|w|,l}(\mathbb{R}^2).
\]
We now turn to the remainder in \eqref{e222}. Set
\begin{multline}\label{RzDef}
\mathcal{R}^{(\rho-|w|)}_z
\left(
\partial_{\underline p}^{\,w}T\right)\bigl(z_1,z_2\bigr)
:=
\int_0^1 \frac{(1-t)^{\rho-|w|}}{(\rho-|w|)!}\,
\\
\times
(z_1-z_2)^{\rho-|w|+1}\,
\partial_{z_2}^{\,\rho-|w|+1}
\partial_{\underline p}^{\,w}
T\bigl((z_1,0),(z_1+t(z_2-z_1),0)\bigr)\,dt.
\end{multline}
Recall that
\[
\mathcal{A}_{l,2}^{\Lambda,\Lambda_0}(z_1,z_2)
=
\sup_{\Lambda\le \Lambda_i\le \Lambda_0}
\sum_{v=1}^l
K\!\left(
\sum_{i=1}^v \frac{1+\delta_l}{\Lambda_i^2};z_1,z_2
\right).
\]
On the other hand, we have using \eqref{Kr}
\[
\bigl|
\partial_Z^{\,\rho-|w|+1}K(\tau;Z,y)
\bigr|
\lesssim
\tau^{-\left(\rho-|w|+1\right)/2}K\left((1+\delta_l)\tau;Z,y\right).
\]
Hence, we have 
\begin{multline}\label{Est2}
\int_{\mathbb{R}}dz_2\,
\mathcal{A}_{l,2}^{\Lambda,\Lambda_0}(z_1,z_2)\,
\bigl|
\partial_Z^{\,\rho-|w|+1}K(\tau;Z,y)\bigr|_{Z=t z_1+(1-t)z_2}
\\
\lesssim
\tau^{-(\rho-|w|+1)/2}
\sup_{\Lambda\le \Lambda_i\le \Lambda_0}
\sum_{v=1}^l
K\!\left(
\sum_{i=1}^v \frac{1+\delta_l}{\Lambda_i^2}+(1+\delta_l)\tau;z_1,y
\right).
\end{multline}
By the semigroup property, the right-hand side is bounded by
\[
\tau^{-(\rho-|w|+1)/2}
\mathcal{A}_{l,2}^{\Lambda,\Lambda_0}\bigl(z_1;(\tau,y)\bigr).
\]
Combining \eqref{Est1*} and \eqref{Est2}, we obtain that
\begin{multline}\label{Bouti}
\bigl|
K\star
\mathcal{R}^{(\rho-|w|)}_z
\left(
\partial_{\underline p}^{\,w}T(0)\right)
\left(z_1;\tau,y\right)
\bigr|
\\
\lesssim
\Lambda_m^{\rho-|w|}
\left(\frac{\tau^{-1/2}}{\Lambda_m}\right)^{\rho-|w|+1}
\mathcal{A}_{l,2}^{\Lambda,\Lambda_0}\bigl(z_1;(\tau,y)\bigr)~\mathcal{N}^{\Lambda_m;(2)}_{-1,l;\infty}\left((z_1-z_2)^{\rho-|w|+1}\partial^w_{\underline{p}}T(0)\right).
\end{multline}
This shows that
\[
\mathcal{R}^{(\rho-|w|)}_z
\left(
\partial_{\underline p}^{\,w}T(0)\right)
\in
\mathcal{D}^{\Lambda_m}_{\rho-|w|,l}(\mathbb{M}^2)
\]
and in particular \eqref{Bouti} implies that 
\begin{equation}\label{286-~}
    \mathcal{N}^{\Lambda_m;(2)}_{\rho-|w|,l;\infty}\left(\mathcal{R}^{(\rho-|w|)}_z
\left(
\partial_{\underline p}^{\,w}T(0)\right)\right)\lesssim \mathcal{N}^{\Lambda_m;(2)}_{-1,l;\infty}\left((z_1-z_2)^{\rho-|w|+1}\partial^w_{\underline{p}}T(0)\right)
\end{equation}
We have therefore shown that every term in \eqref{e222} belongs to $
\mathcal{D}^{\Lambda_m}_{\rho-|w|,l}(\mathbb{M}^2)$. In particular, we obtain combining  \eqref{e222} and \eqref{286-~} 
\begin{multline}\label{287~-}
\mathcal{N}^{\Lambda_m}_{\rho-|w|,l;\infty}\left(\partial^w_{\underline{p}}T(0)\right)\lesssim \sum_{r=0}^{\rho-|w|}\left\|\partial^{w}_{\underline{p}}t_r\left(0\right)\right\|_{\infty}+\int_0^1~\mathcal{N}^{\Lambda_m;(2)}_{\rho-|w|+1,l;\infty}\left(\left(z_1-z_2\right)^{\rho-|w|+1}\partial^w_{\underline{p}}T(0)\right)~.
\end{multline}
Applying Proposition~\ref{PropPoSp} to the previous estimate yields
\[
p^w\partial_{\underline p}^{\,w}T(0)
\in
\mathcal D_{\rho,l}^{\Lambda_m}(\mathbb M^2),
\qquad
|w|\le\rho.
\]
Returning to the momentum expansion \eqref{TaylorPolynomial}, we have already proved that
\[
\mathcal{R}^{(\rho+1)}_p(T(p))
\in
\mathcal{D}^{\Lambda_m}_{\rho,l}(\mathbb{M}^2).
\]
It follows that $T\in \mathcal{D}^{\Lambda_m}_{\rho,l}(\mathbb{M}^2)$. Furthermore, using \eqref{TaylorPolynomial} and \eqref{BorSem} we obtain
\begin{equation}
    \mathcal{N}^{\Lambda_m}_{\rho,l;\infty}\left(T(p)\right)\lesssim \sum_{|w|\le\rho}~p^w~\mathcal{N}^{\Lambda_m}_{\rho-|w|,l;\infty}\left(\partial^w_{\underline{p}}T(0)\right)+|p|^{\rho+1}\sum_{|w|=\rho+1}\int_0^1dt~\mathcal{N}^{\Lambda_m}_{-1,l;\infty}\left(\partial^{w}_{\underline{p}}T\left(tp\right)\right)~,
\end{equation}
which together with \eqref{287~-} gives \eqref{274~-}.
\end{proof}
The four-point sector contains a relevant contribution which is local in
the spatial variables and evaluated at vanishing external momenta. After
the momentum Taylor expansion, the zero-momentum part is still a
distribution in the position variables and must be localized around the
root \(z_1\). This localization is achieved by a first-order telescopic
expansion of the kernel and cutoff factors. The fully localized term is
encoded in the scalar coefficient \(t(z_1)\), while the spatial
remainder is expressed through \(\Delta K\)- and
\(\Delta\chi\)-insertions and is controlled by the assumed
\(\Delta\)-bounds. The following proposition shows that these
ingredients are sufficient to recover the full power-counting estimate
for the four-point distribution.
\begin{proposition}[$4$-point relevant distributions]\label{prop:relevant4}
Fix $\rho\le0$ and let $T\in\mathcal{S}'(\mathbb{M}^4)$. Assume that
\begin{align}
\partial_{\underline{p}} T
&\in
\mathcal{D}^{\Lambda_m}_{\rho-1,l}(\mathbb{M}^4),\label{Asum41}\\
\Delta T
&\in
\mathcal{D}^{\Lambda_m}_{\rho,l}(\mathbb{M}^4)~.\label{Asum42}
\end{align}
Furthermore, assume that 
\begin{equation}\label{Asum43}
\left\| t\right\|_{L^{\infty}(\mathbb{R})}
\le
\mathcal{P}_{l-1}\!\left(\log\frac{\Lambda_m}{m}\right),
\end{equation}
where
\[
t(z_1)
:=
\int_{\mathbb{R}^{3}}
T\bigl((z_1,0),\dots,(z_4,0)\bigr)\,d\vec{z}_{2,4}.
\]
Then
\begin{equation}\label{2190}
T\in\mathcal{D}^{\Lambda_m}_{\rho,l}(\mathbb{M}^4).
\end{equation}
More precisely, we have 
\begin{multline}\label{SemNorm4pt}
\mathcal{N}^{\Lambda_m}_{\rho,l;\infty}\left(T\left(\underline{p}_4\right)\right)\lesssim~\left\|t\right\|_{L^{\infty}(\mathbb{R})}+\mathcal{N}^{\Lambda_m}_{\rho,l;\infty}\left(\Delta T(\vec{0})\right)\\+\sum_{i=1}^4\sum_{\mu=1}^4~p_{i,\mu}~\int_0^1dt\mathcal{N}^{\Lambda_m}_{\rho-1,l;\infty}\left(\left(\partial_{p_{i,\mu}}T\right)\left(\underline{p}_{4}^{(i)}(t)\right)\right),
\end{multline}
where $\underline{p}^{(i)}_{4}(t):=\left(p_1,\cdots,tp_i,\cdots,p_4\right)$.
\end{proposition}

\begin{remark}
The estimate \eqref{SemNorm4pt} remains meaningful whenever the seminorms appearing on the right-hand side are finite. In particular, the assumptions of Proposition~\ref{prop:relevant4} may be relaxed accordingly: it is not necessary to assume \emph{a priori} that $T$ or $\Delta T$ belong exactly to the spaces appearing in \eqref{Asum41}--\eqref{Asum42}, nor that the bound \eqref{Asum43} holds in precisely the stated form.
\end{remark}
\begin{proof}
The proof consists of two steps. We first control the momentum
remainder by expanding the distribution around vanishing external
momenta. It therefore remains to estimate the fully localised
zero-momentum contribution. This is achieved by a telescopic expansion
around the root variable $z_1$, which separates the local term $t$
from kernel and cutoff increments controlled by the $\Delta$-bounds.
Combining the resulting estimates yields the claim. Given $\chi\in\mathcal{K}$ and $\sigma\subseteq\mathfrak{s}\subseteq\{2,3,4\}$, we analyze 
    \begin{multline}
        \left(K^{\otimes \mathfrak{s}}\star \chi^{\otimes \sigma}T\right)\left(z_1;\left(\tau_i,y_i\right)_{i\in\mathfrak{s}};\underline{p}_{4}\right)\\
=
\int_{\mathbb{R}^3}
T\left(\left(z_1,p_1\right),\cdots,\left(z_4,p_4\right)\right)
\prod_{i\in\sigma}\chi(z_i)\prod_{i\in\mathfrak{s}}K\left(\tau_i;z_i,y_i\right)~d\vec{z}_{2,4}~.
    \end{multline}
    First we perform a Taylor expansion of momenta around $0$ of $T$ and we write
    \begin{multline}\label{2322}
    T\left((z_1,p_1),\cdots,(z_4,p_4)\right)=T\left((z_1,0),\cdots,(z_4,0)\right)\\+\sum_{i=1}^4\sum_{\mu=1}^4\int_0^1dt~p_{i,\mu}\left(\partial_{p_{i,\mu}}T\right)\left((z_1,p_1),\cdots,(z_i,tp_i),\cdots,(z_4,p_4)\right)
    \end{multline}
   Using \eqref{Asum41}, we have for all $i$ and $\mu$ in $\left\{1,\cdots,4\right\}$
\[
\partial_{p_{i,\mu}}T
\in
\mathcal D^{\Lambda_m}_{\rho-1,l}(\mathbb M^4)
\] and this implies that there exists polynomials $\mathcal{P}_{l-1}$ and $\tilde{\mathcal{P}}$ such that
    \begin{equation}
        \mathcal N_{\rho-1,l}
\left(
(\partial_{p_{i,\mu}}T)
(\underline p_4^{(i)}(t))
\right)
\leq \mathcal{P}_{l-1}\left(\frac{\log\Lambda_m}{m}\right)\tilde{\mathcal{P}}\left(\frac{\|\underline{p}\|}{\Lambda_m}\right)
    \end{equation}
    since \(0\le t\le1\). This also implies that \begin{equation}\label{deux}
    p_{i,\mu}\partial_{p_{i,\mu}}T\in \mathcal{D}^{\Lambda_m}_{\rho,l}\left(\mathbb{M}^4\right)
    \end{equation}
    by Proposition \ref{PropPoSp}.
    \\
   It therefore remains to control the zero-momentum contribution
\[
T\bigl((z_1,0),\ldots,(z_4,0)\bigr).
\]
We show that this distribution belongs to
\(
\mathcal D^{\Lambda_m}_{\rho,l}(\mathbb M^4)
\). Fix $\chi\in\mathcal{K}$. The strategy consists in deriving suitable estimates for
\begin{equation}
\left(
K^{\otimes\mathfrak{s}}
\star
\chi^{\otimes\sigma}T
\right)
\left(
z_1;(\tau_i,y_i)_{i\in\mathfrak{s}};\vec{0}
\right),
\end{equation}
for fixed subsets
\(
\sigma\subseteq\mathfrak{s}\subseteq\{2,3,4\}.
\)
To estimate the localised zero-momentum contribution, we expand all
kernel and cutoff factors around the distinguished variable \(z_1\).
This separates the fully localised term, which is controlled by the
assumption on \(t\), from increment terms involving either
\(\Delta K\) or \(\Delta\chi\), which are estimated through the
assumption on \(\Delta T\).
We use the general
telescopic identity
\begin{multline}\label{eq:GeneralTelescopicLocalizedChi}
\prod_{i\in\sigma}
\chi(z_i)\prod_{i\in\mathfrak{s}}\phi_i(z_i)
=
\chi^{|\sigma|}(z_1)
\prod_{i\in\mathfrak{s}}\phi_i(z_1)
+
\chi^{|\sigma|}(z_1)
\sum_{i\in\mathfrak{s}}
\left(
\prod_{\substack{j\in\mathfrak{s}\\ j<i}}
\phi_j(z_1)
\right)
\Delta_{z_1}\phi_i(z_i)
\left(
\prod_{\substack{j\in\mathfrak{s}\\ j>i}}
\phi_j(z_j)
\right)
\\
+
\sum_{i\in\sigma}
\left(
\prod_{\substack{j\in\sigma\\ j<i}}
\chi(z_1)
\right)
\Delta_{z_1}\chi(z_i)
\left(
\prod_{\substack{j\in\sigma\\ j>i}}
\chi(z_j)
\right)
\prod_{k\in\mathfrak{s}}\phi_k(z_k).
\end{multline}
The localisation procedure separates the relevant contribution,
represented by the coefficient \(t(z_1)\), from increment terms
involving either \(\Delta K\) or \(\Delta\chi\). The latter are
precisely the objects controlled by the assumption
\eqref{Asum42}. We write 
\begin{multline}\label{eq:SigmaDoubleExpansion}
\int_{\mathbb{R}^3}
T\bigl((z_1,0),\ldots,(z_4,0)\bigr)
\prod_{i\in\mathfrak{s}}
K_{\tau_i,y_i}(z_i)
\prod_{i\in\sigma}\chi(z_i)\,d\vec z_{2,4}
=
t(z_1)\chi^{|\sigma|}(z_1)
\prod_{i\in\mathfrak{s}}
K_{\tau_i,y_i}(z_1)
\\
+\chi^{|\sigma|}(z_1)
\sum_{j\in\mathfrak{s}}
\int_{\mathbb{R}^3}
\Delta_{z_1}K_{\tau_j,y_j}(z_j)
T\bigl((z_1,0),\ldots,(z_4,0)\bigr)
\prod_{\substack{i\in\mathfrak{s}\\ i<j}}
K_{\tau_i,y_i}(z_1)
\prod_{\substack{i\in\mathfrak{s}\\ i>j}}
K_{\tau_i,y_i}(z_i)
\,d\vec z_{2,4}\\
+\sum_{j\in\sigma}
\int_{\mathbb{R}^3}
\Delta_{z_1}\chi(z_j)
T\bigl((z_1,0),\ldots,(z_4,0)\bigr)
\prod_{\substack{i\in\sigma\\ i<j}}
\chi(z_1)
\prod_{\substack{i\in\sigma\\ i>j}}
\chi(z_i)
\prod_{i\in\mathfrak{s}}
K_{\tau_i,y_i}(z_i)
\,d\vec z_{2,4}~.
\end{multline}
 The first one is the fully localised contribution
\begin{equation}\label{I1}
\chi^{|\sigma|}(z_1)t(z_1)
\prod_{i\in\mathfrak{s}}
K_{\tau_i,y_i}(z_1).
\end{equation}
The second type consists of kernel increment terms of the form
\begin{multline}\label{I2}
\chi^{|\sigma|}(z_1)
\sum_{j\in\mathfrak{s}}
\int_{\mathbb{R}^3}
\Delta_{z_1}K_{\tau_j,y_j}(z_j)
T\bigl((z_1,0),\ldots,(z_4,0)\bigr)
\prod_{\substack{i\in\mathfrak{s}\\ i<j}}
K_{\tau_i,y_i}(z_1)
\prod_{\substack{i\in\mathfrak{s}\\ i>j}}
K_{\tau_i,y_i}(z_i)
\,d\vec z_{2,4}.
\end{multline}
Finally, we obtain cut-off increment terms of the form
\begin{multline}\label{I3}
\sum_{j\in\sigma}
\int_{\mathbb{R}^3}
\Delta_{z_1}\chi(z_j)
T\bigl((z_1,0),\ldots,(z_4,0)\bigr)
\prod_{\substack{i\in\sigma\\ i<j}}
\chi(z_1)
\prod_{\substack{i\in\sigma\\ i>j}}
\chi(z_i)
\prod_{i\in\mathfrak{s}}
K_{\tau_i,y_i}(z_i)
\,d\vec z_{2,4}.
\end{multline}
Using \eqref{Asum43} we obtain that the term \eqref{I1} is bounded as follows 
\begin{equation}\label{uno}
    \left|\chi^{|\sigma|}(z_1)~t(z_1)\prod_{i\in\mathfrak{s}}K_{\tau_i,y_i}(z_1)\right|\le \|t\|_{L^{\infty}(\mathbb{R})}~\mathcal{A}_{l,4}^{\Lambda,\Lambda_0}\left(z_1;(\tau_i,y_i)_{i\in\mathfrak{s}};\chi^{\sigma}\right)~.
\end{equation}
Since 
\begin{equation}\label{bornesanssigma}
\Delta T\in\mathcal{D}^{\Lambda_m}_{\rho,l}\left(\mathbb{M}^4\right)
\end{equation}
we deduce that for all $j\in\mathfrak{s}$ we have 
\begin{multline}\label{bornesansigma}
    \left|\left(K^{\otimes\mathfrak{s}\setminus\{j\}}\otimes\Delta_{z_1}K_{\tau_j,y_j}\right)\star T\left(z_1;(\tau_i,y_i)_{i\in\mathfrak{s}};\vec{0}\right)\right|\leq \tau_j^{-\frac{1}{2}}\Lambda_m^{\rho-1}~\mathcal{A}_{l,4}^{\Lambda,\Lambda_0}\left(z_1;(\tau_i,y_i)_{i\in\mathfrak{s}}\right)\\\times \mathcal{N}^{\Lambda_m}_{\rho,l;\infty}\left(\Delta T\left(\vec{0}\right)\right)~Q_1\left(\tau_{\mathfrak{s}\setminus\{j\}}^{\Lambda_m}\right).
\end{multline}
Using \eqref{bornesanssigma} together with \eqref{Magnetic} from Lemma \ref{MagnLemm} we obtain that 
\begin{multline}\label{dos}
\left|\chi^{|\sigma|}(z_1)
\sum_{j\in\mathfrak{s}}
\prod_{\substack{i\in\mathfrak{s}\\ i<j}}
K_{\tau_i,y_i}(z_1)~\int_{\mathbb{R}^3}
\Delta_{z_1}K_{\tau_j,y_j}(z_j)
T\bigl((z_1,0),\ldots,(z_4,0)\bigr)~
\prod_{\substack{i\in\mathfrak{s}\\ i>j}}
K_{\tau_i,y_i}(z_i)
\,d\vec z_{2,4}\right|\\\leq ~\mathcal{A}_{l,4}^{\Lambda,\Lambda_0}\left(z_1;(\tau_i,y_i)_{i\in\mathfrak{s}};\chi^{\sigma}\right)\mathcal{N}^{\Lambda_m}_{\rho,l;\infty}\left(\Delta T(\vec{0})\right)~Q_1\left(\tau_{\mathfrak{s}}^{\Lambda_m}\right).
\end{multline}
Using again the assumption \eqref{Asum42}, this time through the
seminorm involving \(\Delta\chi\)-insertions, we obtain for all $\sigma\subseteq\mathfrak{s}\subseteq \left\{2,3,4\right\}$ and  $j\in\sigma$ that 
\begin{multline}\label{bornesigma}
    \left|\left(K^{\otimes\mathfrak{s}}\star \chi^{\otimes\sigma\setminus\{j\}}\otimes\Delta_{z_1}\chi T\right)\left(z_1;(\tau_i,y_i)_{i\in\mathfrak{s}};\vec{0}\right)\right|\leq \Lambda_m^{\rho-1}~\hat{\mathcal{A}}_{l,4}^{\Lambda,\Lambda_0}\left(z_1;(\tau_i,y_i)_{i\in\mathfrak{s}};\chi^{\sigma}\right)\\\times \mathcal{N}^{\Lambda_m}_{\rho,l;\infty}\left(\Delta T\left(\vec{0}\right)\right)~Q^{\sigma^{(j)}}_{q-1}\left(\tau_{\mathfrak{s}}^{\Lambda_m}\right)
\end{multline}
Using \eqref{bornesigma} and again \eqref{Magnetic} with $\mathcal{A}_{l,n}^{\Lambda,\Lambda_0}$ replaced by $\hat{\mathcal{A}}_{l,n}^{\Lambda,\Lambda_0}$ we deduce that 
\begin{multline}\label{tres}
\sum_{j\in\sigma}
\int_{\mathbb{R}^3}
\Delta_{z_1}\chi(z_j)
T\bigl((z_1,0),\ldots,(z_4,0)\bigr)
\prod_{\substack{i\in\sigma\\ i<j}}
\chi(z_1)
\prod_{\substack{i\in\sigma\\ i>j}}
\chi(z_i)
\prod_{i\in\mathfrak{s}}
K_{\tau_i,y_i}(z_i)
\,d\vec z_{2,4}\\\leq \Lambda_m^{\rho-1}~\hat{\mathcal{A}}_{l,4}^{\Lambda,\Lambda_0}\left(z_1;(\tau_i,y_i)_{i\in\mathfrak{s}};\chi^{\sigma}\right)\mathcal{N}^{\Lambda_m}_{\rho,l;\infty}\left(\Delta T\left(\vec{0}\right)\right)~Q^{\sigma^{(j)}}_{q-1}\left(\tau_{\mathfrak{s}}^{\Lambda_m}\right)~.
\end{multline}
Combining \eqref{uno}, \eqref{dos} and \eqref{tres} we deduce that for all $\sigma\subseteq\mathfrak{s}\subseteq\left\{2,3,4\right\}$ the following bound holds 

\begin{multline}
    \left|\left(K^{\otimes\mathfrak{s}}\star \chi^{\otimes\sigma}T\right)\left(z_1;(\tau_i,y_i)_{i\in\mathfrak{s}};\vec{0}\right)\right|~\le \left(\|t\|_{L^{\infty}(\mathbb{R})}+\mathcal{N}^{\Lambda_m}_{\rho,l;\infty}\left(\Delta T(\vec{0})\right)\right)\\\times \left(\mathcal{A}_{l,4}^{\Lambda,\Lambda_0}+\Lambda_m^{-1}~\hat{\mathcal{A}}_{l,4}^{\Lambda,\Lambda_0}\right)\left(z_1;(\tau_i,y_i)_{i\in\mathfrak{s}};\chi^{\sigma}\right)~Q_1\left(\tau_{\mathfrak{s}}^{\Lambda_m}\right)~.
\end{multline}
Taking the supremum over $z_1$, $\sigma\subseteq\mathfrak{s}\subseteq\{2,3,4\}$ and $\chi\in\mathcal{K}$ we deduce that 
\begin{equation}\label{rigolo}
    \mathcal{N}^{\Lambda_m}_{\rho,l;\infty}\left(T(\vec{0})\right)\lesssim \|t\|_{L^{\infty}(\mathbb{R})}+\mathcal{N}^{\Lambda_m}_{\rho,l;\infty}\left(\Delta T(\vec{0})\right)\le~\mathcal{P}_{l-1}\left(\log\frac{\Lambda_m}{m}\right).
\end{equation}
Combining this with \eqref{2322} we deduce that 
\begin{multline}
\mathcal{N}^{\Lambda_m}_{\rho,l;\infty}\left(T\left(\underline{p}_4\right)\right)\lesssim~\left\|t\right\|_{L^{\infty}(\mathbb{R})}+\mathcal{N}^{\Lambda_m}_{\rho,l;\infty}\left(\Delta T(\vec{0})\right)\\+\sum_{i=1}^4\sum_{\mu=1}^4~p_{i,\mu}~\int_0^1dt\mathcal{N}^{\Lambda_m}_{\rho-1,l;\infty}\left(\left(\partial_{p_{i,\mu}}T\right)\left(\underline{p}_{4}^{(i)}(t)\right)\right),
\end{multline}
This combined with \eqref{deux} and \eqref{rigolo} imply that 
\begin{equation}
    \mathcal{N}^{\Lambda_m}_{\rho,l;\infty}\left(T\left(\underline{p}_4\right)\right)\le \mathcal{P}_{l-1}\left(\log\frac{\Lambda_m}{m}\right)\widetilde{P}\left(\frac{\|\underline{p}_4\|}{\Lambda_m}\right)
\end{equation}
which concludes that $T\in\mathcal{D}^{\Lambda_m}_{\rho,l}\left(\mathbb{M}^4\right)$.

\end{proof}
\subsection{Proof of renormalizability}
We now prove Theorem~\ref{ThmPrinc}. The proof is the synthesis of the
stability estimates established above. We first show that the
\(\Lambda\)-derivative of each quantity appearing in
\eqref{StaAssum1}--\eqref{StaAssum2} satisfies the corresponding
power-counting estimate with the degree lowered by one. We then
integrate in the scale parameter. The irrelevant sector is integrated
from the ultraviolet scale \(\Lambda_0\), while the relevant two- and
four-point correlators are reconstructed from the renormalization
conditions at \(\Lambda=0\).

\begin{proof}
The proof is by induction on the total degree \(n+2l\).
For fixed \(n+2l\), the induction proceeds by increasing loop order
\(l\). The linear term in the flow equation involves the pair
\((l-1,n+2)\), while the quadratic term involves strictly smaller
degrees \((l_1,n_1+1)\) and \((l_2,n_2+1)\), so that all terms on the
right-hand side are covered by the induction hypothesis. Recall that the flow equation
can be written in the form
\begin{equation}\label{FE-main}
\partial_\Lambda
\partial^w
\mathcal L_{l,n}^{\Lambda,\Lambda_0}
=
\frac12
\mathbf L^\Lambda
\bigl(
\partial^w
\mathcal L_{l-1,n+2}^{\Lambda,\Lambda_0}
\bigr)
-
\frac12
\sum_{\substack{\pi,l,w}}'
e^{-\frac{m^2}{2\Lambda^2}}
\mathbf B^\Lambda_\pi
\Bigl(
\partial^{w_1}
\mathcal L_{l_1,n_1+1}^{\Lambda,\Lambda_0},
\partial^{w_2}
\mathcal L_{l_2,n_2+1}^{\Lambda,\Lambda_0}
\Bigr)
\partial^{w_3}\dot C_0^\Lambda(p_\pi),
\end{equation}
where the prime denotes the summation over partitions
\(\pi=\{\pi_1,\pi_2\}\) of $\{1,\cdots,n\}$, loop splittings \(l_1+l_2=l\), and
decompositions \(w=w_1+w_2+w_3\).

For the weighted quantities one obtains the corresponding identity
\begin{align}\label{FEr}
\partial_\Lambda
\Bigl(
&(z_1-z_i)^r
\partial^w
\mathcal L_{l,n}^{\Lambda,\Lambda_0}
\Bigr)
=
\frac12
\mathbf L^\Lambda
\Bigl(
(z_1-z_i)^r
\partial^w
\mathcal L_{l-1,n+2}^{\Lambda,\Lambda_0}
\Bigr)
\nonumber
\\
&\quad
-\frac12
\sum_{\substack{\pi=\{\pi_1,\pi_2\}\\ \{1,i\}\subset\pi_1}}'
e^{-\frac{m^2}{2\Lambda^2}}
\mathbf B^\Lambda_\pi
\Bigl(
(z_1-z_i)^r
\partial^{w_1}
\mathcal L_{l_1,n_1+1}^{\Lambda,\Lambda_0},
\partial^{w_2}
\mathcal L_{l_2,n_2+1}^{\Lambda,\Lambda_0}
\Bigr)
\partial^{w_3}\dot C_0^\Lambda(p_\pi)
\nonumber
\\
&\quad
-\frac12
\sum_{\substack{\pi=\{\pi_1,\pi_2\}\\ \{1,i\}\subset\pi_2}}'
e^{-\frac{m^2}{2\Lambda^2}}
\mathbf B^\Lambda_\pi
\Bigl(
\partial^{w_1}
\mathcal L_{l_1,n_1+1}^{\Lambda,\Lambda_0},
(z_1-z_i)^r
\partial^{w_2}
\mathcal L_{l_2,n_2+1}^{\Lambda,\Lambda_0}
\Bigr)
\partial^{w_3}\dot C_0^\Lambda(p_\pi)
\nonumber
\\
&\quad
-\frac12
\sum_{\substack{\pi=\{\pi_1,\pi_2\}\\ 1\in\pi_k,\ i\in\pi_{k'},\ k\neq k'}}'
e^{-\frac{m^2}{2\Lambda^2}}
(z_1-z_i)^r
\mathbf B^\Lambda_\pi
\Bigl(
\partial^{w_1}
\mathcal L_{l_1,n_1+1}^{\Lambda,\Lambda_0},
\partial^{w_2}
\mathcal L_{l_2,n_2+1}^{\Lambda,\Lambda_0}
\Bigr)
\partial^{w_3}\dot C_0^\Lambda(p_\pi).
\end{align}

By the induction hypothesis, the arguments of the linear and bilinear
operators in \eqref{FEr} satisfy the corresponding bounds at lower
total degree. Proposition~\ref{PropL} therefore gives
\[
\mathbf L^\Lambda
\Bigl(
(z_1-z_i)^r
\partial^w
\mathcal L_{l-1,n+2}^{\Lambda,\Lambda_0}
\Bigr)
\in
\mathcal D^{\Lambda_m}_{3-n-r-|w|,l}(\mathbb M^n).
\]
Similarly, Propositions~\ref{propQu} and \ref{Bilinearr} control the
three bilinear contributions in \eqref{FEr}. After multiplication by
\[
e^{-m^2/(2\Lambda^2)}
\partial^{w_3}\dot C_0^\Lambda(p_\pi),
\]
the covariance estimate \eqref{CovEB} and Proposition~\ref{PropPoSp} show that each
bilinear term belongs to $
\mathcal D^{\Lambda_m}_{3-n-r-|w|,l}(\mathbb M^n)$. Hence
\begin{equation}\label{Boobo0}
\partial_\Lambda
\Bigl(
(z_1-z_i)^r
\partial^w
\mathcal L_{l,n}^{\Lambda,\Lambda_0}
\Bigr)
\in
\mathcal D^{\Lambda_m}_{3-n-r-|w|,l}(\mathbb M^n).
\end{equation}

The same argument, using the second-order stability estimates in
Propositions~\ref{PropL} and \ref{propQu}, yields for \(r+|w|=3\)
\begin{equation}\label{Boobo1}
\partial_\Lambda
\Bigl(
(z_1-z_2)^r
\partial^w
\mathcal L_{l,n}^{\Lambda,\Lambda_0}
\Bigr)
\in
\mathcal D^{\Lambda_m;(2)}_{-n,l}(\mathbb M^n).
\end{equation}

It remains to treat the discrete derivative estimates. Applying
\(\Delta_{z_1}^{(i)}\) to the flow equation, the linear term is
controlled by Proposition~\ref{PropL}, while the bilinear term is
controlled by Proposition~\ref{prop:B-estimate}. Thus
\begin{equation}\label{Boobo2}
\mathcal N^{\Lambda_m}_{3-n,l}
\Bigl(
\partial_\Lambda
\bigl(
\Delta_{z_1}^{(i)}
\mathcal L_{l,n}^{\Lambda,\Lambda_0}
\bigr)
\Bigr)
\lesssim
\mathcal P_{l-1}
\left(
\log\frac{\Lambda_m}{m}
\right)
\widetilde{\mathcal P}
\left(
\frac{\|\underline p_n\|}{\Lambda_m}
\right).
\end{equation}
The same reasoning, using Proposition~\ref{prop:B-estimateInsertion}
for the \(\Delta\chi\)-insertions, gives
\begin{equation}\label{Boobo3}
\mathcal N^{\Lambda_m}_{3-n,l}
\Bigl(
\partial_\Lambda
\bigl(
\Delta_{z_1}^{(i)}\chi\,
\mathcal L_{l,n}^{\Lambda,\Lambda_0}
\bigr)
\Bigr)
\lesssim
\mathcal P_{l-1}
\left(
\log\frac{\Lambda_m}{m}
\right)
\widetilde{\mathcal P}
\left(
\frac{\|\underline p_n\|}{\Lambda_m}
\right).
\end{equation}
Combining~\eqref{Ninf2}, \eqref{Boobo2} and \eqref{Boobo3} we obtain the
desired estimate for
\(\partial_\Lambda(\Delta\mathcal L_{l,n}^{\Lambda,\Lambda_0})\).\\
We now integrate in the scale. If $
n+r+|w|\ge5$
then the ultraviolet boundary condition gives
\[
(z_1-z_i)^r
\partial^w
\mathcal L_{l,n}^{\Lambda_0,\Lambda_0}
=0.
\]
Combining this with \eqref{Boobo0} and Proposition~\ref{integration}
yields
\begin{equation}\label{I385}
(z_1-z_i)^r
\partial^w
\mathcal L_{l,n}^{\Lambda,\Lambda_0}
\in
\mathcal D^{\Lambda_m}_{4-n-r-|w|,l}(\mathbb M^n).
\end{equation}
Likewise, \eqref{Boobo1} implies
\[
(z_1-z_2)^r
\partial^w
\mathcal L_{l,n}^{\Lambda,\Lambda_0}
\in
\mathcal D^{\Lambda_m;(2)}_{1-n,l}(\mathbb M^n),
\qquad r+|w|=3.
\]
The bounds \eqref{Boobo2} and \eqref{Boobo3}, together with the
ultraviolet boundary conditions for the discrete derivative terms,
yield
\[
\Delta
\mathcal L_{l,n}^{\Lambda,\Lambda_0}
\in
\mathcal D^{\Lambda_m}_{4-n,l}(\mathbb M^n),
\qquad n\ge4.
\]
The previous argument applies only to the irrelevant sector, where the
ultraviolet boundary condition at \(\Lambda=\Lambda_0\) is available.
The remaining relevant contributions are reconstructed from the
renormalization conditions imposed at \(\Lambda=0\) using
propositions~\ref{prop:relevant} and \ref{prop:relevant4}. It is enough to prove the bounds
\begin{equation}\label{389}
\left|
\partial^w
l^{\Lambda,\Lambda_0}_{l,2;r}
(z_1;\vec 0)
\right|
\lesssim
\Lambda_m^{2-r-|w|}
\mathcal P_{l-1}
\left(
\log\frac{\Lambda_m}{m}
\right),
\qquad
\left|
l^{\Lambda,\Lambda_0}_{l,4}
(z_1;\vec 0)
\right|
\lesssim
\mathcal P_{l-1}
\left(
\log\frac{\Lambda_m}{m}
\right).
\end{equation}
Here
\[
\partial^w
l^{\Lambda,\Lambda_0}_{l,2;r}
(z_1;\vec 0)
:=
\int_{\mathbb R}
(z_1-z_2)^r
\partial^w
\mathcal L_{l,2}^{\Lambda,\Lambda_0}
\bigl((z_1,0),(z_2,0)\bigr)
\,dz_2,
\]
and
\[
l^{\Lambda,\Lambda_0}_{l,4}
(z_1;\vec 0)
:=
\int_{\mathbb R^3}
\mathcal L_{l,4}^{\Lambda,\Lambda_0}
\bigl((z_1,0),\cdots,(z_4,0)\bigr)
\,d\vec z_{2,4}.
\]
To apply Propositions~\ref{prop:relevant} and
\ref{prop:relevant4}, it remains to control the corresponding Taylor
coefficients. These coefficients are precisely the quantities obtained
from the seminorms in the case \(\sigma=\mathfrak s=\varnothing\).\\
The derivative estimates already obtained imply
\[
\left|
\partial_\lambda
l_{l,4}^{\lambda,\Lambda_0}
(z_1;\vec 0)
\right|
\lesssim
\mathcal P_{l-1}
\left(
\log\frac{\lambda_m}{m}
\right),
\]
and
\[
\left|
\partial_\lambda
\partial^w
l_{l,2;r}^{\lambda,\Lambda_0}
(z_1;\vec 0)
\right|
\lesssim
\lambda_m^{1-r-|w|}
\mathcal P_{l-1}
\left(
\log\frac{\lambda_m}{m}
\right).
\]
The renormalization conditions give
\[
\partial^w
l_{l,2;r}^{0,\Lambda_0}
(z_1;\vec 0)
=0,
\qquad r+|w|\le2,
\qquad
l_{l,4}^{0,\Lambda_0}(z_1;\vec 0)=0.
\]
Integrating from \(0\) to \(\Lambda\), and using the elementary
polynomial estimate \eqref{IrrePoly}, gives \eqref{389}. Therefore
Propositions~\ref{prop:relevant} and \ref{prop:relevant4} yield the
remaining relevant bounds.

Combining the bounds obtained for the irrelevant sector with the
reconstruction of the relevant two- and four-point contributions yields
\eqref{StaAssum1}, \eqref{StatAssum3}, and \eqref{StaAssum2}. This
closes the induction and completes the proof of
Theorem~\ref{ThmPrinc}.
\end{proof}
\begin{remark}
The renormalization conditions are in fact compatible with the
symmetries of the theory. In particular, reflection symmetry implies
\[
\int_{\mathbb{R}}\left(z_1-z_2\right)^r 
\partial^w\mathcal{L}_{l,2}^{\Lambda,\Lambda_0}
\left((z_1,p),(z_2,-p)\right.\left|\right._{p\equiv 0}=0,~~
\]
for all pairs $(r,|w|)\in\{(1,0),(0,1),(1,1)\}$. Indeed, the two-point function
$\mathcal{L}_{l,2}^{\Lambda,\Lambda_0}(p,-p)$ is an even function of the external
momentum $p$ by invariance under $p\mapsto -p$. Consequently, all odd momentum
derivatives vanish at $p=0$. Moreover, the operator
$\mathbb{Z}_{1,2}^{(r)}$ with odd $r$ produces integrands that are odd under the
spatial reflection $z_2-z_1\mapsto -(z_2-z_1)$. The integration over $z_2$ therefore vanishes identically.
\end{remark}
\subsection{Proof of Theorem \ref{BesovNorm}}
\noindent The proof consists in translating the power-counting estimates of
Theorem~\ref{ThmPrinc} into bounds on the heat-kernel regularisations
appearing in the definition of the Besov norms. The seminorms defining
the spaces
\(
\mathcal D^{\Lambda_m}_{\rho,l}
\)
are precisely formulated in terms of such regularisations. Combining
these estimates with the uniform bounds on the tree amplitudes then
yields the claimed Besov regularity.\\
By Definition~\ref{DefPowerCountingSeminorms}, the inclusion
\(
\mathcal L_{l,n}^{\Lambda,\Lambda_0}
\in
\mathcal D^{\Lambda_m}_{4-n,l}
\)
immediately yields the heat-kernel estimates below: for $l\ge 1$ and $n \ge 4$
\begin{multline}\label{323-}
\Big|
\big(
K^{\otimes(n-1)}
\star
\mathcal L_{l,n}^{\Lambda,\Lambda_0}
\big)
\big(
z_1;
(\tau,y_i)_{i\in\mathcal I_n};
\underline{p}_n 
\big)
\Big|
\le
\Lambda_m^{4-n}\mathcal{N}^{\Lambda_m}_{4-n,l;\infty}\left(\mathcal{L}_{l,n}^{\Lambda,\Lambda_0}\left(\underline{p}_n\right)\right)\\\times
\sum_{i=0}^{n-1}\left(\frac{\tau^{-1/2}}{\Lambda_m}\right)^i~
\mathcal A_{l,n}^{\Lambda,\Lambda_0}
\big(
z_1;
(\tau,y_i)_{i\in\mathcal I_n}
\big)~.
\end{multline}
For $l=0$ and $n\ge4$, we obtain that
\begin{multline}
\Big|
\big(
K^{\otimes(n-1)}
\star
\mathcal L_{0,n}^{\Lambda,\Lambda_0}
\big)
\big(
z_1;
(\tau,y_i)_{i\in\mathcal I_n};
\underline{p}_{n}
\big)
\Big|\\
\le~\mathcal{N}^{\Lambda_m}_{4-n,0;\infty}\left(\mathcal{L}_{0,n}^{\Lambda,\Lambda_0}\left(\underline{p}_n\right)\right)
\Lambda_m^{4-n}
\mathcal A_{0,n}^{\Lambda,\Lambda_0}
\big(
z_1;
(\tau,y_i)_{i\in\mathcal I_n}
\big).
\end{multline}
For $l\ge2$ one has
\begin{multline}\label{ind4}
\Big|
\big(
K_{\tau,y_2}
\star
\mathcal L_{l,2}^{\Lambda,\Lambda_0}
\big)
(z_1;p,-p)
\Big|
\le
\Lambda_m^{2}~\mathcal{N}^{\Lambda_m}_{2,l;\infty}\left(\mathcal{L}_{l,2}^{\Lambda,\Lambda_0}\left(\underline{p}\right)\right)~
\sum_{j=0}^{3}
\left(
\frac{\tau^{-1/2}}{\Lambda_m}
\right)^j
\mathcal A_{l,2}^{\Lambda,\Lambda_0}
\big(
z_1;(\tau,y_2)
\big).
\end{multline}
For $l=1$, we have using the flow equations and the renormalization condition \eqref{renoc1} that 
\begin{equation}
    \mathcal{L}_{1,2}^{\Lambda,\Lambda_0}\left((z_1,p),(z_2,-p)\right)=\lambda~C^{\Lambda}(m)~\delta\left(z_1-z_2\right)~,
\end{equation}
where $$C^{\Lambda}(m):=\int_0^{\Lambda}t~e^{-\frac{m^2}{t^2}}dt~\int_{\mathbb{R}^3}e^{-k^2}d^3k~$$ and $\lambda$ is the coupling constant. Hence we obtain that
\begin{equation}\label{326-}
\Big|
\big(
K_{\tau,y_2}
\star
\mathcal L_{1,2}^{\Lambda,\Lambda_0}
\big)
(z_1;p,-p)
\Big|
\lesssim
\Lambda_m^{2}~
K\left(\tau;z_1,y_2\right).
\end{equation}
We now estimate
\begin{equation}
\left(K^{\otimes n}\star
\mathcal L_{l,n}^{\Lambda,\Lambda_0}\right)
(\tau,\vec y_{1,n},\underline{p}_{n})
=
\int_{\mathbb R}
K(\tau;z_1,y_1)
\left(
K^{\otimes(n-1)}\star
\mathcal L_{l,n}^{\Lambda,\Lambda_0}
\right)
\left(
z_1;(\tau,y_i)_{i\in\mathcal I_n};
\underline{p}_n
\right)~dz_1 
\end{equation}
for $l\ge0$ and $n\ge2$. Combining \eqref{323-}--\eqref{326-} with the defining bounds of the
spaces
\(
\mathcal D^{\Lambda_m}_{\rho,l}
\),
we obtain
\begin{align}
\left|
\left(K_\tau^{\otimes n}\star
\mathcal L_{l,n}^{\Lambda,\Lambda_0}\right)
(\vec y_{1,n};\underline{p}_{n})
\right|
&\le
\Lambda_m^{4-n}
\mathcal P_{l-1}\!\left(\log\frac{\Lambda_m}{m}\right)
\mathcal P\!\left(\frac{\|\underline{p}_n\|}{\Lambda_m}\right)
\sum_{i=0}^{n-1}
\left(\frac{\tau^{-1/2}}{\Lambda_m}\right)^i
\\
&\quad\times
\mathcal A_{l,n}^{\Lambda,\Lambda_0}
\big(
(\tau,y_1),(\tau,y_i)_{i\in\mathcal I_n}
\big),
\qquad l\ge1,\; n\ge4 ,\nonumber\\
\Big|
\big(
K^{\otimes n}
\star
\mathcal L_{0,n}^{\Lambda,\Lambda_0}
\big)
\big(
\vec{y}_{1,n};
\underline{p}_{n}
\big)
\Big|
&\le
\Lambda_m^{4-n}
\mathcal A_{0,n}^{\Lambda,\Lambda_0}
\big((\tau,y_1),
(\tau,y_i)_{i\in\mathcal I_n}
\big),\\
\left|
\left(K_\tau^{\otimes2}\star
\mathcal L_{1,2}^{\Lambda,\Lambda_0}\right)
(\vec y_{1,2};\underline{p})
\right|
&\le
\Lambda_m^{2}~K\left(\tau;y_1,y_2\right)
~,
\\
\left|
\left(K_\tau^{\otimes2}\star
\mathcal L_{l,2}^{\Lambda,\Lambda_0}\right)
(\vec y_{1,2};\underline{p})
\right|
&\le
\Lambda_m^{2}~
\mathcal P_{l-1}\!\left(\log\frac{\Lambda_m}{m}\right)~\mathcal P\!\left(\frac{|p|}{\Lambda_m}\right)
\sum_{i=0}^3
\left(\frac{\tau^{-1/2}}{\Lambda_m}\right)^i~\\~~~&\times\mathcal{A}_{l,2}^{\Lambda,\Lambda_0}\left((\tau,y_1),(\tau,y_2)\right),
~~~l\ge2 ~.
\end{align}
Since
\[
\sup_{(y_1,\ldots,y_n)}
\mathcal A_{l,n}^{\Lambda,\Lambda_0}
((\tau,y_i))
\lesssim
\tau^{-(n-1)/2}~,
\]
the most singular contribution in the preceding estimate is of order $
\tau^{-(n-1)}$. At tree level, the factor
\(
\sum_{i=0}^{n-1}
(\tau^{-1/2}/\Lambda_m)^i
\)
is absent, so that the leading singularity is only
 $
\tau^{-(n-1)/2}
$,
yielding the improved regularity exponent $
-n+1$.
By the definition \eqref{BesovCorr}, this corresponds to the Besov
regularity exponent $
\alpha=-2n+2$. Hence we obtain 
\begin{align}
\|\mathcal L_{l,n}^{\Lambda,\Lambda_0}\left(\underline{p}_n \right)\|_{\mathcal C^{-2n+2}(\mathbb R^n)}
&\le
\Big(\sum_{i=0}^{n-1}\Lambda_m^{4-n-i}\Big)
\mathcal P_l\!\left(\log\frac{\Lambda_m}{m}\right)
\mathcal P\!\left(\frac{\|\vec p_n\|}{\Lambda_m}\right),
\\
\|\mathcal L_{0,n}^{\Lambda,\Lambda_0}\left(\underline{p}_{n}\right)\|_{\mathcal C^{-n+1}(\mathbb R^n)}
&\le
\sum_{i=0}^{n-1}\Lambda_m^{4-n-i},
\qquad n\ge4 ,
\\
\|\mathcal L_{l,2}^{\Lambda,\Lambda_0}\left(\underline{p}\right)\|_{\mathcal C^{-4}(\mathbb R^2)}
&\le
\Big(\sum_{i=0}^{3}\Lambda_m^{2-i}\Big)
\mathcal P_{l-1}\!\left(\log\frac{\Lambda_m}{m}\right)
\mathcal P\!\left(\frac{|p|}{\Lambda_m}\right),
\qquad l\ge2 ,
\\
\|\mathcal L_{1,2}^{\Lambda,\Lambda_0}\left(\underline{p}\right)\|_{\mathcal C^{-1}(\mathbb R^2)}
&\le
\Big(\sum_{i=0}^{3}\Lambda_m^{2-i}\Big)~.
\end{align}
Combining the preceding estimates with the definition of the
Besov--Hölder seminorms \eqref{BesovCorr} yields the claimed bounds.
This completes the proof of Theorem~\ref{BesovNorm}. Note that the estimates are uniform in the ultraviolet cutoff
\(\Lambda_0\), since all constants appearing in
Theorem~\ref{ThmPrinc} are independent of \(\Lambda_0\).
\subsection{Besov control of the truncated correlators}\label{BesovTrSec}
We now turn to the construction of the truncated correlators and the
proof of Theorem~\ref{BesovNorm2}. Let
\(
(\chi_\varepsilon)_{\varepsilon>0}
\subset C_b^\infty(\mathbb R)
\)
be a family of smooth monotone cutoffs such that $
\chi_\varepsilon \longrightarrow \mathds 1^+$ pointwise, and
\[
\sup_{\varepsilon>0}
\|\chi_\varepsilon\|_{L^\infty}
<\infty.
\]
Our goal is to show that the distributions
\[
\chi_\varepsilon^{\otimes\sigma}
\mathcal L_{l,n}^{\Lambda,\Lambda_0}
\]
admit a limit as \(\varepsilon\to0\). The proof proceeds in three
steps. First, we establish uniform bounds in suitable Besov spaces.
These bounds imply relative compactness in slightly rougher Besov
spaces by a standard compact embedding argument. We then identify the
limit through the heat-kernel regularisations appearing in the
definition of the Besov norms. The combination of compactness and
uniqueness yields the existence of the truncated correlators.

\begin{proposition}\label{prop:BesovUniform}
Let
$0\le \Lambda\le \Lambda_0$,
$n\ge4$,
$\sigma\subseteq\mathcal I_n$ and let $\mathcal{L}_{l,n}^{\Lambda,\Lambda_0}$ verify the flow boundary conditions \eqref{Bc1} and \eqref{renoc1}-\eqref{renoc2}. 
Then the following bounds hold uniformly in $\epsilon$ and $\Lambda_0$:
\begin{align}
\|\chi_{\epsilon}^{\otimes\sigma}\mathcal{L}_{l,n}^{\Lambda,\Lambda_0}\left(\underline{p}_n\right)\|_{\mathcal{C}^{-|\sigma|n+|\sigma|-n}(\mathbb{R}^n)}
&\leq
\Lambda_m^{4-n}\,
\mathcal{P}_{l-1}\!\left(\log\frac{\Lambda_m}{m}\right)
\mathcal{P}\!\left(\frac{\|\underline{p}_{n}\|}{\Lambda_m}\right),
\qquad l\ge1,
\label{unif1}
\\
\|\chi_{\epsilon}^{\otimes\sigma}\mathcal{L}_{0,n}^{\Lambda,\Lambda_0}\left(\underline{p}_n\right)\|_{\mathcal{C}^{-n+1}(\mathbb{R}^n)}
&\leq
\Lambda_m^{4-n}.
\label{unif2}
\end{align}
The coefficients of the polynomials $\mathcal P_{l-1}$ and $\mathcal P$ depend only on $l$ and $n$, and are independent of $\epsilon$, $\Lambda_0$, and the momenta.
\end{proposition}

\begin{proof}
Applying Theorem~\ref{ThmPrinc} with $\chi=\chi_\epsilon$ we obtain for $n\ge4$ and $l\ge1$
\begin{multline}
\Big|
\left(K^{\otimes(n-1)} \star \chi_{\epsilon}^{\otimes\sigma}
\mathcal L_{l,n}^{\Lambda,\Lambda_0}\right)
\big(
z_1;(\tau,y_i)_{i\in\mathcal{I}_n};
\underline{p}_{n}
\big)
\Big|
\\
\le
\Lambda_m^{4-n}
\left(
\mathcal A_{l,n}^{\Lambda,\Lambda_0}
+
\Lambda_m^{-1}\hat{\mathcal A}_{l,n}^{\Lambda,\Lambda_0}
\right)
\left(
z_1;(\tau,y_i)_{i\in\mathcal{I}_n};\chi_\epsilon^\sigma
\right)
\\
\times
\mathcal P_{l-1}\!\left(\log\frac{\Lambda_m}{m}\right)
\mathcal P\!\left(\frac{\|\underline{p}_{n}\|}{\Lambda_m}\right)
\left(\sum_{i=0}^{|\sigma|}\left(
\frac{\tau^{-1/2}}{\Lambda_m}
\right)^i\right)^{n-1} .
\end{multline}
To estimate the Besov seminorm, we convolve once more with the heat
kernel in the root variable. Using Corollary~\ref{CorTransport} and
\eqref{equationNuevo}, we obtain
\begin{equation}
\int_{\mathbb R}
K(\tau;z_1,y_1)
\hat{\mathcal A}_{l,n}^{\Lambda,\Lambda_0}
\big(
z_1;(\tau,y_i)_{i\in\mathcal I_n};\chi_\epsilon^\sigma
\big)
\,dz_1
\lesssim
\left(\frac{\tau^{-1/2}}{\Lambda_m}\right)
\mathcal A_{l,n}^{\Lambda,\Lambda_0}
\big((\tau,y_1),(\tau,y_i)_{i\in\mathcal I_n}\big),
\end{equation}
where we also used the uniform bound
$\sup_{\epsilon>0}\|\chi_\epsilon\|_\infty<\infty$. Consequently, for $l\ge1$ and $n\ge4$,
\begin{multline}\label{325--}
\Big|
(K^{\otimes n}\star \chi_\epsilon^{\otimes\sigma}
\mathcal L_{l,n}^{\Lambda,\Lambda_0})
\big(
(\tau,y_1),(\tau,y_i)_{i\in\mathcal I_n};
\underline{p}_{n}
\big)
\Big|
\\
\lesssim
\Lambda_m^{4-n}
\mathcal A_{l,n}^{\Lambda,\Lambda_0}
\big(
(\tau,y_1),(\tau,y_i)_{i\in\mathcal I_n}
\big)
\mathcal P_{l-1}\!\left(\log\frac{\Lambda_m}{m}\right)
\mathcal P\!\left(\frac{\|\underline{p}_{n}\|}{\Lambda_m}\right)
\sum_{i=0}^{|\sigma|\left(n-1\right)+1}
\left(\frac{\tau^{-1/2}}{\Lambda_m}\right)^i .
\end{multline}
We use the uniform amplitude bound from Corollary~\ref{UniLinfty},
\[
\sup_{(y_1,\dots,y_n)}
\mathcal A_{l,n}^{\Lambda,\Lambda_0}
\big((\tau,y_1),(\tau,y_i)_{i\in\mathcal I_n}\big)
\lesssim
\tau^{-(n-1)/2}.
\]
Since
\[
\sum_{i=0}^{|\sigma|(n-1)+1}
\left(
\frac{\tau^{-1/2}}{\Lambda_m}
\right)^i
\lesssim
\tau^{-\frac{|\sigma|(n-1)+1}{2}}
\]
for \(0<\tau\le1\), the preceding estimate yields
\[
(K^{\otimes n}\star
\chi_\varepsilon^{\otimes\sigma}
\mathcal L_{l,n}^{\Lambda,\Lambda_0})
=
O\!\left(
\tau^{-\frac{(|\sigma|+1)(n-1)+1}{2}}
\right).
\]
Comparing with the definition \eqref{BesovCorr} gives the regularity
exponent
$
\alpha
=
-|\sigma|n+|\sigma|-n$ and the estimate \eqref{unif1} follows.

At tree level the polynomial factor in
\(
\tau^{-1/2}/\Lambda_m
\)
is absent, so that the leading singularity is
\(
\tau^{-(n-1)/2}
\),
which corresponds to the regularity exponent
\(
-n+1.
\)
\end{proof}
\begin{corollary}\label{cor:precompact}
The family
\[
\Bigl(
\chi_\varepsilon^{\otimes\sigma}
\mathcal L_{l,n}^{\Lambda,\Lambda_0}
\Bigr)_{\varepsilon>0}
\]
is bounded in
$
\mathcal C^{-|\sigma|n+|\sigma|-n}
(\mathbb R^n)$ uniformly in $\epsilon$. Consequently, it is relatively compact in $
\mathcal C^{-|\sigma|n+|\sigma|-n-\kappa}_{\mathrm{loc}}
(\mathbb R^n)$ for $\kappa>0$. In particular, 
there exist a subsequence
\(
\varepsilon_j
\),
and a distribution
\[
\mathds 1_+^\sigma
\mathcal L_{l,n}^{\Lambda,\Lambda_0}
\in
\mathcal C^{-|\sigma|n+|\sigma|-n}_{\mathrm{loc}}
(\mathbb R^n)
\]
such that
\[
\chi_{\varepsilon_j}^{\otimes\sigma}
\mathcal L_{l,n}^{\Lambda,\Lambda_0}
\longrightarrow
\mathds 1_+^\sigma
\mathcal L_{l,n}^{\Lambda,\Lambda_0}
\]
in
\[
\mathcal C^{-|\sigma|n+|\sigma|-n-\kappa}_{\mathrm{loc}}
(\mathbb R^n)\qquad\mathrm{for~all~}\kappa>0~.
\]
The analogous statement holds for \(l=0\), with
\(
\mathcal C^{-n+1}
\)
and
\(
\mathcal C^{-n+1}_{\mathrm{loc}}
\)
in place of the spaces above.
\end{corollary}

\begin{proof}
Proposition~\ref{prop:BesovUniform} yields an $\epsilon$-uniform bound
\[
\sup_{\varepsilon>0}
\Bigl\|
\chi_\varepsilon^{\otimes\sigma}
\mathcal L_{l,n}^{\Lambda,\Lambda_0}
\Bigr\|_{
\mathcal C^{-|\sigma|n+|\sigma|-n}
}
<\infty.
\]
The compact embedding
\[
\mathcal C^\alpha(\mathbb R^n)
\hookrightarrow
\mathcal C^{\alpha-\kappa}_{\mathrm{loc}}(\mathbb R^n),
\qquad
,
\]
therefore implies relative compactness in
\(
\mathcal C^{-|\sigma|n+|\sigma|-n-\kappa}_{\mathrm{loc}}
\) for all $\kappa>0$.
Hence every sequence admits a convergent subsequence.
\end{proof}
The previous corollary provides compactness of the family
\(
(\chi_\varepsilon^{\otimes\sigma}
\mathcal L_{l,n}^{\Lambda,\Lambda_0})_{\varepsilon>0}
\).
To conclude the existence of the truncated correlators, it remains to
identify the limit and to show that it is independent of the chosen
subsequence. This is achieved below by studying the corresponding
heat-kernel regularisations.

\begin{proposition}[Existence of the heat-regularised truncation limit]\label{Existence}
Fix \(\tau_i>0\), \(1\le i\le n\).
Then for every
\((\vec y_{1,n},\underline p_n)\) the limit
\[
\lim_{\epsilon\to0}
\bigl(K^{\otimes n}\star
\chi_{\epsilon}^{\otimes\sigma}
\mathcal L_{l,n}^{\Lambda,\Lambda_0}\bigr)
\left(\vec{\tau}_{1,n};\vec y_{1,n},\underline{p}_n\right)
\]
exists for all $n\ge 4$ and does not depend on the particular choice of the admissible
family $(\chi_\epsilon)_\epsilon\in\mathcal{K}$.
\end{proposition}
\begin{proof}
We argue by induction on the quantity
\(
n+2l
\).
For fixed \(n+2l\), the induction proceeds by increasing loop order. For fixed $n$ we first treat the tree order $l=0$. For $n=4$ one has
\[
\mathcal{L}_{0,4}^{\Lambda,\Lambda_0}
\bigl((z_1,p_1),\ldots,(z_4,p_4)\bigr)
=
\lambda \prod_{i=2}^4 \delta(z_1-z_i).
\]
After heat-kernel regularisation, one obtains
\begin{equation}\label{eq:IH-object}
K^{\otimes 4}\star\chi_{\epsilon}^{\otimes\sigma}
\mathcal{L}_{0,4}^{\Lambda,\Lambda_0}
\left(\vec{y}_{1,4};\underline{p}_{4}\right)
=
\int_{\mathbb{R}}
\prod_{i=1}^4 K(\tau_i;z,y_i)\,
\chi_\epsilon^{|\sigma|}(z)\,dz .
\end{equation}
Since $\chi_\epsilon\to\mathds{1}_{+}$ pointwise and
$\sup_{\epsilon>0}\|\chi_\epsilon\|_\infty<\infty$,
the dominated convergence theorem yields
\[
K^{\otimes 4}\star\chi_{\epsilon}^{\otimes\sigma}
\mathcal{L}_{0,4}^{\Lambda,\Lambda_0}\left(\vec{y}_{1,4};\underline{p}_{4}\right)
\longrightarrow
\int_{\mathbb{R}_+}
\prod_{i=1}^4 K(\tau_i;z,y_i)\,dz .
\]
For general $n$ the flow equation at tree level gives
\begin{multline}
K^{\otimes n}\star
\left(\chi_{\epsilon}^{\otimes\sigma}
\partial_{\Lambda}\mathcal{L}_{0,n}^{\Lambda,\Lambda_0}\right)\\
=
-\frac12
\sum_{\substack{n_1+n_2=n\\ n_1,n_2\ge4}}
\sum_{\pi\in\mathcal{P}_n}
\left(K^{\otimes (n_1+1)}\star
\chi_{\epsilon}^{\otimes\sigma_1}
\mathcal{L}_{0,n_1+1}^{\Lambda,\Lambda_0}\right)\otimes_{\pi}
\left(K^{\otimes (n_2+1)}\star
\chi_{\epsilon}^{\otimes\sigma_2}
\mathcal{L}_{0,n_2+1}^{\Lambda,\Lambda_0}\right)
\dot C^\Lambda(p_\pi),
\end{multline}
where $\mathcal{L}_{0,2}^{\Lambda,\Lambda_0}=0$.
Theorem~\ref{ThmPrinc} implies
\begin{multline}
\Bigl|
\left(K^{\otimes (n_i+1)}\star
\chi_{\epsilon}^{\otimes\sigma_i}
\mathcal{L}_{0,n_i+1}^{\Lambda,\Lambda_0}\right)
\left((\tau_j,y_j)_{j\in\pi_i},\left(\frac{1}{2\Lambda^2},u\right);\underline{p}_{\pi_j},p_{\pi}\right)
\Bigr|
\\
\le
\Lambda_m^{3-n_i}
\mathcal A_{0,n_i+1}^{\Lambda,\Lambda_0}
\left((\tau_j,y_j)_{j\in\pi_i},(\tfrac{1}{2\Lambda^2},u);\chi_\epsilon^{\sigma_i}\right).
\end{multline}
Since $\sup_\epsilon\|\chi_\epsilon\|_\infty<\infty$,
\begin{equation}\label{UniEps}
\mathcal A_{0,n_i+1}^{\Lambda,\Lambda_0}
\left((\tau_j,y_j)_{j\in\pi_i};\chi_\epsilon^{\sigma_i}\right)
\lesssim
\mathcal A_{0,n_i+1}^{\Lambda,\Lambda_0}
\left((\tau_j,y_j)_{j\in\pi_i}\right).
\end{equation}
Using \eqref{UniEps} together with the fusion estimate
\eqref{eq:FusionTight1} yields
\[
\mathcal A_{0,n_1+1}^{\Lambda,\Lambda_0}\otimes_{\pi}
\mathcal A_{0,n_2+1}^{\Lambda,\Lambda_0}
\lesssim
\mathcal A_{0,n}^{\Lambda,\Lambda_0}\left((\tau_1,y_1),(\tau_i,y_i)_{i\in\mathcal{I}_n}\right).
\]
Hence the right-hand side of the differentiated flow equation is
dominated by an \(\epsilon\)-independent integrable majorant in the
integration variables \((u,\Lambda)\). Dominated convergence implies
the existence of the pointwise limit of
\[
K^{\otimes n}\star
\left(\chi_{\epsilon}^{\otimes\sigma}
\partial_{\Lambda}\mathcal{L}_{0,n}^{\Lambda,\Lambda_0}\right).
\]
Using the boundary condition
\[
\mathcal L_{0,n}^{\Lambda_0,\Lambda_0}=0,
\qquad n\ge5,
\]
and the monotonicity
\[
\mathcal A_{0,n}^{\lambda,\Lambda_0}
\lesssim
\mathcal A_{0,n}^{\Lambda,\Lambda_0},
\qquad \lambda\le\Lambda,
\]
integration of the flow equation in $\Lambda$ yields the existence of
\[
\lim_{\epsilon\to0}
K^{\otimes n}\star
\left(\chi_{\epsilon}^{\otimes\sigma}
\mathcal{L}_{0,n}^{\Lambda,\Lambda_0}\right).
\]
For $l\ge1$ we differentiate in $\Lambda$ and use the Polchinski flow
equation. Theorem~\ref{ThmPrinc} together with
\eqref{bBo} and \eqref{UniEps} gives
\begin{multline}\label{eq:ind-bound}
\Bigl|
\bigl(K^{\otimes (n+2)}\!\star \chi_{\varepsilon}^{\otimes\sigma}\,
\mathcal{L}_{l-1,n+2}^{\Lambda,\Lambda_0}\bigr)
\Bigl(\cdot,\bigl(\tfrac{1}{2\Lambda^{2}},u\bigr)^{\otimes2},k,-k\Bigr)
\Bigr|
\\
\le
\Lambda_m^{2-n}\,
\mathcal P_{l-1}\!\left(\log \frac{\Lambda_m}{m}\right)
\mathcal P\!\left(\frac{\|\vec p_n\|}{\Lambda_m},\frac{|k|}{\Lambda_m}\right)
\mathcal A_{l-1,n+2}^{\Lambda,\Lambda_0}
\Bigl((\vec\tau_{1,n},\vec y_{1,n}),
(\tfrac{1}{2\Lambda^2},u)^{\otimes2}\Bigr).
\end{multline}
Combining the reduction estimate \eqref{Redu1},
the fusion estimate \eqref{eq:FusionTight1},
and the covariance estimate,
\[
\int_{\mathbb R^3}|k|^\alpha \dot C^\Lambda(k)\,dk\lesssim1
\]
shows that the integrands in the flow equation admit an
$\epsilon$–uniform $L^1$ bound.
Dominated convergence therefore yields the existence of the pointwise
limit of the differentiated flow equation.
The quadratic term in case $l>0$ is treated similarly as the tree order case.
Integrating the flow equation in $\Lambda$ then gives
\[
\lim_{\epsilon\to0}
K^{\otimes n}\!\star
\chi_{\epsilon}^{\otimes\sigma}\mathcal L_{l,n}^{\Lambda,\Lambda_0}
=\int_\Lambda^{\Lambda_0}\lim_{\epsilon\to0}
K^{\otimes n}\!\star
\partial_\lambda \left(\chi_{\epsilon}^{\otimes\sigma}\mathcal L_{l,n}^{\lambda,\Lambda_0}\right)\,d\lambda
+
\lim_{\epsilon\to0}
K^{\otimes n}\!\star
\chi_{\epsilon}^{\otimes\sigma}\mathcal L_{l,n}^{\Lambda_0,\Lambda_0}.
\]
For $n\ge5$ the boundary term vanishes.
For $n=4$ one has
\[
K^{\otimes4}\!\star
\chi_{\varepsilon}^{\otimes\sigma}
\mathcal L_{l,4}^{\Lambda_0,\Lambda_0}
=
c_l^{\Lambda_0}
\int_{\mathbb R}\chi_\varepsilon^{|\sigma|}(z)
\prod_{i=1}^4K(\tau_i;z,y_i)\,dz
\to
c_l^{\Lambda_0}
\int_{\mathbb R_+}
\prod_{i=1}^4K(\tau_i;z,y_i)\,dz .
\]
This proves the existence of the limit.
To prove independence of the approximation, let
\((\chi_\varepsilon)_\varepsilon\)
and
\((\widetilde\chi_\varepsilon)_\varepsilon\)
be two admissible families and denote by
\(U_{l,n}^{\Lambda,\Lambda_0}\)
the difference of the corresponding limits.

The induction hypothesis implies that all lower-order terms appearing in
the limiting flow equation coincide. Consequently,
\(U_{l,n}^{\Lambda,\Lambda_0}\)
satisfies the homogeneous limiting flow equation.

For \(n\ge5\), the ultraviolet boundary contribution vanishes.
For \(n=4\), the boundary term is proportional to
\[
\chi_\varepsilon^{|\sigma|}
-
\widetilde\chi_\varepsilon^{|\sigma|}.
\]
Since both admissible families converge pointwise to
\(\mathds 1_+\) and are uniformly bounded, dominated convergence implies
that the corresponding boundary contribution converges to zero.

Hence
\(U_{l,n}^{\Lambda,\Lambda_0}\)
has vanishing boundary data and satisfies the homogeneous flow
equation. It follows that
\[
U_{l,n}^{\Lambda,\Lambda_0}=0,
\]
which proves uniqueness of the limit.
\end{proof}
We may now define the truncated correlators as the limits of the smeared quantities.
\begin{definition}[Truncated correlators]
\label{def:truncated-correlators}

For \(l\ge0\), \(n\ge4\), and
\(\sigma\subseteq\mathcal I_n\), we define the truncated
\(n\)-point correlator by
\[
\mathds 1_+^\sigma
\mathcal L_{l,n}^{\Lambda,\Lambda_0}
:=
\lim_{\varepsilon\to0}
\chi_\varepsilon^{\otimes\sigma}
\mathcal L_{l,n}^{\Lambda,\Lambda_0},
\]
where the limit is taken in
\[
\mathcal C^{-|\sigma|n+|\sigma|-n-\kappa}_{\mathrm{loc}}
(\mathbb R^n),~~~\forall\kappa>0
\]
and where
\((\chi_\varepsilon)_\varepsilon\subset\mathcal K\)
is any admissible family.
\end{definition}
\begin{theorem}[Existence and approximation-independence]
\label{thm:existence-correlators}
Let $l\ge0$, $n\ge4$ and $\sigma\subseteq\mathcal I_n$.
Then the truncated correlators
\[
\mathds{1}_+^{\sigma}\mathcal L_{l,n}^{\Lambda,\Lambda_0}
\]
exist as distributions in
$\mathcal C^{-|\sigma|n+|\sigma|-n}_{\mathrm{loc}}(\mathbb R^n)$.

Moreover, the resulting distribution is independent of the particular
choice of the family $(\chi_\epsilon)_\epsilon$ within the admissible
class of cutoffs $\mathcal{K}$.
\end{theorem}

\begin{proof}
By Proposition~\ref{prop:BesovUniform}, the family
\[
\bigl(\chi_{\epsilon}^{\otimes\sigma}
\mathcal L_{l,n}^{\Lambda,\Lambda_0}\bigr)_{\epsilon>0}
\]
is uniformly bounded in $\mathcal C^{-|\sigma|n+|\sigma|-n}(\mathbb R^n)$.
Hence, by Corollary~\ref{cor:precompact}, it is relatively compact in
$\mathcal C^{-|\sigma|n+|\sigma|-n-\kappa}_{\mathrm{loc}}(\mathbb R^n)$ for $\kappa>0$.
On the other hand, Proposition~\ref{Existence} shows that 
\[
K^{\otimes n}\star\left(
\chi_{\epsilon}^{\otimes\sigma}
\mathcal L_{l,n}^{\Lambda,\Lambda_0}
\right)
\]
admit a pointwise limit as $\epsilon\to0$, and that this limit does not
depend on the particular choice of the family $(\chi_\epsilon)_\epsilon$ in $\mathcal{K}$.
Therefore all subsequential limits coincide, and the whole family
converges in
$\mathcal C^{-|\sigma|n+|\sigma|-n-\kappa}_{\mathrm{loc}}(\mathbb R^n)$ for $\kappa>0$. This proves existence and universality of the truncated correlators.
\end{proof}
\begin{remark}\label{rem:minus-limit}
Recall that the correlators are invariant under reflection of the
normal coordinate, namely
\begin{equation}\label{eq:reflection}
\mathcal{L}_{l,n}^{\Lambda,\Lambda_0}
\bigl((z_1,p_1),\dots,(z_n,p_n)\bigr)
=
\mathcal{L}_{l,n}^{\Lambda,\Lambda_0}
\bigl((-z_1,p_1),\dots,(-z_n,p_n)\bigr)
\end{equation}
as distributions in \((z_1,\dots,z_n)\).
Let
\[
R(z_1,\dots,z_n):=(-z_1,\dots,-z_n)
\]
denote the reflection map. Then
\[
\Bigl(\prod_{i\in\sigma}\chi_\varepsilon(-z_i)\Bigr)
\mathcal L_{l,n}^{\Lambda,\Lambda_0}
=
R^*\!\left[
\Bigl(\prod_{i\in\sigma}\chi_\varepsilon(z_i)\Bigr)
\mathcal L_{l,n}^{\Lambda,\Lambda_0}
\right].
\]

Since pullback by the smooth diffeomorphism \(R\) acts continuously on
local Besov spaces, the convergence established in
Theorem~\ref{thm:existence-correlators} implies
\[
\Bigl(\prod_{i\in\sigma}\chi_\varepsilon(-z_i)\Bigr)
\mathcal L_{l,n}^{\Lambda,\Lambda_0}
\longrightarrow
R^*
\Bigl(
\mathds1_+^\sigma
\mathcal L_{l,n}^{\Lambda,\Lambda_0}
\Bigr).
\]

This motivates the definition
\[
\mathds1_-^\sigma
\mathcal L_{l,n}^{\Lambda,\Lambda_0}
:=
R^*
\Bigl(
\mathds1_+^\sigma
\mathcal L_{l,n}^{\Lambda,\Lambda_0}
\Bigr),
\]
which may be interpreted as the correlator truncated to the negative
half-space.
\end{remark}
\begin{remark}\label{rem:other-truncated-limits}

The proof of Proposition~\ref{Existence} immediately extends to the
weighted and partially regularised correlators appearing in the
power-counting analysis. More precisely, for every
\(
\mathfrak s\subseteq\mathcal I_n
\),
\(
\sigma\subseteq\mathfrak s
\),
every \(r\in\{0,1,2,3\}\),
and every pair of distinct indices
\(i,j\in\mathcal I_n\),
the limit
\[
\lim_{\varepsilon\to0}
\Bigl(
K^{\otimes\mathfrak s}
\star
\mathbb Z^{(r)}_{i,j}
\chi_\varepsilon^{\otimes\sigma}
\mathcal L_{l,n}^{\Lambda,\Lambda_0}
\Bigr)
\Bigl(
z_1;
(\tau_k,y_k)_{k\in\mathfrak s};
\underline p_n
\Bigr)
\]
exists.

We denote the resulting quantity by
\[
\Bigl(
K^{\otimes\mathfrak s}
\star
\mathbb Z^{(r)}_{i,j}
\mathds1_+^\sigma
\mathcal L_{l,n}^{\Lambda,\Lambda_0}
\Bigr)
\Bigl(
z_1;
(\tau_k,y_k)_{k\in\mathfrak s};
\underline p_n
\Bigr).
\]

Passing to the limit in the flow equations by dominated convergence
yields the corresponding flow equations for the truncated correlators.
These identities provide the natural starting point for the analysis of
renormalisation problems involving interactions supported on a
half-space.
\end{remark}

For $n=2$, the power-counting degree is positive and the resulting
Besov regularity differs from the case $n\ge4$. We therefore analyse
the two-point sector separately. The estimate obtained below is a direct
consequence of the power-counting spaces and semi-norms and their connection to the $n$-point correlators established in
Theorem~\ref{ThmPrinc} together with the uniform control of the tree
amplitudes. For the two-point correlators, Theorem~\ref{ThmPrinc} yields
\begin{equation}\label{eq:n2-bound}
\left|
\bigl(
K_{\tau,y_2}
\star
\mathcal L_{l,2}^{\Lambda,\Lambda_0}
\bigr)
\bigl(z_1;p,-p\bigr)
\right|
\le
\Lambda_m^2
\mathcal P_{l-1}
\!\left(
\log\frac{\Lambda_m}{m}
\right)
\mathcal P
\!\left(
\frac{|p|}{\Lambda_m}
\right)
\sum_{j=0}^{3}
\left(
\frac{\tau^{-1/2}}{\Lambda_m}
\right)^j
\mathcal A_{l,2}^{\Lambda,\Lambda_0}
\bigl(
z_1;(\tau,y_2)
\bigr).
\end{equation}

Multiplying by \(K(\tau;z_1,y_1)\) and integrating over \(z_1\), we obtain
\begin{multline}
\left|
\int_{\mathbb R_+}\!\!\int_{\mathbb R}
\mathcal L_{l,2}^{\Lambda,\Lambda_0}
\bigl(
(z_1,p),(z_2,-p)
\bigr)
\prod_{i=1}^{2}
K(\tau;z_i,y_i)
\,dz_1\,dz_2
\right|
\\
\le
\Lambda_m^2
\mathcal P_{l-1}
\!\left(
\log\frac{\Lambda_m}{m}
\right)
\mathcal P
\!\left(
\frac{|p|}{\Lambda_m}
\right)
\sum_{j=0}^{3}
\left(
\frac{\tau^{-1/2}}{\Lambda_m}
\right)^j
\mathcal A_{l,2}^{\Lambda,\Lambda_0}
\bigl(
(\tau,y_1),(\tau,y_2)
\bigr).
\end{multline}
Here we used 
\[
\int_{\mathbb R_+}
K(\tau;z_1,y_1)\,
\mathcal A_{l,2}^{\Lambda,\Lambda_0}
\bigl(
z_1;(\tau,y_2)
\bigr)
\,dz_1\lesssim
\mathcal A_{l,2}^{\Lambda,\Lambda_0}
\bigl(
(\tau,y_1),(\tau,y_2)
\bigr),
\]
together with the uniform bound
\[
\sup_{y_1,y_2}
\mathcal A_{l,2}^{\Lambda,\Lambda_0}
\bigl(
(\tau,y_1),(\tau,y_2)
\bigr)
\lesssim \tau^{-\frac{1}{2}}.
\]
Consequently, for every \(l\ge2\),
\begin{align}
\|
\mathds 1_+
\mathcal L_{l,2}^{\Lambda,\Lambda_0}
(\underline{p})
\|_{\mathcal C^{-4}(\mathbb R^2)}
&\lesssim
\Bigl(
\sum_{j=0}^{3}
\Lambda_m^{2-j}
\Bigr)
\mathcal P_{l-1}
\!\left(
\log\frac{\Lambda_m}{m}
\right)
\mathcal P
\!\left(
\frac{|p|}{\Lambda_m}
\right),
\\
\|
\mathds 1_-
\mathcal L_{l,2}^{\Lambda,\Lambda_0}
(\underline{p})
\|_{\mathcal C^{-4}(\mathbb R^2)}
&\lesssim
\Bigl(
\sum_{j=0}^{3}
\Lambda_m^{2-j}
\Bigr)
\mathcal P_{l-1}
\!\left(
\log\frac{\Lambda_m}{m}
\right)
\mathcal P
\!\left(
\frac{|p|}{\Lambda_m}
\right).
\end{align}
The case \(l=1\) is exceptional. Indeed,
\[
\mathcal L_{1,2}^{\Lambda,\Lambda_0}
\bigl(
(z_1,p),(z_2,-p)
\bigr)
=
\ell_{1,2}^{\Lambda,\Lambda_0}(z_1)
\,
\delta(z_1-z_2),
\]
where
\[
\ell_{1,2}^{\Lambda,\Lambda_0}(z_1)
:=
\int_{\mathbb R}
\mathcal L_{1,2}^{\Lambda,\Lambda_0}
\bigl(
(z_1,0),(z,0)
\bigr)
\,dz .
\]
Hence
\[
\mathds 1_\pm(z_1)
\mathcal L_{1,2}^{\Lambda,\Lambda_0}
\bigl(
(z_1,p),(z_2,-p)
\bigr)
=
\ell_{1,2}^{\Lambda,\Lambda_0}(z_1)
\,
\mathds 1_\pm(z_1)
\delta(z_1-z_2),
\]
which belongs to
\(
\mathcal C^{-1}(\mathbb R^2).
\)
\newpage

\newpage
\section{The ultraviolet limit of the correlators and the half-space truncated correlators}\label{ConvSection}
In this section we establish convergence of the renormalised
correlators in the ultraviolet limit
\(
\Lambda_0\to\infty
\).
More precisely, we show that the families $
\mathcal L_{l,n}^{0,\Lambda_0}$
and their half-space truncations form Cauchy sequences in suitable
Besov spaces. As a consequence, the continuum limit exists and defines
renormalised correlators independent of the ultraviolet cutoff.
\begin{theorem}[Ultraviolet convergence]\label{Convergence}
    For $0<\Lambda_0'<\Lambda_0$ we have for all $n\ge 4$ and $l\ge 1$ that 
    \begin{equation}
\left\|\mathcal{L}_{l,n}^{0,\Lambda_0}(\underline{p_n})-\mathcal{L}_{l,n}^{0,\Lambda'_0}(\underline{p}_n)\right\|_{\mathcal{C}^{-2n+2}\left(\mathbb{R}^n\right)}\leq \Lambda_0^{-1}\Big(\sum_{i=0}^{n-1}m^{5-n}\Big)
\left(\log\frac{\Lambda_0+m}{m}\right)^{\nu}~
\mathcal{P}\!\left(\frac{\|\underline{p}_n\|}{\Lambda_m}\right)~.
    \end{equation}
Furthermore, every
$\sigma\subseteq\mathcal I_n$ and every compact $K\Subset \mathbb{R}^n$ there exists $C_K>0$ such that 
    \begin{multline}\label{337---}
\|\mathds{1}_+^{\sigma}\mathcal{L}_{l,n}^{0,\Lambda_0}(\underline{p}_n)-\mathds{1}_+^{\sigma}\mathcal{L}_{l,n}^{0,\Lambda'_0}(\underline{p}_n)\|_{\mathcal{C}^{-|\sigma|n+|\sigma|-n}(K)}
\\\leq~C_K~
\Lambda_0^{-1}\left(\sum_{i=1}^{\left(|\sigma|+1\right)(n-1)}m^{5-n-i}\right)~\left(\log\frac{\Lambda_0+m}{m}\right)^{\tilde{\nu}}
\mathcal{P}\!\left(\frac{\|\underline{p}_{n}\|}{\Lambda_m}\right),
\qquad l\ge1,
    \end{multline}
At tree level, one has for all $n\ge4$ 
\begin{align}\label{BesovnTree}
\left\|\mathcal{L}_{0,n}^{0,\Lambda_0}(\underline{p}_n)-\mathcal{L}_{0,n}^{0,\Lambda_0'}(\underline{p}_{n})\right\|_{\mathcal{C}^{-n+1}(\mathbb{R}^n)}
\le~\Lambda_0^{-1}~
\left(\sum_{i=0}^{n-1}m^{5-n-i}\right).
\end{align}
The same estimate holds for every compact $K\Subset \mathbb{R}^n$
\begin{equation}
\left\|\mathds{1}^{\sigma}_+\mathcal{L}_{0,n}^{0,\Lambda_0}(\underline{p}_n)-\mathds{1}^{\sigma}_+\mathcal{L}_{0,n}^{0,\Lambda_0'}(\underline{p}_{n})\right\|_{\mathcal{C}^{-n+1}(K)}.
\end{equation}
For the two–point functions and $l\ge2$,
\begin{align}\label{2PTBesov}
\left\|\mathcal{L}_{l,2}^{0,\Lambda_0}(\underline{p})-\mathcal{L}_{l,2}^{0,\Lambda_0'}(\underline{p})\right\|_{\mathcal{C}^{-4}(\mathbb{R}^2)}
&\le~\Lambda_0^{-1}~
\Big(\sum_{i=0}^{3}m^{2-i}\Big)
\left(\log\frac{\Lambda_0+m}{m}\right)^{\nu'}
\mathcal{P}\!\left(\frac{|p|}{m}\right).
\end{align}
Finally,
\begin{align}\label{2PTBesov2}
\left\|\mathcal{L}_{1,2}^{0,\Lambda_0}\left(\underline{p}\right)-\mathcal{L}_{1,2}^{0,\Lambda_0'}\left(
\underline{p}\right)\right\|_{\mathcal{C}^{-1}(\mathbb{R}^2)}
&\le~\Lambda_0^{-1}~
\Big(\sum_{i=0}^{3}m^{3-i}\Big)
\mathcal{P}\!\left(\frac{|p|}{m}\right).
\end{align}
The same estimates holds for the truncated correlators
\(
\mathds1_+
\mathcal L_{l,2}^{0,\Lambda_0}
\)
and
\(
\mathds1_+
\mathcal L_{l,2}^{0,\Lambda_0'}
\),
with the same right-hand side in \eqref{2PTBesov} and \eqref{2PTBesov2}.

As a consequence, the limits
\[
\lim_{\Lambda_0\to\infty}
\mathcal L_{l,n}^{0,\Lambda_0}\quad\mathrm{and}\qquad \lim_{\Lambda_0\to\infty}
\mathds1_+^\sigma
\mathcal L_{l,n}^{0,\Lambda_0}
\]
exist in the corresponding Besov spaces. In particular, the
renormalised correlators and their half-space truncations admit a
well-defined ultraviolet limit.
\end{theorem}

\noindent The proof of Theorem~\ref{Convergence} is based on power-counting bounds for the ultraviolet derivatives
\(
\partial_{\Lambda_0}\mathcal L_{l,n}^{\Lambda,\Lambda_0}\).
These bounds are the analogue of Theorem~\ref{ThmPrinc} for the cutoff derivatives and reflect the additional decay generated by differentiation with respect to the ultraviolet cutoff. Once established, they imply the convergence estimates of Theorem~\ref{Convergence} by integration in \(\Lambda_0\). The precise statement is given in the following proposition.
\begin{proposition}\label{PropConv}[Power-counting bounds for ultraviolet derivatives]
  Under the assumptions of Theorem~\ref{ThmPrinc}, one has for every
\(\Lambda\in[0,\Lambda_0]\)
\begin{multline}\label{Premiere}
\mathcal{N}^{\Lambda_m}_{5-n-r-|w|,l;\infty}\left((z_1-z_i)^r~
\partial^w
\partial_{\Lambda_0}\mathcal L_{l,n}^{\Lambda,\Lambda_0}
\bigl(
\underline{p}_{n}
\bigr)\right)\\\le \left(\Lambda_0+m\right)^{-2}\mathcal{P}_{l-1}\left(\log \frac{\Lambda_0+m}{m}\right)\widetilde{\mathcal{P}}\left(\frac{\|\underline{p}_n\|}{\Lambda_m}\right)\qquad
\forall n\ge2,
\end{multline}
\begin{multline}\label{deuxieme}
\mathcal{N}^{\Lambda_m}_{5-n,l;\infty}\left(\Delta\left(
\partial_{\Lambda_0}\mathcal L_{l,n}^{\Lambda,\Lambda_0}\right)
\bigl(
\underline{p}_n
\bigr)\right)\\\le \left(\Lambda_0+m\right)^{-2}\mathcal{P}_{l-1}\left(\log \frac{\Lambda_0+m}{m}\right)\widetilde{\mathcal{P}}\left(\frac{\|\underline{p}_n\|}{\Lambda_m}\right)\qquad
\forall n\ge4,
\end{multline}
and 
\begin{multline}\label{troisieme}
~\mathcal{N}^{\Lambda_m;(2)}_{2-n,l;\infty}\left((z_1-z_2)^r
\partial^w\partial_{\Lambda_0}\mathcal L_{l,n}^{\Lambda,\Lambda_0}
\bigl(
\underline{p}_{n}
\bigr)\right)\\\le \left(\Lambda_0+m\right)^{-2}\mathcal{P}_{l-1}\left(\log \frac{\Lambda_0+m}{m}\right)\widetilde{\mathcal{P}}\left(\frac{\|\underline{p}_n\|}{\Lambda_m}\right)\qquad
\forall n\ge2,~~\forall r+|w|=3
\end{multline} 
\end{proposition}
\begin{proof}
 The proof follows the same strategy as that of Theorem~\ref{ThmPrinc}. The key observation is that differentiating the flow equation with respect to the ultraviolet cutoff $\Lambda_0$
 yields a closed flow equation for $\partial_{\Lambda_0}\mathcal{L}_{l,n}^{\Lambda,\Lambda_0}$. The resulting equation has exactly the same structure as the original flow equation, while the ultraviolet boundary conditions provide an additional factor $\left(\Lambda_0+m\right)^{-2}$
. Consequently, the stability estimates established in the subsection \ref{Continuitt} may be applied  to the $\Lambda_0$-derivatives. The proof therefore reduces to repeating the renormalization argument with shifted power counting.
  \begin{itemize}
      \item \underline{Case 1: $n+r+|w|\ge5$:} recall that for these terms the boundary condition is fixed at $\Lambda=\Lambda_0$ such that
      \begin{equation}\label{bcdt}
          \left(z_1-z_i\right)^r\partial^w\mathcal{L}_{l,n}^{\Lambda_0,\Lambda_0}\left((z_1,p_1),\cdots,(z_n,p_n)\right)=0~.
      \end{equation}
      As a consequence we also have 
      \begin{equation}\label{bcdt2}
          \Delta\mathcal{L}_{l,n}^{\Lambda_0,\Lambda_0}=0,~~~\forall n\ge 4
      \end{equation}
      and 
       \begin{equation}\label{bcdt3}
\left(z_1-z_2\right)^r\partial^w\mathcal{L}_{l,n}^{\Lambda_0,\Lambda_0}\left((z_1,p_1),\cdots,(z_n,p_n)\right)=0,~~~\forall n\ge 2,~~r+|w|=3~.
      \end{equation}
      We establish \eqref{Premiere}. The proofs of \eqref{deuxieme} and \eqref{troisieme}
are identical and will therefore be omitted.
    Hence we integrate the FEs from $\Lambda$ to $\Lambda_0$ and use \eqref{bcdt} to write
    \begin{equation}\label{408}
       \left(z_1-z_i\right)^r \partial^w\mathcal{L}_{l,n}^{\Lambda,\Lambda_0}\left(\vec{z}_{1,n};\underline{p}_n\right)=\int_{\Lambda}^{\Lambda_0}d\lambda~\left(z_1-z_i\right)^r\partial^w\mathcal{R}_{l,n}^{\lambda,\Lambda_0}\left(\vec{z}_{1,n};\underline{p}_n\right)
    \end{equation}
    where $\mathcal{R}_{l,n}^{\Lambda,\Lambda_0}$ denotes the r.h.s. of the FEs. Now we differentiate w.r.t. $\Lambda_0$ and we obtain 
    \begin{multline}
       \left(z_1-z_i\right)^r \partial^w\partial_{\Lambda_0}\mathcal{L}_{l,n}^{\Lambda,\Lambda_0}\left(\underline{z}_n;\underline{p}_n\right)=\left(z_1-z_i\right)^r\partial^w\mathcal{R}_{l,n}^{\Lambda_0,\Lambda_0}\left(\vec{z}_{1,n};\underline{p}_n\right)+\\\int_{\Lambda}^{\Lambda_0}d\lambda~\left(z_1-z_i\right)^r\partial^w\partial_{\Lambda_0}\mathcal{R}_{l,n}^{\lambda,\Lambda_0}\left(\vec{z}_{1,n};\underline{p}_n\right)~.
    \end{multline}
    The r.h.s. of the FEs is given in terms of the Polchinski operators $\mathbf{L}^{\Lambda}$ and $\mathbf{B}^{\Lambda}_{\pi}$ by 
    \begin{equation}
        \mathcal{R}^{\Lambda,\Lambda_0}_{l,n}=\mathbf{L}^{\Lambda}\left(\mathcal{L}_{l-1,n+2}^{\Lambda,\Lambda_0}\right)-\sum_{\substack{\{\pi_1,\pi_2\}\in\mathcal{P}_n\\l_1+l_2=l}}\mathbf{B}^{\Lambda}_{\pi}\left(\mathcal{L}_{l_1,n_1+1}^{\Lambda,\Lambda_0},\mathcal{L}_{l_2,n_2+1}^{\Lambda,\Lambda_0}\right)
    \end{equation}
    In the proof of Theorem \ref{ThmPrinc} we have exactly established using the continuity maps of the Polchinski operators $\mathbf{L}^{\Lambda}$ and $\mathbf{B}^{\Lambda}_{\pi}$ that 
    \begin{equation}
        (z_1-z_i)^r\partial^w\mathcal{R}_{l,n}^{\Lambda,\Lambda_0}\in\mathcal{D}^{\Lambda_m}_{3-n-r,l}\left(\mathbb{M}^n\right)~.
    \end{equation}
    This in particular implies for $\Lambda=\Lambda_0$ and for all $\mathfrak{s}\subseteq \sigma \subseteq \mathcal{I}_n$ that 
    \begin{multline}\label{351--}
\left|\left(K^{\otimes\mathfrak{s}}\otimes \chi^{\otimes\sigma}\mathbb{Z}_{1,i}^r\partial^{w}\mathcal{R}_{l,n}^{\Lambda_0,\Lambda_0}\right)\left(z_1;\underline{y}_{\mathfrak{s}};\underline{p}_n\right)\right|\le \left(\Lambda_0+m\right)^{3-n-r-|w|}\mathcal{P}_{l-1}\left(\log\frac{\Lambda_0+m}{m}\right)\\\times\mathcal{Q}^{\sigma}\left(\left\{\frac{\tau_i^{-\frac{1}{2}}}{\Lambda_0+m}\right\}_{i\in\mathfrak{s}}\right)\widetilde{\mathcal{P}}\left(\frac{\|\vec{p}_n\|}{\Lambda_0+m}\right)\left(\mathcal{A}_{l,n}^{\Lambda_0,\Lambda_0}+\Lambda_m^{-1}\hat{\mathcal{A}}_{l,n}^{\Lambda_0,\Lambda_0}\right)\left(z_1;\underline{y}_{\mathfrak{s}};\chi^{\sigma}\right)
    \end{multline}
   Recall that $\Lambda\le \Lambda_0$. Since $5-n-r-|w|\ge 0$ the l.h.s. of \eqref{351--} is bounded by 
    \begin{multline} \frac{\left(\Lambda+m\right)^{5-n-r-|w|}}{\left(\Lambda_0+m\right)^2}\mathcal{P}_{l-1}\left(\log\frac{\Lambda_0+m}{m}\right)\mathcal{Q}^{\sigma}\left(\left\{\frac{\tau_i^{-\frac{1}{2}}}{\Lambda+m}\right\}_{i\in\mathfrak{s}}\right)\widetilde{\mathcal{P}}\left(\frac{\|\underline{p}_n\|}{\Lambda_m}\right)\\\times\left(\mathcal{A}_{l,n}^{\Lambda,\Lambda_0}+\Lambda_m^{-1}\hat{\mathcal{A}}_{l,n}^{\Lambda,\Lambda_0}\right)\left(z_1;\underline{y}_{\mathfrak{s}};\chi^{\sigma}\right)~.
    \end{multline}
    This implies that
    \begin{multline}\label{413}
\mathcal{N}^{\Lambda_m}_{5-n-r-|w|,l;\infty}\left((z_1-z_i)^r~
\partial^w
\mathcal R_{l,n}^{\Lambda_0,\Lambda_0}
\bigl(
\underline{p}_n
\bigr)\right)\\\le \left(\Lambda_0+m\right)^{-2}\mathcal{P}_{l-1}\left(\log \frac{\Lambda_0+m}{m}\right)\widetilde{\mathcal{P}}\left(\frac{\|\underline{p}_n\|}{\Lambda_m}\right)\qquad
\end{multline}
Differentiating \eqref{408} w.r.t. $\Lambda_0$ we obtain that 
\begin{multline}
    \partial_{\Lambda_0}\partial^w\mathcal{R}_{l,n}^{\Lambda,\Lambda_0}=\mathbf{L}^{\Lambda}\left(\partial_{\Lambda_0}\partial^w\mathcal{L}_{l-1,n+2}^{\Lambda,\Lambda_0}\right)\\-\sum_{\substack{\{\pi_1,\pi_2\}\in\mathcal{P}_n\\l_1+l_2=l}}\left(\mathbf{B}^{\Lambda}_{\pi}\left(\partial_{\Lambda_0}\mathcal{L}_{l_1,n_1+1}^{\Lambda,\Lambda_0},\mathcal{L}_{l_2,n_2+1}^{\Lambda,\Lambda_0}\right)+\mathbf{B}^{\Lambda}_{\pi}\left(\mathcal{L}_{l_1,n_1+1}^{\Lambda,\Lambda_0},\partial_{\Lambda_0}\mathcal{L}_{l_2,n_2+1}^{\Lambda,\Lambda_0}\right)\right)
\end{multline}
We present the argument for $r=0$. The case $r>0$ is obtained by differentiating the weighted flow equation \eqref{FEr} and applying the continuity estimates established for weighted insertions. Since no new ingredient is required, the details are omitted. The induction hypothesis, together with Propositions \ref{PropL} and \ref{propQu} gives the following
\begin{multline}
    \mathcal{N}^{\Lambda_m}_{4-n-|w|,l;\infty}\left(\partial_{\Lambda_0}\partial^{w}\mathcal{R}^{\Lambda,\Lambda_0}_{l,n}\left(\underline{p}_{n}\right)\right)\\\lesssim ~\Lambda_m~\left(\int_k\mathcal{N}^{\Lambda_m}_{3-n-|w|,l-1;\infty}\left(\partial_{\Lambda_0}\partial^{w}\mathcal{L}^{\Lambda,\Lambda_0}_{l-1,n+2}\left(\underline{p}_{n},k,-k\right)\right)\dot{C}^{\Lambda}_m(k)\right)\\+\sum_{\substack{l_1+l_2=l\\n_1+n_2=n}}\mathcal{N}^{\Lambda_m}_{4-n_1-|w_1|,l_1;\infty}\left(\partial_{\Lambda_0}\partial^{w_1}\mathcal{L}^{\Lambda,\Lambda_0}_{l_1,n_1+1}\left(\underline{p}_{\pi_1},p_{\pi}\right)\right)\\\times e^{-\frac{m^2}{2\Lambda^2}}~\dot{C}^{\Lambda}_0\left(p_{\pi}\right)\mathcal{N}^{\Lambda_m}_{3-n_2-|w_2|,l_2;\infty}\left(\partial^{w_2}\mathcal{L}^{\Lambda,\Lambda_0}_{l_2,n_2+1}\left(\underline{p}_{\pi_2},-p_{\pi}\right)\right)
\end{multline}
Applying the induction hypothesis, Theorem~\ref{ThmPrinc}, and the continuity estimates of Propositions~\ref{PropL} and \ref{propQu}, we obtain
\begin{equation}
    \mathcal{N}^{\Lambda_m}_{4-n-|w|,l;\infty}\left(\partial_{\Lambda_0}\partial^{w}\mathcal{R}^{\Lambda,\Lambda_0}_{l,n}\left(\underline{p}_{n}\right)\right)\lesssim \left(\Lambda_0+m\right)^{-2}\mathcal{P}_{l-1}\left(\log\frac{\Lambda_0+m}{m}\right)\widetilde{\mathcal{P}}\left(\frac{\|\underline{p}_n\|}{\Lambda_m}\right)~.
\end{equation}
This implies that 
\begin{multline}\label{premiereEToile}
    \left|\left(K^{\otimes\mathfrak{s}}\star \chi^{\otimes\sigma}\partial_{\Lambda_0}\mathcal{R}_{l,n}^{\Lambda,\Lambda_0}\right)\left(z_1;\underline{y}_{\mathfrak{s}},\underline{p}_n\right)\right|\le~\frac{\left(\Lambda+m\right)^{4-n-|w|}}{\left(\Lambda_0+m\right)^2}\mathcal{P}_{l-1}\left(\log\frac{\Lambda_0+m}{m}\right)\widetilde{\mathcal{P}}\left(\frac{\|\underline{p}_n\|}{\Lambda_m}\right)\\\times \mathcal{Q}^{\sigma}\left(\tau^{\Lambda_m}_{\mathfrak{s}}\right)\left(\mathcal{A}_{l,n}^{\Lambda,\Lambda_0}+\Lambda_m^{-1}\hat{\mathcal{A}}_{l,n}^{\Lambda,\Lambda_0}\right)\left(z_1;\underline{y}_{\mathfrak{s}};\chi^{\sigma}\right)~.
\end{multline}
Since $4-n-|w|\le -1$ we have that 
$$\int_{\Lambda}^{\Lambda_0}\left(\lambda+m\right)^{4-n-|w|}d\lambda\lesssim \left(\Lambda+m\right)^{5-n-|w|}~.$$
This combined with \eqref{premiereEToile} implies that 
\begin{equation}\label{417}
\mathcal{N}^{\Lambda_m}_{5-n-|w|,l;\infty}\left(\int_{\Lambda}^{\Lambda_0}\partial_{\Lambda_0}\partial^{w}\mathcal{R}^{\lambda,\Lambda_0}_{l,n}\left(\underline{p}_{n}\right)d\lambda\right)\lesssim \left(\Lambda_0+m\right)^{-2}\mathcal{P}_{l-1}\left(\log\frac{\Lambda_0+m}{m}\right)\widetilde{\mathcal{P}}\left(\frac{\|\underline{p}_n\|}{\Lambda_m}\right)~.
\end{equation}
Combining \eqref{417} for general $r$, \eqref{413} and \eqref{408} gives the bound \eqref{Premiere} for $n+r+|w|\ge 5$. 
\item \underline{Case 2: $n+r+|w|\le 4$} It remains to treat the relevant components. As in the proof of Theorem~\ref{ThmPrinc}, the power-counting estimates for the full distributions follow from estimates on the corresponding Taylor coefficients together with the bounds \eqref{deuxieme} and \eqref{troisieme}. We therefore analyse the four-point and two-point sectors separately. We have
\begin{multline}
\partial_{\Lambda}\partial_{\Lambda_0}l_{l,4}^{\Lambda,\Lambda_0}\left(z_1;\underline{p}_{4}\right)=\int_{\mathbb{R}^3}\left(\mathbf{L}^{\Lambda}\left(\partial_{\Lambda_0}\mathcal{L}_{l-1,6}^{\Lambda,\Lambda_0}\right)\left(z_1,\vec{z}_{2,4};\underline{p}_4\right)\right.
\\-\sum_{\substack{\pi\in\mathcal{P}_n\\l_1+l_2=l}}~e^{-\frac{m^2}{2\Lambda^2}}\left[\mathbf{B}^{\Lambda}_{\pi}\left(\partial_{\Lambda_0}\mathcal{L}_{l_1,2}^{\Lambda,\Lambda_0},\mathcal{L}_{l_2,4}^{\Lambda,\Lambda_0}\right)+\mathbf{B}^{\Lambda}_{\pi}\left(\mathcal{L}_{l_1,2}^{\Lambda,\Lambda_0},\partial_{\Lambda_0}\mathcal{L}_{l_2,4}^{\Lambda,\Lambda_0}\right)\right]\left(z_1,\vec{z}_{2,4};\underline{p}_4\right)\big)~\dot{C}^{\Lambda}_0(p_{\pi})d\vec{z}_{2,4}~.
\end{multline}
Using the induction hypothesis together with  Theorem \ref{ThmPrinc} we deduce that 
\begin{equation}
    \left|\partial_{\Lambda}\partial_{\Lambda_0}l_{l,4}^{\Lambda,\Lambda_0}\left(z_1;\underline{p}_{4}\right)\right|\le \left(\Lambda_0+m\right)^{-2}\mathcal{P}_{l-1}\left(\log\frac{\Lambda_m}{m}\right)\widetilde{\mathcal{P}}\left(\frac{{\|\underline{p}_4\|}}{\Lambda_m}\right)
\end{equation}
Integrating from $0$ to $\Lambda$ using the renormalization condition $l_{l,4}^{0,\Lambda_0}\left(z_1,0\right)=0$ gives that 
\begin{equation}
    \left|\partial_{\Lambda_0}l_{l,4}^{\Lambda,\Lambda_0}\left(z_1;\vec{0}\right)\right|\le \frac{\Lambda+m}{\left(\Lambda_0+m\right)^{2}}\mathcal{P}_{l-1}\left(\log\frac{\Lambda_m}{m}\right)~.
\end{equation}
This together with the bound \eqref{deuxieme} that was established in the first part of this proof
combined with \eqref{SemNorm4pt} leads to the following bound 
\begin{multline}
\mathcal{N}^{\Lambda_m}_{1,l;\infty}\left(\partial_{\Lambda_0}\mathcal{L}_{l,4}^{\Lambda,\Lambda_0}\left(\underline{p}_4\right)\right)\lesssim~\left\|\partial_{\Lambda_0}l_{l,4}^{\Lambda,\Lambda_0}\left(\vec{0}\right)\right\|_{L^{\infty}(\mathbb{R})}\\+|\underline{p}_4|\mathcal{N}^{\Lambda_m}_{0,l;\infty}\left(\partial_{\underline{p}}\partial_{\Lambda_0}\mathcal{L}_{l,4}^{\Lambda,\Lambda_0}\left(\underline{p}_4\right)\right)+\mathcal{N}^{\Lambda_m}_{1,l;\infty}\left(\Delta \left(\partial_{\Lambda_0}\mathcal{L}_{l,4}^{\Lambda,\Lambda_0}\right)(\vec{0})\right)~
\end{multline}
which leads directly to \eqref{Premiere} in the case $(n,r,w)=(4,0,0)$. 
For $n=2$, we proceed in a similar way to bound inductively $\partial_{\Lambda_0}l_{l,2;r}^{\Lambda,\Lambda_0}(z_1;p,-p)$  using the induction hypothesis and Theorem \ref{ThmPrinc} together with the semi-norm continuity estimates of the Polchinski operators.  Combining this with \eqref{274~-} and the bound \eqref{troisieme} implies the bound \eqref{Premiere} for $(n,r,|w|)\in\left\{(2,0,0),(2,0,2)\right\}$.
  \end{itemize}
\end{proof}
\subsection*{Proof of Theorem \ref{Convergence}}
\begin{proof}
   We prove only the estimate \eqref{337---}, since all remaining bounds
are obtained by the same argument.

The starting point is Proposition~\ref{PropConv}, which provides
uniform bounds on the ultraviolet derivative
\(
\partial_{\Lambda_0}
\mathcal L_{l,n}^{\Lambda,\Lambda_0}
\)
in the power-counting spaces. Applying the same heat-kernel
regularisation argument as in the proofs of
Theorems~\ref{BesovNorm} and \ref{BesovNorm2},
we obtain corresponding Besov estimates for the ultraviolet
derivative of the truncated correlators. Using the bound \eqref{Premiere} and proceeding as in the proof of Theorems \ref{BesovNorm} and \ref{BesovNorm2} after integrating
the $z_1$ variable with the heat kernel, we obtain for all $\chi\in\mathcal{K}$ and $\sigma\subseteq\mathcal{I}_n$ the following bound
    \begin{multline*}
       \left|\left(K^{\otimes n}\star\partial_{\Lambda_0}\chi^{\otimes\sigma}\mathcal{L}_{l,n}^{\Lambda,\Lambda_0}\right)\left(\underline{p}_n;\tau,\vec{y}_{1,n}\right)\right|\\
       \le \frac{\Lambda_m^{5-n}}{\left(\Lambda_0+m\right)^2}~\mathcal{P}_{l-1}\left(\log\frac{\Lambda_0+m}{m}\right)\widetilde{\mathcal{P}}\left(\frac{\|\underline{p}_n\|}{\Lambda_m}\right)\left(\sum_{i=0}^{|\sigma|\left(n-1\right)+1}\tau^{-\frac{n-1+i}{2}}m^{-i}~\right)~.
    \end{multline*}
   We now take \(\chi=\chi_\varepsilon\), where \((\chi_\varepsilon)_\varepsilon\) is an admissible approximation of the half-space indicator. Integrating the previous estimate with respect to the ultraviolet
cutoff parameter over the interval
\(
[\Lambda_0',\Lambda_0]
\),
and using \eqref{irrePol},
we obtain
     \begin{multline*}
       \left|K^{\otimes n}\star\chi_{\epsilon}^{\otimes\sigma}\left(\mathcal{L}_{l,n}^{\Lambda,\Lambda_0}-\mathcal{L}_{l,n}^{\Lambda,\Lambda'_0}\right)\left(\underline{p}_n;\tau,\vec{y}_{1,n}\right)\right|\\
       \le \frac{\Lambda_m^{5-n}}{\left(\Lambda_0+m\right)^2}~\mathcal{P}_{l-1}\left(\log\frac{\Lambda_0+m}{m}\right)\widetilde{\mathcal{P}}\left(\frac{\|\underline{p}_n\|}{\Lambda_m}\right)\sum_{i=0}^{|\sigma|\left(n-1\right)+1}\tau^{-\frac{n-1+i}{2}}m^{-i}~.
    \end{multline*}
     Setting $\Lambda=0$ and using the elementary estimate 
\begin{equation*}
\mathcal{P}_{l-1}\left(\log\frac{\Lambda_0+m}{m}\right)\lesssim~\left(\log\frac{\Lambda_0+m}{m}\right)^{\nu}
\end{equation*}
where $\nu$ is the polynomial degree of $\mathcal{P}_{l-1}$ we obtain
\begin{multline}\label{364--}
    \left|K^{\otimes n}\star\chi_{\epsilon}^{\otimes\sigma}\left(\mathcal{L}_{l,n}^{0,\Lambda_0}-\mathcal{L}_{l,n}^{0,\Lambda'_0}\right)\left(\underline{p}_n;\tau,\vec{y}_{1,n}\right)\right|\\
       \le \frac{m^{5-n}}{\Lambda_0+m}~\left(\log\frac{\Lambda_0+m}{m}\right)^{\nu}\widetilde{\mathcal{P}}\left(\frac{\|\underline{p}_n\|}{\Lambda_m}\right)\left(\sum_{i=0}^{|\sigma|\left(n-1\right)+1}\tau^{-\frac{n-1+i}{2}}m^{-i}~\right)~.
    \end{multline}
    Let $\alpha := -\bigl((|\sigma|+1)(n-1)+1\bigr) = -|\sigma|n+|\sigma|-n$.
    Recalling the definition of the Besov seminorm, the previous estimate
implies
    \begin{multline}
       \left\|\left(\chi_{\epsilon}^{\otimes\sigma}\mathcal{L}_{l,n}^{0,\Lambda_0}-\chi_{\epsilon}^{\otimes\sigma}\mathcal{L}_{l,n}^{0,\Lambda'_0}\right)\left(\underline{p}_n\right)\right\|_{\mathcal{C}^{\alpha}\left(\mathbb{R}^n\right)}\\
       \le \left(\sum_{i=0}^{|\sigma|\left(n-1\right)+1}m^{5-n-i}\right)~\left(\Lambda_0+m\right)^{-1}~\left(\log\frac{\Lambda_0+m}{m}\right)^{\nu}\widetilde{\mathcal{P}}\left(\frac{\|\underline{p}_n\|}{\Lambda_m}\right)~. 
    \end{multline}
For every compact set \(K\Subset\mathbb R^n\), the definition of the
truncated correlators and the uniform Besov bound in
\(\varepsilon\) imply the lower-semicontinuity estimate
\[
\left\|
\bigl(\mathds 1_+^\sigma
\mathcal L_{l,n}^{0,\Lambda_0}
-\mathds 1_+^\sigma
\mathcal L_{l,n}^{0,\Lambda_0'}
\bigr)
(\underline p_n)
\right\|_{\mathcal C^\alpha(K)}
\le
C_K
\liminf_{\varepsilon\to0}
\left\|
\chi_\varepsilon^{\otimes\sigma}
\bigl(
\mathcal L_{l,n}^{0,\Lambda_0}
-
\mathcal L_{l,n}^{0,\Lambda_0'}
\bigr)
(\underline p_n)
\right\|_{\mathcal C^\alpha(\mathbb R^n)} .
\]
Using the estimate obtained above for the right-hand side, we deduce
\[
\left\|
\bigl(\mathds 1_+^\sigma
\mathcal L_{l,n}^{0,\Lambda_0}
-\mathds 1_+^\sigma
\mathcal L_{l,n}^{0,\Lambda_0'}
\bigr)
(\underline p_n)
\right\|_{\mathcal C^\alpha(K)}
\le
C_K
\frac{m^{5-n}}{\Lambda_0+m}
\left(
\log\frac{\Lambda_0+m}{m}
\right)^\nu
\widetilde{\mathcal P}
\left(
\frac{\|\underline p_n\|}{m}
\right).
\]
    Since the right-hand side converges to zero as
\(
\Lambda_0\to\infty
\),
uniformly in
\(
\Lambda_0'\ge\Lambda_0
\),
we conclude that, for every compact set
\(
K\Subset\mathbb R^n
\),
the family
$
\Bigl(
\mathds1_+^\sigma
\mathcal L_{l,n}^{0,\Lambda_0}
\Bigr)_{\Lambda_0\ge1}
$
is Cauchy in
\(
\mathcal C^\alpha(K)
\).
Since
\(
\mathcal C^\alpha(K)
\)
is complete, there exists $
\mathds1_+^\sigma
\mathcal L_{l,n}
\in
\mathcal C^\alpha(K)$
such that
\[
\mathds1_+^\sigma
\mathcal L_{l,n}^{0,\Lambda_0}
\longrightarrow
\mathds1_+^\sigma
\mathcal L_{l,n}
\qquad
\text{in }
\mathcal C^\alpha(K).
\]
Because \(K\Subset\mathbb R^n\) is arbitrary, the convergence holds in
\(
\mathcal C^\alpha_{\mathrm{loc}}(\mathbb R^n)
\).
\end{proof}

\newpage
\appendix

\section{Elementary estimates}

We collect here a number of elementary integral bounds that are repeatedly used in the derivation of inductive estimates for the $n$-point correlation functions.

\begin{itemize}

\item[i)] \emph{Irrelevant contributions.}  
The integrals to be estimated are of the form
\begin{equation}
\int_a^b dx\, x^{-q-1}(\log x)^s,
\qquad
1 \le a \le b,\quad q \in \mathbb{N},\quad s \in \mathbb{N}_0 .
\end{equation}
Define the polynomial
\begin{equation}
P_{q,s}(\log x)
:= \frac{1}{q}
\left(
(\log x)^s
+ \frac{s}{q}(\log x)^{s-1}
+ \frac{s(s-1)}{q^2}(\log x)^{s-2}
+ \cdots
+ \frac{s!}{q^s}
\right),
\end{equation}
and set $f_{q,s}(x) := x^{-q} P_{q,s}(\log x)$.  
A direct computation shows that
\[
f'_{q,s}(x) = -x^{-q-1}(\log x)^s < 0,
\qquad
f_{q,s}(x) > 0 \quad \text{for } x>1 .
\]
Consequently,
\[
\int_a^b dx\, x^{-q-1}(\log x)^s
= f_{q,s}(a) - f_{q,s}(b)
\le f_{q,s}(a).
\]

In the sequel, this estimate is applied in the form
\begin{equation}\label{irrePol}
\int_{\Lambda}^{\Lambda_0} d\lambda\, (\lambda+m)^{-q-1}
\left(\log\frac{\lambda+m}{m}\right)^{q'}
\;\lesssim\;
\Lambda_m^{-q}\,
P_{q,q'}\!\left(\log\frac{\Lambda_m}{m}\right),
\end{equation}
where $\Lambda_m := \Lambda + m$.  
Note that in the irrelevant case one has $q>1$.

\item[ii)] \emph{Relevant contributions.}  
Here the integrals take the form
\[
\int_1^b dx\, x^{q-1}(\log x)^s,
\qquad
b \ge 1,\quad q,s \in \mathbb{N}_0 .
\]

For $q=0$ the integral can be evaluated explicitly.  
For $q>0$, define
\[
g_{q,s}(x)
:= \frac{1}{q} x^q
\left(
(\log x)^s
- \frac{s}{q}(\log x)^{s-1}
+ \frac{s(s-1)}{q^2}(\log x)^{s-2}
+ \cdots
+ (-1)^s \frac{s!}{q^s}
\right).
\]
Then
\[
g'_{q,s}(x) = x^{q-1}(\log x)^s,
\]
and therefore
\[
\int_1^b dx\, x^{q-1}(\log x)^s
= g_{q,s}(b) - g_{q,s}(1)
\le \frac{s!}{q^s} + |g_{q,s}(b)|.
\]

In the applications below, this estimate is used in the form
\begin{align}\label{BornePolyRe}
\int_0^{\Lambda} d\lambda\, (\lambda+m)^{q-1}
\left(\log\frac{\lambda+m}{m}\right)^s
&= m^q \int_1^{\frac{\Lambda+m}{m}} d\lambda\,
\lambda^{q-1}(\log\lambda)^s\\
&\le \Lambda_m^q
\left(
\frac{s!}{q^s}
+ \left| g_{q,s}\!\left(\frac{\Lambda_m}{m}\right) \right|
\right)\nonumber\\
&\leq \Lambda_m^q~\mathcal{P}_s\!\left(\log\frac{\Lambda_m}{m}\right)~,\label{IrrePoly}
\end{align}
where $\mathcal{P}_s$ denotes a polynomial of degree $s$ with positive coefficients depending on $q$ and $s$.
We suppress the dependence on $q$, since no control over these coefficients is required in the sequel, provided they remain independent of the mass as well as of the flow and ultraviolet scales.
Later, the parameter $s$ will be fixed to $s=l-1$, where $l$ denotes the loop order, and hence tracks the degree of the corresponding polynomial.

\end{itemize}
\section*{Acknowledgements}

Parts of this work were completed while the author was supported by the Max Planck Institute for Mathematics in the Sciences. The hospitality and support of the Institute are gratefully acknowledged. The author would particularly like to thank Stefan Hollands, Leonard Ferdinand, Robert Schlesier, Nikolay Barashkov and Rhys Steele for many helpful discussions.
\newpage
\section*{References}
\bibliographystyle{abbrv}
\bibliography{aipsamp}

\end{document}